%% file: mainMap.tex
\documentclass[twoside,11pt]{article}

\usepackage{jmlr2e}
\pdfoutput=1
\usepackage[dvipsnames]{xcolor} 
\RequirePackage[OT1]{fontenc}
\RequirePackage{amsthm,amsmath} %
\RequirePackage[numbers]{natbib}
\usepackage[colorlinks=true, citecolor=blue, filecolor=black,  urlcolor=blue, hypertexnames=false]{hyperref}
\usepackage[normalem]{ulem}
\usepackage[ansinew]{inputenc}
\usepackage[letterpaper,centering]{geometry}
\usepackage{bm}
\usepackage{mathtools}
\usepackage{rotating}
\usepackage{yfonts} %
\usepackage[textsize=tiny]{todonotes}
\usepackage{float}
\usepackage[caption=false]{subfig} %
\usepackage{algorithm}
\usepackage{algorithmicx}
\usepackage{caption}
\usepackage{algpseudocode}
\usepackage{varwidth}
\usepackage{dsfont}
\usepackage{multirow}
\usepackage{enumerate}
\usepackage{enumitem}
\usepackage{mathrsfs}
\usepackage{grffile}
\usepackage{tikz}
\usetikzlibrary{shapes,arrows,decorations.pathmorphing,backgrounds,positioning,fit,petri}
\usetikzlibrary{calc}
\usetikzlibrary{matrix}
\usetikzlibrary{backgrounds}
\usepackage{pgfplots}
\pgfplotsset{compat=newest}
\usepgfplotslibrary{groupplots}
\include{stylesTikzGraphs}
\include{optionsDecompositionGraphs}

\pgfdeclarelayer{background}
\pgfdeclarelayer{bBack}
\pgfsetlayers{bBack,background,main}
\usetikzlibrary{fit}
\numberwithin{equation}{section}
\theoremstyle{plain}
\newtheorem{thm}{Theorem}[section]
\newtheorem{lem}{Lemma}[section]
\newtheorem{cor}{Corollary}[section]
\theoremstyle{definition}
\newtheorem{definition}{Definition}[section]
\input{macroMap.tex} %

\newcommand{\hrevone}{\textcolor{black} }

\ShortHeadings{Inference via low-dimensional couplings}{Spantini, Bigoni, Marzouk}
\firstpageno{1}{}

\begin{document}

\title{Inference via low-dimensional couplings}

\author{\name Alessio Spantini \email spantini@mit.edu \\
        \name Daniele Bigoni \email dabi@mit.edu \\
        \name Youssef Marzouk \email ymarz@mit.edu \\
       \addr %
       Massachusetts Institute of Technology\\
       Cambridge, MA 02139 USA}

\editor{}

\maketitle

\begin{abstract}%
\input{sec_abstract}
\end{abstract}


\tableofcontents %
\section{Introduction}
\label{sec:intro}
\input{sec_intro}

\section{Notation}
\input{sec_notation}
\section{Triangular transport maps: a building block}
\label{sec:compTransport}

\input{sec_compTransport}
\section{Markov networks}
\input{sec_Markov}
\label{sec:markov}

\section{Sparsity of triangular transport maps}
\label{sec:sparse}
\input{sec_sparse}

\subsection{Connection to Gaussian Markov random fields}
\label{sec_gmrf}

\input{sec_gmrf}

\subsection{Ordering of triangular maps}
\label{sec_order}
\input{sec_order}

\section{Decomposability of transport maps}
\label{sec:decomp}

\input{sec_decomp}

\section{Sequential inference on state-space models: variational algorithms}
\label{sec:dataAss}
\input{sec_dataAss}

\section{Numerical illustration}
\label{sec:numerics}
\input{sec_numerics}

\section{Discussion} %
\label{sec:discus}
\input{sec_discuss}

\acks{We would like to thank Ricardo Baptista, Chi Feng,  Jeremy Heng, Pierre Jacob, 
Qiang Liu, Rebecca Morrison, Zheng Wang, Alan Willsky, Olivier Zahm, and Benjamin Zhang 
for many insightful discussions and for pointing us to key references
in the literature. This work was supported in part by the US
Department of Energy, Office of Advanced Scientific Computing (ASCR),
under grant numbers DE-SC0003908 and DE-SC0009297.}

\appendix

\section{Generalized Knothe-Rosenblatt rearrangement}
\label{sec:genKR}
\input{sec_genKR}

\section{Proofs of the main results}
\label{sec:proofs}

\input{sec_proofs}
\section{Algorithms for inference on state-space models}
\label{sec:algo}
\input{sec_algs}

\section{Additional results for the stochastic volatility model}
\label{sec:add_res}

\input{sec_additional_results}

\vskip 0.2in
\bibliographystyle{plainnat}
\bibliography{mapBib}

\end{document}

%% file: stylesTikzGraphs.tex
\tikzset{
	default/.style={circle,draw=blue!50,fill=blue!20,thick,
		inner sep=0pt,minimum size=6mm},
	defaultOrange/.style={circle,draw=orange!50,fill=orange!20,thick,
		inner sep=0pt,minimum size=6mm},		
	defaultGreen/.style={circle,draw=ForestGreen!50,fill=ForestGreen!20,thick,
		inner sep=0pt,minimum size=6mm},
	defaultGray/.style={circle,draw=black,fill=Gray!10,thick,
		inner sep=0pt,minimum size=6mm},			
	defaultCyan/.style={circle,draw=cyan!50,fill=cyan!20,thick,
		inner sep=0pt,minimum size=6mm},				
	defaultBlank/.style={circle,draw=white,fill=white,thick,
		inner sep=0pt,minimum size=6mm},	
	paramSetNodes/.style={circle,draw=blue!50,fill=blue!20,thick,
		inner sep=0pt,minimum size=11mm},	
	dataNode/.style={circle,draw=red!50,fill=red!20,thick,
		inner sep=0pt,minimum size=6mm},		
	dataSetNodes/.style={circle,draw=red!50,fill=red!20,thick,
		inner sep=0pt,minimum size=11mm},	
	default2/.style={circle,draw=black,fill=blue!20,thick,
		inner sep=0pt,minimum size=6mm}, 
	edgeStyle/.style={ semithick},
	entriesMatrixZero/.style={rectangle,draw=white,fill=white, 
		inner sep=0pt,minimum size=3mm}, 
	entriesMatrixZeroDashed/.style={rectangle,draw=black,fill=white, dashed,
		inner sep=0pt,minimum size=3mm},
	entriesMatrix/.style={rectangle,draw=black,fill=blue!60!white, 
		inner sep=0pt,minimum size=3mm},
	entriesMatrixRed/.style={rectangle,draw=black,fill=red!60!white, 
		inner sep=0pt,minimum size=3mm},	
	entriesMatrixOrange/.style={rectangle,draw=black,fill=orange!60!white, 
		inner sep=0pt,minimum size=3mm},				
        shadedNode/.style={circle,very thick,draw=orange!50,fill=orange!20,inner sep=0pt,minimum size=6mm},	background/.style={ draw, fill=yellow!30, align=right}												
}

%% file: optionsDecompositionGraphs.tex
\newcommand{\colorA}{red!40}  
\newcommand{\colorB}{yellow!50} 
\newcommand{\colorS}{orange!40}  

\newcommand{\valueInnerSep}{1em}

%% file: macroMap.tex
\newcommand{\Ex}{\mathbb{E}}                                                   %
\newcommand{\Var}[1]{\underset{#1}{\mathbb{V}\mathrm{ar}}}   %
\newcommand{\Gauss}{\mathcal{N}}                                      %

\newcommand{\push}{_\sharp}                                      %
\newcommand{\orth}{ \perp\!\!\!\perp }  %

\newcommand{\borelm}{\mathscr{M}}
\newcommand{\borelmp}{\borelm_+}
\newcommand{\lebm}{\boldsymbol{\lambda} } %
\newcommand{\genm}{\boldsymbol{\nu} }	  %
\newcommand{\agenm}{\hat{\boldsymbol{\nu}}}	  %
\newcommand{\agengr}{\hat{\Gcb}}	  %

\newcommand{\re}{\mathbb{R}}
\newcommand{\ra}{\rightarrow}

\def\argmin{\mathop{\mathrm{argmin}}}
\newcommand{\Bc}{\mathcal{B}}
\newcommand{\Cc}{\mathcal{C}}

\newcommand{\Vc}{\mathcal{V}}
\newcommand{\Fc}{\mathcal{F}}

\newcommand{\Ec}{\mathcal{E}}

\newcommand{\Ac}{\mathcal{A}}
\newcommand{\Sc}{\mathcal{S}}

\newcommand{\Xb}{\boldsymbol{X}}
\newcommand{\xb}{\boldsymbol{x}}
\newcommand{\Yb}{\boldsymbol{Y}}
\newcommand{\yb}{\boldsymbol{y}}
\newcommand{\Pb}{\boldsymbol{P}}
\newcommand{\Zb}{\boldsymbol{Z}}
\newcommand{\zb}{\boldsymbol{z}}
\newcommand{\eb}{\boldsymbol{e}}
\newcommand{\Fb}{\boldsymbol{F}}
\newcommand{\Hb}{\boldsymbol{H}}
\newcommand{\hb}{\boldsymbol{h}}
\newcommand{\Qb}{\boldsymbol{Q}}
\newcommand{\Rb}{\boldsymbol{R}}
\newcommand{\Gammab}{\boldsymbol{\Gamma}}
\newcommand{\Ab}{\boldsymbol{A}}
\newcommand{\ab}{\boldsymbol{a}}
\newcommand{\Bb}{\boldsymbol{B}}
\newcommand{\Cb}{\boldsymbol{C}}
\newcommand{\cb}{\boldsymbol{c}}
\newcommand{\Jb}{\boldsymbol{J}}

\newcommand{\Sigmab}{\boldsymbol{\Sigma}}
\newcommand{\mub}{\boldsymbol{\mu}}
\newcommand{\phib}{\boldsymbol{\phi}}

\newcommand{\thetab}{\boldsymbol{\theta}}                              %

\newcommand{\Gcb}{\boldsymbol{\mathcal{G}}}
\newcommand{\Ccb}{\boldsymbol{\mathcal{C}}}
\newcommand{\Kcb}{\boldsymbol{\mathcal{K}}}

\newcommand{\Pbb}{\mathbb{P}}
\newcommand{\Nbb}{\mathbb{N}}

\newcommand{\cfrak}{\mathfrak{c}}

\newcommand{\sparse}{ \mathfrak{I} }

\newcommand{\neigh}{ {\rm Nb} }
\newcommand{\fillin}{\textswab{F} } 
\newcommand{\permset}{\mathfrak{P} } 

\newcommand{\vhyp}{\boldsymbol{\Theta}}
\newcommand{\vhyps}{\theta}

\newcommand{\dhyp}{p}
\newcommand{\ddata}{d} %

\newcommand{\submap}{\mathfrak{M}}

\newcommand{\lmap}{L} %
\newcommand{\rmap}{R} %
\newcommand{\decset}{\mathfrak{D}} %

\newcommand{\forop}{ \Fb }
\newcommand{\noised}{ \boldsymbol{\varepsilon} }
\newcommand{\noisedcov}{ \Qb }
\newcommand{\noiseo}{ \boldsymbol{\xi} }
\newcommand{\noiseocov}{ \Rb }
\newcommand{\obsop}{ \Hb }

\newcommand{\Aset}{ \Ac }
\newcommand{\Bset}{ \Bc }
\newcommand{\Sset}{ \Sc }
\newcommand{\Aind}{ \Aset^\perp }

\newcommand{\emap}{T}                          %
\newcommand{\Dkl}{\mathcal{D}_{\rm KL}}         %
\newcommand{\spaceMap}{\mathcal{T}}             %

\newcommand{\spaceMapT}{\mathcal{T}_\triangle}  %

\newcommand{\pull}{^\sharp}  

%% file: sec_abstract.tex
We investigate the low-dimensional structure of deterministic
transformations between random variables, i.e., transport maps between
probability measures. 
In the context of statistics and machine learning, these
transformations can be used to couple a tractable ``reference''
measure (e.g., a standard Gaussian) with a target measure of
interest.
Direct simulation from the desired measure can then be achieved by %
pushing forward reference samples through the map. %
Yet characterizing such a map---e.g., representing and evaluating it---grows challenging in high dimensions.
The central contribution of this paper is to establish a link between
the Markov properties of the target measure %
and the existence of low-dimensional couplings, induced by transport maps
that are \textit{sparse} and/or \textit{decomposable}.  Our analysis
not only facilitates the construction of transformations in
high-dimensional settings, but also suggests new inference
methodologies for continuous non-Gaussian graphical models.  For
instance, in the context of nonlinear state-space models, we describe
new variational algorithms for filtering, smoothing, and sequential
parameter inference.
These algorithms can be understood as the natural
generalization---to the non-Gaussian case---of the square-root
Rauch--Tung--Striebel Gaussian smoother.

%% file: sec_intro.tex
This paper studies the low-dimensional structure of transformations
between random variables. Such transformations, which can be
understood as transport maps between probability measures \cite{villani2008optimal}, are ubiquitous in statistics and machine learning: %
they can be used for posterior sampling \cite{el2012bayesian}, possibly via deep neural networks \cite{rezende2015variational}; for accelerating Markov chain Monte Carlo or importance sampling algorithms \cite{parno2015transport,han2017stein}; or as the building blocks of implicit generative models \cite{kingma2013auto,goodfellow2014generative,mohamed2016learning} 
and flexible methods for density estimation \cite{tabak2013family,dinh2016density}.

In the context of variational inference \cite{blei2016variational}, a transport map can be used to
define a deterministic coupling between a tractable reference measure $\genm_\eta$ that we can easily simulate (e.g., a standard Gaussian) and an arbitrary target measure $\genm_\pi$ that we wish to
characterize (e.g., a posterior distribution).
Given i.i.d.\ samples $( \Xb_i )$ 
from the reference measure, 
we can evaluate the transport map to obtain i.i.d.\ samples 
$( T(\Xb_i) )$ from the target. 
In other words,
the map
allows any expectation
$\int g \, \rm{d}\genm_\pi$ over the target measure to be rewritten as
an integral over the reference measure,
\begin{equation*}
 \int g(\xb)\,{\rm d}\genm_\pi(\xb) = \int g(T(\xb))\,{\rm d}\genm_\eta(\xb) \,,
\end{equation*}
thus
enabling the use of 
standard integration techniques for the tractable $\genm_\eta$, including Monte Carlo
sampling and deterministic quadrature 
\cite{meng2002warp,wang2016warp,el2012bayesian,dick2016higher}. 

We focus on absolutely continuous measures 
$(\genm_\eta, \genm_\pi)$ on $\re^n$,
for which the existence of %
a transport map $T:\re^n \ra \re^n$ is guaranteed 
\cite{santambrogio2015optimal}.
Such a map, however, is seldom unique.
 Identifying a particular 
 map
 requires imposing additional structure on the problem. Optimal transport maps, for instance, 
 define %
 couplings that 
 minimize a particular integrated \textit{transport cost} expressing the effort required to rearrange samples \cite{villani2008optimal}.
 The analysis of such maps underpins a vast field that links geometry
 and partial differential equations, with applications in fluid
 dynamics, economics, 
 statistics
 \cite{douglas1999applications,kantorovich1965best},
 and beyond. In recent years, several other couplings have been proposed for
 use in statistical problems, e.g., parametric approximations
 \cite{el2012bayesian} of the Knothe--Rosenblatt rearrangement
 \cite{rosenblatt1952remarks,knothe1957contributions}, couplings
 induced by the flows of ODEs %
 \cite{heng2015gibbs,anderes2012general,daum2008particle}, and
 couplings induced by the composition of many simple maps, including deep
 neural networks %
 \cite{liu2016stein,tabak2013family,rezende2015variational}.
Yet the construction, representation, and evaluation of all these 
maps
grows challenging in high dimensions. In the setting considered here, a 
transport map
is a function 
from $\mathbb{R}^n$ onto itself; without specifying further structure, representing
such a map or even realizing its action 
is often
intractable as $n$ increases. 

The central contribution of this paper is to establish a link between
the conditional independence structure of the reference-target pair---the
so-called Markov properties \cite{lauritzen1996graphical} of $\genm_\eta$ and
$\genm_\pi$---and the existence of low-dimensional couplings.
These couplings are induced by transport maps that are \textit{sparse}
and/or \textit{decomposable}.
A sparse map %
consists of 
scalar-valued component functions that each depend only on a few input variables, whereas
a decomposable map
factorizes as the
{\it exact} composition of finitely many functions of low effective dimension
(i.e., $T = T_1 \circ \cdots \circ T_\ell$, where each $T_i$ differs from the
identity map only along a subset of its components).
These properties, and their combinations, dramatically reduce the complexity of 
representing a transport map and can be deduced \textit{before} the
map 
is explicitly computed.

The utility of these results is twofold.
First, they
make the construction of couplings---and hence the characterization of complex probability distributions---tractable for a large class of inference problems. 
In particular, these results can be exploited in
state-of-the-art approaches for the numerical computation of
transport maps, including normalizing flows \cite{rezende2015variational} or Stein variational algorithms
\cite{anderes2012general,liu2016stein,detommaso2018stein}.
Second, %
these results suggest new algorithmic approaches %
for important classes of statistical models. %
For instance, our analysis of sparse triangular maps provides a
general framework for describing continuous and non-Gaussian Markov
random fields, and for exploiting the conditional independence
structure of these fields in computation.  Our analysis of
decomposable transport maps yields new variational algorithms for
sequential inference in nonlinear and non-Gaussian state space
models. These algorithms characterize the full Bayesian solution to
the smoothing and joint state--parameter
inference problems by means of a decomposable
transport map, which is constructed (recursively) in a
{\it single} forward pass using local operations.
These algorithms can be
understood as the natural generalization, to the non-Gaussian case, of
the square-root Rauch-Tung-Striebel Gaussian smoother.  Moreover, the
results presented in this paper underpin %
recent
efforts in structure learning for 
non-Gaussian
graphical models \cite{morrison2017beyond}, and novel approaches to the
filtering of high-dimensional spatiotemporal processes \cite[Chapter
6]{spantini2017inference}.  Overall, we propose a range of techniques
to address  problems of  inference in continuous
non-Gaussian graphical models.

The paper is organized as follows.  
Section \ref{sec:notation} introduces some notation used throughout the paper.
Section \ref{sec:compTransport} reviews the %
Knothe-Rosenblatt rearrangement, a key coupling for our analysis,
while Section \ref{sec:markov} briefly recalls some standard
terminology for Markov random fields and graphical models.
The main results are in Sections \ref{sec:sparse}--\ref{sec:dataAss}:
Section \ref{sec:sparse} addresses the sparsity of triangular
transports, while Section \ref{sec:decomp} introduces and develops the
concept of decomposable transport maps for general Markov networks.
These two sections can be read independently.
Section \ref{sec:dataAss} specializes the theory of Section
\ref{sec:decomp} to state-space models, 
introducing new variational %
algorithms for filtering, smoothing, and parameter inference.
Section \ref{sec:numerics} illustrates aspects of the theory with numerical examples. A final discussion is presented in Section
\ref{sec:discus}.
\hrevone{ 
Appendix \ref{sec:genKR} collects some technical details on the Knothe-Rosenblatt rearrangement and its generalizations. Appendix \ref{sec:proofs} contains the proofs of the main results.
Appendix \ref{sec:algo} provides pseudocode for our variational algorithms applied to state-space models, and additional numerical experiments 
are described in Appendix \ref{sec:add_res}.
Code and all numerical examples are available online.\footnote{\url{http://transportmaps.mit.edu}}
}

%% file: sec_notation.tex
\label{sec:notation} 
Here, we collect some
useful notation used throughout the paper.
\paragraph{Notation for functions, sets, and graphs}

For a pair of functions $f$ and $g$, we denote their composition by
$f \circ g$.
We denote by $\partial_k f$ the partial derivative of $f$ with
respect to its $k$th input variable.
By $\partial_k f = 0$, we mean
that the function $f$ does not depend on its $k$th input variable.
Depending on the context,
we can  identify a matrix $Q$ with its corresponding linear map, given by
$\xb \mapsto Q \xb$. 

For all $n>0$, we let $\mathbb{N}_n=\{1,\ldots,n\}$ denote the set
of the first $n$ integers.
For any pair of sets, $\Ac \subset \Bc$ means that $\Ac$ 
is a subset of $\Bc$ (including the possibility of $\Ac = \Bc$).
We denote by $\vert \Ac \vert$ the cardinality of $\Ac$.

Given a graph $\Gcb=(\Vc, \Ec)$ with vertices
$\Vc$ and edges $\Ec$, we denote by $\neigh (k,\Gcb)$ the
neighborhood of a node $k$ in $\Gcb$, while for 
any set $\Ac \subset \Vc$, we denote by
$\Gcb_{\Ac}=(\Vc', \Ec')$ the subgraph  given by
$\Vc' = \Ac$ and $\Ec' = \Ec \cap (\Ac \times \Ac)$.
\paragraph{Notation for measures and densities}

In this paper, we mostly consider  probability measures on
$\re^n$
that are absolutely
continuous with respect to the Lebesgue measure, $\lebm$, and that
are fully supported.
We denote the set of such measures by 
$\borelmp(\re^n)$.
The {\it density} of a measure will  always be intended 
with respect to $\lebm$.
For a pair of measures $\genm_1,\genm_2$, 
$\genm_1 \ll \genm_2$ means that $\genm_1$ is
absolutely continuous with respect to $\genm_2$.

For any measure $\genm$ and measurable map $T$, we denote by
$T\push \genm$ 
the pushforward measure given by $\genm \circ T^{-1}$, where for
any set $\Bset$, 
$T^{-1}(\Bset)$ is the set-valued preimage of $\Bset$ under $T$.
Similarly, we denote by
$T\pull \genm$ 
the pullback measure given by $\genm \circ T$.
Given %
a measure $\genm$
with
density $\pi$ and a %
map $T$, 
we denote
by $T\push \pi$ the density of $T\push \genm$, provided it exists (depending on $T$).
We call $T\push \pi$ the \emph{pushforward density} of $\pi$ by $T$.
Similarly, we define the pullback density 
$T\pull \pi$ as the density of $T\pull \genm$, provided it exists.
Whether the map $T$ preserves the absolute continuity of the measure 
depends
on the regularity of $T$.
For instance, if $T:\re^n \ra \re^n$ is a diffeomorphism---i.e., 
a differentiable bijection with differentiable inverse---then
one has: %
\begin{equation} \label{eq:smoothPushPull}
	T\push \pi(\xb) = \pi(T^{-1}(\xb))\,|\det \nabla T^{-1}(\xb)|,
	\qquad
	T\pull \pi(\xb) = \pi(T(\xb))\,|\det \nabla T(\xb)|,
\end{equation}
where $\nabla T(\xb)$ denotes the Jacobian of $T$ at $\xb$.
The regularity assumptions on $T$ can be substantially weakened as long as
one modifies
\eqref{eq:smoothPushPull} appropriately 
\cite{rudin1987real,spivak1965calculus,fremlin2000measure}.
We will give one such example shortly when dealing
with triangular maps 
(see Section \ref{sec:compTransport} or Appendix \ref{sec:genKR}).
We denote by  $\int f(\xb)\, \genm({\rm d}\xb)$ the
integration of a measurable function $f:\re^n \ra \re$ with
respect to a measure $\genm$.
For  
the Lebesgue measure, we simplify our notation as
$\int f(\xb)\, \lebm({\rm d}\xb) = \int f(\xb) \,{\rm d}\xb$. 
Given a pair $\eta,\pi$ of probability densities 
and a 
map $T:\re^n \ra \re^n$,
we say that $T$ \emph{pushes forward} 
$\eta$ to $\pi$ if and only if
$T$ couples the corresponding probability measures, i.e.,
$T\push \genm_{\eta} = \genm_{\pi}$, with
$\genm_{\eta}(\Bset) = \int_{\Bset} \eta(\xb)\,{\rm d}\xb$
and
$\genm_{\pi}(\Bset) = \int_{\Bset} \pi(\xb)\,{\rm d}\xb$ for all
measurable sets $\Bset$. 
(Notice that $T\push \eta$ need not be given  by \eqref{eq:smoothPushPull} since we are
not specifying any regularity on $T$.)

When it is clear from context, we will freely omit the qualifier 
a.e.\ to indicate a property that
holds up to a set of measure zero.
\paragraph{Notation for random variables}

We use boldface capital letters, e.g., $\Xb$, to 
denote random variables on
$\re^n$
with $n>1$, while we write 
scalar-valued
random variables as $X$.
The law of a random variable $\Xb$ defined on a probability space
$(\Omega,\Pbb)$ is given by %
$\Xb \push\Pbb$. 
For a measure $\genm$,
$\Xb \sim \genm$ means that
$\Xb$ has law $\genm$.
If $\Xb=(\Xb_1,\ldots,\Xb_p)$ 
is a collection of 
random variables and $\Ac \subset \mathbb{N}_p$, then 
$\Xb_{\Ac}=(\Xb_i, i\in\Ac)$ denotes a subcollection of $\Xb$.
In the same way, for $j < k$, $\Xb_{j:k}=(\Xb_j,\Xb_{j+1},\ldots, \Xb_k)$.
If $\Xb=(\Xb_1,\ldots,\Xb_p)$ has joint
density $\pi$ and $\Ac \subset \mathbb{N}_p$,  
we denote by $\pi_{\Xb_{\Ac}}$ the marginal of $\pi$ along $\Xb_{\Ac}$, i.e.,
$\pi_{\Xb_{\Ac}}(\xb_{\Ac}) = \int \pi(\xb)\,{\rm d}\xb_{\mathbb{N}_p \setminus \Ac}$.
If %
$\pi$  %
is the density of 
$\Zb=(\Xb, \Yb)$, we denote by 
$\pi_{\Xb \vert \Yb}$
the density
of $\Xb$ given $\Yb$, where
\begin{equation} \label{eq:cond_density}
	\pi_{\Xb \vert \Yb}(\xb \vert \yb) = 
	\begin{cases} 
	\pi_{\Xb,\Yb}(\xb, \yb)/\pi_{\Yb}(\yb) &\mbox{if }  \pi_{\Yb}(\yb) \neq 0  \\ 
	0 & \mbox{otherwise}.
	\end{cases}
\end{equation}
We denote independence of a pair of random variables $\Xb,\Yb$ by 
	$\Xb \orth \Yb$. In the same way, $\Xb \orth \Yb \vert \Rb$ means
	that $\Xb$ and $\Yb$ are independent given a third random variable $\Rb$.

%% file: sec_compTransport.tex
An important transport for our analysis is the  Knothe-Rosenblatt (KR)
rearrangement on $\re^n$ 
\cite{rosenblatt1952remarks,knothe1957contributions,bogachev2005triangular}. 
For a pair of measures
$\genm_{\eta}, \genm_{\pi} \in \borelmp(\re^n)$, with
densities $\eta$ and $\pi$, respectively, 
the KR rearrangement is the unique  
monotone increasing lower triangular measurable map %
that pushes forward $\genm_{\eta}$ to $\genm_{\pi}$, i.e.,  
$T\push \genm_{\eta} = \genm_{\pi}$
\cite{carlier2010knothe,bogachev2005triangular}. 
Here, monotonicity is with respect to the 
lexicographic order on $\re^n$, while uniqueness is up to
$\genm_{\eta}$-null sets \cite{bogachev2005triangular}.
A lower triangular map $T:\re^n\ra\re^n$ is a multivariate 
function whose $k$th component depends only on the first $k$ input variables, i.e.,
\begin{equation}   
\emap( \xb ) = \left[\begin{array}{l}
T^1(x_1)\\ 
T^2(x_1,x_2)\\ 
\vdots \\ 
T^n(x_1,x_2,\dots x_n)
\end{array}\right]%
\label{eq:lowerTri}
\end{equation}
for some collection of functions
$(T^k)$ and for all $\xb=(x_1,\ldots,x_n)$.
The distinction between lower, upper, or other
more general forms of triangular map is a matter
of convention. We will revisit this important point
in Section \ref{sec:decomp}.
See Appendix \ref{sec:genKR} for a constructive definition of the KR
rearrangement based on a sequence of one-dimensional transports.
In our hypothesis, the KR rearrangement is always a bijection
on $\re^n$, while
each map %
\begin{equation} \label{eq:componentMap}
  \xi \mapsto T^k( x_1 ,\ldots, x_{k-1}, \xi )
\end{equation}
is
homeomorphic (continuous bijection with
continuous inverse), strictly increasing, and
differentiable 
a.e.\ \cite{santambrogio2015optimal,tao2011introduction}. 
Here, monotonicity with respect to the lexicographic order is equivalent
to %
each \eqref{eq:componentMap} being an increasing function.
The resulting rearrangement $T$ is far from
being a diffeomorphism but is still regular enough to
define a useful change of variables, as the following lemma
proven in \cite{bogachev2005triangular} shows.
\begin{lem}\label{lem:changeVarTriWeak}
If $T$ is a KR rearrangement pushing forward  
$\genm_{\eta}$ to  $\genm_{\pi}$, then $\genm_{\eta}$-a.e., 
\begin{equation} \label{eq:pullbackDensTri}
  T\pull \pi(\xb) = \pi(T(\xb))\,\det \nabla T(\xb) = \eta(\xb),
\end{equation}
where $\det \nabla T \coloneqq \prod_{i=1}^n \partial_k T^k$ exists a.e., and
where $T\pull \pi$ is the density of $T\pull \genm_{\pi}$.
\end{lem}
In general, $\det \nabla T$ in \eqref{eq:pullbackDensTri}  
is not the determinant of the
Jacobian of $T$ since the map may not be differentiable, in which case
it would not be possible to define
$\nabla T$ %
in the classical sense; this is
why $\det \nabla T$ is {\it redefined} in the lemma.
Nevertheless, it is  known that $T$ inherits the same
regularity as $\eta$ and $\pi$, but not more
\cite{santambrogio2015optimal, bogachev2005triangular}.
See Appendix \ref{sec:genKR} for additional remarks on the regularity of the map.

An essential feature of the triangular transport map is its {\it
  anisotropic} dependence on the input variables.  That is, even
though each component of the transport map does not depend on all $n$
inputs, the map is still capable of coupling arbitrary probability
distributions.  Informally, we can think of the KR
rearrangement as imposing the {\it sparsest} possible structure that
preserves generality of the coupling---in that the rearrangement is
guaranteed to exist for any 
$\genm_{\eta},\genm_{\pi} \in \borelmp(\re^n)$.
(In fact, the transport can be defined under much weaker conditions \cite{santambrogio2015optimal}.) 
In Section \ref{sec:decomp}, we will show that the
anisotropy of the KR rearrangement is crucial to
proving that certain ``complex'' (and generally non-triangular)
transports can be factorized into compositions of a few
\textit{lower-dimensional} {triangular} maps.  Thus we can think of the
KR rearrangement as the fundamental building block of a more general
class of non-triangular transports.

The KR rearrangement also enjoys many attractive
computational features.  As shown in
\cite{el2012bayesian,marzouk2016introduction}, it can be characterized
as the unique minimizer of the Kullback--Leibler (KL) divergence
$\Dkl(\,T\push \genm_{\eta} \,\vert\vert \, \genm_{\pi} \, )$ over the cone
$\spaceMapT$ of monotone increasing triangular maps.
From the perspective of function approximation, 
parameterizing a monotone triangular map
is straightforward: it suffices to write each
component of the map as\footnote{\hrevone{For computational
efficiency, one may substitue the exponential function with any other strictly positive expression, like a positively shifted square function.
}
}
\begin{equation} \label{eq:monotone}
  T^k(\xb)= a_k(x_1,\ldots,x_{k-1}) + 
  \int_0^{x_k} \exp \left ( b_k(x_1,\ldots,x_{k-1}, t) \right )\, {\rm d} t,
\end{equation}
for some arbitrary functions $a_k:\re^{k-1}\ra\re$ and
$b_k:\re^{k}\ra\re$ \cite{bigoni2016monotone,ramsay1998estimating}.
\hrevone{ 
For example,
one could parameterize each 
$a_k,b_k$ using a 
linear expansion 
\begin{equation}
  a_k(\xb) = \sum_i a_{k,i} \,\psi_i(\xb),\qquad
  b_k(\xb) = \sum_j b_{k,j} \,\psi_j(\xb)
\end{equation}
in terms of multivariate 
Hermite polynomials $(\psi_i)$ and unknown 
coefficients $\cb=(a_{k,i}, b_{k,j})$; alternatively, one could use
a neural network representation \cite{goodfellow2016deep} of
$a_k$ and $b_k$.
The resulting transport map $T[\cb]$---parameterized by  
the coefficients $\cb$---is monotone and invertible
for all choices of $\cb$.  
}
(In
contrast, parameterizing general classes of monotone
\textit{non-triangular} maps is a difficult task.)
The minimization of
$\Dkl(\,T\push \genm_{\eta} \,\vert\vert \, \genm_{\pi} \, )$ for a map in
$\spaceMapT$ and for a pair of nonvanishing target ($\pi$) and
reference ($\eta$) densities  can be rewritten as
\cite{el2012bayesian,marzouk2016introduction}:
\begin{eqnarray}  \label{OptimDirect}
      & \min\limits_T     &  -\Ex \left [ \log \pi(T(\Xb)) + 
      \sum_k \log \partial_k T^k(\Xb)
      - \log \eta(\Xb) \right ]   \\
      & {\rm s.t.}      &  T \in \spaceMapT,       \nonumber
\end{eqnarray}
where 
the expectation is taken with respect to
the reference measure---which is the law of
$\Xb$. 

Two aspects of \eqref{OptimDirect} are particularly important. First,
for the purpose of optimization, the target density can be replaced
with its unnormalized version $\bar{\pi}$. (This replacement is
essential in Bayesian inference, where the posterior normalizing
constant is usually unknown.) 
Second, \eqref{OptimDirect} can be treated as a stochastic program and
solved by means of sample-average approximation (SAA)
\cite{Shapiro2013}
or stochastic approximation
\cite{kushner2003stochastic, spall2005introduction,bottou2016optimization}.  
Recall that the
reference measure is a degree of freedom of the problem and is chosen
precisely to make the integration in \eqref{OptimDirect} feasible
using, for instance, quadrature, Monte Carlo, or quasi-Monte Carlo
methods \cite{davis2007methods,robert2013monte,dick2013high,dick2016higher}.

Assuming some additional regularity for $\pi$ (e.g., at least
differentiability) and using
the monotone parameterization of \eqref{eq:monotone}, then
\eqref{OptimDirect} becomes an unconstrained and differentiable
optimization problem.  In particular, we can use the gradient of
$\log \pi$ to obtain an unbiased estimator for the gradient of
\eqref{OptimDirect}
\cite{ho1983perturbation,glasserman1991gradient,asmussen2007stochastic}.
Alternatively, if $\nabla \log \pi$ is unavailable, we can use the
{\it score method} \cite{glynn1990likelihood,ranganath2013black} to
produce an estimator that is still unbiased, but with higher variance.
For concreteness, consider the realization of an i.i.d.\ sample
$(\xb_i)_{i=1}^M$ from $\genm_\eta$.  Then 
a
SAA of
\eqref{OptimDirect} reads as:
\begin{eqnarray}  \label{OptimDirectSAA}
			& \min\limits_{T}     &  - \sum_{i=1}^M \left
                                                ( \,\log \bar{\pi}(T(\xb_i)) + 
			\sum_k \log \partial_k T^k(\xb_i) 
			- \log \eta(\xb_i)   \right )   \\
			& {\rm s.t.}      &  T \in \spaceMapT,       \nonumber
\end{eqnarray}
which is now amenable to deterministic optimization techniques. 
The numerical solution of
\eqref{OptimDirectSAA} by means of an iterative optimization method
(e.g., BFGS \cite{wright1999numerical}) produces a sequence of maps
$\widetilde{T}_1,\widetilde{T}_2,\ldots$ that are increasingly
better approximations of the KR rearrangement, in the
sense defined by \eqref{OptimDirectSAA}.  In particular, we can
interpret $(\widetilde{T}_k)_k$ as a discrete time flow that pushes
forward the collection of reference samples, $(\xb_i)_{i=1}^M$, to the
target distribution. See Figure \ref{fig:movieToyProblem} for a
simple illustration.
As shown in \cite{el2012bayesian}, the KL divergence
$\Dkl(\,\widetilde{T}\push \genm_{\eta} \,\vert\vert \, \genm_{\pi} \, )$
for an approximate map $\widetilde{T}$ can be estimated as:
\begin{equation} \label{eq:var_diag}
  \Dkl(\,\widetilde{T}\push \genm_{\eta} \,\vert\vert \, \genm_{\pi} \, ) \approx
  \frac{1}{2} \Var \,\left[ \log \bar{\pi}(\widetilde{T}(\Xb)) + 
      \sum_k \log \partial_k \widetilde{T}^k(\Xb)
      - \log \eta(\Xb)   \right],
\end{equation}
up to second-order terms, in the limit of 
$\Dkl(\,\widetilde{T}\push \genm_{\eta} \,\vert\vert \, \genm_{\pi} \, ) \ra 0$, even
if the normalizing constant of $\pi$ is unknown.
This convergence criterion is rather useful for {\it any} variational inference
method, and is usually not available for techniques
like MCMC. %
In the same way, 
\cite{el2012bayesian}
constructs effective 
estimators for the normalizing constant
$\beta \coloneqq \bar{\pi} / \pi$ as
\begin{equation} \label{eq:normConst}
  \hat{\beta} = \exp \Ex\left[\log \bar{\pi}(\widetilde{T}(\Xb)) + 
      \sum_k \log \partial_k \widetilde{T}^k(\Xb) - \log\eta(\Xb)\right].
\end{equation}

We refer the reader to
\cite{parno2015transport,marzouk2016introduction} for an alternative
construction of the transport map that is useful when only \textit{samples
from the target measure} are available.  An interesting application of
the latter construction is the problem of density estimation
\cite{tabak2013family} or Bayesian inference with intractable
likelihoods \cite{wilkinson2011stochastic,Marin2011,Csillery2010}. In
this case, it turns out that the {\it inverse} transport $S = T^{-1}$
can be easily computed via convex optimization
\cite{parno2014transport}.  (Notice that $S$ is just an ordinary
triangular transport map that pushes forward $\genm_\pi$ to $\genm_\eta$.  The 
``inverse'' descriptor will help distinguish $S$ from the
map $T$ that pushes forward the reference to the target distribution.  We
refer to $T$ as the {\it direct} transport.)  We can then invert $S$
at $\xb\in\re^n$ to obtain the evaluation of the direct transport
$T(\xb)$.  Inverting a monotone triangular function is a
computationally trivial task since it requires the solution of a
sequence of one-dimensional root finding problems
\cite{parno2015transport,marzouk2016introduction}.  
In practice, one just needs to invert \eqref{eq:componentMap}
for $k=1,\ldots,n$.
It is also possible to compute the
inverse transport from the unnormalized target density, rather than
from samples; here, it suffices to minimize
$\Dkl(\,\genm_{\eta} \,\vert\vert \, S\push \genm_{\pi} \, )$ for
$S\in\spaceMapT$.  The resulting variational problem is equivalent to
\eqref{OptimDirect} with the identity $S=T^{-1}$.
By symmetry of our formulation, 
$S$ 
has the same regularity as $T$. 
In particular, Lemma \ref{lem:changeVarTriWeak} holds for $S$ as
well, and gives a formula for the pushforward density $T\push \eta$ as:
\begin{equation} \label{eq:pushforDensTri}
  T\push \eta(\zb) = \eta(S(\zb))\,\det \nabla S(\zb) = \pi(\zb),
\end{equation}
where $\det \nabla S \coloneqq \prod_{i=1}^n \partial_k S^k$ exists a.e., and
where $T\push \eta$ is the density of $T\push \genm_{\eta}$.

There is a growing body of literature on the efficient numerical approximation
of transport maps 
(see, e.g., 
\cite{mendoza2018bayesian,rezende2015variational,liu2016stein,tabak2013family,el2012bayesian,bigoni2016monotone}). 
Essentially all of these approaches employ numerical optimization to construct or
realize the action of a map, and thus harness \textit{optimization} to
enhance \textit{integration}.
Yet all these approaches face a fundamental challenge: the
transport map is a function from $\re^n$ onto itself, and in high
dimensions (i.e., for large $n$) the representation and approximation
of such functions becomes increasingly intractable.
In the ensuing sections, on the other hand, we will show that a large
class of transport maps are in fact only superficially high-dimensional;
that is, they possess some \textit{hidden} low-dimensional structure
that can facilitate their fast and reliable computation.  This
low-dimensional structure is linked to the Markov properties of the
target measure, which we briefly review
in the next section.
\newcommand{\scaleToolMovie}{0.5}
\newcommand{\hspaceToolMovie}{15pt}
\newcommand{\hspaceCaptionToolMovie}{9pt}
\captionsetup[subfigure]{labelformat=empty}
\begin{figure}[h]
\begin{center}
\subfloat[\hspace{\hspaceCaptionToolMovie}
$\widetilde{T}_0$]{\includegraphics[scale=\scaleToolMovie]{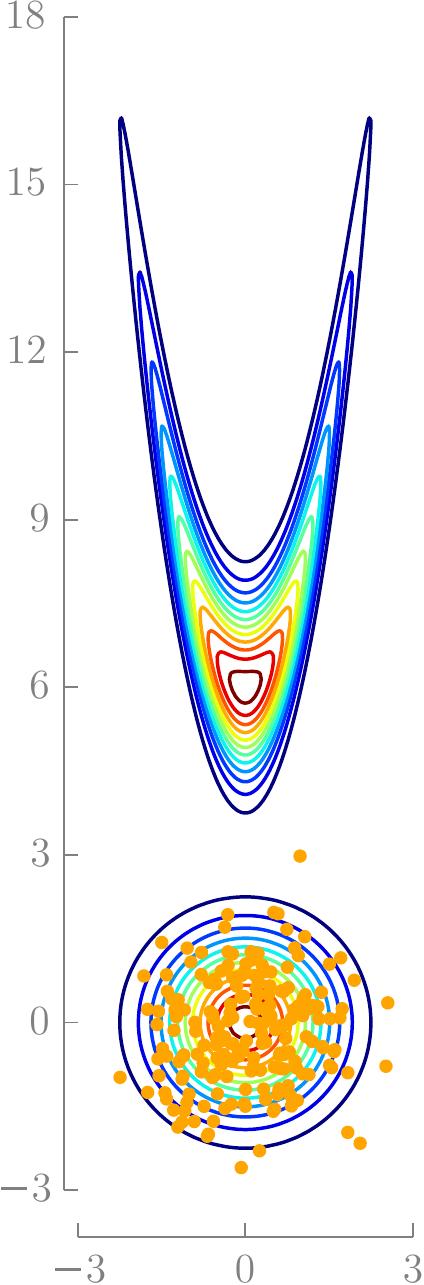}}
\hspace{\hspaceToolMovie}
\subfloat[\hspace{\hspaceCaptionToolMovie}
$\widetilde{T}_1$]{\includegraphics[scale=\scaleToolMovie]{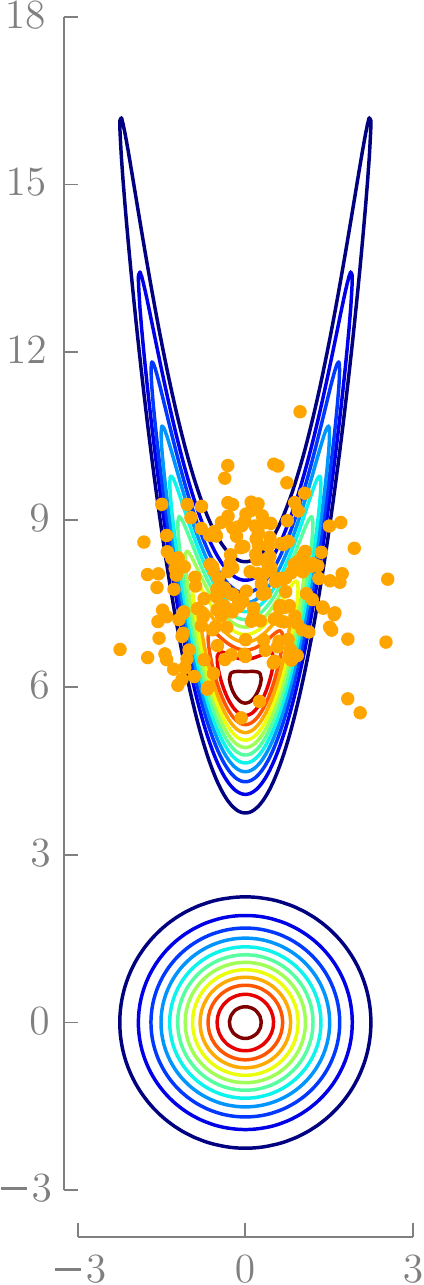}}
\hspace{\hspaceToolMovie}
\subfloat[\hspace{\hspaceCaptionToolMovie}
$\widetilde{T}_2$]{\includegraphics[scale=\scaleToolMovie]{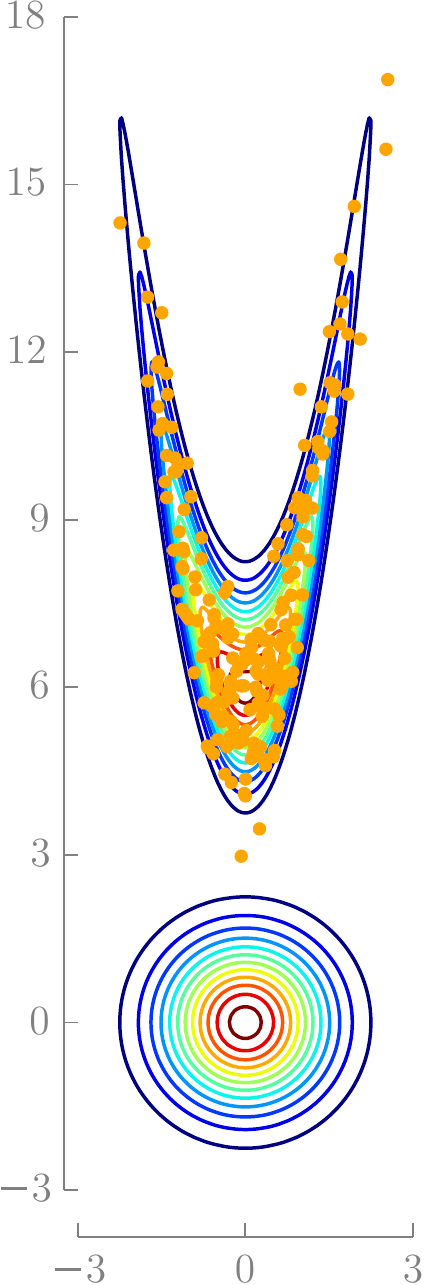}}
\hspace{\hspaceToolMovie}
\subfloat[\hspace{\hspaceCaptionToolMovie}
$\widetilde{T}_3$]{\includegraphics[scale=\scaleToolMovie]{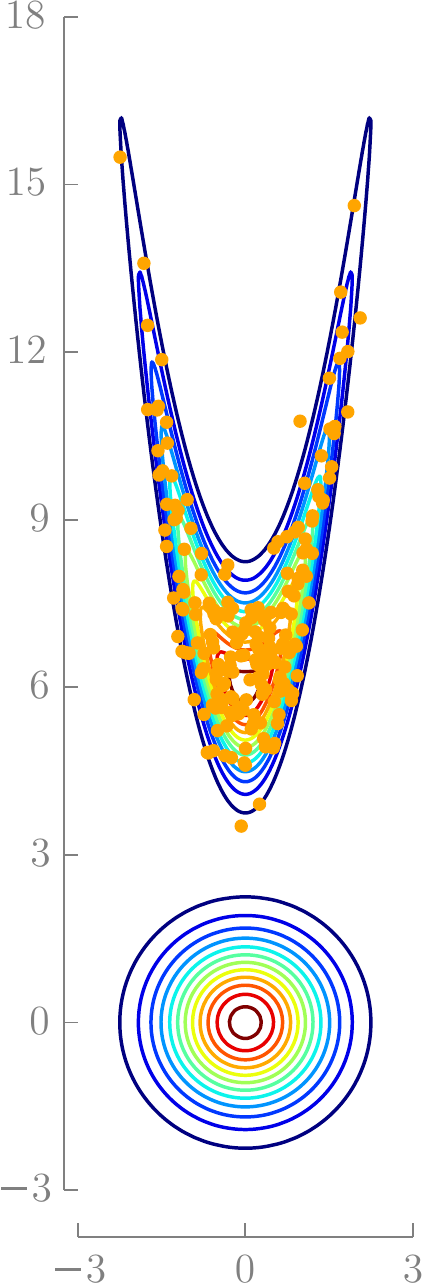}}
\caption[]{Computation of a simple transport map in two dimensions:
  The leftmost figure shows contours of the reference density
  $\eta$, which is a standard Gaussian, and of the target density
  $\pi$, which is a banana-shaped distribution in the tails of $\eta$.
  The target distribution has a nonlinear dependence structure. The
  orange dots in the leftmost figure correspond to 100 samples
  $(\xb_i)$ from $\eta$ and are used to make a sample-average
  approximation of \eqref{OptimDirect}.  We adopt the triangular
  monotone parameterization of \eqref{eq:monotone} for the candidate
  transport map, where the functions $a_k,b_k$ are expanded in a
  multivariate Hermite polynomial basis of total degree two
  \cite{xiu2010numerical}.  The resulting optimization problem is
  solved with a quasi-Newton method (BFGS).  The $k$th figure from the
  left shows the pushforward of the original reference samples through
  the approximate transport map, $\widetilde{T}_k$, after $k$
  iterations of BFGS. The initial map $\widetilde{T}_0$ is chosen to
  be the identity.
  The reference samples flow \emph{collectively} towards the target
  density and eventually settle on the support of $\pi$, capturing its
  structure after just a few iterations.}
 \label{fig:movieToyProblem} 
\end{center}
\end{figure}

%% file: sec_Markov.tex
Let $\Zb=(Z_1,\ldots,Z_n)$ be a collection of
random variables with law
$\genm_\pi$ and density $\pi$.
We can represent a list of conditional
independences satisfied by 
$\Zb$---the so-called Markov properties---using a simple
undirected graph $\Gcb=(\Vc,\Ec)$,
where each node %
$k \in \Vc$ %
is associated
with 
a  distinct random variable, $Z_k$,
and where the
edges in $\Ec$  encode
a specific notion of probabilistic interaction
among these random variables \cite{koller2009probabilistic}.
In particular, we say that $\Zb$ is a Markov
network---or a Markov random field (MRF)---with respect to $\Gcb$ if for any triplet 
$\Aset,\Sset,\Bset$ of
disjoint subsets of $\Vc$, where $\Sset$ is a 
separator set for $\Aset$ and $\Bset$,\footnote{
$\Sset$ is a separator set for 
$\Aset$ and $\Bset$ if (1) $\Sset$ is disjoint from
$\Aset$ and $\Bset$ (2)
Every path %
from $\alpha \in \Aset$ to $\beta \in \Bset$ intersects $\Sset$.
If $\Aset$ and $\Bset$ are disconnected components of $\Gcb$, then
$\Sset = \emptyset$ is a separator set for  $\Aset$ and $\Bset$. 
} 
the subcollections
$\Zb_{\Aset} $
and $\Zb_{\Bset}$ 
are conditionally independent given $\Zb_{\Sset}$, i.e.,
\begin{equation} \label{eq:globalMarkov}
 \Zb_{\Aset} \orth	\Zb_{\Bset} \, \vert \, \Zb_{\Sset}.
\end{equation}
The measure $\genm_\pi$
is said to satisfy the global Markov property,
relative to $\Gcb$, if \eqref{eq:globalMarkov}  holds \cite{lauritzen1996graphical}.
We can also say that $\genm_\pi$ is globally Markov with respect to $\Gcb$.
The corresponding graph is then called an independence map (I-map) for 
$\genm_\pi$  \cite{koller2009probabilistic}.
Intuitively, a sparse graph  represents a family of distributions that enjoy many conditional
independence properties. 
I-maps are in general not unique.
Of particular interest are \textit{minimal} I-maps, i.e., the sparsest
graphs compatible with the conditional independence structure of $\genm_\pi$.
Conditional independence is associated with factorization properties of 
$\pi$.
For instance,  $\Zb_{\Aset} \orth     \Zb_{\Bset} \, \vert \, \Zb_{\Sset}$ if and only if
$\pi_{\Zb_{\Aset} ,  \Zb_{\Bset} \vert \Zb_{\Sset}} = 
\pi_{\Zb_{\Aset}  \vert \Zb_{\Sset}} \, 
\pi_{\Zb_{\Bset}   \vert \Zb_{\Sset}} $ a.e.\ \cite{lauritzen1996graphical}.
We then say that $\genm_\pi$ {\it factorizes} according to some graph
$\Gcb$ if there exists a version of the density of $\genm_\pi$ such that
\begin{equation} \label{eq:factHammer}
     \pi(\zb) = \frac{1}{\cfrak} \prod_{\Cc \in\Ccb} \, \psi_{\Cc}(\zb_{\Cc}),
\end{equation}
for some nonnegative functions $(\psi_{\Cc})$ called \emph{potentials}, 
where
$\Ccb$ is the set of maximal cliques\footnote{ 
A clique is a fully connected subset of the vertices, 
whereas a maximal clique is a clique that is not a strict subset of another clique.
}
 of $\Gcb$ and $\cfrak$ is a 
 normalizing constant.
It is immediate to show that if $\genm_\pi$ factorizes according to $\Gcb$, then
$\genm_\pi$ satisfies the global Markov property relative to $\Gcb$
\cite[Proposition 3.8]{lauritzen1996graphical}.
The converse is true only under additional assumptions: 
for instance, if $\genm_\pi$ admits a  continuous and strictly positive density 
(see the Hammersley-Clifford theorem \cite{hammersley1971markov,lauritzen1996graphical}).

A critical question then is how to characterize a suitable I-map for a
given measure. There are several answers. 
First of all, in many applications that involve probabilistic modeling, 
the target distribution is defined 
in terms of %
its potentials, as in 
\eqref{eq:factHammer},
because this is just a more convenient
way to specify a high-dimensional
distribution and to perform inference (or general probabilistic reasoning)
with it \cite{koller2009probabilistic}. 
Finding a graph for which $\genm_\pi$ factorizes is then a trivial task.
See Figure \ref{fig:sparsityVolHyper} ({\it left}) for an example. 
Applications where this commonly holds range from spatial statistics and image analysis to speech recognition \cite{koller2009probabilistic,rue2005gaussian}. 
In Section \ref{sec:dataAss}, for example, we focus exclusively
on discrete-time Markov processes, where the
Markov structure of the problem is self-evident.
\hrevone{ 
More specifically, Section \ref{sec:dataAss} 
tackles the problem of recursive smoothing and
static parameter estimation for a state-space model.
In this context, the target measure $\genm_\pi$ could
represent the joint distribution of state and parameters, conditioned on
all the available observations (see   
Figures \ref{fig:sparsityVolHyper} and \ref{fig:dataAssHyper}).
The reader might want to consider this sequential inference problem
as a guiding application while reading the forthcoming Sections
\ref{sec:sparse} and \ref{sec:decomp}. We emphasize, however, that our
theory is far more general and by no means restricted to any specific Markov structure.
}

In other settings, the graph is unknown and must be estimated. When only samples from 
$\genm_\pi$ are available, this is a question of model learning, as described in \cite[Part III]{koller2009probabilistic}; see also \cite{hyvarinen2005estimation,meinshausen2006high,yuan2007model,lin2015high} for various applications.
In case of a known and smooth target density, we can 
characterize pairwise conditional independence in terms of mixed
second-order partial derivatives, as shown by the following lemma. 
\begin{lem}[Pairwise conditional independence] \label{lem:pairwiseIndep}
If $\Zb \sim \genm_\pi$  for a measure $\genm_\pi$ with smooth and strictly positive
density $\pi$, we have:
\begin{equation}
Z_i \orth Z_j \,\vert\, \Zb_{\Vc\setminus (i,j)} \quad
\Longleftrightarrow \quad
\partial^2_{i,j}\log \pi = 0 %
 \,\, {\text on} \,\, \re^n.
\end{equation}
\end{lem}
Thus, if we can evaluate $\pi$ and its derivatives (up to
a normalizing constant), we can
use Lemma \ref{lem:pairwiseIndep} to assess pairwise
conditional independence and to define a minimal I-map for $\genm_\pi$
as follows: add an edge between every pair
of distinct nodes unless the corresponding random variables
are conditionally independent  
\cite[Theorem 4.5]{koller2009probabilistic}.

Regardless of the many ways to obtain an I-map, 
there is a fundamental connection between Markov properties of
a distribution and the existence
of low-dimensional transport maps.
The rest of the paper will elaborate precisely on this connection.

%% file: sec_sparse.tex
We begin our investigation of low dimensional structure 
by considering
the notion of sparse transport map.
A sparse  map    
is a multivariate function where each component does not
depend on all of its input variables. According to this definition, a
triangular 
transport 
is already sparse.  In this section, however, we
show that the KR rearrangement can be even 
\textit{sparser}, depending on the Markov structure of the target
distribution.

\subsection{Sparsity bounds} \label{sec:sparsityBounds}
Given a lower triangular function $T$, we define its  
  sparsity pattern,  $\sparse_T$, as the set of all integer pairs $(j,k)$, with
$j<k$, such that the $k$th component of the map does not depend on the
$j$th input variable, i.e., 
$\sparse_T=\{ (j,k) : j<k , \,\partial_j T^k=0 \}$.
(We do not include pairs $j>k$ in the definition of $\sparse_T$  since, for a lower triangular
function, $\partial_j \, T^k = 0 $ for $j>k$ by construction.)

Knowing the sparsity pattern of the KR rearrangement
\textit{before} computing the actual transport has important
computational implications.  For instance, in the variational
characterization of the transport described in \eqref{OptimDirect}, we can
restrict the feasible domain %
to the set of triangular maps with sparsity pattern
given by $\sparse_T$, and still recover the desired KR
rearrangement.  That is, if $(j,k)\in \sparse_T$, %
we can
parameterize any candidate transport map 
by
removing the dependence on the $j$th input variable from the $k$th
component of the map.  
Thus, analyzing the Markov structure of the target distribution
enables the representation and computation of maps in possibly higher-dimensional
settings.

The following theorem, which is the main result of this section,
characterizes \textit{bounds} on the sparsity patterns of triangular
transport maps
given an I-map for the target measure.  In the statement of
the theorem, we denote the direct transport by $T$ and the inverse
transport by $S=T^{-1}$ (see Section \ref{sec:compTransport}).  The
theorem suggests that $S$ and $T$ can have quite different sparsity
patterns.\footnote{
A note: as we already saw, the KR rearrangement is unique up to a
set of measure zero. Theorem \ref{thm:sparsityRosenblatt} 
characterizes the sparsity pattern of a particular {\it version} of the map,
the one given by Definition \ref{def:KRrearr} in Appendix \ref{sec:genKR}.
We will implicitly make this assumption throughout the paper.
} %

\begin{thm}[Sparsity of Knothe--Rosenblatt rearrangements] \label{thm:sparsityRosenblatt}
Let $\Xb \sim \genm_\eta$, $\Zb \sim \genm_\pi$ with 
$\genm_\eta,\genm_\pi \in \borelmp(\re^n)$ and 
$\genm_\eta$ a product measure on 
$\times_{i=1}^n\mathbb{R}$.
Moreover, assume that
$\genm_\pi$ is globally Markov with respect to $\Gcb$,
and
define, recursively, the sequence of graphs $(\Gcb^k)_{k=1}^n$ as:
(1) $\Gcb^n \coloneqq \Gcb$ and (2) for all
$1\le k < n$, $\Gcb^{k-1}$ is obtained from $\Gcb^{k}$ by removing node $k$ and by 
turning its neighborhood $\neigh (k,\Gcb^k)$ into a clique.
Then the following hold:
	\begin{enumerate}
		\item \label{thm:sparsityRosenblatt:inverse} 	
		If $\sparse_S$ is the sparsity pattern
		of the inverse transport map $S$, then 
			\begin{equation} \label{eq:superSetInverse}
			\widehat{\sparse}_S \subset \sparse_S,
			\end{equation}										
		where 	$\widehat{\sparse}_S$ is the set of integer pairs $(j,k)$ 
		such that $j \notin \neigh (k, \Gcb^k) $.				
		\medskip
						  
		\item \label{thm:sparsityRosenblatt:direct} 
		If $\sparse_T$ is the sparsity pattern
		of the direct transport map $T$, then
			\begin{equation}  \label{eq:superSetDirect}
			\widehat{\sparse}_T \subset \sparse_T,
			\end{equation}										
		where 	$\widehat{\sparse}_T$ 
		is defined recursively as follows:
		for $k=2,\ldots,n$ the pair $(j,k) \in \widehat{\sparse}_T$
		if and only if
		$(j,i)\in \widehat{\sparse}_T$
		for all $ i\in \neigh (k, \Gcb^k)$.		
		\smallskip

		\item \label{thm:sparsityRosenblatt:inclusion}  
		The predicted sparsity pattern of $S$ is always greater than or
		equal to that of $T$, i.e.,

			\begin{equation} \label{eq:inclusion}
			\widehat{\sparse}_T \subset \widehat{\sparse}_S. 
			\end{equation}						
	\end{enumerate}				 					 
\end{thm}

Several remarks are in order. First, we emphasize the fact that
Theorem \ref{thm:sparsityRosenblatt} characterizes sparsity patterns
using \textit{only} an I-map for $\genm_\pi$, without requiring any actual
computation of the 
transports. 
One only needs to perform simple graph
operations on $\Gcb$ to build the sequence of graphs $(\Gcb^k)$. See
Figure \ref{fig:marginalGraph} for an illustration of this procedure,
with the corresponding sparsity patterns in
Figure \ref{fig:sparsityPatt}.
We refer to $(\Gcb^k)$ as the {\it marginal} graphs.  
In fact, the sequence $(\Gcb^k)$ is precisely the set of intermediate graphs 
produced by the variable elimination algorithm \cite[Ch.\ 9]{koller2009probabilistic}, 
when marginalizing 
with elimination ordering $(n, n-1, \ldots, 1)$. 
This should not be
surprising as the KR rearrangement is essentially a
sequence of ordered marginalizations \cite{villani2008optimal}.
The hypothesis that $\genm_\eta$ is a product measure 
is important for the theorem to hold. 
If we pick a reference measure with an arbitrary Markov structure, there need
not exist a sparse transport map coupling $\genm_\eta$ and $\genm_\pi$, even if
$\genm_\pi$ has a sparse I-map.
The role of a reference measure is somewhat peculiar to the world of
couplings and is usually not addressed in classical
treatments of graphical models. %
Nonetheless, this assumption on $\genm_\eta$
is not restrictive in the present framework,
since the reference distribution is considered a degree of freedom
of the problem.
Theorem \ref{thm:sparsityRosenblatt} gives sufficient but not
necessary conditions on $(\genm_\eta, \genm_\pi)$ for the existence of a sparse map.
And it could not be otherwise: if $\genm_\eta = \genm_\pi$ then the identity
map---the sparsest possible map---would be a valid coupling.

We also note that Theorem \ref{thm:sparsityRosenblatt} does not
provide the exact sparsity patterns of the triangular transport maps;
instead, \eqref{eq:superSetInverse} and \eqref{eq:superSetDirect}
provide \textit{subsets} of $\sparse_T$ and $\sparse_S$.
In other words, the actual
transport maps might be sparser than
predicted by the sets $\widehat{\sparse}_S$ and
$\widehat{\sparse}_T$---but, crucially, they cannot be less sparse.
Thus, we can think of Theorem \ref{thm:sparsityRosenblatt} as
providing \textit{bounds} on the sparsity of triangular transports.
An important fact is that, without additional information on $\genm_\pi$,
these
bounds are sharp.  That is, we can always find a pair of measures
$(\genm_\eta, \genm_\pi)$ satisfying the hypotheses of Theorem
\ref{thm:sparsityRosenblatt} and such that the predicted and actual
sparsity patterns coincide, i.e.,
$\widehat{\sparse}_T = \sparse_T $ or
$\widehat{\sparse}_S = \sparse_S$. 

Part \ref{thm:sparsityRosenblatt:inclusion} of Theorem
\ref{thm:sparsityRosenblatt} shows that the predicted sparsity pattern
of the inverse KR rearrangement is always larger than
or equal to that of the direct transport, i.e.,
$\widehat{\sparse}_T \subset \widehat{\sparse}_S $.  This does not
mean that for every pair of measures $(\genm_\eta, \genm_\pi)$, the inverse
triangular transport is always at least as sparse as the direct
transport; in fact, it is possible to provide simple counterexamples.
However, this result does imply that if we are only given an I-map for
$\genm_\pi$, then parameterizing candidate \textit{inverse} triangular transports
allows the imposition of more sparsity constraints than parameterizing
candidate direct transports. 
In general, sparser transports are easier to represent.
See Figure \ref{fig:sparsityVolHyper} ({\it right}) for a nontrivial
example of sparsity patterns for a stochastic volatility model.

Indeed, \eqref{eq:inclusion} hints at a typical trend: inverse
transport maps tend to be sparser (in many practical cases, \textit{much}
sparser) than their direct counterparts. Intuitively, the sparsity of
a direct transport is %
associated with
marginal independence in $\Zb$, 
whereas the inverse transport inherits sparsity from
the conditional independence structure of $\Zb$.
The latter is a weaker condition
than mutual independence;
for instance, the correlation length of a process modeled by a Markov
random field may be much larger than the typical neighborhood size
\cite{rue2005gaussian,blake2011markov}. 
Thus, given a sparse I-map for the target measure, it can be
computationally advantageous to characterize an inverse transport rather
than a direct one,
because the inverse transport can inherit a
larger sparsity pattern.  Given an inverse triangular transport $S$,
we can then easily evaluate the direct transport $T = S^{-1}$ at any point
$\xb\in\re^n$ by inverting $S$ pointwise, as described in Section~\ref{sec:compTransport}. There is no
need to have an explicit representation of the direct transport as
long as it can be implicitly defined through its inverse.

\begin{figure}[h]
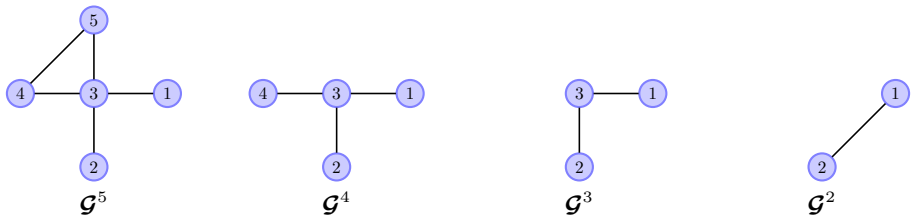

\centering
	\include{tikz_marginalGraphs}
        \caption[]{Sequence of graphs $(\Gcb^k)$ described in Theorem
          \ref{thm:sparsityRosenblatt} for a
          target measure in $\borelmp(\re^5)$
          with I-map illustrated by the leftmost graph,
          $\Gcb^5$.  Notice that to generate the graph $\Gcb^2$, we
          remove node $3$ from $\Gcb^3$ and turn its neighborhood into
          a clique by adding the edge $(1,2)$.  }
\label{fig:marginalGraph}
\end{figure}
\begin{figure}[h]
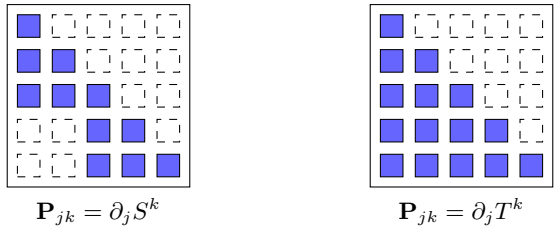

\centering
	\include{tikz_sparsePatternsExample}
\caption[]{Sparsity patterns predicted by Theorem \ref{thm:sparsityRosenblatt}
for the target measure analyzed in Figure \ref{fig:marginalGraph}. We
represent the sparsity patterns using a symbolic matrix notation: the
$(j,k)$-th entry of the matrix is {\it not} colored if
the $k$th component of the map ($S$ or $T$) does not depend on the $j$th input variable, or,
equivalently, if $(j,k)\in \widehat{\sparse}_S$ (resp. $\widehat{\sparse}_T$)
\eqref{eq:superSetInverse}. (Since we are considering lower triangular
transports,
 all entries $j>k$ are uncolored. Note also that $S^k$ and
$T^k$ are always functions of their $k$th input by strict monotonicity of
the map.) The predicted sparsity pattern for the direct
transport in this example is $\widehat{\sparse}_T = \emptyset$.
}
\label{fig:sparsityPatt}
\end{figure}

\begin{figure}[h]
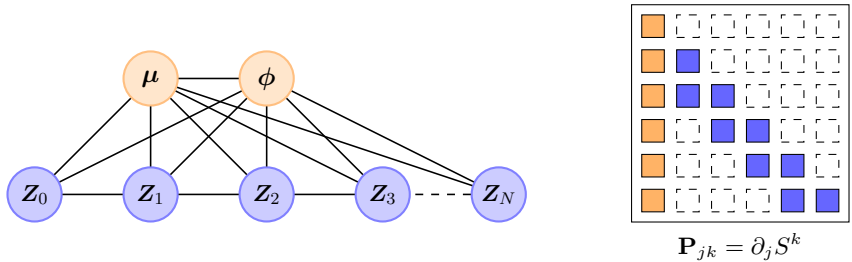

  \centering
  \include{tikz_sparsityVolHyper}
  \caption[]{ 
  	({\it left})
  	Markov network for a stochastic volatility model in \cite{kim1998stochastic}. 
  	Blue nodes represent the discrete-time latent log-volatility process $(\Zb_k)_{k=0}^N$, which obeys a simple autoregressive model 
with
  	hyperparameters $\mub,\phib$. 
  	The graph above is a minimal I-map for the posterior density described in 
  	Section \ref{sec:numerics},
        $\pi_{\mub, \phib, \Zb_{0:N} \vert \yb_{0:N}}$, where $\yb_{0:N}$ are some (fixed) observations.
  	({\it right}) The predicted sparsity pattern
    $\widehat{\sparse}_S$ (only the top
    $6 \times 6$ block is shown) for the inverse transport corresponding to
    the model on the left:
	the first
    column/row of the matrix refer jointly to all of the hyperparameters. 
    Each component $S^k$ of the inverse transport
    can depend at most on four input variables, namely
    $\mub,\phib,\Zb_{k-1},\Zb_k$, regardless of the overall
    dimension $N$ of the problem.
    In order to
    apply the results of Theorem \ref{thm:sparsityRosenblatt}, we must
    select an ordering of the input variables; here, we used the
    ordering $(\mub,\phib,\Zb_0,\ldots,\Zb_N)$. 
    Optimal orderings are further discussed in
    Section~\ref{sec_order}. 
  }
\label{fig:sparsityVolHyper}
\end{figure}

%% file: tikz_marginalGraphs.tex
\newcommand{\wStarGraph}{0.5cm} 
\newcommand{\scaleStarGraph}{.6} 

\subfloat[$\Gcb^5$]{
\begin{tikzpicture}[transform shape, scale = \scaleStarGraph]
	\node[default]  (center) 										{$3$};
 	\node[default]	 (belowCenter)		  		[below=of center]		{$2$} 
 		edge[edgeStyle] (center);
 	\node[default] (aboveCenter)			 		[above=of center]		{$5$}
 		edge[edgeStyle] (center);
 	\node[default] (rightCenter)		  			[right=of center]			{$1$}
 		edge[edgeStyle] (center);
	\node[default] (leftCenter)		  			[left=of center]			{$4$}
 		edge[edgeStyle] (aboveCenter)
 		edge[edgeStyle] (center);		
\end{tikzpicture}
} \hspace{\wStarGraph}
\subfloat[$\Gcb^4$]{
\begin{tikzpicture}[transform shape, scale = \scaleStarGraph]
	\node[default]  (center) 										{$3$};
 	\node[default]	 (belowCenter)		  		[below=of center]		{$2$} 
 		edge[edgeStyle] (center);
 	\node[defaultBlank] (aboveCenter)			 		[above=of center]		{\textcolor{white}{$5$}};
 	\node[default] (rightCenter)		  			[right=of center]			{$1$}
 		edge[edgeStyle] (center);
	\node[default] (leftCenter)		  			[left=of center]			{$4$}
 		edge[edgeStyle] (center);		
\end{tikzpicture}
} \hspace{\wStarGraph}
\subfloat[$\Gcb^3$]{
\begin{tikzpicture}[transform shape, scale = \scaleStarGraph]
	\node[default]  (center) 										{$3$};
 	\node[default]	 (belowCenter)		  		[below=of center]		{$2$} 
 		edge[edgeStyle] (center);
 	\node[defaultBlank] (aboveCenter)			 		[above=of center]		{\textcolor{white}{$5$}};
 	\node[default] (rightCenter)		  			[right=of center]			{$1$}
 		edge[edgeStyle] (center);
	\node[defaultBlank] (leftCenter)		  			
	[left=of center]			{\textcolor{white}{$4$}};
\end{tikzpicture}
} \hspace{\wStarGraph}
\subfloat[$\Gcb^2$]{
\begin{tikzpicture}[transform shape, scale = \scaleStarGraph]
	\node[defaultBlank]  (center) 										{\textcolor{white}{$5$}};
 	\node[default]	 (belowCenter)		  		[below=of center]		{$2$} ;
 	\node[defaultBlank] (aboveCenter)			 		[above=of center]		{\textcolor{white}{$4$}};
 	\node[default] (rightCenter)		  			[right=of center]			{$1$}
 		edge[edgeStyle] (belowCenter)	;	
	\node[defaultBlank] (leftCenter)		  			[left=of center]			{\textcolor{white}{$3$}};
\end{tikzpicture}
}

%% file: tikz_sparsePatternsExample.tex
\newcommand{\wStarGraph}{2cm} 
\newcommand{\scaleStarGraph}{.4} 

\subfloat[${\bf P}_{jk}=\partial_{j}S^k$]{
\begin{tikzpicture}[transform shape, scale = \scaleStarGraph ]
\matrix [draw ,column sep=.15cm, row sep=.15cm, ampersand replacement=\& ]
	{
\node[entriesMatrix] { };\& \node[entriesMatrixZeroDashed]{ };\& \node[entriesMatrixZeroDashed]{ }; \& \node[entriesMatrixZeroDashed] { }; \& \node[entriesMatrixZeroDashed] { };   \\
\node[entriesMatrix] { };\& \node[entriesMatrix]{ };\& \node[entriesMatrixZeroDashed]{ }; \& \node[entriesMatrixZeroDashed] { }; \& \node[entriesMatrixZeroDashed] { };   \\
\node[entriesMatrix] { };\& \node[entriesMatrix]{ };\& \node[entriesMatrix]{ }; \& \node[entriesMatrixZeroDashed] { }; \& \node[entriesMatrixZeroDashed] { };   \\
\node[entriesMatrixZeroDashed] { };\& \node[entriesMatrixZeroDashed]{ };\& \node[entriesMatrix]{ }; \& \node[entriesMatrix] { }; \& \node[entriesMatrixZeroDashed] { };   \\
\node[entriesMatrixZeroDashed] { };  \& \node[entriesMatrixZeroDashed]{ };     \& \node[entriesMatrix]{ }; \& \node[entriesMatrix] { }; \& \node[entriesMatrix] { };    \\
	};
\end{tikzpicture}
}
\hspace{\wStarGraph}
\subfloat[${\bf P}_{jk}=\partial_{j}T^k$]{
\begin{tikzpicture}[transform shape, scale = \scaleStarGraph ]
\matrix [draw ,column sep=.15cm, row sep=.15cm, ampersand replacement=\& ]
	{
\node[entriesMatrix] { };\& \node[entriesMatrixZeroDashed]{ };\& \node[entriesMatrixZeroDashed]{ }; \& \node[entriesMatrixZeroDashed] { }; \& \node[entriesMatrixZeroDashed] { };   \\
\node[entriesMatrix] { };\& \node[entriesMatrix]{ };\& \node[entriesMatrixZeroDashed]{ }; \& \node[entriesMatrixZeroDashed] { }; \& \node[entriesMatrixZeroDashed] { };   \\
\node[entriesMatrix] { };\& \node[entriesMatrix]{ };\& \node[entriesMatrix]{ }; \& \node[entriesMatrixZeroDashed] { }; \& \node[entriesMatrixZeroDashed] { };   \\
\node[entriesMatrix] { };\& \node[entriesMatrix]{ };\& \node[entriesMatrix]{ }; \& \node[entriesMatrix] { }; \& \node[entriesMatrixZeroDashed] { };   \\
\node[entriesMatrix] { };  \& \node[entriesMatrix]{ };     \& \node[entriesMatrix]{ }; \& \node[entriesMatrix] { }; \& \node[entriesMatrix] { };    \\
	};
\end{tikzpicture}
} 

%% file: tikz_sparsityVolHyper.tex
\newcommand{\wStarGraph}{1cm} 
\newcommand{\scaleStarGraph}{.8} 
\newcommand{\minsize}{9mm} 
\newcommand{\varState}{\Zb} 

\subfloat{
\begin{tikzpicture}[transform shape, scale = \scaleStarGraph]
	\node[default,minimum size=\minsize]  (nd1) 									{$\varState_0$};
 	\node[default,minimum size=\minsize]	 (nd2)		  		[right=of nd1]		{$\varState_1$} 
 		edge[edgeStyle] (nd1);
 	\node[default,minimum size=\minsize]	 (nd3)		  		[right=of nd2]		{$\varState_2$} 
 		edge[edgeStyle] (nd2);
 	\node[default,minimum size=\minsize]	 (nd4)		  		[right=of nd3]		{$\varState_3$} 
 		edge[edgeStyle] (nd3); 		 		
 	\node[default,minimum size=\minsize]	 (ndT)		  		[right=of nd4]		{$\varState_N$} 
 		edge[edgeStyle,dashed] (nd4); 	 		
 	\node[defaultOrange,minimum size=\minsize]	 (above2)		  		[above=of nd2]	{\large $\mub$} 
 		edge[edgeStyle] (nd1)
 		edge[edgeStyle] (nd2)
 		edge[edgeStyle] (nd3)  	
 		edge[edgeStyle] (nd4)
 		edge[edgeStyle] (ndT);
 	\node[defaultOrange,minimum size=\minsize]	 (above3) [above=of nd3] {\large $\phib$} 
 		edge[edgeStyle] (nd1)
 		edge[edgeStyle] (nd2)
 		edge[edgeStyle] (nd3)  	
 		edge[edgeStyle] (nd4)
 		edge[edgeStyle] (ndT)
 		edge[edgeStyle] (above2); 
\end{tikzpicture}
}  
\hspace{\wStarGraph}
\subfloat[${\bf P}_{jk}=\partial_{j}S^k$]{
\begin{tikzpicture}[transform shape, scale = \scaleStarGraph ]
\matrix [draw ,column sep=.15cm, row sep=.15cm, ampersand replacement=\& ]
	{
\node[entriesMatrixOrange] { };\& \node[entriesMatrixZeroDashed]{ };\& \node[entriesMatrixZeroDashed]{ }; \& \node[entriesMatrixZeroDashed] { }; \& \node[entriesMatrixZeroDashed]{}; \& \node[entriesMatrixZeroDashed] { };   \\
\node[entriesMatrixOrange] { };\& \node[entriesMatrix]{ };\& \node[entriesMatrixZeroDashed]{ }; \& \node[entriesMatrixZeroDashed] { }; \& \node[entriesMatrixZeroDashed]{}; \& \node[entriesMatrixZeroDashed] { };   \\
\node[entriesMatrixOrange] { };\& \node[entriesMatrix]{ };\& \node[entriesMatrix]{ }; \& \node[entriesMatrixZeroDashed] { }; \& \node[entriesMatrixZeroDashed]{}; \& \node[entriesMatrixZeroDashed] { };   \\
\node[entriesMatrixOrange] { };\& \node[entriesMatrixZeroDashed]{ };\& \node[entriesMatrix]{ }; \& \node[entriesMatrix] { }; \& \node[entriesMatrixZeroDashed]{}; \& \node[entriesMatrixZeroDashed] { };   \\
\node[entriesMatrixOrange] { };  \& \node[entriesMatrixZeroDashed]{ };     \& \node[entriesMatrixZeroDashed]{ }; \& \node[entriesMatrix] { }; \& \node[entriesMatrix]{}; \& \node[entriesMatrixZeroDashed] { };    \\
\node[entriesMatrixOrange] { };  \& \node[entriesMatrixZeroDashed]{ };     \& \node[entriesMatrixZeroDashed]{ }; \& \node[entriesMatrixZeroDashed] { }; \& \node[entriesMatrix]{}; \& \node[entriesMatrix] { };    \\
	};
\end{tikzpicture}
}	

%% file: sec_gmrf.tex
The reader familiar with Gaussian Markov random fields (GMRFs), might
see links between the preceding results and widespread approaches to the
modeling of Gaussian fields. 
In this section, we clarify the extent of these connections.

Many applications (e.g., image analysis, spatial statistics, time
series) involve modeling by means of high-dimensional Gaussian
fields.  Dealing with large and dense covariances, however, is often
impractical; both storage and sampling of the Gaussian field are
problematic.  The usual workaround is to replace or approximate the
Gaussian field with a {\it sparse} GMRF---i.e., a Gaussian Markov
network that enforces locality in the probabilistic interactions among
the underlying random variables.  The minimal I-map for the GMRF is
thus sparse, and so is the precision matrix $\Lambda$ 
of the field \cite{rue2005gaussian}.  The covariance matrix is still
in general dense, but dealing with the sparse precision matrix is much
easier.  If $L L^\top$ is a (sparse) Cholesky decomposition of
$\Lambda$, 
then $L^\top$ represents a linear triangular
transport that pushes forward 
the joint distribution of the GMRF,
$\genm_\pi = \Gauss(0,\Lambda^{-1})$, to
a standard normal,
$\genm_\eta = \Gauss(0,{\bf I})$.
The key point is that for many Markov
structures of interest, the Cholesky factor inherits sparsity from the
underlying graph, so that sampling from $\genm_\pi$ can be achieved at low
cost as follows: if $\Xb$ is a sample from $\genm_\eta$, then we can obtain a sample
$\Zb$ from $\genm_\pi$ simply by solving the sparse triangular linear system
$L^\top\,\Zb = \Xb$.  There is no need to explicitly represent or
store the dense factor $L^{-\top}$, since we can implicitly represent
its action by inverting a sparse triangular function.

Now the connection with Section \ref{sec:sparsityBounds} is clear: 
$L^\top$ is an inverse triangular
transport,\footnote{Actually, this transport is upper rather than
lower triangular. This distinction plays no role in the following
discussion, and the fact that a KR rearrangement is a
lower triangular function is merely a matter of convention.}
while $L^{-\top}$ is a direct one.
Moreover, solving
a triangular linear system is just a particular instance of inverting
a nonlinear triangular function by performing a sequence of
one-dimensional root-findings.
Thus the developments of the previous section, which consider
arbitrary \textit{nonlinear} maps, are a natural generalization---to the
{\it non-Gaussian} case---of modeling and sampling techniques for
high-dimensional GMRFs \cite{rue2005gaussian}.

%% file: sec_order.tex
The results of Theorem \ref{thm:sparsityRosenblatt} suggest that the
sparsity of a triangular transport map depends on the ordering of the
input variables.  See Figure \ref{fig:example_badOrder} for a simple
illustration. Indeed, the triangular transport itself
depends anisotropically on the input variables and requires the
definition of a proper ordering.  A natural approach is then to seek
the ordering that promotes the {\it sparsest} transport map possible.
Consider a pair of measures 
$(\genm_\eta, \genm_\pi)$ 
that satisfies the hypotheses of Theorem \ref{thm:sparsityRosenblatt}.
We associate an ordering of the input variables with
a permutation $\sigma$ on $\Nbb_n=\{1,\ldots,n\}$,
and define the {\it reordered} target measure $\genm_{\pi}^\sigma$ 
as the 
pushforward of $\genm_\pi$ by 
the matrix $Q^\sigma$
that represents the permutation $\sigma$.
In particular, $(Q^\sigma)_{ij}=(\eb_{\sigma(i)})_j$, where
$\eb_i$ is the $i$th 
standard basis vector on $\re^n$.
Moreover, if $\Gcb$ is an I-map for $\genm_\pi$, then
we denote an I-map for $\genm_\pi^\sigma$ by $\Gcb^\sigma$. 
Notice that $\Gcb^\sigma$ can be derived from $\Gcb$ simply 
by relabeling its nodes according to the 
permutation $\sigma$. 
Then we can cast a variational problem for the
{\it best} ordering $\sigma^\ast$ as:
\begin{eqnarray}  \label{bestOrdering}
	\sigma^* \in \, & {\rm arg \, max_\sigma} \quad    &  |  \sparse_S |   \\
	& {\rm s.t.} 		 \quad    & S\push \, \genm_\pi^\sigma = \genm_\eta      \nonumber  \\
	&  				    & \sigma \in \permset(\Nbb_n),   \nonumber 
\end{eqnarray}
where $S$ is the KR rearrangement that pushes forward
the reordered target $\genm_\pi^\sigma$ to $\genm_\eta$ and $\permset(\Nbb_n)$ is
the set of permutations of $\mathbb{N}_n$. 
The goal is to maximize the
cardinality of the sparsity pattern of the inverse map,
$| \sparse_S |$. 
We
restrict our attention to the sparsity of the inverse transport, since
we know from Section \ref{sec:sparsityBounds} that the direct transport tends
to be dense, even for the most trivial Markov structures.

Ideally, we would like to determine a good ordering for the map
\textit{before} computing the actual transport, and to use the
resulting information about the sparsity pattern to simplify the
optimization problem for $S$.
However, evaluating the objective function of \eqref{bestOrdering}
requires computing a different inverse transport for each permutation
$\sigma$.  One possible way to relax \eqref{bestOrdering} is to
replace $\sparse_S$ with the predicted sparsity pattern
$\widehat{\sparse}_S$ introduced in \eqref{eq:superSetInverse}.  The
advantage of this approach is that the objective function of the
relaxed problem can now be evaluated in closed form without computing
any transport map, but rather by performing the simple sequence of graph
operations on $\Gcb^\sigma$ %
described by Theorem
\ref{thm:sparsityRosenblatt}.  The caveat is that, in general,
$\widehat{\sparse}_S \subset \sparse_S$, and thus maximizing
$| \widehat{\sparse}_S |$ amounts to seeking the tightest lower bound on
the sparsity pattern of the inverse transport.
From the definition of $\widehat{\sparse}_S$, it follows that the best
ordering $\sigma^*$ for the \textit{relaxed} problem is one that introduces
the fewest edges in the construction of the 
marginal graphs $\Gcb^n,\ldots,\Gcb^1$, 
whenever $\Gcb^n=\Gcb^{\sigma^*}$.
Thus, for a given I-map $\Gcb$, we
denote by $\fillin (\sigma;\Gcb)$ the ${\it fill}$-${\it in}$ produced 
by the ordering $\sigma$.
That is, $\fillin (\sigma;\Gcb)$ is a set containing
all the edges introduced in the construction of the 
 marginal graphs $(\Gcb^k)$ from $\Gcb^\sigma$.
A computationally feasible
relaxation of \eqref{bestOrdering} is then given by:
\begin{eqnarray}  \label{bestOrderingGraph}
	\sigma^* \in \, & {\rm arg \, min_\sigma} \quad    &  | \fillin (\sigma;\Gcb) |   \\
	& {\rm s.t.} 		 \quad    
	& \sigma \in \permset(\Nbb_n)\, .  \nonumber 
\end{eqnarray} 
\eqref{bestOrderingGraph} is a standard problem in graph theory; it
arises in a variety of practical settings, including (most relatedly)
finding the best elimination ordering for variable elimination in
graphical models \cite{koller2009probabilistic}, or finding the
permutation that minimizes the fill-in of the Cholesky factor of a
positive definite matrix \cite{george1989evolution,saad2003iterative}.
From an algorithmic point of view, \eqref{bestOrderingGraph} is
NP-complete \cite{yannakakis1981computing}.  This should not be
surprising, as best--ordering problems are typically combinatorial in
nature.  Nevertheless, given its widespread applicability, a host of
effective polynomial-time heuristics for \eqref{bestOrderingGraph}
have been developed in past years (e.g., min-fill or
weighted-min-fill \cite{koller2009probabilistic}).  Most importantly,
\eqref{bestOrderingGraph} can be solved without 
ever 
touching the target measure
(assuming, of course, that an I-map $\Gcb$ for $\genm_\pi$ is
known).
As a result, the cost of finding a good ordering is often negligible
compared to the cost of characterizing a nonlinear transport map via optimization.

\begin{figure}[]
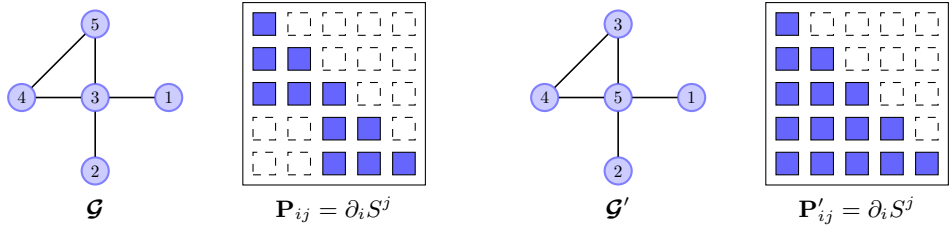

\centering
	\include{tikz_ordering}
\caption[]{
Illustration of how a simple re-ordering of the input variables can
change the (predicted) sparsity pattern of the inverse map.
On the left, $\Gcb$ represents an I-map for the  
target measure considered in Figure \ref{fig:marginalGraph}, with
ordering $(Z_1,Z_2,Z_3,Z_4,Z_5)$,
together with its sparsity pattern $\widehat{\sparse}_S$. (See Figure
\ref{fig:sparsityPatt} for details on the ``matrix'' representation of
sparsity patterns.) On the right, $\Gcb'$ is an I-map for the same target measure but with
the ordering $(Z_1,Z_2,Z_5,Z_4,Z_3)$. 
The corresponding sparsity pattern $\widehat{\sparse}_{S'}$ is now the empty set.
}
\label{fig:example_badOrder}
\end{figure}

%% file: tikz_ordering.tex
\newcommand{\wStarGraph}{0.4cm} 
\newcommand{\wStarGraphtwo}{1cm} 
\newcommand{\scaleStarGraph}{.6} 
\subfloat[$\Gcb$]{
\begin{tikzpicture}[transform shape, scale = \scaleStarGraph]
	\node[default]  (center) 										{$3$};
 	\node[default]	 (belowCenter)		  		[below=of center]		{$2$} 
 		edge[edgeStyle] (center);
 	\node[default] (aboveCenter)			 		[above=of center]		{$5$}
 		edge[edgeStyle] (center);
 	\node[default] (rightCenter)		  			[right=of center]			{$1$}
 		edge[edgeStyle] (center);
	\node[default] (leftCenter)		  			[left=of center]			{$4$}
 		edge[edgeStyle] (aboveCenter)
 		edge[edgeStyle] (center);		
\end{tikzpicture}
} \hspace{\wStarGraph}
\subfloat[${\bf P}_{ij}=\partial_{i}S^j$]{
\begin{tikzpicture}[transform shape, scale = \scaleStarGraph ]
\matrix [draw ,column sep=.15cm, row sep=.15cm, ampersand replacement=\& ]
	{
\node[entriesMatrix] { };\& \node[entriesMatrixZeroDashed]{ };\& \node[entriesMatrixZeroDashed]{ }; \& \node[entriesMatrixZeroDashed] { }; \& \node[entriesMatrixZeroDashed] { };   \\
\node[entriesMatrix] { };\& \node[entriesMatrix]{ };\& \node[entriesMatrixZeroDashed]{ }; \& \node[entriesMatrixZeroDashed] { }; \& \node[entriesMatrixZeroDashed] { };   \\
\node[entriesMatrix] { };\& \node[entriesMatrix]{ };\& \node[entriesMatrix]{ }; \& \node[entriesMatrixZeroDashed] { }; \& \node[entriesMatrixZeroDashed] { };   \\
\node[entriesMatrixZeroDashed] { };\& \node[entriesMatrixZeroDashed]{ };\& \node[entriesMatrix]{ }; \& \node[entriesMatrix] { }; \& \node[entriesMatrixZeroDashed] { };   \\
\node[entriesMatrixZeroDashed] { };  \& \node[entriesMatrixZeroDashed]{ };     \& \node[entriesMatrix]{ }; \& \node[entriesMatrix] { }; \& \node[entriesMatrix] { };    \\
	};
\end{tikzpicture}
}
\hspace{\wStarGraphtwo}
\subfloat[$\Gcb'$]{
\begin{tikzpicture}[transform shape, scale = \scaleStarGraph]
	\node[default]  (center) 										{$5$};
 	\node[default]	 (belowCenter)		  		[below=of center]		{$2$} 
 		edge[edgeStyle] (center);
 	\node[default] (aboveCenter)			 		[above=of center]		{$3$}
 		edge[edgeStyle] (center);
 	\node[default] (rightCenter)		  			[right=of center]			{$1$}
 		edge[edgeStyle] (center);
	\node[default] (leftCenter)		  			[left=of center]			{$4$}
 		edge[edgeStyle] (aboveCenter)
 		edge[edgeStyle] (center);		
\end{tikzpicture}
} \hspace{\wStarGraph}
\subfloat[${\bf P}_{ij}'=\partial_{i}S^j$]{
\begin{tikzpicture}[transform shape, scale = \scaleStarGraph ]
\matrix [draw ,column sep=.15cm, row sep=.15cm, ampersand replacement=\& ]
	{
\node[entriesMatrix] { };\& \node[entriesMatrixZeroDashed]{ };\& \node[entriesMatrixZeroDashed]{ }; \& \node[entriesMatrixZeroDashed] { }; \& \node[entriesMatrixZeroDashed] { };   \\
\node[entriesMatrix] { };\& \node[entriesMatrix]{ };\& \node[entriesMatrixZeroDashed]{ }; \& \node[entriesMatrixZeroDashed] { }; \& \node[entriesMatrixZeroDashed] { };   \\
\node[entriesMatrix] { };\& \node[entriesMatrix]{ };\& \node[entriesMatrix]{ }; \& \node[entriesMatrixZeroDashed] { }; \& \node[entriesMatrixZeroDashed] { };   \\
\node[entriesMatrix] { };\& \node[entriesMatrix]{ };\& \node[entriesMatrix]{ }; \& \node[entriesMatrix] { }; \& \node[entriesMatrixZeroDashed] { };   \\
\node[entriesMatrix] { };  \& \node[entriesMatrix]{ };     \& \node[entriesMatrix]{ }; \& \node[entriesMatrix] { }; \& \node[entriesMatrix] { };    \\
	};
\end{tikzpicture}
} 

%% file: sec_decomp.tex
Thus far, we have investigated the sparsity of triangular transport maps and
found that inverse transports tend to inherit sparsity from the
underlying Markov structure of the target measure.  Though direct
triangular transports also inherit some sparsity according to
Theorem~\ref{thm:sparsityRosenblatt}, they tend to be more dense.

This section shows that direct transports enjoy a different form of
low-dimensional structure: \textit{decomposability}.
A decomposable transport map is a function that can be written as the
composition of a finite number of low-dimensional maps, e.g., 
$T=T_1\circ \cdots \circ T_\ell$ for some integer $\ell \ge 2$.
We use a very specific notion of low-dimensional map, as follows.
\begin{definition}[Low-dimensional map with respect to a set] \label{def:lowdim}
A map $M : \re^n \ra \re^n$ is low-dimensional with respect to a
nonempty set $\Cc \subset \Vc \simeq \mathbb{N}_n$ if
\begin{enumerate} 
	\item  $M^k(\xb)=x_k$ for $k \in \Cc$
	\smallskip
	\item $\partial_j M^k = 0$ for $j\in \Cc$ and $k \in \Vc \setminus \Cc$.
\end{enumerate}
The effective dimension  of $M$ is the minimum 
cardinality $|\Vc \setminus \Cc|$ over all sets $\Cc$ 
with respect to %
which $M$
is low-dimensional.
\end{definition}
In particular,  
up to a permutation of its components, we can rewrite $M$ as:
\begin{equation}   
M( \xb ) = 
\left[\begin{array}{l}
M^{ \bar{\Cc}}( \xb_{\bar{\Cc}} )\\
\xb_{\Cc} 
\end{array}\right],
\end{equation}
where $\bar{\Cc} = \Vc \setminus \Cc$ denotes the complement of
$\Cc$ in $\Vc$, and where for any map $M$ and set $\Ac=\{a_1,\ldots,a_k\}$, $M^{\Ac}$ denotes the multivariate
function $\xb \mapsto (M^{a_1}(\xb),\ldots,M^{a_k}(\xb))$ obtained by stacking together
the components of $M$ with index in $\Ac$.
Thus $M$ is the trivial embedding of a $|\bar{\Cc}|$-dimensional function  
into the identity map and  has
{\it effective dimension} bounded by $|\bar{\Cc}|<n$. 
It is not surprising, then, that a decomposable transport $T=T_1\circ \cdots \circ T_\ell$
should be easier to represent than an ordinary map.
A perhaps less intuitive feature, however, 
is that the computation of a high-dimensional decomposable
transport can
be broken down into multiple simpler
steps, each associated with the computation of a low-dimensional map
$T_j$ that accounts only for local features of the target measure. %
The forthcoming analysis will consider \textit{general}, and hence
possibly non-triangular, transports. Thus its scope is much broader
than that of Section \ref{sec:sparse}, where we only focused on the %
sparsity of triangular transports.
Yet, we will show that  triangular maps
are the building block of decomposable transports.
The cornerstone of our analysis is
Theorem \ref{thm:decompTrans}, which characterizes the existence and
structure of decomposable transports given only the Markov structure
of the underlying target measure. 

Our discussion will proceed in two stages: first, we show how to
identify direct transports that decompose into two maps, i.e.,
$T=T_1\circ T_2$, and then we explain how to apply this result
recursively to obtain a general decomposition of the form
$T=T_1\circ \cdots \circ T_\ell$.

\subsection{Preliminary notions}
\label{sec:prelnotions}
Before addressing the decomposability of transport maps, we need to
introduce two useful concepts: proper graph decompositions and generalized
triangular functions. The decomposition of a graph is a standard
notion \cite{lauritzen1996graphical}.
\begin{definition}[Proper graph decomposition] \label{def:graphDec}
Given a graph $\Gcb=(\Vc,\Ec)$, a triple 
$(\Aset,\Sset,\Bset)$ of disjoint  subsets of the vertex set $\Vc$ forms a 
proper decomposition of $\Gcb$ if (1)  $\Vc = \Aset \cup \Sset \cup \Bset$, 
(2) $\Aset$ and $\Bset$ are nonempty,
(3) $\Sset$ separates $\Aset$ from $\Bset$, and (4) $\Sset$ is a clique. %
\end{definition}
See Figure
\ref{fig:graphSparsification} (\textit{top left}) for an example of
a decomposition.
Clearly, not every graph admits a proper decomposition; for
instance, a fully connected graph does not have a separator set for %
nonempty $\Aset$ and $\Bset$.  The idea we will
pursue here is that 
graph decompositions
lead to the existence of
decomposable transports.
\input{sec_generalTri}
\subsection{Decomposition and graph sparsification}
\label{sec:decompMainResults}
We now characterize transports that decompose into a pair of
low-dimensional maps, as described in the following theorem.
We formulate the theorem for a generic target measure
$\genm_i$.
Later we will apply the theorem recursively
to a sequence $(\genm_i)$ of different targets.

\begin{thm}[Decomposition of transport maps] \label{thm:decompTrans}
Let $\Xb \sim \genm_\eta$, $\Zb^i \sim \genm_i$, with
$\genm_\eta,\genm_i \in \borelmp(\re^n)$ and
$\genm_\eta$ tensor product measure. %
Denote by $\eta,\pi_i$ a pair of nonvanishing densities for 
$\genm_\eta$ and $\genm_\pi$, respectively, and
assume that
$\genm_i$ factorizes according to a graph 
$\Gcb^i$, which admits a proper decomposition
$(\Aset,\Sset,\Bset)$.
Then the following hold:
\begin{enumerate}
	\item  
		\label{thm:decompTrans_partFactorizationPi}
		There exists a factorization of $\pi_i$ of the form
		\begin{equation} \label{eq:formFactoriz}
		\pi_i(\zb)= 
		\frac{1}{\mathfrak{c}} \,
		\psi_{\Aset \cup \Sset}(\zb_{\Aset \cup \Sset}) \, 
		\psi_{\Sset \cup \Bset }(\zb_{ \Sset \cup \Bset }),
		\end{equation}
		where 
		$\psi_{\Aset \cup \Sset}$ is strictly positive and integrable, 
		with $\mathfrak{c} = \int \psi_{ \Aset \cup \Sset}$.
	\smallskip
		
	\item 
	\label{thm:decompTrans_partT}
	For any factorization \eqref{eq:formFactoriz}	 and for any permutation
		$\sigma$ of $\mathbb{N}_n$ with 
			\begin{equation} \label{eq:formPerm}
			\sigma(k)\in 
			\begin{cases} 
			\Sset &\mbox{if }  k = 1,\ldots,|\Sset|  \\ 
			\Aset & \mbox{if } k = |\Sset|+1,\ldots, |\Aset\cup \Sset| \\
			\Bset & \mbox{otherwise},
			\end{cases}
			\end{equation} 
	there exists a nonempty family, $\decset_i$,
	of decomposable transport maps $T = \lmap_i \circ \rmap$ parameterized
	by $\rmap \in \mathfrak{R}_i$ such that each $T\in \decset_i$
	pushes forward $\genm_\eta$ to $\genm_i$ and where: %
		\begin{enumerate}
			\item 
			\label{thm:decompTrans_partLeftMap}
			$\lmap_i$ is a $\sigma$-generalized KR rearrangement 
			that pushes forward $\genm_\eta$ to a measure with density
			$\psi_{\Aset \cup \Sset}(\zb_{\Aset \cup \Sset}) \, 
			\eta_{\Xb_{\Bset}}(\zb_{\Bset})/\mathfrak{c}$
			and
			is low-dimensional with respect to $\Bset$.

			\item 
			\label{thm:decompTrans_partRightMap}
			$\mathfrak{R}_i$ is the set of maps
			$\re^n\ra \re^n$ that are low-dimensional 
			with respect to $\Aset$ and 
			that push forward $\genm_\eta$ to the pullback 
			$\lmap_i\pull \, \genm_i \in  \borelmp(\re^n)$.
			\item 
			\label{thm:decompTrans_partRVimap}
			If 
			$\Zb^{i+1} \sim \lmap_i\pull \, \genm_i$, 
			then
			$\Zb^{i+1}_{\Aset} \orth \Zb^{i+1}_{\Sset \cup \Bset}$ and %
			$\Zb^{i+1}_{\Aset}=\Xb_{\Aset} $ in distribution.
			\item 
			\label{graphSparsification_partImap}
             	  $\lmap_i\pull \,\genm_i$ factorizes according to
             	  a graph $\Gcb^{i+1}$ that can be derived from $\Gcb^i$
                  as follows: 

            		\begin{itemize}[label={--}]
                		\item Remove any edge from $\Gcb^i$ that is
                		incident to any node in $\Aset$.  

                		\item
                		For any maximal clique $\Cc \subset \Sset \cup \Bset$ with
                		nonempty intersection
                		$\Cc \cap \Sset$, 
                		let $j_{\Cc}$ be the maximum integer $j$ such that 
                		$\sigma(j) \in  \Cc \cap \Sset$
                		and turn $\Cc \cup \{ \sigma(1) ,\ldots, \sigma(j_{\Cc})\}$ into a
                		clique.
            		\end{itemize}	 
		\end{enumerate}
\end{enumerate}
\end{thm}

We first look at the theorem for $i=1$ and let 
$\genm_1 \coloneqq \genm_\pi$ and $\Gcb^1 \coloneqq \Gcb$,
where 
$\genm_\pi$ denotes our usual target measure with I-map $\Gcb$ and where $(\Aset ,\Sset, \Bset)$
denotes a decomposition of $\Gcb$.

Among the infinitely many transport maps from $\genm_\eta$ to  $\genm_\pi$, 
Theorem \ref{thm:decompTrans} identifies a family of decomposable ones. 
The existence of these maps relies exclusively 
on the
Markov structure of $\genm_\pi$: we just require $\Gcb$ to 
admit a (proper) decomposition.\footnote{
To obtain a proper decomposition
of $\Gcb$, one is free to add edges to $\Gcb$ 
in order to turn the separator set $\Sset$ into a clique (see Definition \ref{def:graphDec});
$\genm_\pi$  still factorizes according to any less sparse
version of $\Gcb$.
}

Each transport $T \in \decset_1$ pushes forward $\genm_\eta$ to $\genm_\pi$ and is the composition of two low-dimensional maps, i.e., 
$T = \lmap_1 \circ \rmap$ for a fixed $\lmap_1$ defined in 
Theorem \ref{thm:decompTrans}[Part \ref{thm:decompTrans_partLeftMap}]
and for some $\rmap \in \mathfrak{R}_1$.
(We also write $\decset_1 \coloneqq \lmap_1 \circ \mathfrak{R}_1$.\footnote{
The notation here is intuitive: for a given $g:\re^n \ra \re^n$ and for a 
given set 
of functions $\Fc$ from $\re^n$ to $\re^n$, $g \circ \Fc$ denotes the set of
maps that can be written as $g \circ f$ for some $f \in \Fc$. 	
})
The structure of these low-dimensional maps is quite interesting.  
Up to a reordering of their components, 
Theorem \ref{thm:decompTrans}[Parts \ref{thm:decompTrans_partLeftMap} and 
\ref{thm:decompTrans_partRightMap}] show that
$\lmap_1$ and $\rmap$ have an
intuitive complementary form:
\newcommand{\spacinglines}{2pt}
\begin{equation}  \label{eq:sparsityPatternGenTriMaps}
\lmap_1( \xb ) = \left[\begin{array}{l}
\lmap_1^{\Aset}( \xb_{\Sset} , \xb_{\Aset} )\\[\spacinglines] 
\lmap_1^{\Sset}(\xb_{\Sset})\\[\spacinglines] 
\xb_{\Bset} 
\end{array}\right],\qquad %
\rmap( \xb ) = \left[\begin{array}{l}
\xb_{\Aset}\\[\spacinglines] 
\rmap^{\Sset}(\xb_{\Sset},\xb_{\Bset})\\[\spacinglines] 
\rmap^{\Bset}(\xb_{\Sset},\xb_{\Bset}) 
\end{array}\right].
\end{equation}
(If $\Sset = \emptyset$, one can just remove $\lmap_1^{\Sset}$ and
$\rmap^{\Sset}$ from \eqref{eq:sparsityPatternGenTriMaps}, 
and drop the dependence of the remaining components
on $\xb_{\Sset}$.)
In particular, $\lmap_1$ and $\rmap$ have effective dimensions 
bounded by 
$|\Aset \cup \Sset|$ and $|\Sset \cup \Bset|=|\Vc \setminus \Ac|$, respectively (see Definition \ref{def:lowdim}). 
Even though $\lmap_1$ and $\rmap$ are low-dimensional maps, 
their composition is quite dense---in the sense of Section \ref{sec:sparse}---and is
in general nontriangular:
\begin{equation}   
T( \xb ) = (\lmap_1 \circ \rmap)(\xb) = 
\left[\begin{array}{l}
\lmap_1^{\Aset}(\, \rmap^{\Sset}(\xb_{\Sset},\xb_{\Bset}) , \,\xb_{\Aset} )\\[\spacinglines] 
\lmap_1^{\Sset}(\, \rmap^{\Sset}(\xb_{\Sset},\xb_{\Bset})\,)\\[\spacinglines] 
\rmap^{\Bset}(\xb_{\Sset},\xb_{\Bset}) 
\end{array}\right],
\end{equation}
and thus more difficult to represent and to work with.
The key idea of decomposable transports is that they can be represented
implicitly through the composition of their low-dimensional factors, similar to the way that 
direct transports can be represented implicitly
through their sparse inverses (Section \ref{sec:sparse}).

The sparsity patterns of $\lmap_1$ and $\rmap$ in \eqref{eq:sparsityPatternGenTriMaps}
are needed for the theorem to hold.
In particular, $\lmap_1$ must be a $\sigma$-triangular function with $\sigma$ specified by 
\eqref{eq:formPerm}.
Notice that \eqref{eq:formPerm} does not prescribe an exact
permutation, but just a few constraints on a feasible $\sigma$.  
Intuitively,
these constraints say that $\lmap_1$ should be a function whose components
with indices in $\Sset$ depend only on the variables in $\Sset$ 
(whenever $\Sset \neq \emptyset$), and whose
components with indices in $\Aset$ depend only on the variables in
$\Aset \cup \Sset$.  
Thus,
there is usually some freedom in the choice of $\sigma$.  
Different permutations lead to different families of decomposable transports, and
can induce different sparsity patterns in an 
I-map, $\Gcb^2$, for $\lmap_1\pull\,\genm_\pi$ 
(Theorem \ref{thm:decompTrans}[Part \ref{graphSparsification_partImap}]).
Part \ref{graphSparsification_partImap} of the theorem 
shows how to derive a possible I-map $\Gcb^2$---not necessarily
minimal---by
performing a sequence of graph operations on
$\Gcb$.
There are two steps: one that does not depend on $\sigma$ and one that does.
Let us focus first on the former: the idea is to
remove from $\Gcb$ any edge that is incident to any
node in $\Aset$, effectively disconnecting $\Aset$ from the rest of the graph.
That is, if $\Zb^{2} \sim  \lmap_1\pull \genm_\pi$, then, regardless of $\sigma$, 
$\lmap_1$ makes $\Zb_{\Aset}^2$ marginally independent of
$\Zb_{\Sset \cup \Bset}^2$ by 
acting {\it locally} on $\Gcb$.
And not only that: $\lmap_1$ also ensures that the marginals of $\genm_\eta$ and
$\lmap_1\pull \genm_\pi$ agree along $\Aset$ 
(see Theorem \ref{thm:decompTrans}[Part \ref{thm:decompTrans_partRVimap}]).
Thus we should really interpret $\lmap_1$ as the first step towards a progressive
transport of %
$\genm_\eta$ to $\genm_\pi$.
$\lmap_1$ is a local map: it can
depend nontrivially only upon variables in $\xb_{\Aset \cup \Sset}$.
Indeed, in the most general case, $| \Aset \cup \Sset|$
is the minimum effective dimension of a
low-dimensional map  
necessary to decouple $\Aset$ from the rest of
the graph.  
The more edges incident to $\Aset$, the higher-dimensional a
transport is needed.
This type of {\it graph sparsification} requires a peculiar
``block triangular'' structure for $\lmap_1$ as shown by
\eqref{eq:sparsityPatternGenTriMaps}:
any $\sigma$-triangular function with $\sigma$ given by \eqref{eq:formPerm}
achieves this special structure.
The second step of Part
\ref{graphSparsification_partImap} shows that 
if $\Sset \neq \emptyset$, then it might be necessary to add edges
to the subgraph $\Gcb_{\Sset \cup \Bset}$, depending on $\sigma$.\footnote{
This is not always the case. 	
For instance, if $\Sset$ is a   subset of every maximal
clique of $\Gcb$ in $\Sset \cup \Bset$ that has nonempty intersection with $\Sset$, then,
by Theorem \ref{thm:decompTrans}[Part \ref{graphSparsification_partImap}], 
no edges need to be added.
}  
The relevant aspect of $\sigma$ for this discussion is
the definition of the permutation onto the first $|\Sset|$ integers. 
In general, there are  $|\Sset|!$ different permutations that could
induce different sparsity patterns in $\Gcb^2$.
We shall see that permutations that add the fewest edges possible are
of particular relevance.

\subsection{Recursive decompositions}
\label{sec:rec_dec}
The sparsity of $\Gcb^2$ is important because 
it affects the ``complexity'' of the maps in $\mathfrak{R}_1$:
each $R \in \mathfrak{R}_1$ pushes forward 
$\genm_\eta$ to $\lmap_1\pull\,\genm_\pi$.
More specifically, by the previous discussion, we can see how
the role of each $R \in \mathfrak{R}_1$ is
really only that of matching the marginals of $\genm_\eta$ and $\lmap_1\pull\,\genm_\pi$ 
along $\Vc \setminus \Ac$.
A natural question then is whether we can break this matching step into %
simpler tasks, or, in the language of this section, whether
 $\mathfrak{R}_1$ contains transports that are
further decomposable. %
Intuitively, we are seeking a finer-grained representation for some of the 
transports in $\mathfrak{R}_1$.
The following lemma (for $i=1$) provides a positive answer
to this question as long as
$\Vc \setminus \Ac$
is not fully connected in $\Gcb^2$.
From now on, we denote $(\Aset ,\Sset, \Bset)$ by 
$(\Aset_1 ,\Sset_1, \Bset_1)$, since we will be
dealing with a sequence of different graph decompositions.

\begin{lem}[Recursive decompositions] \label{lem:recursiveDecomp}
Let $\genm_\eta , \genm_i,\Gcb^i$ be defined as in the
assumptions of  Theorem \ref{thm:decompTrans}
for a proper decomposition 
$(\Aset_i,\Sset_i,\Bset_i)$ of  $\Gcb^i$, while let
$\Gcb^{i+1}$ and  
$\decset_i =  \lmap_i \circ \mathfrak{R}_i$ %
be the resulting  graph 
(Part \ref{graphSparsification_partImap}) and
family of decomposable transports,\footnote{
Whenever we do not specify a permutation $\sigma_i$ or a 
factorization \eqref{eq:formFactoriz} in the definition of $\lmap_i$, it means
that the claim holds true for any feasible choice of these
parameters.
}
respectively. 
Then there are two possibilities:
\begin{enumerate}
	\item 
	\label{lem:recursiveDecomp_partPos}
	$\Sset_i \cup \Bset_i$ is not a clique in $\Gcb^{i+1}$. 
	In this
	case, it is possible to identify a
	proper decomposition $(\Aset_{i+1},\Sset_{i+1},\Bset_{i+1})$ of $\Gcb^{i+1}$
	for some $\Aset_{i+1}$ that is a strict superset of $\Aset_i$ by
	(possibly) adding edges to $\Gcb^{i+1}$ in order
	to turn $\Sset_{i+1}$ into a clique. 
	Let $\decset_{i+1} = \lmap_{i+1} \circ \mathfrak{R}_{i+1}$ be defined as in 
	Theorem \ref{thm:decompTrans}
	for the pair of measures $\genm_\eta , \genm_{i+1} \coloneqq \lmap_i\pull \,\genm_i$ and \
	$(\Aset_{i+1},\Sset_{i+1},\Bset_{i+1})$. 
		Then the following hold:
		\begin{enumerate}
			\item 
			\label{lem:recursiveDecomp_partInclusion}
			$\mathfrak{R}_i \supset \decset_{i+1}$  
			and   
			$\lmap_i \circ \mathfrak{R}_i \supset 
			\lmap_i \circ \lmap_{i+1} \circ \mathfrak{R}_{i+1}$.

			\item
			\label{lem:recursiveDecomp_lmap}
			$\lmap_{i+1}$ is low-dimensional with respect to 
			$\Aset_{i} \cup \Bset_{i+1}$ and has effective dimension  bounded by
			$|(\Aset_{i+1} \setminus \Aset_{i}) \cup \Sset_{i+1}|$.
			\item 
			\label{lem:recursiveDecomp_rmap}
			Each $\rmap \in 
			\mathfrak{R}_{i+1}$ has effective dimension bounded by
			$|\Vc \setminus \Aset_{i+1}|$.
		\end{enumerate}

	\item
	\label{lem:recursiveDecomp_partNeg}
	$\Sset_i \cup \Bset_i$ is a clique in $\Gcb^{i+1}$.
	In this case, %
	the decomposition of Part \ref{lem:recursiveDecomp_partPos} does not exist.
\end{enumerate}

\end{lem}

Lemma \ref{lem:recursiveDecomp}[Part \ref{lem:recursiveDecomp_partPos}]
shows that if $\Sset_1 \cup \Bset_1$  is not
fully connected in $\Gcb^2$, then
there exists a
proper decomposition $(\Aset_2,\Sset_2,\Bset_2)$ of $\Gcb^2$ (obtained, possibly, by adding edges to $\Gcb^2$ in
$\Vc \setminus \Aset_1$) for which $\Aset_2$ is a strict superset of $\Aset_1$.
One can then apply Theorem \ref{thm:decompTrans}
for the pair  $\genm_\eta , \genm_{2} = \lmap_1\pull \,\genm_1$ and the
decomposition 
$(\Aset_{2},\Sset_{2},\Bset_{2})$.
As a result, Part \ref{lem:recursiveDecomp_partInclusion} of the lemma shows that
$\mathfrak{R}_1$ contains a subset 
$\decset_2 = \lmap_2 \circ \mathfrak{R}_2$ of decomposable
transport maps where both $\lmap_2$ and each $R \in \mathfrak{R}_2$ are
local transports on $\Vc \setminus \Aset_1$, i.e., they are both
low-dimensional with respect to $\Aset_1$. 
In particular, $\lmap_2$ is
responsible for decoupling $\Aset_{2}\setminus \Aset_{1}$ from the
rest of the graph and for matching the marginals of $\genm_\eta$ and
$\lmap_2 \pull\,\lmap_1 \pull \, \genm_\pi = (\lmap_1 \circ \lmap_2)\pull \,
\genm_\pi $ along $\Aset_{2}\setminus \Aset_{1}$.  The effective dimension
of $\lmap_2$ is bounded above by the size of the 
separator
set $\Sset_2$ plus the number of 
nodes in $\Aset_2 \setminus \Aset_1$
(Part \ref{lem:recursiveDecomp_lmap}
of the lemma). 
The effective dimension of each
$R \in \mathfrak{R}_2$ is bounded by the cardinality of
$\Vc \setminus \Aset_2$ and is, in the most general case, lower than
that of the maps in $\mathfrak{R}_1$ (Part \ref{lem:recursiveDecomp_rmap}).
Moreover, by Part
\ref{lem:recursiveDecomp_partInclusion},
$\decset_1 = \lmap_1 \circ \mathfrak{R}_1 \supset 
\lmap_1 \circ \lmap_{2} \circ \mathfrak{R}_{2}$, which means that among 
the infinitely many decomposable transports that push forward 
$\genm_\eta$ to $\genm_\pi$, there 
exists at least one that factorizes as the composition of
\textit{three} low-dimensional maps as opposed to two, i.e., $T=\lmap_1 \circ \lmap_2 \circ R$ for some
$R \in \mathfrak{R}_{2}$.

If, on the other hand, 
$\Sset_1 \cup \Bset_1$  is  
fully connected in $\Gcb^2$,
then by Lemma \ref{lem:recursiveDecomp}[Part \ref{lem:recursiveDecomp_partNeg}] 
we know that 
the decomposition of Part \ref{lem:recursiveDecomp_partPos}
does not exist.
As a result, we cannot use Theorem \ref{thm:decompTrans} to prove the
existence of more finely decomposable transports in $\mathfrak{R}_{1}$.
In other words, if we want to match the marginals of $\genm_\eta$ and 
$\lmap_1\pull\,\genm_\pi$ 
along $\Vc \setminus \Ac_1 = \Sset_1 \cup \Bset_1$, then 
we must do so in one shot, using a {\it single} transport map; no more
intermediate steps are feasible.
The main idea, then, is to apply
Lemma \ref{lem:recursiveDecomp}[Part \ref{lem:recursiveDecomp_partPos}], recursively, 
$k$ times, where $k$ is the first integer (possibly zero) for which
$\Sset_{k+1} \cup \Bset_{k+1}$ is a clique in 
$\Gcb^{k+2}$. %
After $k$ iterations, the following
inclusion must hold:
\begin{equation} \label{eq:longerInclusion}
	\decset_1 = \lmap_1 \circ \mathfrak{R}_1 \supset
	\lmap_1 \circ \cdots \circ \lmap_{k+1} \circ \mathfrak{R}_{k+1},
\end{equation}
which shows that there
exists a decomposable transport,
\begin{equation} \label{eq:longerComposition}
T=\lmap_1 \circ \cdots \circ \lmap_{k+1} \circ R,	
\end{equation}
for some $R \in \mathfrak{R}_{k+1}$, that pushes forward $\genm_\eta$ to $\genm_\pi$.
(Note that we can  apply 
Lemma \ref{lem:recursiveDecomp}[Part \ref{lem:recursiveDecomp_partPos}] 
only finitely many times 
since $|\Vc \setminus \Aset_{i+1}|$ is an integer function strictly
decreasing in $i$ and 
bounded away from zero.)
Each $\lmap_i$ in \eqref{eq:longerInclusion} 
is a $\sigma_i$-triangular map for some permutation $\sigma_i$
that satisfies \eqref{eq:longerComposition}, and is low-dimensional with respect to 
$\Aset_{i-1} \cup \Bset_{i}$, i.e., for $i>1$ and up
to a permutation of its components,
\begin{equation}  
\lmap_i( \xb ) = \left[\begin{array}{l}
\xb_{\Aset_{i-1}} \\[\spacinglines] 
\lmap_i^{\Aset_{i} \setminus \Aset_{i-1}}( \xb_{\Sset_i} , \xb_{\Aset_{i} \setminus \Aset_{i-1}} )\\[\spacinglines] 
\lmap_i^{\Sset_i}(\xb_{\Sset_i})\\[\spacinglines] 
\xb_{\Bset_i} 
\end{array}\right].
\end{equation}
The map $R$ is low-dimensional with respect to
$\Aset_{k+1}$ and can also be chosen as a generalized triangular function.
Intuitively, we can think of $\lmap_i$ as
decoupling nodes in $\Aset_i\setminus \Aset_{i-1}$ 
from the rest of the graph in an I-map for 
$(\lmap_1 \circ\cdots\circ \lmap_{i-1})\pull \genm_\pi$.
(Recall that by Lemma \ref{lem:recursiveDecomp} all the sets $(\Aset_i)$ are nested, i.e.,
$\Aset_1 \subset \cdots  \subset \Aset_{k+1}$.)
Figure \ref{fig:graphSparsification} 
illustrates the mechanics underlying the recursive application of 
Lemma \ref{lem:recursiveDecomp}.

We emphasize that the existence and structure
of \eqref{eq:longerComposition} follow from simple 
graph operations on $\Gcb$, and do not entail any actual
computation with the target measure $\genm_\pi$.
Notice also that even if each map in the decomposition \eqref{eq:longerInclusion}
is $\sigma$-triangular, the resulting transport map $T$ need not be
triangular at all.  In other words, we  obtain 
factorizations of general and possibly non-triangular
transport maps in terms of low-dimensional generalized triangular functions.
In this sense, we can regard triangular maps as a fundamental 
``building block'' of a much larger class of transport maps.

Decomposable transports are clearly not unique. 
In particular, there are two factors that affect both the sparsity pattern and the
number $k$ of composed maps
in the family $\lmap_1 \circ \cdots \circ \lmap_{k+1} \circ \mathfrak{R}_{k+1}$:
the sequence of decompositions $(\Aset_i,\Sset_i,\Bset_i)$ and the sequence
of permutations $(\sigma_i)$. 
Usually, there is a certain freedom in the choice of these
parameters, and each configuration might lead to a different family
of decomposable transports.
Of course some families might be more desirable than others:
ideally, we would like the low-dimensional maps in the
composition to have the smallest effective dimension possible.
Recall that by Lemma \ref{lem:recursiveDecomp}
the effective dimension of each $\lmap_i$ can be bounded
above by $|(\Aset_{i} \setminus \Aset_{i-1}) \cup \Sset_{i}|$ 
(with the convention $\Aset_{0} = \emptyset$).
Thus we should intuitively choose a decomposition
$(\Aset_i,\Sset_i,\Bset_i)$ of $\Gcb^i$ and a permutation $\sigma_i$ for $\lmap_i$ 
that
minimize the cardinality of $(\Aset_{i} \setminus \Aset_{i-1}) \cup \Sset_{i}$, and that,
at the same time, minimize the number of edges added from $\Gcb^i$
to $\Gcb^{i+1}$.
In principle, we
should also account for the dimensions of all future maps in the
recursion. In the most general case, this graph theoretic
question could be addressed using dynamic programming
\cite{bertsekas1995dynamic}.  In practice, however, we will often
consider graphs for which a {\it good} sequence of decompositions and
permutations is
rather obvious (see Section \ref{sec:dataAss}). 
For instance, if the target distribution $\genm_\pi$ %
factorizes according to 
a tree $\Gcb$, then it is immediate to 
show the existence of a decomposable transport $T=T_1 \circ \cdots \circ T_{n-1}$ 
that pushes forward $\genm_\eta$ to $\genm_\pi$ and that factorizes
as the composition of $n-1$ low-dimensional maps
$(T_i)_{i=1}^{n-1}$, each associated to an edge of $\Gcb$:
it suffices to consider a sequence of decompositions $(\Aset_i,\Sset_i,\Bset_i)$
with $\Aset_1 \subset \Aset_2 \subset \cdots$, where,
for a
given rooted version of $\Gcb$,
$\Aset_{i}\setminus \Aset_{i-1}$ consists of a single node $a_{i}$ with
the largest depth in $\Gcb_{\Vc \setminus \Aset_{i-1}}$, and where $\Sset_{i}$
contains the unique parent of that node.
Remarkably, each map $T_i$ has effective dimension less than or equal
to two, independent of $n$---the size
of the tree.

At this point, it might be natural to consider a junction tree decomposition of 
a triangulation of $\Gcb$ \cite{koller2009probabilistic} as a convenient
graphical tool to schedule the
sequence of decompositions $(\Aset_i,\Sset_i,\Bset_i)$ needed to apply
Lemma \ref{lem:recursiveDecomp} recursively. 
Decomposable graphs are in fact ultimately chordal \cite{lauritzen1996graphical}.
However, the situation might not be as straightforward.
The problem
is that the clique structure of $\Gcb^i$, an I-map for $\genm_i$, can be
{\it very} different than that of $\Gcb^{i+1}$, an I-map for $\lmap_i
\pull \, \genm_i$; Theorem 
\ref{thm:decompTrans}[Part \ref{graphSparsification_partImap}]
shows that  $\Gcb^{i+1}$ might contain larger maximal cliques than those in 
$\Gcb^i$,  even
if $\Gcb^i$ is chordal (see Figure \ref{fig:graphSparsification} for an
example). 
Thus, working with a junction tree might require a bit of extra care.

\subsection{Computation of decomposable transports}
\label{sec:compDecompTrans}
Given the existence and structure of a decomposable transport
like \eqref{eq:longerComposition}, what to do with it?
There are at least two possible ways of exploiting this type of information.
First, one could solve a variational
problem like \eqref{OptimDirect} 
and enforce an explicit
parameterization of the transport map as the composition
$T=\lmap_1 \circ \cdots \circ \lmap_{k+1} \circ R$. 	
 In this scenario, one need only
parameterize the low-dimensional maps $(\lmap_i, R)$ and optimize, jointly,
over their composition.  The advantage of this approach is that it
bypasses the parameterization of a single high-dimensional function,
$T$, altogether.
See the literature on normalizing flows \cite{rezende2015variational} 
for possible computational
ideas in this direction.  

An alternative---and perhaps more
intriguing---possibility is to compute the maps $(\lmap_i)$ sequentially
by solving \textit{separate} low-dimensional optimization problems---one for
each map $\lmap_i$.  
By Theorem \ref{thm:decompTrans}[Part \ref{thm:decompTrans_partLeftMap}]
and Lemma \ref{lem:recursiveDecomp}, 
there exists a factorization
\eqref{eq:formFactoriz} of
$\pi_i$---a density of $L_{i-1}\pull\,\genm_{i-1}$---for
which $\lmap_i$ is a $\sigma_i$-generalized
KR rearrangement that pushes forward $\genm_\eta$ to a measure
with density proportional to
$\psi_{\Aset_i \cup \Sset_i}  \, 
\eta_{\Xb_{\Bset_i}}$, where
$(\Aset_i , \Sset_i, \Bset_i)$ is a decomposition of
$\Gcb^i$ and $\Gcb^i$  is an I-map for $\genm_i$.
In general $\psi_{\Aset_i \cup \Sset_i}$ depends  on
$\lmap_{i-1}$, and so the maps $(\lmap_i)$ must be computed sequentially.\footnote{
This is not always the case.
For instance, given a rooted version of $\Gcb$ and a pair of consecutive {\it
  depths} (see the discussion at
the end of Section \ref{sec:rec_dec}), all the maps $(\lmap_i)$
associated with edges connecting nodes at these two depths can be computed in parallel.}
In essence, 
decomposable transports break the inference task into 
smaller and possibly easier steps.

Note that we could define $\lmap_i$ with respect to any
factorization \eqref{eq:formFactoriz} with
$\psi_{\Aset_i \cup \Sset_i}$ integrable:
these different factorizations
would lead to a family of decomposable transports with the same
low-dimensional 
structure and sparsity patterns %
(as predicted by Theorem \ref{thm:decompTrans}).
Thus, as long as we have access to a sequence of
integrable factors $(\psi_{\Aset_i \cup \Sset_i})$, we can 
compute each map $\lmap_i$ individually by solving a
low-dimensional  
optimization problem. (See Appendix \ref{sec:genKR} for computational
remarks on generalized triangular functions.) 
Intuitively, since by
Lemma \ref{lem:recursiveDecomp}[Part \ref{lem:recursiveDecomp_lmap}] 
$\lmap_i$ is low-dimensional with respect to $\Aset_{i-1} \cup \Bset_{i}$,
we really only need to optimize for a portion of the map, namely
$\lmap_i^{\Cc}$ for $\Cc = (\Aset_{i} \setminus \Aset_{i-1}) \cup \Sset_{i}$,
which can be regarded effectively as a multivariate map on 
$ \re^{|\Cc|} $. 
In the same way, the map $\rmap$ can be computed as any transport 
(possibly triangular) that pushes forward $\genm_\eta$ to 
$L_{k+1}\pull\,\genm_{k+1}$.
Theorem \ref{thm:decompTrans}[Part \ref{thm:decompTrans_partRightMap}] 
tells us that once again we only need to optimize for
a low-dimensional portion of the map, namely, $\rmap^{\Sset_{k+1} \cup \Bset_{k+1}}$.

While it might be difficult to access a sequence
of factorizations \eqref{eq:formFactoriz} for a general
problem, there are important applications, such as Bayesian
filtering, smoothing, and joint parameter/state estimation, where the 
sequential computation of the transports $(\lmap_i,R)$ is always possible by
construction.  
We discuss these applications in the next section.

\begin{figure}[]
\centering
	\include{tikz_graphSparsif_1}
	\include{tikz_graphSparsif_2}
        \caption[]{
        Sequence of graph decompositions associated with
        the recursive application of Lemma \ref{lem:recursiveDecomp}.
        At 
        \emph{(top left)} 
        is an I-map, $\Gcb^1$, for $\genm_\pi$, with
        $\genm_\pi \in \borelmp(\re^6)$. 
        We first
        decompose this graph into $(\Aset_1,\Sset_1,\Bset_1)$ as indicated,
        and apply Theorem \ref{thm:decompTrans} to the pair 
        $\genm_\eta, \genm_\pi$. 
        To do so, we first need to add edge $(2,3)$ to $\Gcb^1$ in order to turn
        $(\Aset_1,\Sset_1,\Bset_1)$ into a proper decomposition of
        $\Gcb^1$ with a fully connected $\Sset_1$.%
        The resulting graph, $\Gcb^{1}_\star$, is now chordal 
        (in fact, a triangulation of
        $\Gcb^1$ \cite{lauritzen1996graphical}), but still an 
        I-map for $\genm_\pi$.
        The first
        map $\lmap_1$ is $\sigma_1$-triangular with %
        $\sigma_1(\mathbb{N}_6)= \{2,3,1,4,5,6\}$ and it is
        low-dimensional with respect to $\Bset_1$;
        The \emph{(top right)} figure shows  the 
        I-map, $\Gcb^2$, for %
        $\lmap_1\pull \, \genm_\pi$ as
        given by Theorem \ref{thm:decompTrans}[Part \ref{graphSparsification_partImap}]:
        as expected, $\Aset_1$ is disconnected from $\Sset_1 \cup \Bset_1$;
        moreover, a new maximal clique $\{2,3,4,5\}$ appears in
        $\Gcb^2$.
        This new clique is larger than any of the maximal cliques in $\Gcb^{1}_\star$,
        even though
        $\Gcb^{1}_\star$ is chordal.
        (Notice that $\sigma_1$ is not the permutation that adds the fewest edges 
        possible in
        $\Gcb^2$. An example of such ``best'' permutation would be
        $\sigma(\mathbb{N}_6)= \{3,2,1,4,5,6\}$.)
        Though Theorem
        \ref{thm:decompTrans} guarantees the existence of a
        low-dimensional map $R \in \mathfrak{R}_1$ that pushes forward 
        $\genm_\eta$ to
        $\lmap_1\pull \, \genm_\pi$, we instead proceed recursively by
        applying Lemma  \ref{lem:recursiveDecomp}[Part \ref{lem:recursiveDecomp_partPos}] 
        for a proper
        decomposition, $(\Aset_2,\Sset_2,\Bset_2)$, of $\Gcb^2$, where 
        $\Aset_2$ is a strict superset of $\Aset_1$ \emph{(bottom left)}.
        The lemma shows that 
        $\mathfrak{R}_1 \supset \lmap_2 \circ \mathfrak{R}_2$
        for some $\sigma_2$-triangular map $\lmap_2$,
        which  
        is low-dimensional with respect to $\Aset_1 \cup \Bset_2$, and
        where each $\rmap \in \mathfrak{R}_2$ pushes forward $\genm_\eta$   to
        $(\lmap_1 \circ \lmap_2)\pull \genm_\pi$. 
        Can we apply Lemma  \ref{lem:recursiveDecomp} one more
        time to characterize decomposable transports in $\mathfrak{R}_2$?
        The answer is no, as the I-map for 
        $(\lmap_1 \circ \lmap_2)\pull \genm_\pi$ 
        \emph{(bottom right)} 
        consists of a
        single clique in $\Sset_2 \cup \Bset_2$. 
        Nevertheless, each $\rmap \in \mathfrak{R}_2$ is still low-dimensional
        with respect to $\Aset_2$. 
        Overall, we showed the existence of a transport
        map $T:\re^6 \ra \re^6$ pushing forward $\genm_\eta$ to $\genm_\pi$ that decomposes as
        $T=\lmap_1 \circ \lmap_2 \circ R$, where $\lmap_1$, $\lmap_2$, 
        $\rmap$ are effectively $\{3,4,3\}$-dimensional maps, respectively.

        }          
\label{fig:graphSparsification}
\end{figure}

%% file: sec_generalTri.tex
The notion of a generalized triangular function is perhaps less
standard, but still relatively straightforward:
\begin{definition}[Generalized triangular function] \label{def:genTri} 
  A
  function $T:\re^n\ra\re^n$ is said to be generalized triangular, or
  simply $\sigma$-triangular, if there exists a permutation $\sigma$
  of $\mathbb{N}_n$ such that the $\sigma(k)$th
  component of $T$ depends only on 
  the variables
  $x_{\sigma(1)},\ldots,x_{\sigma(k)}$,
  i.e., $T^{\sigma(k)}(\xb)=T^{\sigma(k)}(x_{\sigma(1)},\ldots,x_{\sigma(k)})$ for all
  $\xb=(x_1,\ldots,x_n)$ and for all $k=1,\ldots,n$.
\end{definition}
We can think of a generalized triangular function as a map that is
lower triangular up to a permutation. 
In particular, if $\sigma$ is the identity on $\mathbb{N}_n$, then a $\sigma$-triangular function
is simply a lower triangular map (see Section \ref{sec:compTransport}).
To represent the permutation $\sigma$,
we use the notation $\sigma(\{ i_1 ,\ldots , i_k \})=\{ \sigma(i_1) , \ldots , \sigma(i_k) \}$
to denote an ordered set that collects the action of the permutation
on the %
elements $(i_j)$. 
For example, if $T:\re^4\ra\re^4$ is a
$\sigma$-triangular map with $\sigma$ defined as $\sigma( \mathbb{N}_4 )=\{ 1,4,2,3 \} $,
then
$T$ will be of the form:
\begin{equation}  \label{eq:genTriMap}
T( \xb ) = \left[\begin{array}{l}
T^1(x_1)\\ 
T^2(x_1,x_4,x_2)\\ 
T^3(x_1,x_4, x_2, x_3) \\ 
T^4(x_1,x_4)
\end{array}\right]%
\end{equation}
for some collection $(T^k)$. 
We regard each component $T^{\sigma(k)}$ as a map %
$\re^k \ra \re$.
We say that a $\sigma$-triangular function 
$T$ is monotone increasing if 
each component $T^k$ is a monotone increasing
function of the input $x_k$.
Moreover, for any 
$\genm_{\eta}, \genm_{\pi} \in \borelmp(\re^n)$
and for any permutation $\sigma$ of $\mathbb{N}_n$, there
exists a 
($\genm_\eta$-unique) monotone increasing $\sigma$-triangular
map---which we call a
$\sigma$-generalized
KR rearrangement---that
pushes forward 
$\genm_{\eta}$ to
$\genm_{\pi}$.
We give a 
constructive
definition for a generalized
KR rearrangement in Appendix \ref{sec:genKR}.
A key
property of a $\sigma$-generalized KR rearrangement
is that it allows different
sparsity patterns to be engineered, depending on $\sigma$, in
a map that is
otherwise fully general---in the sense of being able to couple
arbitrary 
measures in $\borelmp(\re^n)$.
This feature will be essential to characterizing decomposable transport maps.

%% file: tikz_graphSparsif_1.tex
\newcommand{\wStarGraph}{1cm} 
\newcommand{\scaleStarGraph}{.8} 
\newcommand{\fontsizelabelsgraph}{\LARGE} 

\subfloat{
\begin{tikzpicture}[transform shape, scale = \scaleStarGraph]

	\begin{pgfonlayer}{bBack}
		\node[defaultBlank]  (center) {};	
		\node[defaultBlank] (dxCenter)	  			[right=of center]			{};		
	\end{pgfonlayer}		
									
 	\node[default]	 (blCenter)		  		[below=of center]		{$2$} 
 		;
 	\node[default] 	(abCenter)	 		[above=of center]		{$3$}
 		;
 		
	\node[default] 	(sxCenter)  			[left=of center]			{$1$}
 		edge[edgeStyle] (abCenter)
 		edge[edgeStyle] (blCenter);

 	\begin{pgfonlayer}{bBack}
		\node[defaultBlank]  (absxCenter) [above=of sxCenter] {};
		\node[defaultBlank]  (blsxCenter) [below=of sxCenter] {};
	\end{pgfonlayer}

	\node[default] (dxdxCenter) 	  			[right=of dxCenter]			{$6$}
 		;
	\node[default] 	(abdxCenter)  			[above=of dxCenter]			{ $4$ }
 		edge[edgeStyle] (abCenter)
 		edge[edgeStyle] (dxdxCenter); 	
	\node[default] 	(bldxCenter)  			[below=of dxCenter]			{ $5$}
 		edge[edgeStyle] (abdxCenter)
 		edge[edgeStyle] (blCenter)
 		edge[edgeStyle] (abCenter)
 		edge[edgeStyle] (dxdxCenter);

 	\begin{pgfonlayer}{bBack}
		\node[defaultBlank]  (abdxdxCenter) [above=of dxdxCenter] {};
		\node[defaultBlank]  (bldxdxCenter) [below=of dxdxCenter] {};
	\end{pgfonlayer}

	 \begin{pgfonlayer}{background}
		\node [fill=\colorA, inner sep=.35cm, rounded corners, inner sep=\valueInnerSep,
				label={[black,font=\fontsizelabelsgraph]below:$\Aset_1$}, 
			  fit=(blsxCenter)(absxCenter)] (setA) {};
		\node [fill=\colorB, inner sep=.35cm, rounded corners, inner sep=\valueInnerSep,
				label={[black,font=\fontsizelabelsgraph]below:$\Bset_1$}, 
			  fit=(abdxdxCenter)(bldxCenter)] (setB) {};	
		\node [fill=\colorS, inner sep=.35cm, rounded corners, inner sep=\valueInnerSep,
				label={[black,font=\fontsizelabelsgraph]below:$\Sset_1$}, 
			  fit=(abCenter)(blCenter)] (setS) {};
	\end{pgfonlayer}	
\end{tikzpicture}
}
\hspace{\wStarGraph}
\subfloat{
\begin{tikzpicture}[transform shape, scale = \scaleStarGraph]

	\begin{pgfonlayer}{bBack}
		\node[defaultBlank]  (center) {};	
		\node[defaultBlank] (dxCenter)	  			[right=of center]			{};		
	\end{pgfonlayer}		
									
 	\node[default]	 (blCenter)		  		[below=of center]		{$2$} 
 		;
 	\node[default] 	(abCenter)	 		[above=of center]		{$3$}
 		edge[edgeStyle] (blCenter);	
 		
	\node[default] 	(sxCenter)  			[left=of center]			{$1$}
 ;

 	\begin{pgfonlayer}{bBack}
		\node[defaultBlank]  (absxCenter) [above=of sxCenter] {};
		\node[defaultBlank]  (blsxCenter) [below=of sxCenter] {};
	\end{pgfonlayer}

	\node[default] (dxdxCenter) 	  			[right=of dxCenter]			{ $6$ }
 		;
	\node[default] 	(abdxCenter)  			[above=of dxCenter]			{ $4$ }
 		edge[edgeStyle] (abCenter)
 		edge[edgeStyle] (blCenter)
 		edge[edgeStyle] (dxdxCenter); 	
	\node[default] 	(bldxCenter)  			[below=of dxCenter]			{ $5$ }
 		edge[edgeStyle] (abdxCenter)
 		edge[edgeStyle] (blCenter)
 		edge[edgeStyle] (abCenter)
 		edge[edgeStyle] (dxdxCenter);

 	\begin{pgfonlayer}{bBack}
		\node[defaultBlank]  (abdxdxCenter) [above=of dxdxCenter] {};
		\node[defaultBlank]  (bldxdxCenter) [below=of dxdxCenter] {};
	\end{pgfonlayer}

	 \begin{pgfonlayer}{background}
		\node [fill=\colorA, inner sep=.35cm, rounded corners, inner sep=\valueInnerSep,
				label={[black,font=\fontsizelabelsgraph]below:$\Aset_1$}, 
			  fit=(blsxCenter)(absxCenter)] (setA) {};
		\node [fill=\colorB, inner sep=.35cm, rounded corners, inner sep=\valueInnerSep,
				label={[black,font=\fontsizelabelsgraph]below:$\Bset_1$}, 
			  fit=(abdxdxCenter)(bldxCenter)] (setB) {};	
		\node [fill=\colorS, inner sep=.35cm, rounded corners, inner sep=\valueInnerSep,
				label={[black,font=\fontsizelabelsgraph]below:$\Sset_1$}, 
			  fit=(abCenter)(blCenter)] (setS) {};
	\end{pgfonlayer}	
\end{tikzpicture}
} 	 	

%% file: tikz_graphSparsif_2.tex
\subfloat{
\begin{tikzpicture}[transform shape, scale = \scaleStarGraph]

	\begin{pgfonlayer}{bBack}
		\node[defaultBlank]  (center) {};	
		\node[defaultBlank] (dxCenter)	  			[right=of center]			{};		
	\end{pgfonlayer}		
									
 	\node[default]	 (blCenter)		  		[below=of center]		{$2$} 
 		;
 	\node[default] 	(abCenter)	 		[above=of center]		{$3$}
 		edge[edgeStyle] (blCenter);	
 		
	\node[default] 	(sxCenter)  			[left=of center]			{$1$}
 ;

 	\begin{pgfonlayer}{bBack}
		\node[defaultBlank]  (absxCenter) [above=of sxCenter] {};
		\node[defaultBlank]  (blsxCenter) [below=of sxCenter] {};
	\end{pgfonlayer}

	\node[default] (dxdxCenter) 	  			[right=of dxCenter]			{ $6$}
 		;
	\node[default] 	(abdxCenter)  			[above=of dxCenter]			{ $4$}
 		edge[edgeStyle] (abCenter)
 		edge[edgeStyle] (blCenter)
 		edge[edgeStyle] (dxdxCenter); 	
	\node[default] 	(bldxCenter)  			[below=of dxCenter]			{ $5$ }
 		edge[edgeStyle] (abdxCenter)
 		edge[edgeStyle] (blCenter)
 		edge[edgeStyle] (abCenter)
 		edge[edgeStyle] (dxdxCenter);

 	\begin{pgfonlayer}{bBack}
		\node[defaultBlank]  (abdxdxCenter) [above=of dxdxCenter] {};
		\node[defaultBlank]  (bldxdxCenter) [below=of dxdxCenter] {};
	\end{pgfonlayer}

	 \begin{pgfonlayer}{background}
		\node [fill=\colorA, inner sep=.35cm, rounded corners, inner sep=\valueInnerSep,
				label={[black,font=\fontsizelabelsgraph]below:$\Aset_2$}, 
			  fit=(absxCenter)(blCenter)] (setA) {};
		\node [fill=\colorB, inner sep=.35cm, rounded corners, inner sep=\valueInnerSep,
				label={[black,font=\fontsizelabelsgraph]below:$\Bset_2$}, 
			  fit=(abdxdxCenter)(bldxdxCenter)] (setB) {};	
		\node [fill=\colorS, inner sep=.35cm, rounded corners, inner sep=\valueInnerSep,
				label={[black,font=\fontsizelabelsgraph]below:$\Sset_2$}, 
			  fit=(bldxCenter)(abdxCenter)] (setS) {};
	\end{pgfonlayer}	
\end{tikzpicture}
}
\hspace{\wStarGraph}
\subfloat{
\begin{tikzpicture}[transform shape, scale = \scaleStarGraph]

	\begin{pgfonlayer}{bBack}
		\node[defaultBlank]  (center) {};	
		\node[defaultBlank] (dxCenter)	  			[right=of center]			{};		
	\end{pgfonlayer}		
									
 	\node[default]	 (blCenter)		  		[below=of center]		{$2$} 
 		;
 	\node[default] 	(abCenter)	 		[above=of center]		{$3$}
 		;	
 		
	\node[default] 	(sxCenter)  			[left=of center]			{$1$}
 ;

 	\begin{pgfonlayer}{bBack}
		\node[defaultBlank]  (absxCenter) [above=of sxCenter] {};
		\node[defaultBlank]  (blsxCenter) [below=of sxCenter] {};
	\end{pgfonlayer}

	\node[default] (dxdxCenter) 	  			[right=of dxCenter]			{$6$}
 		;
	\node[default] 	(abdxCenter)  			[above=of dxCenter]			{ $4$ }
 		edge[edgeStyle] (dxdxCenter); 	
	\node[default] 	(bldxCenter)  			[below=of dxCenter]			{ $5$ }
 		edge[edgeStyle] (abdxCenter)
 		edge[edgeStyle] (dxdxCenter);

 	\begin{pgfonlayer}{bBack}
		\node[defaultBlank]  (abdxdxCenter) [above=of dxdxCenter] {};
		\node[defaultBlank]  (bldxdxCenter) [below=of dxdxCenter] {};
	\end{pgfonlayer}


	\begin{pgfonlayer}{background}
		\node [fill=\colorA, inner sep=.35cm, rounded corners, inner sep=\valueInnerSep,
				label={[black,font=\fontsizelabelsgraph]below:$\Aset_2$}, 
			  fit=(absxCenter)(blCenter)] (setA) {};
		\node [fill=\colorB, inner sep=.35cm, rounded corners, inner sep=\valueInnerSep,
				label={[black,font=\fontsizelabelsgraph]below:$\Bset_2$}, 
			  fit=(abdxdxCenter)(bldxdxCenter)] (setB) {};	
		\node [fill=\colorS, inner sep=.35cm, rounded corners, inner sep=\valueInnerSep,
				label={[black,font=\fontsizelabelsgraph]below:$\Sset_2$}, 
			  fit=(bldxCenter)(abdxCenter)] (setS) {};
	\end{pgfonlayer}
\end{tikzpicture}
} 	 	

%% file: sec_dataAss.tex
In this section, we consider the problem of sequential Bayesian inference 
(or discrete-time data assimilation 
\cite{law2015data,reich2015probabilistic,sarkka2013bayesian})
for continuous, nonlinear, and non-Gaussian 
state-space models.

Our goal is to specialize the theory developed in Section
\ref{sec:decomp} %
to the solution of Bayesian filtering and smoothing problems.
The key result of this section is a new variational %
algorithm for characterizing the full posterior distribution of the
sequential inference problem---e.g., not just a few filtering or smoothing
marginals---via recursive lag--$1$ smoothing with transport maps. The
proposed algorithm builds a decomposable high-dimensional transport map in
a \emph{single forward pass} by solving a sequence of local
small-dimensional problems, without resorting to any backward pass on
the state space model (see Theorem \ref{thm:decompSmooth}).  These
results extend naturally to the case of joint parameter and state
estimation (see Section \ref{sec:joint} and Theorem
\ref{thm:decompJoint}).
\hrevone{
Pseudocode for the algorithm is given in Appendix \ref{sec:algo}.
}

A state-space model consists of a pair of discrete-time stochastic
processes $(\Zb_k,\Yb_k)_{k\ge 0}$ indexed by the time $k$, where
$(\Zb_k)$ is a latent
Markov process of interest
and where $(\Yb_k)$ is the observed process.
We can think of each $\Yb_k$ as a noisy and perhaps indirect measurement of $\Zb_k$.
The Markov structure corresponding to the joint process $(\Zb_k,\Yb_k)$ 
is shown in Figure \ref{fig:dataAssSmoothing}.
The generalization of the results of this section to the case of missing
observations is straightforward and will not be addressed here.

We assume that we are given the transition densities
$\pi_{\Zb_{k+1}\vert \Zb_k}$ for all
$k\ge0$, sometimes 
referred to as the ``prior dynamic,'' together with 
the marginal 
density
of the initial conditions $\pi_{\Zb_{0}}$.
(For instance, the prior dynamic could stem from the discretization of a 
continuous time stochastic differential equation
\cite{oksendal2013stochastic}.)
We denote by $\pi_{\Yb_k \vert \Zb_k }$ the likelihood function, i.e., the
density of $\Yb_k$ given $\Zb_k$, and
assume that $\Zb_k$ and
$\Yb_k$ are
random variables taking values on $\re^n$
and $\re^{\ddata}$, respectively.
Moreover, we denote by $(\yb_k)_{k \ge 0}$ a sequence of realizations of the
observed process $(\Yb_k)$ that will define the posterior
distribution over the unobserved ({hidden}) states of the model, and make the
following regularity assumption in our theorems:
$\pi_{\Zb_{0:k-1}, \Yb_{0:k-1}} > 0$
for all $k\ge 1$.
(The existence of underlying fully supported measures will be left implicit throughout
the section for notational convenience.)

\begin{figure}[h]
\centering
	\input{tikz_dataAssSmoothing}
\caption[Markov structure of a typical state-space model]{
\emph{(above)} I-map for the joint process $(\Zb_k,\Yb_k)_{k \ge 0}$ defining the
state-space model.
\emph{(below)} I-map for the independent reference 
process  $(\Xb_k)_{k \ge 0}$ used in Theorem \ref{thm:decompSmooth}.
}
\label{fig:dataAssSmoothing}
\end{figure}
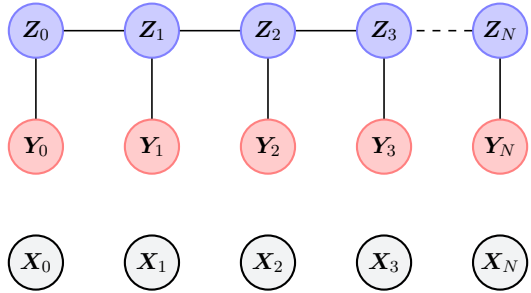

\subsection{Smoothing and filtering: the full Bayesian solution}  
\label{sec:filt}

In typical applications of state-space modeling, the process
$(\Yb_k)$ is only observed sequentially, and thus the goal of
inference is to characterize---{sequentially} in time and via a
{recursive} algorithm---the joint distribution of the current and past states
given currently available measurements, i.e.,
	\begin{equation} \label{eq:dataAssBayes}
	\pi_{\Zb_{0:k}  \vert  \yb_{0:k}}( \zb_{0:k} )
	\coloneqq 
		\pi_{\Zb_{0:k} \vert  
		\Yb_{0:k}}( \zb_{0:k} \vert \yb_{0:k} )		
	\end{equation}
for all $k\ge0$. 
That is, we wish to characterize %
$\pi_{\Zb_{0:k} \vert \yb_{0:k}}$
based on our knowledge of the posterior distribution at the previous timestep,
$\pi_{\Zb_{0:k-1} \vert \yb_{0:k-1}}$, and with an effort that is
constant over time.
We regard \eqref{eq:dataAssBayes} as the full Bayesian
solution to the sequential inference problem \cite{sarkka2013bayesian}. 

Usually, the task of updating $\pi_{\Zb_{0:k-1} \vert \yb_{0:k-1}}$
to yield $\pi_{\Zb_{0:k} \vert \yb_{0:k}}$
becomes increasingly challenging over time due to the widening
inference horizon, making characterization of the full Bayesian
solution impractical for large $k$ %
\cite{sarkka2013bayesian}.  
Thus, two simplifications of the sequential inference problem are
frequently considered: filtering and smoothing
\cite{sarkka2013bayesian}.  In filtering, we characterize
$\pi_{\Zb_{k} \vert \yb_{0:k}}$ for all $k\ge 0$, while in
smoothing we recursively
update $\pi_{\Zb_{j} \vert \yb_{0:k}}$
for increasing $k>j$, where $\Zb_j$ is some past
state of the unobserved process.
Both  filtering and  smoothing deliver
particular low-dimensional
\emph{marginals} of the full Bayesian solution to the inference problem, and
hence are often considered good candidates
for numerical approximation \cite{doucet2009tutorial,sarkka2013bayesian,crisan2002survey}.

The following theorem shows that
characterizing  the full Bayesian solution to the 
sequential inference problem via a decomposable transport map  
is essentially no harder than performing lag--$1$ smoothing, which,
in turn, amounts
to characterizing $\pi_{\Zb_{k-1},\Zb_{k} \vert \yb_{0:k}}$
for all $k\ge 0$
(an operation only incrementally harder than regular filtering).
This result relies on
the recursive application of the decomposition theorem for
couplings (Theorem \ref{thm:decompTrans}) 
to the {\it tree}
Markov structure of $\pi_{\Zb_{0:k} \vert \yb_{0:k}}$.
In what follows, let $(\Xb_k)_{k \ge 0}$ be an independent 
(reference) process with nonvanishing marginal densities %
$(\eta_{\Xb_k})$, with each $\Xb_k$ taking values on
$\re^n$. %
See Figure \ref{fig:dataAssSmoothing} for the corresponding Markov
network.
\begin{thm}[Decomposition theorem for state-space models] \label{thm:decompSmooth}
Let $(\submap_i)_{i\ge 0}$ be a
sequence of $(\sigma_i)$-generalized 
KR rearrangements on $\re^{n} \times \re^{n} $,
which are %
of the form
\begin{equation}  \label{eq:thm:decompSmooth:upperTriMap} 
\submap_i( \xb_i , \xb_{i+1} ) = \left[\begin{array}{l}
\submap_i^0( \xb_i , \xb_{i+1} ) \\[\spacinglines]
\submap_i^1(\xb_{i+1})
\end{array}\right] %
\end{equation}
for some $\sigma_i$, $\submap_i^0 : \re^{n} \times \re^{n} \ra \re^n$, 
$\submap_i^1 : \re^{n} \ra \re^n$, 
and that are defined by the recursion: 
	\begin{itemize}[label={--}]
		\item 
		$\submap_0$ pushes forward 
		$\eta_{\Xb_0,\Xb_{1}}$ to
		$\pi^0 = \widetilde{\pi}^0/ \mathfrak{c}_0$,
		\item 
		$\submap_i$
		pushes forward $\eta_{\Xb_i,\Xb_{i+1}}$ to
		$\pi^i(\zb_i, \zb_{i+1}) =
		\eta_{\Xb_i}(\zb_i) \,
		\widetilde{\pi}^i( \submap_{i-1}^1(\zb_i) , \zb_{i+1})/\mathfrak{c}_i$,		
	\end{itemize}
where 
$\mathfrak{c}_i$ is a normalizing constant
and where
$(\widetilde{\pi}^i)_{i \ge 0}$ 
are
functions on $\re^{n} \times \re^{n} $ given by: 
\begin{itemize}[label={--}]
	\item $\widetilde{\pi}^0(\zb_0, \zb_1) = \pi_{\Zb_0,\Zb_1}(\zb_0, \zb_1) \,
		\pi_{\Yb_0 \vert \Zb_0}(\yb_0 \vert \zb_0)  \, 
		\pi_{\Yb_1 \vert \Zb_1}(\yb_1 \vert \zb_1)$,
        
	\item $\widetilde{\pi}^i(\zb_i, \zb_{i+1})=
	   \pi_{\Zb_{i+1} \vert \Zb_i}( \zb_{i+1} \vert \zb_{i})  \,
		\pi_{\Yb_{i+1} \vert \Zb_{i+1} }(\yb_{i+1} \vert \zb_{i+1})$ for $i\ge 1$.	
\end{itemize}
Then, for all $k\ge 0$, the following hold:	
	\begin{enumerate}
		\item 
		\label{thm:decompSmooth:partFilt}
		The map $\submap_k^1$ pushes forward $\eta_{\Xb_{k+1}}$ to
		$\pi_{\Zb_{k+1} \vert \yb_{0:k+1}}$.
		\hfill [filtering]
		\item 
		\label{thm:decompSmooth:partSmot}
		The map $\overline{\submap}_k$, defined as ($\submap_{k-1}^1(\xb)=\xb$ for $k=0$)
		\begin{equation}    \label{eq:thm:decompSmooth:upperTriMapSmooth} 
			\overline{\submap}_k( \xb_k , \xb_{k+1} ) = \left[\begin{array}{l}
			\submap_{k-1}^1( \submap_k^0( \xb_k , \xb_{k+1} ) )\\[\spacinglines] 
			\submap_k^1(\xb_{k+1})
			\end{array}\right], %
		\end{equation}
		pushes forward 
		$\eta_{\Xb_k,\Xb_{k+1}}$ 
		to
		$\pi_{ \Zb_k ,\Zb_{k+1} \vert \yb_{0:k+1} }$. 
		\hfill
		[lag--$1$ smoothing]
		\item 
		\label{thm:decompSmooth:partFull}
		The composition of transport maps $\mathfrak{T}_k = T_0 \circ \cdots \circ T_k$,
		where each $T_i$ is defined 
		as
		\begin{equation}  \label{eq:thm:decompSmooth:upperTriMapEmb} 
		T_i( \xb_0 , \ldots, \xb_{k+1} ) = \left[\begin{array}{l}
		\xb_0 \\ 
		\vdots \\
		\xb_{i-1} \\[\spacinglines]
		\submap_i^0( \xb_i , \xb_{i+1}) \\[\spacinglines]
		\submap_i^1( \xb_{i+1}) \\
		\xb_{i+2} \\
		\vdots \\
		\xb_{k+1}
		\end{array}\right],%
		\end{equation}
		pushes forward 
		$\eta_{\Xb_{0:k+1}}$  
		to 
		$\pi_{\Zb_{0:k+1} \vert \yb_{0:k+1} }$. \hfill
		[full Bayesian solution]
		\item 
		\label{thm:decompSmooth:partEvidence}
		The model evidence (marginal likelihood) is given by 
		\begin{equation} \label{thm:decompSmooth:evidence}
		\pi_{\Yb_{0:k+1}}(\yb_{0:k+1}) 
		= \prod_{i=0}^k \mathfrak{c}_i.	
		\end{equation}		
	\end{enumerate}
\end{thm}
Theorem \ref{thm:decompSmooth} suggests a variational %
algorithm for smoothing and filtering a continuous state-space
model: compute the sequence of maps $(\submap_i)$, each of dimension
$2n$; embed them into higher-dimensional identity maps to form $(T_i)$
according to \eqref{eq:thm:decompSmooth:upperTriMapEmb}; 
then evaluate
the composition $\mathfrak{T}_k = T_0 \circ \cdots \circ T_k$ to
sample directly from $\pi_{\Zb_{0:k+1} \vert \yb_{0:k+1} }$ (i.e.,
the full Bayesian solution) and obtain information about any
smoothing or filtering distribution of interest.

Successive transports in the composition $(\mathfrak{T}_k)_{k \ge 0}$ are 
{\it nested} and thus
ideal for sequential assimilation:
given $\mathfrak{T}_{k-1}$, we can 
obtain
$\mathfrak{T}_{k}$ simply by computing an additional map $\submap_{k}$
of dimension $2n$---with no need
to recompute $(\submap_i)_{i < k}$. This step converts a transport map that samples $\pi_{\Zb_{0:k}\vert \yb_{0:k}}$ into
one
that samples 
$\pi_{\Zb_{0:k+1}\vert \yb_{0:k+1}}$.
This
feature is important since $\submap_{k}$ is always a $2n$-dimensional map, while
$\pi_{\Zb_{0:k+1}\vert \yb_{0:k+1}}$ is a density
on $\re^{n(k+2)}$---a space whose dimension increases with time $k$.
In fact, from the perspective of Section~\ref{sec:decomp}, Theorem
\ref{thm:decompSmooth} simply shows that each
$\pi_{\Zb_{0:k+1}\vert \yb_{0:k+1}}$ can be represented via a
\emph{decomposable} transport $\mathfrak{T}_k = T_0 \circ \cdots \circ T_k$.
The sparsity pattern of each map $\submap_i$, specified  in 
\eqref{eq:thm:decompSmooth:upperTriMap}, is necessary for Theorem \ref{thm:decompSmooth} 
to hold: $\submap_i$ cannot be {\it any} transport map; it must be
block upper triangular.

The proposed algorithm consists of a forward pass on the state-space model---wherein the sequence of transport maps $(\submap_i)$ are computed and stored---followed by
a backward pass where the composition %
$\mathfrak{T}_k = T_0 \circ \cdots \circ T_k$ is %
evaluated deterministically  to sample $\pi_{\Zb_{0:k+1} \vert
  \yb_{0:k+1}}$. %
This backward pass does not re-evaluate the potentials of the
state-space model (e.g., transition kernels or likelihoods) at earlier
times, nor does it perform any
additional
computation other than evaluating the maps $(\submap_i)$ in $\mathfrak{T}_k$.
Though each map $T_j$ is usually trivial to evaluate---e.g., the map might be
parameterized in terms of polynomials \cite{marzouk2016introduction}
and %
differ from the identity along only %
$2n$ components---it is
true that the cost of evaluating $\mathfrak{T}_k$ grows linearly
with $k$. %
 This is hardly surprising since
$\pi_{\Zb_{0:k+1} \vert \yb_{0:k+1}}$ is a density
over spaces of increasing dimension. A direct approximation of $\mathfrak{T}_k$ is
usually a bad idea since the map is high-dimensional and dense (in the sense
defined by Section \ref{sec:decomp}); it is better to store $\mathfrak{T}_k$
implicitly through the sequence of maps $(\submap_i)_{i\ge0}^k$, and 
sample smoothed trajectories by evaluating $\mathfrak{T}_k$
only when it is needed.
If we
are only interested in a particular smoothing marginal, 
e.g.,
$\pi_{\Zb_{0} \vert \yb_{0:k+1}}$ for all $k\ge 0$, then we can define
a general forward recursion to sample $\pi_{\Zb_{0} \vert \yb_{0:k+1}}$
with a {\it single} transport map that is updated recursively
over %
time, rather than with a growing composition of maps---and thus with 
a cost independent of $k$. This construction is given in 
Section \ref{sec:margSmooth}.

Also, it is important to emphasize that in order to assimilate a new measurement,
say $\yb_{k+1}$, we do \emph{not} need to evaluate the full composition  $\mathfrak{T}_{k-1}$; we only need
to compute a low-dimensional map $\submap_k$ whose target density $\pi^k$
depends only on $\submap_{k-1}$. The previous maps $(\submap_i)_{i<k-1}$
are unnecessary at this stage. Thus the effort of assimilating a new piece
of data is constant in time---modulo the complexity of each $\submap_k$.

The distribution $\pi_{\Zb_{0:k+1}\vert \yb_{0:k+1}}$ is {not}
represented via a  collection of particles as $k\ge 0$
increases, but rather via a growing composition of low-dimensional 
transport
maps
\textcolor{black}{
that yields {\it fully supported} approximations of $\pi_{\Zb_{0:k+1}\vert \yb_{0:k+1}}$.
}
These maps are computed via deterministic optimization:
there are no importance sampling or resampling steps.
Intuitively, the optimization step for $\submap_k$ {\it moves} the
particles on which
the map is evaluated, rather than reweighing them.
Part \ref{thm:decompSmooth:partFilt} of Theorem \ref{thm:decompSmooth}
shows that the lower subcomponent $\submap_k^1:\re^n \ra \re^n$ of the map
$\submap_k$ characterizes the filtering distribution
$\pi_{\Zb_{k+1} \vert \yb_{0:k+1}}$  for all $k\ge 0$, while
Part \ref{thm:decompSmooth:partSmot}
shows
that each $\submap_k$ also characterizes the lag--$1$ smoothing
distribution $\pi_{ \Zb_k ,\Zb_{k+1} \vert \yb_{0:k+1}}$ up to an invertible transformation of the marginal over $\Zb_k$.
Thus, Theorem \ref{thm:decompSmooth} 
implies a deterministic algorithm for lag--$1$ smoothing 
that in fact fully characterizes the posterior 
distribution of the nonlinear state-space model---much in the same spirit as the 
Rauch-Tung-Striebel (RTS) smoothing algorithm for Gaussian models.
We clarify this connection in Section \ref{sec:smoothLinGauss}.

\hrevone{A related perspective on the 
proposed smoothing algorithm is that
the composition of maps $\mathfrak{T}_k = T_0 \circ \cdots \circ T_k$ 
implements 
the following factorization of the full Bayesian solution,
\begin{equation}
\label{eq:backward_cond}
\pi_{\Zb_{0:k+1} \vert \yb_{0:k+1} }  =  \pi_{\Zb_{k+1} \vert \yb_{0:k+1}} \,
\pi_{\Zb_{k} \vert \Zb_{k+1}, \yb_{0:k}} \,  \pi_{\Zb_{k-1} \vert \Zb_{k}, \yb_{0:k-1}}  \cdots \, \pi_{\Zb_0 \vert \Zb_1, \yb_{0}}, 
\end{equation} 
wherein each map 
$\submap_i$, due to its block {\it upper} triangular structure, 
is associated with a specific
factorization of the 
lag--$1$ smoothing density,
\begin{equation}
\pi_{\Zb_{i+1}, \Zb_{i} \vert \yb_{0:i+1} } = \pi_{\Zb_{i+1} \vert \yb_{0:i+1}} \, \pi_{\Zb_{i} \vert \Zb_{i+1}, \yb_{0:i}}. 
\end{equation}
Evaluating $\mathfrak{T}_k$ on samples drawn from the reference process $\eta_{\Xb_{0:k+1}}$ amounts to sampling first from the final filtering marginal $\pi_{\Zb_{k+1} \vert \yb_{0:k+1}}$ and then from the sequence of ``backward'' conditionals
in \eqref{eq:backward_cond}.
See also \cite{kitagawa1987non,doucet2009tutorial,godsill2012monte} for alternative %
approximations of the forward--filtering backward--smoothing formulas.%
}

\hrevone{%
Note that the proposed approach does \emph{not} 
reduce to the ensemble Kalman filter (EnKF) or to the ensemble Kalman smoother
(EnKS) \cite{evensen2003ensemble,evensen2000ensemble}, even if the maps $\{\submap_k\}$ are constrained to be linear.
For one, the EnKF implements a two-step recursive approximation of each filtering marginal, which consists of  
(i) %
a particle approximation of the
``forecast'' distribution $\pi_{\Zb_{k+1} \vert \yb_{0:k}}$ obtained by {\it simulating} the transition kernel $\pi_{\Zb_{k+1}\vert\Zb_k}$, 
followed by 
(ii) a 
{\it linear} approximation of the forecast-to-analysis update (i.e., the update from $\pi_{\Zb_{k+1} \vert \yb_{0:k}}$ to $\pi_{\Zb_{k+1} \vert \yb_{0:k+1}}$).
In contrast, our approach constructs a recursive {\it variational} approximation of each lag--$1$ smoothing distribution,
essentially using
numerical optimization to minimize the KL divergence between
$\pi_{\Zb_{k}, \Zb_{k+1} \vert \yb_{0:k+1}}$ and its transport map approximation. 
We do not %
make a particle approximation of 
the forecast distribution 
by integrating the model dynamics, but instead require %
explicit evaluations of the transition density $\pi_{\Zb_{k+1}\vert\Zb_k}$.
 If, however, the dynamics of the state-space model are linear, with Gaussian transition/observational noise and Gaussian initial conditions, then the proposed algorithm is equivalent to
 filtering and 
 smoothing via
 ``exact'' Kalman formulas;
 in this case, the EnKF and EnKS can be interpreted as Monte Carlo approximations of the recursions defined by the proposed algorithm \cite{raanes2016ensemble}. 
}
\paragraph{Numerical approximations.}
\textcolor{black}
{ 
In general, the maps $(\submap_i)$ must be approximated numerically 
(see Section \ref{sec:compTransport}).
As a result, Monte Carlo estimators
associated with the evaluation of 
$\mathfrak{T}_k = T_0 \circ \cdots \circ T_k$ are {\it biased}, although possibly with negligible variance, since it is trivial to evaluate the map a large number of times.
This bias is {\it only} due to the numerical approximation of $(\submap_i)$, and not
to the particular factorization properties of $\mathfrak{T}_k$.
In practice, one might either accept this bias or try to reduce it.
The bias can be reduced in at least two ways:
(1) by enriching the parameterization of some $(\submap_i)$, and thus 
increasing the accuracy of the
  variational approximation,
 or (2) by using the
map-induced proposal density $(\mathfrak{T}_k)\push \eta_{\Xb_{0:k+1}}$---i.e.,  
the pushforward of a marginal of the reference process through $\mathfrak{T}_k$---within importance sampling or MCMC (see Section \ref{sec:numerics}).
For instance, the weight function
\begin{equation}
	w^{k+1}(\xb) = \frac{\pi_{\Zb_{0:k+1} \vert \yb_{0:k+1}}(\xb)}{ (\mathfrak{T}_k)\push \eta_{\Xb_{0:k+1}}(\xb)}
\end{equation}
is readily available, and can be used to yield consistent estimators with
respect to the smoothing distribution.
However, the resulting weights cannot be computed recursively in time, because even though
the small dimensional maps $\submap_k$ are computed sequentially, the map-induced proposal
$(\mathfrak{T}_k)\push \eta_{\Xb_{0:k+1}}$
changes entirely at every step.
}

In particle filters, the complexity of approximating the underlying
distribution is given by the number of particles $N$.
In the proposed variational approach, the complexity of the
approximation depends on the parameterization of each map $\submap_i$. 
There is no single parameter like $N$ to describe the complexity of the latter---though, broadly, it should depend on the number of degrees of freedom in the parameterization.
In some cases, one might think of using the total order of a multivariate
polynomial expansion of each component of the map as a tuning parameter.
But this is far from general or practical in high dimensions.
The virtue of a functional representation of the transport map is the ability
to carefully select the degrees of freedom of the parameterization.
For instance, we might model local interactions between different 
groups of input variables using different approximation orders or even different 
sets of basis functions.
This freedom should {not} be frightening, but rather embraced as a
rich opportunity to exploit %
the structure
of the particular problem at hand.
(See \cite[Chapter 6]{spantini2017inference} for an  
example of this practice in the context of filtering high-dimensional
spatiotemporal processes with chaotic dynamics.)
\hrevone{In general, richer parameterizations of the maps are more
  costly to characterize because they lead to higher-dimensional
  optimization problems \eqref{OptimDirectSAA}.  Yet, richer
  parameterizations can yield arbitrarily accurate results.  There is
  clearly a tradeoff between computational 
  cost
  and statistical
  accuracy.
We investigate this tradeoff numerically in Section
\ref{sec:numerics}, where we report the cost of computing a transport
map under different parameterizations and inference scenarios.}

\hrevone{Another important note: the {\it sequential} approximation of
  the individual maps $(\submap_i)$ might present additional
  challenges due to the accumulation of error, since the target
  density for the $k$-th map $\submap_k$ depends on the numerical
  approximation of the previous map, $\submap_{k-1}$. This is not an
  issue with the factorization of $\mathfrak{T}_k$ per se, but rather
  with sequentially computing each element of the factorization. The
  analysis of sequential Monte Carlo methods (e.g.,
  \cite{crisan2002survey,del2004feynman,smith2013sequential})
  addresses a similar accumulation of error, but 
has not yet been
    extended to
  sequential variational inference techniques. In Section
  \ref{sec:numerics}, we empirically investigate the stability of
  variational transport map approximations for a problem of very long
  time smoothing (see Figure
  \ref{fig:stoc-vol:very-long-smooth-vs-unbiased}), showing excellent
  results---at least for the reconstruction of low-order smoothing
  marginals.}

As shown in \eqref{eq:normConst}, the computation of each %
$\submap_i$ is also associated with an approximation of the normalizing constant $\mathfrak{c}_i$ of
its own target density, which then leads to a one-pass approximation of the
marginal likelihood using \eqref{thm:decompSmooth:evidence}.

One last remark: 
the proof of Theorem \ref{thm:decompSmooth} shows
that 
the triangular structure hypothesis for each $\submap_i$
can be relaxed provided
that the underlying densities are regular enough. The following
corollary clarifies this point.
\begin{cor} \label{cor:genDecompThmChains}
The results of Theorem \ref{thm:decompSmooth} still hold  if we replace
every KR rearrangement $\submap_i$ with a ``block triangular''  diffeomorphism of the form
\eqref{eq:thm:decompSmooth:upperTriMap}
that
couples the same distributions, provided
that such  regular transport maps exist.
\end{cor}
Filtering and smoothing are of course very rich problems, and in this
section we have by no means attempted to be exhaustive. Rather, our
goal was to highlight some 
implications of decomposable transports
on problems of %
sequential Bayesian inference, in a general non-Gaussian setting.  

\subsection{The linear Gaussian case: connection with the 
RTS smoother}
\label{sec:smoothLinGauss}
\input{sec_linGaussSmooth}

\subsection{Sequential joint parameter and state estimation} %
\label{sec:joint}
In defining a state-space model, it is common to parameterize the
transition densities of the unobserved process or the likelihoods of
the observables in terms of some hyperparameters $\vhyp$. The Markov
structure of the resulting Bayesian hierarchical model, conditioned on
the data, is shown in Figure \ref{fig:dataAssHyper}.
The state-space model is
now fully specified in terms of the conditional densities
$(\pi_{\Yb_{k}\vert\Zb_k, \vhyp})_{k\geq 0}$,
$(\pi_{\Zb_{k+1}\vert\Zb_k, \vhyp})_{k\geq 0}$,
$\pi_{\Zb_0 \vert \vhyp}$, and the marginal $\pi_{\vhyp}$. We assume
that the hyperparameters $\vhyp$ take values on $\re^{\dhyp}$, and that
the following regularity conditions hold:
$\pi_{\vhyp,\Zb_{0:k-1}, \Yb_{0:k-1}} > 0$
for all $k\ge 1$.

Given such a parameterization, one often wishes to \textit{jointly}
infer the hidden states and the hyperparameters of the model as
observations of the process $(\Yb_k)$ become available. That is, the
goal of inference is to characterize, %
via a {\it recursive} algorithm, the sequence of posterior
distributions given by
	\begin{equation} \label{eq:dataAssBayesJoint}
		\pi_{\vhyp,\Zb_{0:k}  \vert  \yb_{0:k}}(\zb_{\vhyps}, \zb_{0:k} )
	\coloneqq 
		\pi_{\vhyp, \Zb_{0:k} \vert  
		\Yb_{0:k}}(\zb_{\vhyps}, \zb_{0:k} \vert \yb_{0:k} )		
	\end{equation}
for all $k\ge0$ and for a sequence $(\yb_k)_{k \ge 0}$ of observations.
The following theorem shows that we can characterize \eqref{eq:dataAssBayesJoint} 
by computing a 
sequence of low-dimensional transport maps in the same spirit as 
Theorem \ref{thm:decompSmooth}.
In what follows, let $(\Xb_k)$ be an independent 
process with marginals
$(\eta_{\Xb_k})$ as defined in Theorem \ref{thm:decompSmooth} and let
$\Xb_{\vhyp}$ be a  random variable
on $\re^p$ that is independent of $(\Xb_k)$ and with
nonvanishing density $\eta_{\Xb_{\vhyp}}$.
\begin{figure}[H]
\centering
	\input{tikz_dataAssHyper}
\caption[Markov structure of a typical state-space model with static parameters]{
I-map for 
$\pi_{\vhyp,\Zb_0 , \ldots, \Zb_N \vert \yb_0 , \ldots, \yb_N}$, for any
$N>0$.
}
\label{fig:dataAssHyper}
\end{figure}

\begin{thm}[Decomposition theorem for joint parameter and state estimation] \label{thm:decompJoint}
Let $(\submap_i)_{i \ge 0}$ be a sequence of $(\sigma_i)$-generalized 
KR rearrangements on $\re^{\dhyp} \times \re^n \times \re^n$, which are of the
form
\begin{equation}   \label{eq:thm:decompJoint:genTriMap} 
\submap_i( \xb_{\vhyps}, \xb_i , \xb_{i+1} ) = \left[\begin{array}{l}
\submap_i^{\vhyp}( \xb_{\vhyps} ) \\[\spacinglines]
\submap_i^0( \xb_{\vhyps}, \xb_i , \xb_{i+1} ) \\[\spacinglines] 
\submap_i^1( \xb_{\vhyps}, \xb_{i+1})
\end{array}\right] %
\end{equation}
for some $\sigma_i$, $\submap_i^{\vhyp} : \re^{\dhyp} \ra \re^\dhyp$,
$\submap_i^0 : \re^{\dhyp} \times \re^n \times \re^n \ra \re^n$, 
$\submap_i^1 : \re^{\dhyp} \times \re^{n} \ra \re^n$, and
that are defined by the recursion:
	\begin{itemize}[label={--}]
		\item  
		$\submap_0$ pushes forward 
		$\eta_{\Xb_{\vhyp},\Xb_0,\Xb_{1}}$ to
		\begin{equation} \label{eq:target_stepmap_0}
		\pi^0 = \widetilde{\pi}^0/ \mathfrak{c}_0,
		\end{equation}
		\item 
		$\submap_i$
		pushes forward $\eta_{\Xb_{\vhyp},\Xb_i,\Xb_{i+1}}$ to
		\begin{equation} \label{eq:target_stepmap}
		\pi^i(\zb_{\vhyps},\zb_i, \zb_{i+1}) =
		\eta_{\Xb_{\vhyp},\Xb_i}(\zb_{\vhyps},\zb_i) \, 
		\widetilde{\pi}^i(\mathfrak{T}^{\vhyp}_{i-1}(\zb_{\vhyps}),
                \,
		\submap_{i-1}^1(\zb_{\vhyps},\zb_i), \zb_{i+1})/\mathfrak{c}_i,
		\end{equation}		
	\end{itemize}
where 
$\mathfrak{c}_i$ is a normalizing constant,
the map
$\mathfrak{T}^{\vhyp}_{j} \coloneqq \submap_{0}^{\vhyp} \circ \cdots
\circ \submap_{j}^{\vhyp}$ for all $j \ge 0$, and where
$(\widetilde{\pi}^i)_{i \ge 0}$ 
are
functions on 
$\re^{\dhyp} \times \re^{n} \times \re^{n} $ given by: 
	\begin{itemize}[label={--}]
		\item  
		$\widetilde{\pi}^0(\zb_{\vhyps},\zb_0, \zb_{1})  
		= \pi_{\vhyp,\Zb_0,\Zb_1}(\zb_{\vhyps},\zb_0, \zb_{1}) \,
		\pi_{\Yb_0 \vert \Zb_0, \vhyp}(\yb_0 \vert \zb_0, \zb_{\vhyps}) \, 
		\pi_{\Yb_1 \vert \Zb_1, \vhyp}(\yb_1 \vert \zb_1, \zb_{\vhyps})$,
		\item 
		$\widetilde{\pi}^i(\zb_{\vhyps},\zb_i, \zb_{i+1}) 
		= \pi_{\Zb_{i+1} \vert \Zb_i, \vhyp}(\zb_{i+1}\vert \zb_i, \zb_{\vhyps}) \,
		\pi_{\Yb_{i+1} \vert \Zb_{i+1} , \vhyp }(\yb_{i+1}\vert \zb_{i+1}, \zb_{\vhyps})$ 
		for $i\ge 1$.
	\end{itemize}
Then, for all $k\ge 0$, the following hold:	
	\begin{enumerate}
		\item
		\label{thm:decompJoint:partFilt}
		The map 
		$\widetilde{\submap}_k$,
		defined as
		\begin{equation}    \label{eq:thm:decompJoint:TriMapFilt} 
			\widetilde{\submap}_k( \xb_{\vhyps} , \xb_{k+1} ) = 
			\left[\begin{array}{l}
			\mathfrak{T}^{\vhyp}_{k}(  \xb_{\vhyps} )\\[\spacinglines] 
			\submap_k^1(  \xb_{\vhyps}, \xb_{k+1})
			\end{array}\right],%
		\end{equation}
		pushes forward 
		$\eta_{\Xb_{\vhyp}, \Xb_{k+1}}$ 
		to
		$\pi_{ \vhyp ,\Zb_{k+1} \vert \yb_{0}, \ldots, \yb_{k+1}}$. 
		\hfill [filtering]
		\item
		\label{thm:decompJoint:partFull}
		The composition of transport maps
		$\mathfrak{T}_k = T_0 \circ \cdots \circ T_k$,
		where each $T_i$ is defined as
		\begin{equation}  \label{eq:thm:decompJoint:genTriMapEmb}
		T_i( \xb_{\vhyps}, \xb_0 , \ldots, \xb_{k+1} ) = \left[\begin{array}{l}
		\submap_i^{\vhyp}( \xb_{\vhyps}) \\
		\xb_0 \\ 
		\vdots \\
		\xb_{i-1} \\[\spacinglines]
		\submap_i^0( \xb_{\vhyps}, \xb_i , \xb_{i+1}) \\[\spacinglines]
		\submap_i^1( \xb_{\vhyps}, \xb_{i+1}) \\
		\xb_{i+2} \\
		\vdots \\
		\xb_{k+1}
		\end{array}\right],
		\end{equation}
		pushes forward
		$\eta_{\Xb_{\vhyp},\Xb_{0}, \ldots, \Xb_{k+1}}$ to 
		$\pi_{\vhyp,\Zb_{0},\ldots, \Zb_{k+1} \vert \yb_{0},\ldots, \yb_{k+1} }$.
		\hfill [full Bayesian solution]
		\item 
		\label{thm:decompJoint:partEvidence}
		The model evidence (marginal likelihood) is given by \eqref{thm:decompSmooth:evidence}. 
	\end{enumerate}
\end{thm}

Theorem \ref{thm:decompJoint} suggests a variational %
algorithm
for the joint parameter and state estimation problem that is
similar to the one proposed in Theorem \ref{thm:decompSmooth}: 
compute the sequence of maps $(\submap_i)$, each of dimension $2n+p$;
embed them into
higher-dimensional identity maps to form $(T_i)$ according to
\eqref{eq:thm:decompJoint:genTriMapEmb};
then evaluate the composition
$\mathfrak{T}_k = T_0 \circ \cdots \circ T_k$ to sample directly from
$\pi_{\vhyp,\Zb_{0:k+1} \vert \yb_{0:k+1} }$ (i.e., the full Bayesian
solution). %
\hrevone{ 
See Appendix \ref{sec:algo} for more details.
}
Each map $\submap_i$ is now of dimension twice that of the
model state plus the dimension of the hyperparameters. This dimension is
slightly higher than that of the maps $(\submap_i)$ considered
in Theorem \ref{thm:decompSmooth},
and should be regarded as the price to pay
for introducing hyperparameters in the state-space
model and having to deal with the Markov structure of Figure \ref{fig:dataAssHyper}
as opposed to the tree structure of Figure \ref{fig:dataAssSmoothing}.
By Theorem \ref{thm:decompJoint}[Part \ref{thm:decompJoint:partFilt}], the composition 
of maps 
$\mathfrak{T}^{\vhyp}_{k} = \submap_{0}^{\vhyp} \circ \cdots \circ \submap_{k}^{\vhyp}$
provides a recursive  characterization of the posterior distribution over the static 
parameters, $\pi_{ \vhyp \vert \yb_{0:k+1}}$, for all $k \ge 0$.
The latter is often the ultimate goal of inference \cite{andrieu2010particle}.
\textcolor{black}{ 
In order to have a sequential algorithm for
parameter estimation, we also need to keep
a running approximation of 
$\mathfrak{T}^{\vhyp}_{k}$ 
using the recursion $\mathfrak{T}^{\vhyp}_{k} = \mathfrak{T}^{\vhyp}_{k-1} \circ 
\submap_{k}^{\vhyp}$---e.g., via regression---so that the
cost of evaluating $\mathfrak{T}^{\vhyp}_{k}$ does not grow with $k$.
}

\textcolor{black}{ 
Even in the joint parameter and state estimation case, only a single
forward pass with local computations is necessary to 
gather
all the
information from
the state-space model needed to sample the collection of posteriors
$(\pi_{\vhyp,\,\Zb_{0:k+1} \vert \yb_{0:k+1} })$.
Notice that the accuracy of the variational procedure is only limited by the
accuracy of each computed map, and that the proposed approach does not prescribe
an
artificial dynamic for the  parameters
\cite{kitagawa1998self,liu2001combined}, or an \textit{a priori} fixed-lag
smoothing approximation \cite{polson2008practical}.
Yet there is no rigorous proof that the performance of the
proposed sequential algorithm for parameter estimation does not deteriorate with time.
Indeed, developing exact, sequential, and online algorithms for parameter estimation 
in general non-Gaussian state-space models is among the chief research challenges in  SMC methods \cite{jacob2015sequential}. 
See \cite{chopin2013smc2,crisan2013nested,del2017biased} for recent contributions in this direction 
and \cite{kantas2015particle}
for a review of SMC approaches to Bayesian parameter inference.
See also \cite{erol2017nearly} for a hybrid approach that combines elements of variational
inference with particle filters.
}

We refer the reader to Section \ref{sec:numerics} for 
a numerical illustration of parameter inference with transport maps involving a stochastic
volatility model. %

\subsection{Fixed-point smoothing}
\label{sec:margSmooth}
Consider again the problem of sequential inference in
a state-space model without static parameters (see Figure
\ref{fig:dataAssSmoothing}), and suppose that we are interested
only in the smoothing marginal $\pi_{\Zb_0 \vert \yb_{0:k}}$ for all
$k \ge 0$; this is the fixed-point smoothing problem \cite{sarkka2013bayesian}.  

In Section \ref{sec:filt} we showed that computing a sequence of maps
$(\submap_i)$---each of dimension $2n$---is sufficient to sample the
joint distribution $\pi_{\Zb_{0:k+1} \vert \yb_{0:k+1}}$ by evaluating
the composition $\mathfrak{T}_k=T_0 \circ \cdots \circ T_k$, where
each $T_i$ is a trivial embedding of $\submap_i$ into an identity map.
If we can sample $\pi_{\Zb_{0:k+1} \vert \yb_{0:k+1}}$, then it is
easy to obtain samples from the marginal
$\pi_{\Zb_{0} \vert \yb_{0:k+1}}$: in fact, it suffices to evaluate
only the first $n$ components of $\mathfrak{T}_k$, which can be
interpreted as a map from $\re^{n\times (k+2)}$ to $\re^n$. To do so,
however, we need to evaluate $k$ maps.
A natural question then is whether it is possible to 
characterize $\pi_{\Zb_{0} \vert \yb_{0:k+1}}$
via a {\it single} transport map that is updated recursively 
in time, as opposed to a growing
composition of maps.

Here we propose a solution---certainly not the only possibility---based on the theory of
Section~\ref{sec:joint}.
The idea is to treat $\Zb_0$ as a static parameter, i.e., to set 
$\vhyp \coloneqq \Zb_0$ and apply the results of Theorem 
\ref{thm:decompJoint} to the Markov structure of
Figure \ref{fig:dataAssMargSmoothing}. The resulting algorithm
computes a sequence of
maps $(\submap_i)$ of dimension $3n$, i.e., {\it three} times the state dimension, and
keeps a running approximation of 
$\mathfrak{T}^{\vhyp}_{k}$
via the recursion $\mathfrak{T}^{\vhyp}_{k} = \mathfrak{T}^{\vhyp}_{k-1} \circ 
\submap_{k}^{\vhyp}$, where each $\submap_{k}^{\vhyp}$ is just 
a subcomponent of $\submap_k$. 
These maps $(\submap_i)$ are higher-dimensional than those considered in
Section \ref{sec:filt}, %
but
they do yield
the desired result:
each $\mathfrak{T}^{\vhyp}_{k}:\re^n \ra \re^n$
characterizes the smoothing 
marginal $\pi_{\Zb_{0} \vert \yb_{0:k+1}}$, for all $k \ge 0$, via a
single transport map that is updated recursively in time with just one forward pass
(see Theorem \ref{thm:decompJoint}[Part \ref{thm:decompJoint:partFilt}]).

\begin{figure}[h]
\centering
	\input{tikz_dataAssMargSmoothing}
\caption[Online fixed-point smoothing]{I-map (certainly not minimal) for
$\pi_{\Zb_0, \Zb_{1:N} \vert \yb_{0:N}}$, for any $N>0$.
Orange edges have been added compared to the tree structure of
Figure \ref{fig:dataAssSmoothing}.
}
\label{fig:dataAssMargSmoothing}
\end{figure}
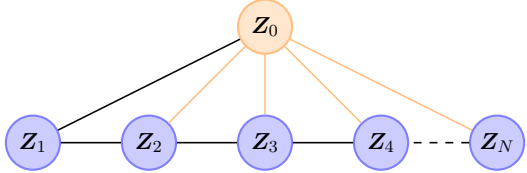

%% file: tikz_dataAssSmoothing.tex
\newcommand{\wStarGraph}{2cm} 
\newcommand{\scaleStarGraph}{.8} 
\newcommand{\minsize}{9mm} 
\newcommand{\varState}{\Zb} 
\newcommand{\varRef}{\Xb} 

\subfloat{
\begin{tikzpicture}[transform shape, scale = \scaleStarGraph]
	\node[default,minimum size=\minsize]  (nd1) 										{$\varState_0$};
 	\node[default,minimum size=\minsize]	 (nd2)		  		[right=of nd1]		{$\varState_1$} 
 		edge[edgeStyle] (nd1);
 	\node[default,minimum size=\minsize]	 (nd3)		  		[right=of nd2]		{$\varState_2$} 
 		edge[edgeStyle] (nd2);
 	\node[default,minimum size=\minsize]	 (nd4)		  		[right=of nd3]		{$\varState_3$} 
 		edge[edgeStyle] (nd3); 		 		
 	\node[default,minimum size=\minsize]	 (nd5)		  		[right=of nd4]		{$\varState_N$} 
 		edge[edgeStyle,dashed] (nd4); 		
  	\node[dataNode,minimum size=\minsize]	 (data1)		  		[below=of nd1]		{$\Yb_0$} 
 		edge[edgeStyle] (nd1);	
  	\node[dataNode,minimum size=\minsize]	 (data2)		  		[below=of nd2]		{$\Yb_1$} 
 		edge[edgeStyle] (nd2); 	
  	\node[dataNode,minimum size=\minsize]	 (data3)		  		[below=of nd3]		{$\Yb_2$} 
 		edge[edgeStyle] (nd3); 
  	\node[dataNode,minimum size=\minsize]	 (data4)		  		[below=of nd4]		{$\Yb_3$} 
 		edge[edgeStyle] (nd4); 
  	\node[dataNode,minimum size=\minsize]	 (data5)		  		[below=of nd5]		{$\Yb_N$} 
 		edge[edgeStyle] (nd5);  
 	%
 	%
 	\node[defaultGray,minimum size=\minsize]	 (ndr1)		  		[below=of data1]		{$\varRef_0$}; 	
 	\node[defaultGray,minimum size=\minsize]	 (ndr2)		  		[below=of data2]		{$\varRef_1$};
 	\node[defaultGray,minimum size=\minsize]	 (ndr3)		  		[below=of data3]		{$\varRef_2$};
 	\node[defaultGray,minimum size=\minsize]	 (ndr4)		  		[below=of data4]		{$\varRef_3$};
 	\node[defaultGray,minimum size=\minsize]	 (ndr5)		  		[below=of data5]		{$\varRef_N$}; 		
\end{tikzpicture}
} 	

%% file: sec_linGaussSmooth.tex
In this section, we specialize the results of Theorem \ref{thm:decompSmooth}
to linear Gaussian state-space models, and make explicit the
connection with the RTS Gaussian smoother \cite{rauch1965maximum}.

Consider a linear Gaussian state-space model defined by 
\begin{eqnarray}
 \Zb_{k+1}  & = & \forop_{k} \, \Zb_k + \noised_{k} \\
 \Yb_k      & = &\obsop_k \, \Zb_k + \noiseo_k \nonumber	
\end{eqnarray}
for all $k \ge 0$, where $\noised_{k}\sim \Gauss( 0 , \noisedcov_k )$,
$\noiseo_{k}\sim \Gauss( 0 , \noiseocov_k )$,
$\forop_{k} \in \re^{n \times n}$, $\obsop_k \in \re^{ d \times n }$, and
$\Zb_0 \sim \Gauss( \mu_0, \Gammab_0 )$.
Both $\noised_{k}$ and $\noiseo_{k}$ are independent of $\Zb_k$, while
$\noisedcov_k, \noiseocov_k$, and $\Gammab_0$ are symmetric positive definite %
matrices for all $k\ge 0$.

If we choose an independent reference process $(\Xb_k)$ with standard normal
marginals, i.e., $\eta_{\Xb_k} = \Gauss(0, {\bf I})$, 
 then the
maps $(\submap_k)$ of Theorem \ref{thm:decompSmooth} can be chosen
to be linear:
\begin{equation}   \label{eq:linearSubMap}
	\submap_k( \zb_k , \zb_{k+1} ) = 
	\left[
		\begin{array}{rr}
		\Ab_k & \Bb_k \\ 
		{\bf 0} & \Cb_k
		\end{array}
	\right]
	\left \{ 
		\begin{array}{l}
		\zb_k \\ 
		\zb_{k+1} 
		\end{array}
	\right \} + 
	\left \{ 
		\begin{array}{l}
		\ab_k \\ 
		\cb_k 
		\end{array}
	\right \}, 	
\end{equation}
for some matrices $\Ab_k,\Bb_k,\Cb_k \in \re^{n \times n}$ and $\ab_k,\cb_k\in \re^n$.
(Notice that in this case Corollary \ref{cor:genDecompThmChains} applies and
the matrices $\Ab_k,\Bb_k$ can be full and not necessarily triangular.)
The following lemma gives a closed form expression for the maps
$(\submap_k)$ with $k\ge 1$. ($\submap_0$ can be
derived analogously with simple algebra.) 
\begin{lem}[The linear Gaussian case] \label{lem:kalmanRec}
For $k\ge 1$, the map $\submap_k$ in \eqref{eq:linearSubMap} can be defined
as follows: if $(\cb_k , \Cb_k)$ is the output of a 
square-root Kalman filter at time $k$ \cite{bierman2006factorization},
i.e., if $\cb_k$ and $\Cb_k$ are, respectively, the mean and square root of the covariance of the
filtering distribution
$\pi_{\Zb_{k+1}\vert\yb_{0:k+1}}$, 
 then one can set:
\begin{eqnarray} \label{eq:recursionKalm}
	\Ab_k & = & \Jb_{k}^{-1/2} \\
	\Bb_k & = & 
	- \Jb_{k}^{-1} \, \Pb_{k} \, \Cb_k \nonumber \\
	\ab_k & = & 
	\Jb_{k}^{-1} \, \Pb_{k} \, (\forop_k\,\cb_{k-1}- \cb_k) \nonumber,
\end{eqnarray}
for $\Jb_{k} \coloneqq {\bf I} + \Cb_{k-1}^\top \, \forop_{k}^\top \, \noisedcov_k^{-1} \, 
		\forop_{k} \, \Cb_{k-1}$ and
$\Pb_{k} = - \Cb_{k-1}^\top\,\Fb_k^\top\,\Qb_k^{-1}$.		
\end{lem}
The formulas in Lemma \ref{lem:kalmanRec} can be interpreted as
one possible implementation of a square-root
RTS smoother for Gaussian models: at each step $k$ 
of a forward pass, the filtering estimates $(\cb_k , \Cb_k)$
are augmented with a collection $(\ab_k,\Ab_k,\Bb_k)$ of stored quantities, which
can then be reused to sample  the full Bayesian
solution (or particular smoothing marginals) whenever needed, and
without ever touching the state-space model again. 
In this sense, the algorithm proposed in Section \ref{sec:filt} can be
understood 
as the natural
generalization---to the non-Gaussian case---of 
the square-root %
RTS smoother.

%% file: tikz_dataAssHyper.tex
\newcommand{\wStarGraph}{2cm} 
\newcommand{\scaleStarGraph}{.8} 
\newcommand{\minsize}{9mm} 
\newcommand{\varState}{\Zb} 

\subfloat{
\begin{tikzpicture}[transform shape, scale = \scaleStarGraph]
	\node[default,minimum size=\minsize]  (nd1) 										{$\varState_0$};
 	\node[default,minimum size=\minsize]	 (nd2)		  		[right=of nd1]		{$\varState_1$} 
 		edge[edgeStyle] (nd1);
 	\node[default,minimum size=\minsize]	 (nd3)		  		[right=of nd2]		{$\varState_2$} 
 		edge[edgeStyle] (nd2);
 	\node[default,minimum size=\minsize]	 (nd4)		  		[right=of nd3]		{$\varState_3$} 
 		edge[edgeStyle] (nd3); 		 		
 	\node[default,minimum size=\minsize]	 (nd5)		  		[right=of nd4]		{$\varState_N$} 
 		edge[edgeStyle,dashed] (nd4); 		
  \node[defaultOrange,minimum size=\minsize]   (ndb)       [above=of nd3]    
  {\Large $\vhyp$} 
    edge[edgeStyle] (nd1)
    edge[edgeStyle] (nd2)
    edge[edgeStyle] (nd3)
    edge[edgeStyle] (nd4)
    edge[edgeStyle] (nd5);     								 	
\end{tikzpicture}
} 	

%% file: tikz_dataAssMargSmoothing.tex
\newcommand{\wStarGraph}{2cm} 
\newcommand{\scaleStarGraph}{.8} 
\newcommand{\minsize}{9mm} 
\newcommand{\varState}{\Zb} 
\newcommand{\colorEdges}{orange!50} 

\subfloat{
\begin{tikzpicture}[transform shape, scale = \scaleStarGraph]
	\node[default,minimum size=\minsize]  (nd1) 										{$\varState_1$};
 	\node[default,minimum size=\minsize]	 (nd2)		  		[right=of nd1]		{$\varState_2$} 
 		edge[edgeStyle] (nd1);
 	\node[default,minimum size=\minsize]	 (nd3)		  		[right=of nd2]		{$\varState_3$} 
 		edge[edgeStyle] (nd2);
 	\node[default,minimum size=\minsize]	 (nd4)		  		[right=of nd3]		{$\varState_4$} 
 		edge[edgeStyle] (nd3); 		 		
 	\node[default,minimum size=\minsize]	 (nd5)		  		[right=of nd4]		{$\varState_N$} 
 		edge[edgeStyle,dashed] (nd4); 		
  \node[defaultOrange,minimum size=\minsize]   (ndb)       [above=of nd3]    
  {$\varState_0$} 
    edge[edgeStyle] (nd1)
    edge[edgeStyle, color = \colorEdges] (nd2)
    edge[edgeStyle, color = \colorEdges] (nd3)
    edge[edgeStyle, color = \colorEdges] (nd4)
    edge[edgeStyle, color = \colorEdges] (nd5);     								 	
\end{tikzpicture}
} 	

%% file: sec_numerics.tex
We illustrate some aspects of the preceding theory using a problem of
sequential inference in a non-Gaussian state-space model.  In particular,
we show the application of decomposable transport maps
(Sections~\ref{sec:decomp} and \ref{sec:dataAss}) to joint state and
parameter inference in a stochastic volatility model. 
This example is
intended as a direct and simple illustration of the theory.
The notion of decomposable transport maps is useful well beyond the sequential
inference setting, and entails the general problem of inference in continuous non-Gaussian graphical 
models.
We refer the reader to \cite{morrison2017beyond} for an application of the theory of
sparse transports (Section \ref{sec:sparse}) to the problem of learning the Markov structure of a 
non-Gaussian distribution, and 
we defer
further numerical investigations to a dedicated paper
\cite{bigoni2016monotone}.

Following \cite{kim1998stochastic,rue2009approximate}, we model the scalar log-volatility 
$(\Zb_k)$
of the return of a financial asset at time $k=0,\ldots,N$ using an autoregressive 
process of
order one, which is fully specified by 
$\Zb_{k+1}=\mub + \phib \,(\Zb_k -\mub) + \,\noised_{k}$, for all $k\ge0$,
where $\noised_{k}\sim \Gauss(0,1/16)$ is independent of $\Zb_k$,
$\Zb_0 \vert \mub,\phib \sim 
\Gauss(\mub,\frac{1}{1-\phib^2})$,
and where
$\phib$ and $\mub$ represent scalar
hyperparameters of the model. 
In particular, $\mub\sim\Gauss(0,1)$ and 
$\phib = 2\exp(\phib^\star)/(1+\exp(\phib^\star))-1$ with
$\phib^\star\sim \Gauss(3,1)$.
We define $\vhyp\coloneqq(\mub,\phib)$.
The process $(\Zb_k)$ and parameters $\vhyp$ are unobserved and
must be estimated from an observed process $(\Yb_k)$, which represents the
mean return on holding the asset at 
time $k$,  
$\Yb_k = \noiseo_k \, 
\exp(\frac{1}{2}\Zb_k)$, where
$\noiseo_k$ is a standard normal random variable
 independent of $\Zb_k$.
As a dataset $(\yb_k)_{k = 0}^N$, we use the $N+1$  daily differences
of the pound/dollar exchange rate starting on 1 October 1981, with
$N=944$
\cite{rue2009approximate,durbin2000time}. 

Our goal is to sequentially characterize
$\pi_{\vhyp,\Zb_{0:k} \vert \yb_{0:k}}$, for all 
$k=0,\ldots,N$, as observations $(\yb_k)$ become 
available. 
The Markov structure of $\pi_{\vhyp,\Zb_{0:N} \vert \yb_{0:N}}$
matches
Figure \ref{fig:dataAssHyper}. 
We solve the problem using the algorithm introduced in Section \ref{sec:joint}: 
we compute a sequence, $(\submap_j)_{j=0}^{N-1}$,
of four-dimensional transport maps 
($n = \text{dim}(\Zb_j)=1$ and $p = \text{dim}(\vhyp)=2$) according to their definition 
in Theorem \ref{thm:decompJoint} and using the
variational form \eqref{OptimDirect}. 
All reference densities are standard Gaussians.
Then, by Theorem \ref{thm:decompJoint}[part \ref{thm:decompJoint:partFilt}], 
for any $k < N$, we can easily sample %
the filtering
marginal $\pi_{\Zb_{k+1} \vert \yb_{0:k+1}}$  by pushing forward
a standard normal through the subcomponent $\submap_k^1$ of $\submap_k$,
and we can also sample  the posterior distribution over the
static parameters $\pi_{ \vhyp \vert \yb_{0:k+1}}$ by pushing forward a 
standard normal through the map $\mathfrak{T}^{\vhyp}_{k}$. 
The map $\mathfrak{T}^{\vhyp}_{k} = \submap_{0}^{\vhyp} \circ \cdots \circ \submap_{k}^{\vhyp}$ 
is updated sequentially over time (via
regression) using the
recursion $\mathfrak{T}^{\vhyp}_{k} = \mathfrak{T}^{\vhyp}_{k-1} \circ 
\submap_{k}^{\vhyp}$, so that the
cost of evaluating $\mathfrak{T}^{\vhyp}_{k}$ does not increase with $k$.
The resulting algorithm for parameter estimation is thus sequential.
Moreover, 
if we want to sample  
$\pi_{\vhyp,\Zb_{0:k+1} \vert \Yb_{0:k+1}}$---the full
Bayesian solution at time $k+1$---we simply need to embed 
each $\submap_j$
into an identity map to form the
transport $T_j$,
for $j=0,\ldots,k$,
and push forward reference samples 
through the composition 
$\mathfrak{T}_k = T_0 \circ \cdots \circ T_k$
(Theorem \ref{thm:decompJoint}[part \ref{thm:decompJoint:partFull}]).
\hrevone{See Appendix \ref{sec:algo} for pseudocode of the relevant algorithms.}

Figures \ref{fig:stoc-vol:filtering} and
\ref{fig:stoc-vol:smoothing-vs-unbiased} show the resulting
filtering and smoothing marginals of the states over time, respectively. 
Figures  \ref{fig:stoc-vol:3d-phi-vs-unbiased} and
\ref{fig:stoc-vol:3d-mu-vs-unbiased}
collect the corresponding posterior marginals of
the static parameters over time.
Figure \ref{fig:stoc-vol:post-pred} 
illustrates marginals of the posterior predictive distribution of the data, together with
the observed data $(\yb_k)$, showing excellent coverage overall.

Our results rely on a numerical approximation of the
desired transport maps.
Each component of $\mathfrak{M}_k$ is 
parameterized via the monotone representation
\eqref{eq:monotone}, 
with $(a_k)$ and $(b_k)$ chosen to be Hermite polynomials and functions, respectively, of total degree seven \cite{boyd2001chebyshev} . 
The %
expectation 
in \eqref{OptimDirect} is approximated using tensorized Gauss quadrature 
rules.
The resulting minimization problems are 
solved sequentially using the Newton-CG method \cite{wright1999numerical}.
This test case was run using the dedicated software package publicly available at
\url{http://transportmaps.mit.edu}.
The website contains details about additional possible parameterizations of the maps.

There are several ways to investigate the quality of these
approximations.  Figures
\ref{fig:stoc-vol:smoothing-vs-unbiased},
\ref{fig:stoc-vol:3d-phi-vs-unbiased}, and
\ref{fig:stoc-vol:3d-mu-vs-unbiased} 
compare the numerical
approximation (via a decomposable transport map) of the smoothing marginals
of the states and the posteriors of the static parameters to a ``reference'' solution
obtained via MCMC. The MCMC chain is run until it
yields $10^5$ effectively independent samples.  The two solutions 
agree remarkably well and are almost indistinguishable in most
places. 
(Of course, MCMC in this context is not a data-sequential
algorithm; it requires that all the data $(\yb_k)_{k = 0}^N$ be
available simultaneously.)
An important fact is that the MCMC chain is generated using an 
{\it independent} proposal \cite{robert2013monte} given by the pushforward of a standard Gaussian
through the numerical
approximation of
$\mathfrak{T}_{N-1}$ (denoted as 
$\widetilde{ \mathfrak{T} }_{N-1}$).
The resulting MCMC chain has an acceptance rate slightly above $75 \%$, confirming the
overall quality of the variational approximation.
We notice, however, a {\it slow} accumulation of error in the posterior marginal
for the static parameter $\mub$ (Figure \ref{fig:stoc-vol:3d-mu-vs-unbiased}).
This is not surprising since we are performing {\it sequential} 
parameter inference \cite{jacob2015sequential}. 

A second {quality} test can proceed as follows:
since we use a standard Gaussian reference distribution $\genm_{\eta}$, we expect the
pullback of $\pi_{\vhyp,\Zb_{0:N} \vert \yb_{0:N}}$ through 
$\widetilde{ \mathfrak{T} }_{N-1}$ to be close to a standard Gaussian.
Figure \ref{fig:stoc-vol:pullback}  %
supports this claim by showing a collection of random two-dimensional conditionals
of the approximate pullback: these ``slices'' of the 947-dimensional 
($N$+1 states plus two hyperparameters)
pullback distribution are identical to a two-dimensional
standard normal, as expected.
The fact that we {\it can}  evaluate the approximate pullback density is one of
the key features of this variational approach to inference.
Even more, we can use this approximate pullback density 
to estimate the KL divergence between our target $\genm_{\pi}$ 
(the full Bayesian solution at time $N$) 
and the approximating measure 
$(\widetilde{ \mathfrak{T} }_{N-1})\push \genm_{\eta}$, via the
variance diagnostic in \eqref{eq:var_diag}.
A numerical realization of \eqref{eq:var_diag} 
yields $\Dkl(\,(\widetilde{ \mathfrak{T} }_{N-1})\push 
\genm_{\eta} \,\vert\vert \, \genm_{\pi} \, ) \approx 1.07 \times 10^{-1}$, which
confirms the good numerical 
approximation of $\genm_\pi$, a 947-dimensional target measure.
For comparison, we note that the KL divergence from
$\genm_\pi$ to its 
Laplace 
approximation (a Gaussian approximation at the mode) is %
$\approx 5.68$---considerably
worse 
than what is achieved through optimization of
a nonlinear transport map.
Moreover, the Laplace approximation {cannot} be computed
sequentially with a constant effort per time step.

\hrevone{While a slow accumulation of errors is expected for sequential parameter inference,
we also wish
to investigate the stability of our transport map
approximation for %
recursive smoothing {\it without} static parameters.
}
We try the following experiment: (1) compute the
posterior medians of the static parameters after $N+1=945$ days, i.e., 
$\thetab^* = {\rm med}[\Theta \vert \yb_0 ,\ldots, \yb_{944} ]$; and then (2) use these
parameters to characterize the smoothing distribution $\pi_{\Zb_{0}
,\ldots, \Zb_{2500} \vert \thetab^\ast, \,\yb_{0}, \ldots,
\,\yb_{2500}}$ of the log-volatility over (roughly) \hrevone{the next ten
years worth of exchanges}, 
using the sequential algorithm proposed in Section \ref{sec:filt}, which in this case
amounts to computing only a sequence of two-dimensional maps $(\submap_k)$.
The resulting smoothing marginals are shown in Figure
\ref{fig:stoc-vol:long-smooth-vs-unbiased}  and compared to those of a
reference MCMC 
simulation with $10^5$ effectively independent samples; we observe excellent agreement despite the long
assimilation window.
\hrevone{We then repeat the same experiment for an even longer assimilation
window, i.e., 9009 steps or roughly 35 years.
Figure \ref{fig:stoc-vol:very-long-smooth-vs-unbiased} 
shows the remarkable stability of the resulting smoothing approximation, at
least for low-order marginals. In fact, even the approximation of the \emph{joint} distribution of the states is quite good, as reported in the last
column of Table \ref{tab:cost_vs_accuracy}.
Understanding how errors propagate in this variational framework---and what could be potential
mechanisms for the ``dissipation'' of errors---is an exciting avenue for future work.
}

\hrevone{ 
The results presented so far are very accurate, but also expensive.
Table \ref{tab:cost_vs_accuracy} collects the computational times
for the joint state-parameter inference problem (approximately two days) and for the
long-time (9009 step) smoothing problem (approximately 40 minutes), using a
degree-seven map.
While there remains a tremendous opportunity to develop %
more performance-oriented versions of our transport map code, specialized to the
problem of sequential inference, 
the present framework also offers a practical and powerful tradeoff
between computational 
cost
and accuracy.  In Appendix \ref{sec:add_res}, we re-run all our test
cases using linear, rather than degree seven, parameterizations of
the maps $\{ \submap_k\}$.  Table \ref{tab:cost_vs_accuracy} shows
that the computational times are dramatically reduced: from two days
to approximately one minute for the joint state-parameter inference
problem, and from 40 minutes to 7 minutes for the long-time smoothing
problem.  The reduction in computational time comes, of course, at the
price of accuracy; see last column of Table
\ref{tab:cost_vs_accuracy}.  This reduction in accuracy may or may not
be acceptable.  For instance, in Figure
\ref{fig:stoc-vol:very-long-smooth-vs-unbiased-linear}, it is
difficult to distinguish the linear map approximation from the
reference MCMC solution.  Quantitatively, we know from Table
\ref{tab:cost_vs_accuracy} that a linear map is worse at approximating
the full Bayesian solution than a degree-7 transformation.
Yet, as far as quantiles of low-order marginals  are concerned, the two solutions are
indistinguishable (Figure \ref{fig:stoc-vol:very-long-smooth-vs-unbiased-linear}); in an applied setting, this accuracy may be more than sufficient.
In other cases, however, a linear map might be inadequate. For example, the parameter marginals in Figures 
\ref{fig:stoc-vol:3d-phi-vs-unbiased-linear}
and
\ref{fig:stoc-vol:3d-mu-vs-unbiased-linear}, estimated using linear maps,
are much worse than their degree-7 counterparts (Figures
\ref{fig:stoc-vol:3d-phi-vs-unbiased} and
\ref{fig:stoc-vol:3d-mu-vs-unbiased}). 
In these cases, we \emph{need} nonlinear transformations.
}

\hrevone{Clearly, there is a rich spectrum of possibilities between a
linear and a high-order transport map. Some parameterizations can
scale with dimension (e.g., separable but nonlinear representations),
while others do not (e.g., total-degree polynomial expansions).  
Depending on the problem, some
parameterizations will lead to
accurate results, while others will not.
Yet, the cost-accuracy tradeoff %
in
the transport framework can be
{\it controlled}, e.g., by estimating %
the {quality} of a given approximation 
using
\eqref{eq:var_diag}.}

\vspace{60pt}

\begin{figure}[H] %
  \begin{center}
        \includegraphics[width=0.90\textwidth, bb=25bp 0bp 650bp 290bp, clip]{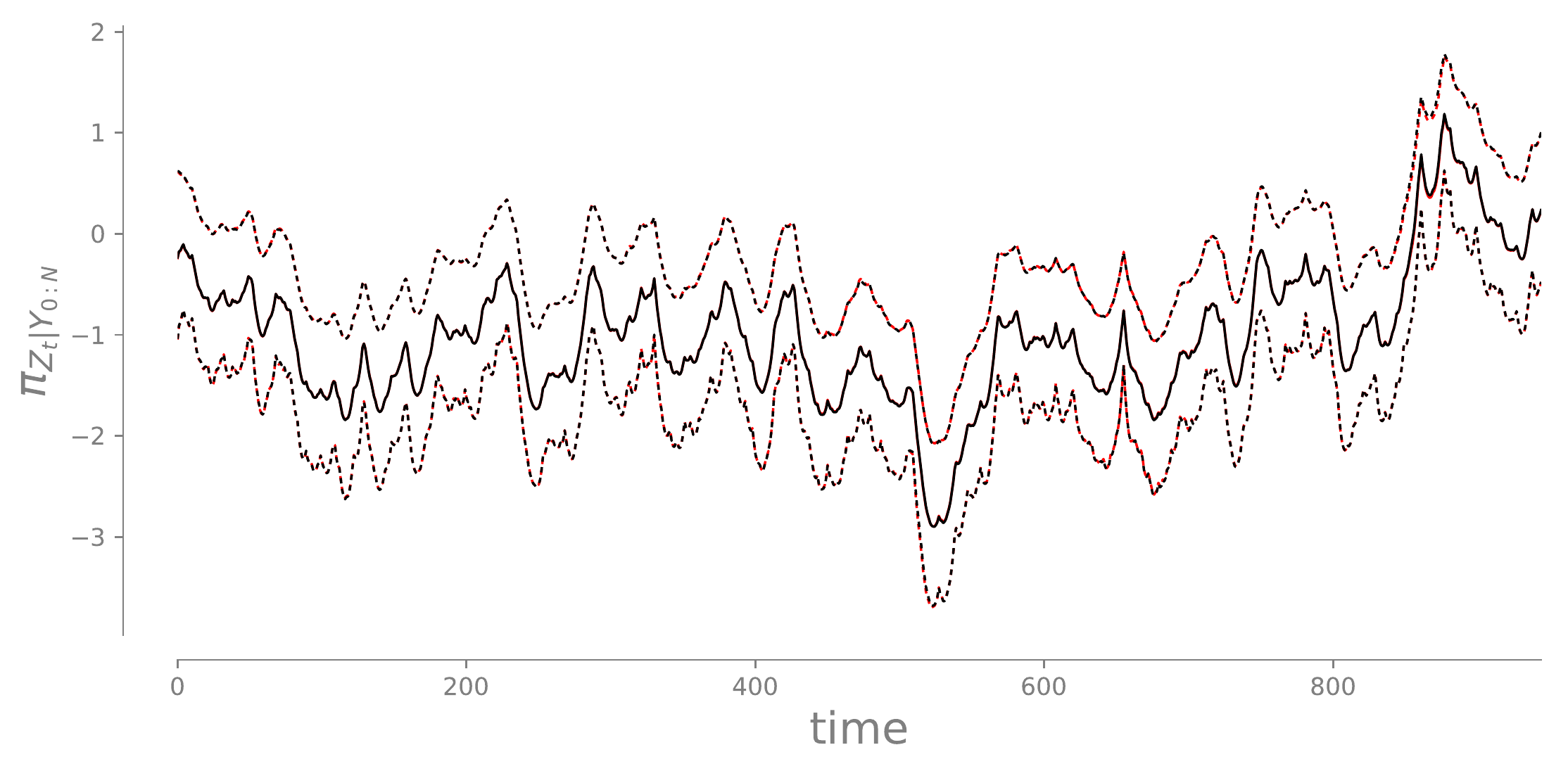}   
    \caption{    
      Comparison between the $\{5,95\}$--percentiles (dashed lines) and the mean (solid line)
      of the numerical approximation of the
      smoothing marginals $\pi_{ \Zb_{k} \vert \yb_{0:N}}$ via
      transport maps (red lines) versus a ``reference'' MCMC solution 
      (black lines), %
      for $k=0,\ldots,N$. The two solutions are indistinguishable.}
    \label{fig:stoc-vol:smoothing-vs-unbiased}
  \end{center}
\end{figure}

\begin{figure}[H] %
  \begin{center}
        \includegraphics[width=0.90\textwidth, bb=25bp 0bp 650bp 290bp, clip]{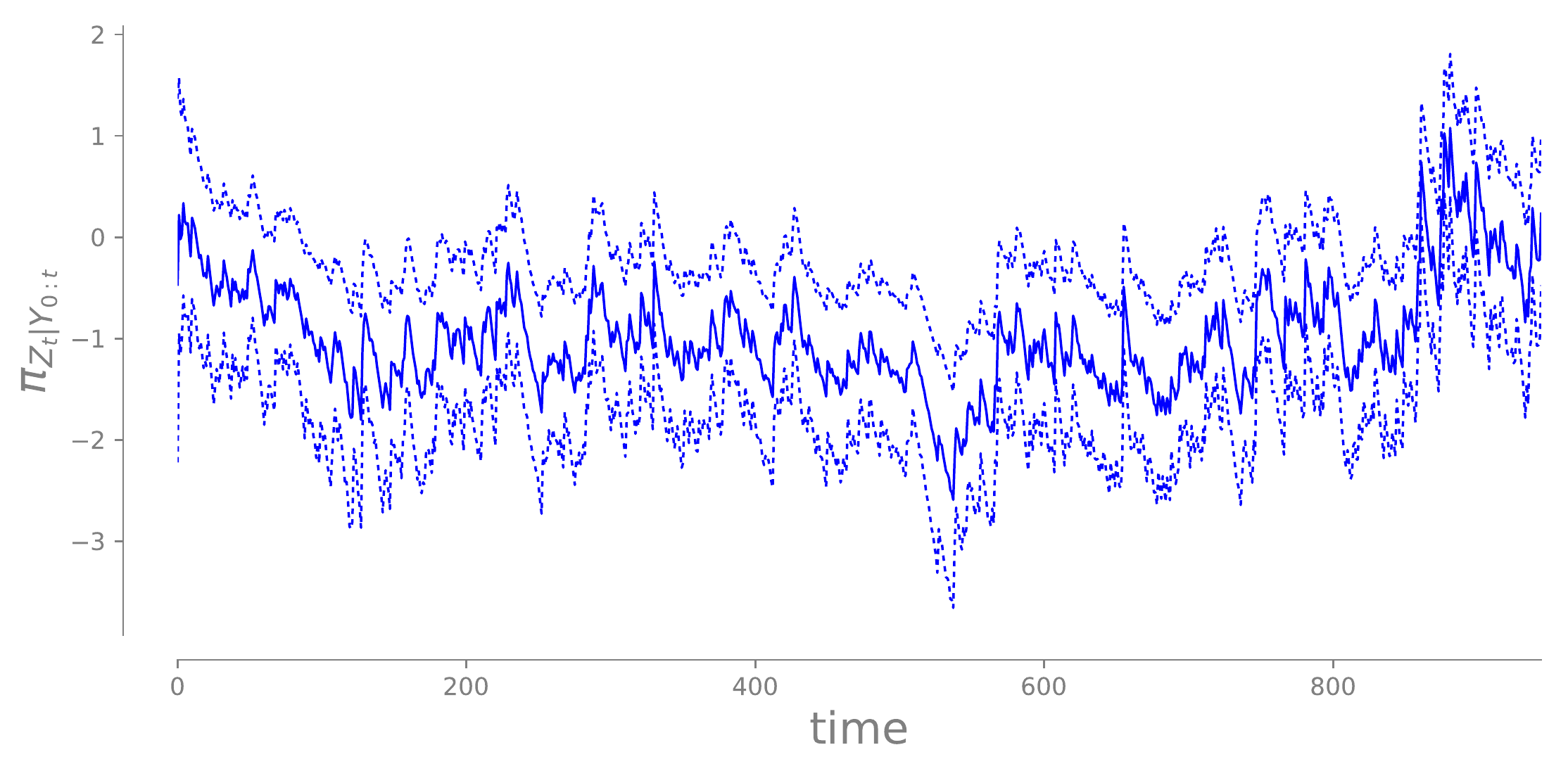}   
    \caption{At each time $k$, we illustrate the
      $\{5,95\}$--percentiles (dotted lines) and the mean (solid line) of the
      numerical approximation of the filtering distribution 
      $\pi_{ \Zb_k \vert \yb_{0:k}}$, 
      for $k=0,\ldots,N$. }
    \label{fig:stoc-vol:filtering} 
  \end{center}
\end{figure}

\begin{figure}[H] %
  \begin{center}
        \includegraphics[width=1.0\textwidth, bb=65bp 0bp 650bp 290bp, clip]{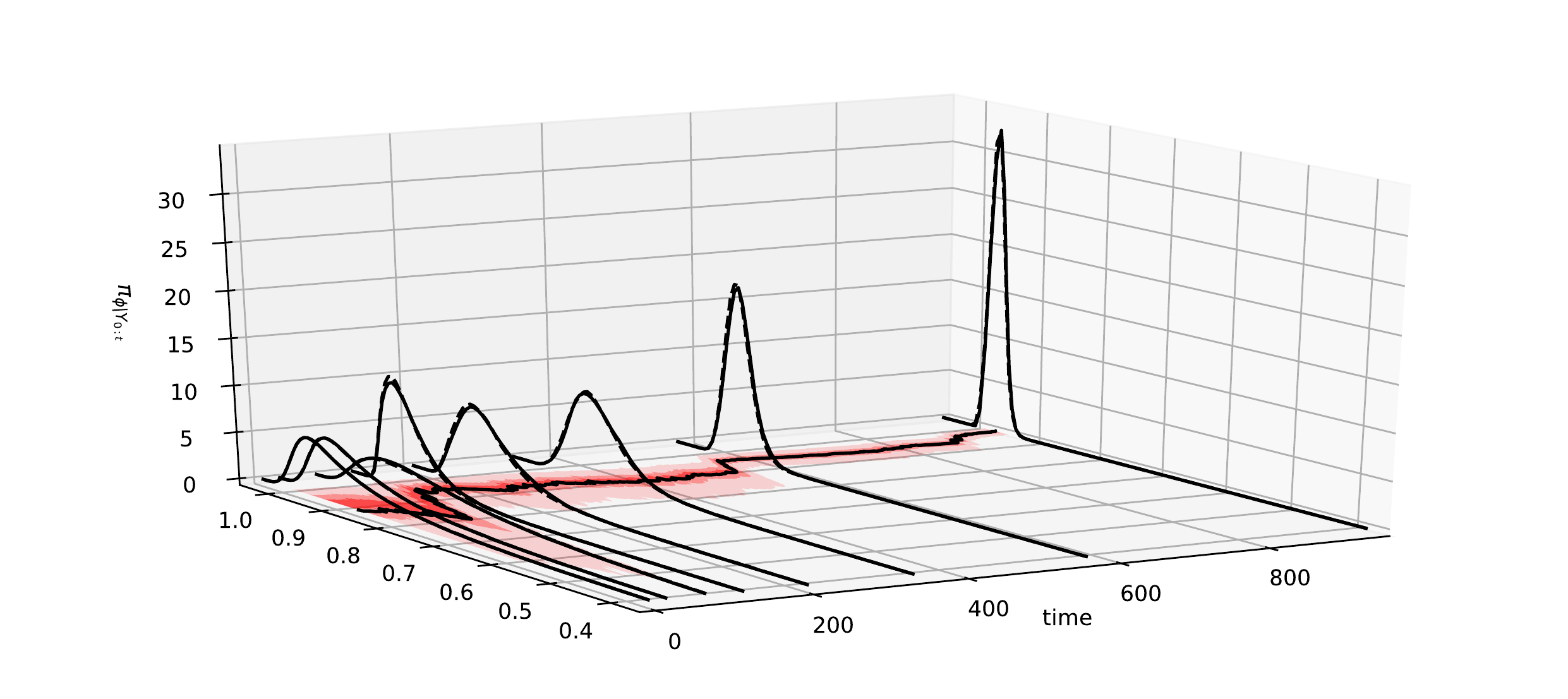}   
    \caption{\emph{(Horizontal plane)} At each time $k$, we illustrate the
      $\{5,25,40,60,75,95\}$--percentiles (shaded regions) and the mean (solid line) of 
      the numerical approximation of $\pi_{ \phib \vert \yb_{0:k}}$,
      the posterior marginal of
      the static parameter $\phib$. 
      \emph{(Vertical axis)} At several times $k$ %
      we also compare 
      the transport map numerical approximation of 
         $\pi_{ \phib \vert \yb_{0:k}}$ 
          (solid lines)
      with a reference MCMC solution (dashed lines).
      The two distributions agree remarkably well.}
  \label{fig:stoc-vol:3d-phi-vs-unbiased} 
  \end{center}
\end{figure}

\begin{table}[H]
  \centering
  \begin{tabular}{c|c|c|c|c|c|c}
    \textbf{Type} & \textbf{\# steps} & \textbf{Order} & \textbf{Time [m\,:\,s]} & \textbf{\# cores} & \textbf{Figures} & \textbf{Var.\ diag.} \eqref{eq:var_diag}\\ \hline\hline
    \multirow{3}{*}{S/P} & \multirow{3}{*}{945} & Laplace & 00\,:\,04 & 1 & & $5.68$ \\
    & & 
    7
    & $\approx$ 2 days & 64 & 
    \ref{fig:stoc-vol:smoothing-vs-unbiased} --
    \ref{fig:stoc-vol:pullback} & $1.07 \times 10^{-1}$ \\
    & & linear & 01\,:\,14 & 1 & 
    \ref{fig:stoc-vol:smoothing-vs-unbiased-linear} --
    \ref{fig:stoc-vol:2d-mu-vs-unbiased-linear} & $1.77$ \\ \hline
    \multirow{3}{*}{S} & \multirow{3}{*}{9009} & Laplace & 00\,:\,42 & 1 & & $10.0$ \\
    & & %
    7
    & 42\,:\,50 & 1 & \ref{fig:stoc-vol:very-long-smooth-vs-unbiased} & $1.19 \times 10^{-1}$ \\
    & & linear & 06\,:\,40 & 1 & \ref{fig:stoc-vol:very-long-smooth-vs-unbiased-linear} & $5.01$ \\ \hline
  \end{tabular}
      \caption{Computational effort required to compute
      a decomposable transport map for different complexities of the
      transformations $\submap_k$---linear versus degree seven---and for
      different inference scenarios---smoothing and static
      parameter estimation {\it(top row)} or long-time smoothing
      without static parameters {\it(bottom row)}, for the stochastic
      volatility model of Section \ref{sec:numerics}. The last column
      reports the variance diagnostic 
      \eqref{eq:var_diag} for the corresponding \emph{joint} posterior, not just a
      few marginals. It highlights a tradeoff between cost and
      accuracy, typical of the transport map approach to variational
      inference. For comparison, we also report the cost and accuracy
      of a simple Laplace approximation, which requires no formal
      optimization.
      }

      \label{tab:cost_vs_accuracy}
\end{table}

\begin{figure}[H]%
  \begin{center}
        \includegraphics[width=1.0\textwidth, bb=65bp 0bp 650bp 290bp, clip]{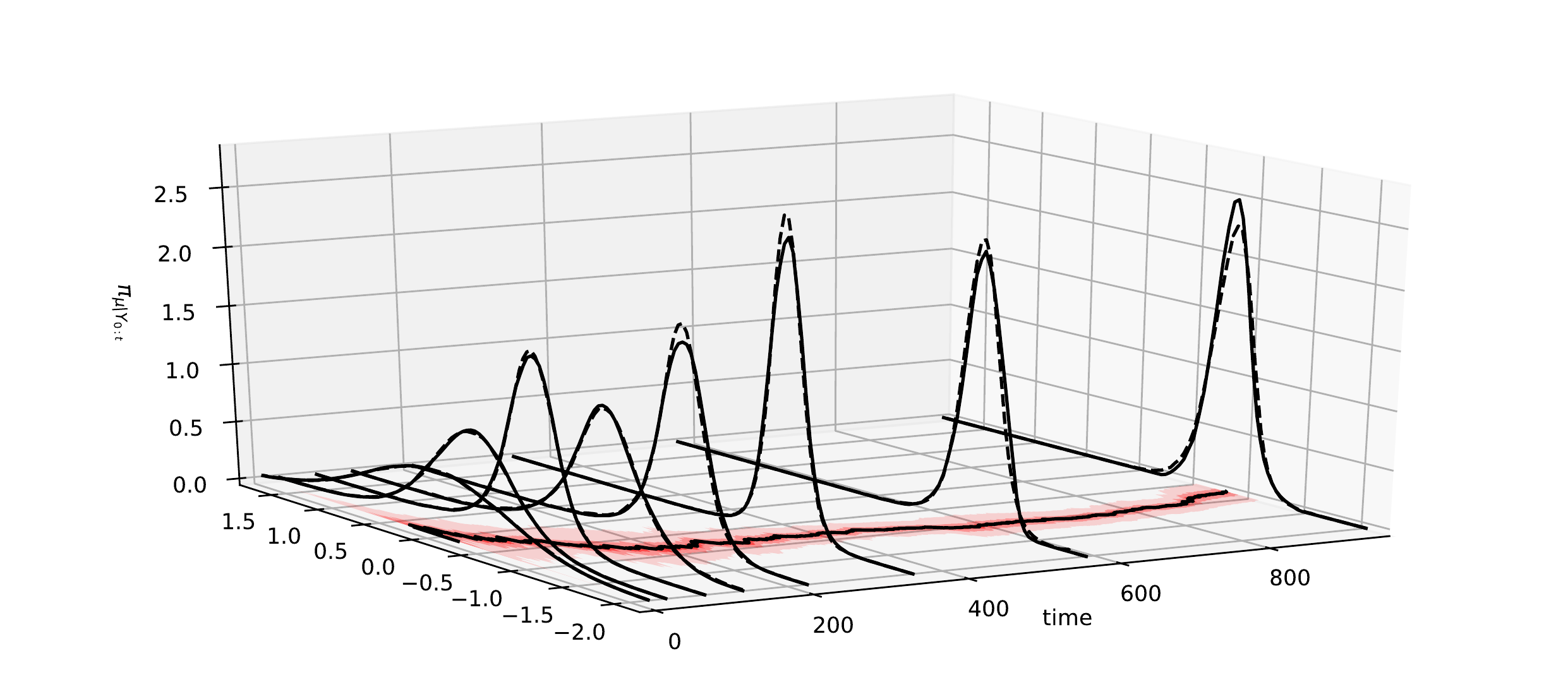}   
    \caption{
    Same as Figure \ref{fig:stoc-vol:3d-phi-vs-unbiased},  but for the static parameter $\mub$.
    }
  \label{fig:stoc-vol:3d-mu-vs-unbiased}   
  \end{center}
\end{figure}

\begin{figure}[H] %
  \begin{center}
        \includegraphics[width=0.90\textwidth, bb=25bp 0bp 650bp 281bp, clip]{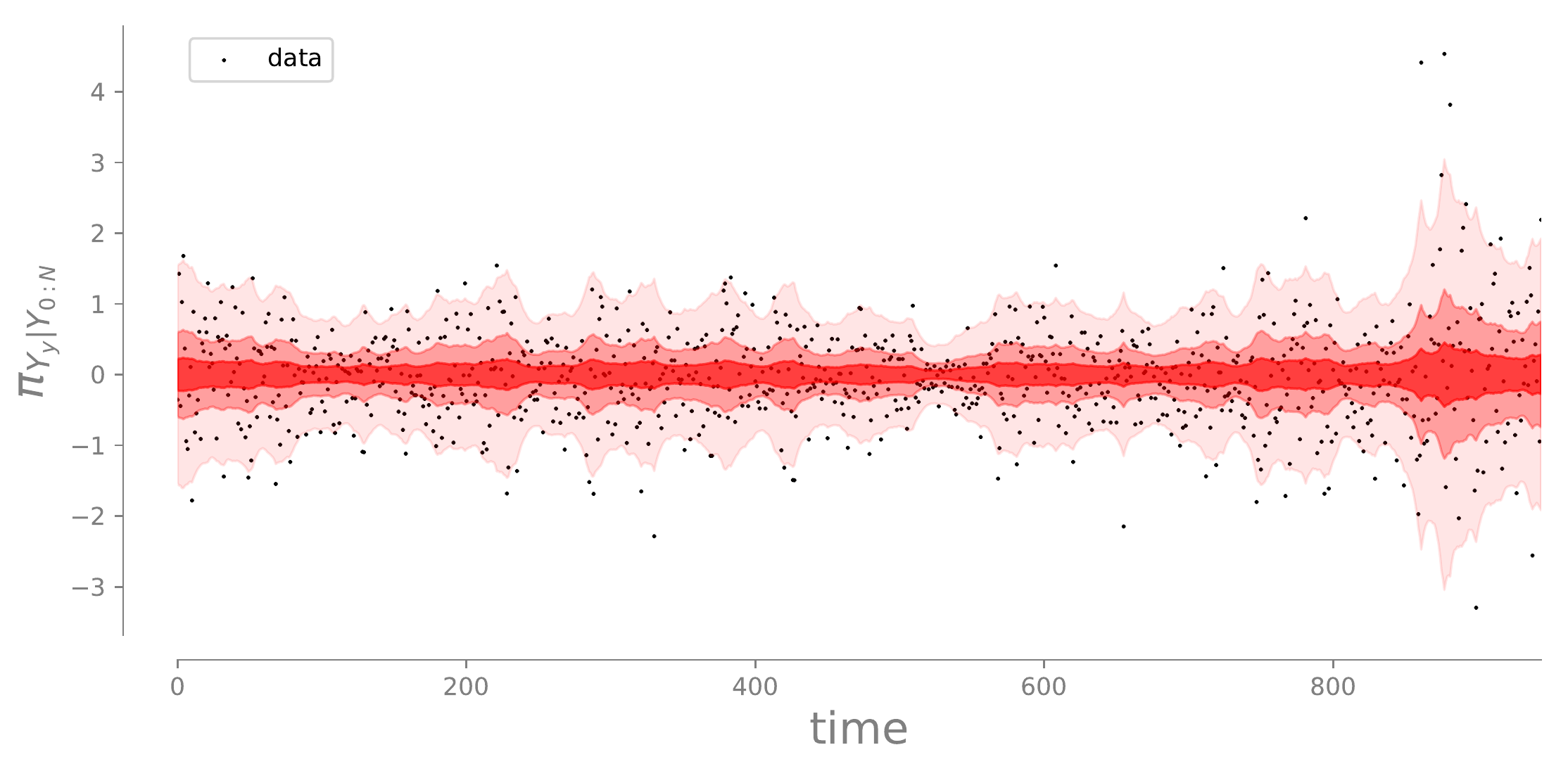}   
    \caption{ %
     Shaded regions represent the $\{5,25,40,60,75,95\}$--percentiles of the
     marginals of the
     posterior predictive distribution (conditioning on all the data),
     along with
     black dots that represent the observed data $(\yb_k)_{k=0}^N$.
      }
    \label{fig:stoc-vol:post-pred} 
  \end{center}
\end{figure}

\begin{figure}[H] %
  \begin{center}
     \includegraphics[width=0.40\textwidth, trim={0 0 0 1cm}, clip]{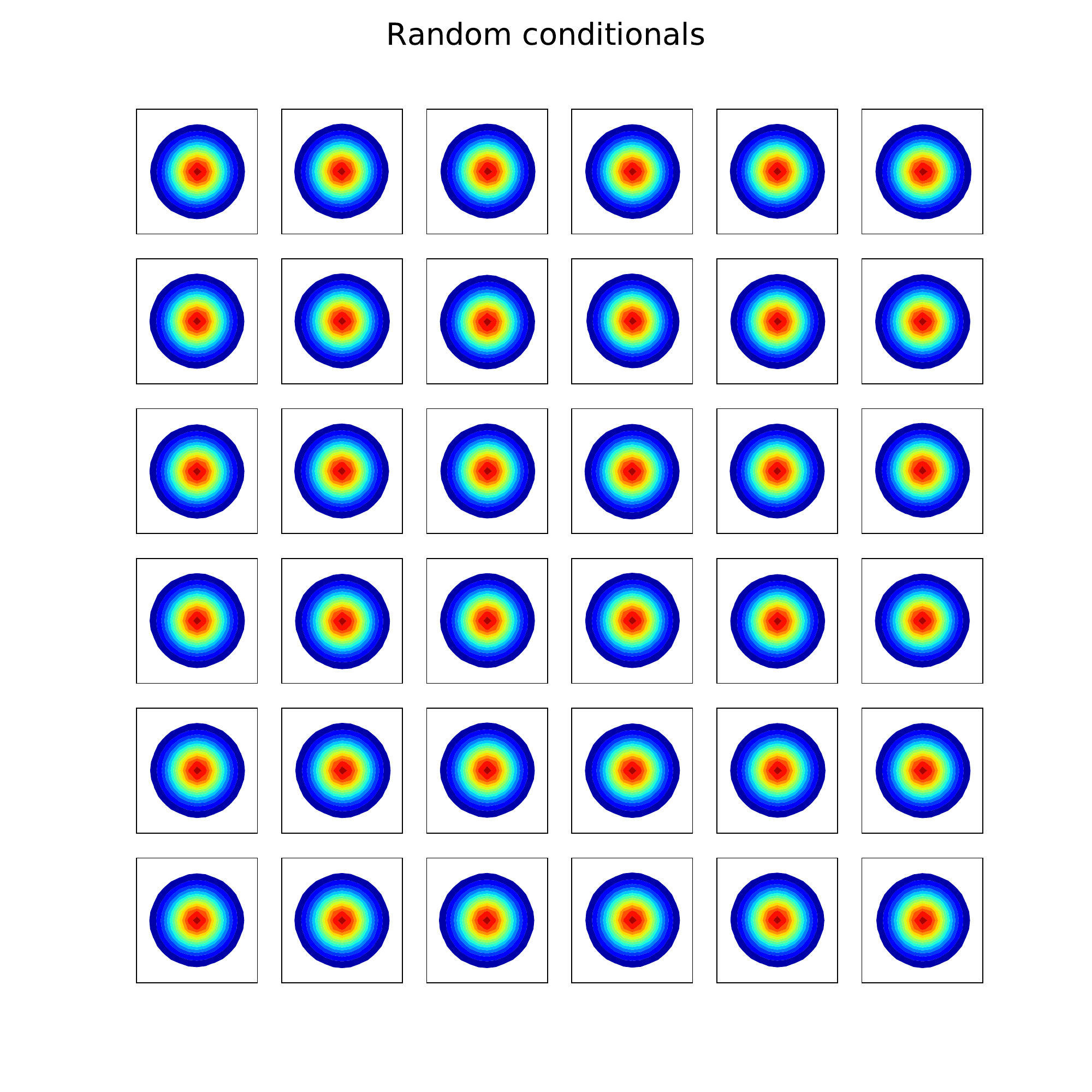} 
    \caption{ Randomly chosen two-dimensional conditionals of the pullback of 
      $\pi_{\vhyp, \Zb_{0:N} \vert \yb_{0:N}}$ 
      through the numerical approximation of
      $\mathfrak{T}_{N-1}$.
      Since we use a standard normal reference distribution, the
      numerical approximation of 
      $\mathfrak{T}_{N-1}$
      should be deemed satisfactory if the %
      pullback density
      is close to a standard normal, as it is here.
      }
    \label{fig:stoc-vol:pullback} 
  \end{center}
\end{figure}

\begin{figure}[H] %
  \begin{center}
        \includegraphics[width=0.90\textwidth, bb=25bp 0bp 650bp 290bp, clip]{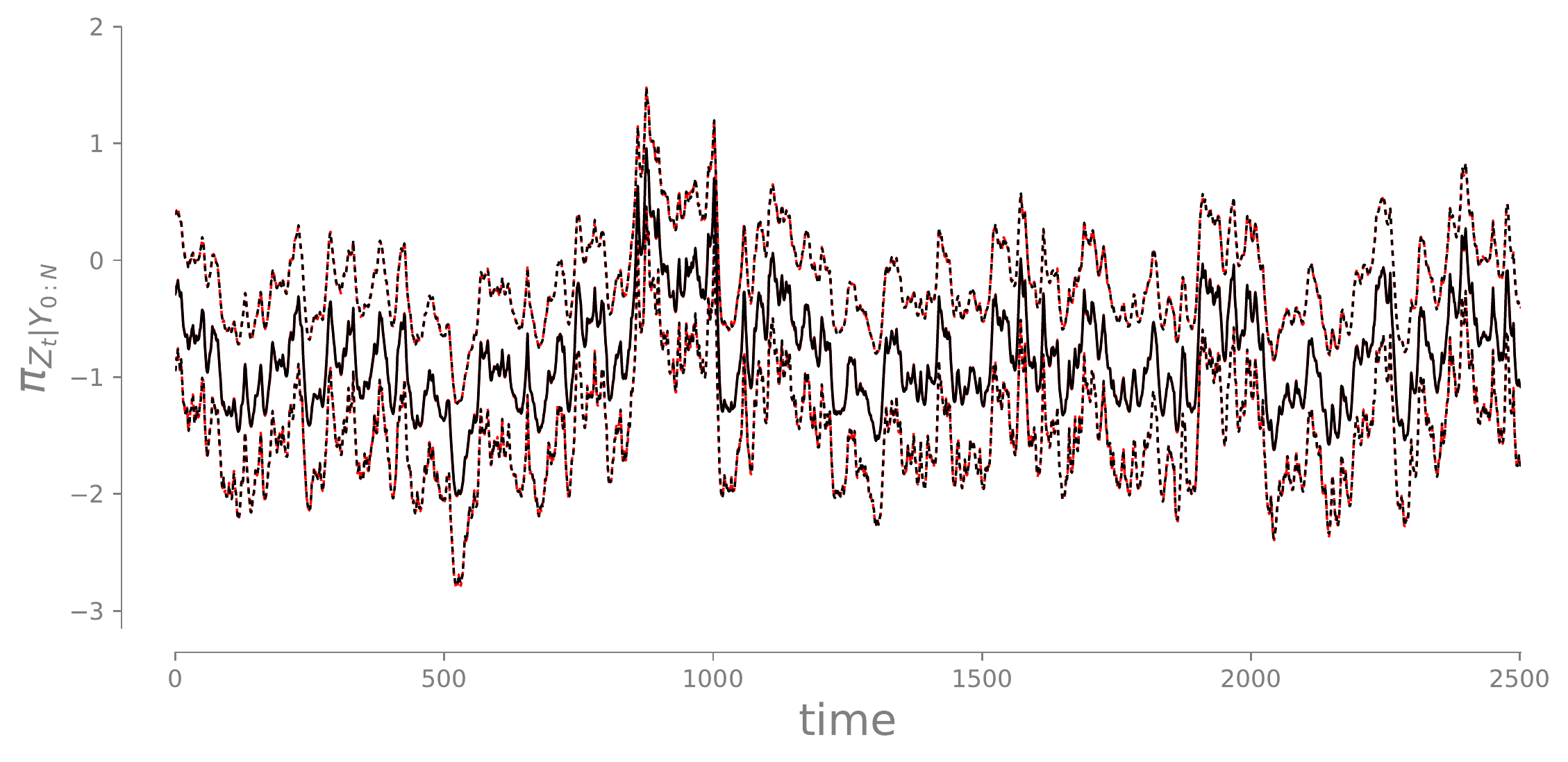}   
    \caption{Comparison between the $\{5,95\}$--percentiles (dashed lines) and the mean 
    (solid line)
      of the transport map 
      numerical 
      approximation of the
      smoothing marginals $\pi_{ \Zb_{k} \vert \thetab^*, \yb_{0:2500}}$,
      with $\thetab^* = {\rm med}[ \Theta \vert \yb_0 ,\ldots, \yb_{N}
      ]$ (red lines), and a reference MCMC solution       %
      (black lines).
      The two solutions are indistinguishable.}
    \label{fig:stoc-vol:long-smooth-vs-unbiased} 
  \end{center}
\end{figure}

\clearpage
\begin{sidewaysfigure}
\begin{center}
\includegraphics[width=.9\textwidth, bb=25bp 0bp 800bp 290bp, clip]{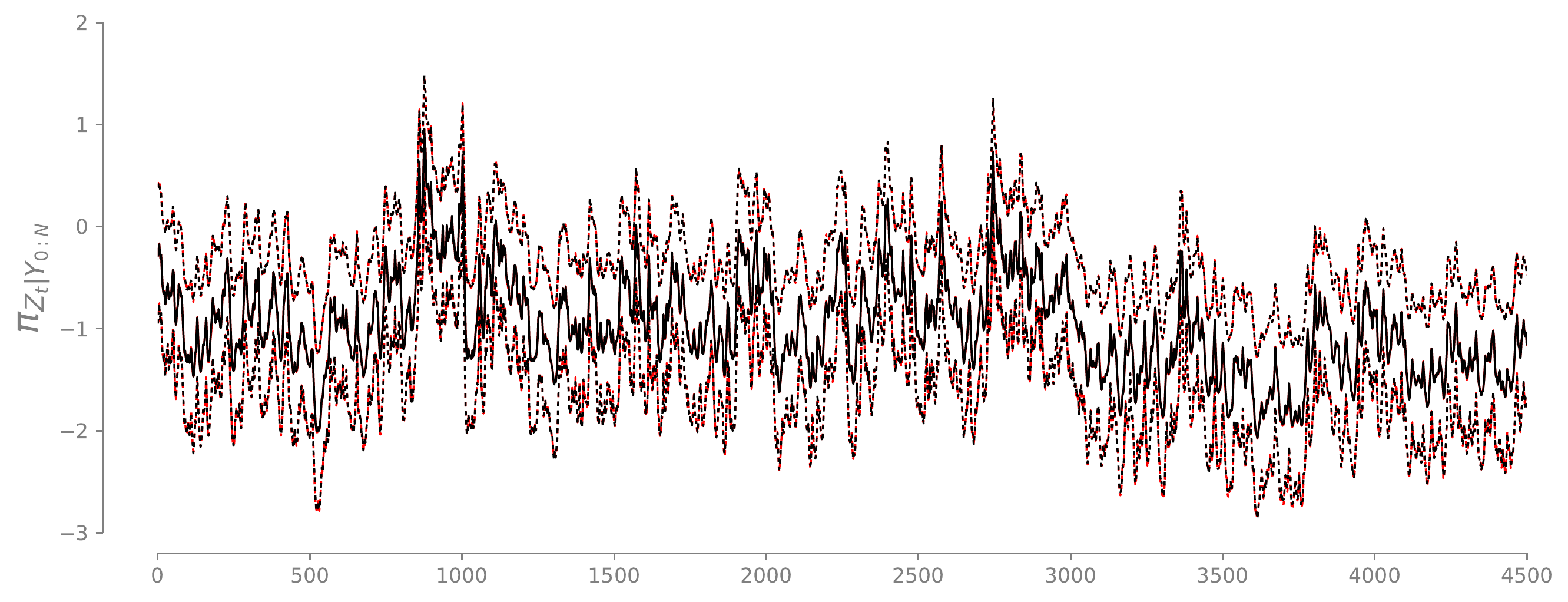}
\includegraphics[width=.9\textwidth, bb=25bp 0bp 800bp 290bp, clip]{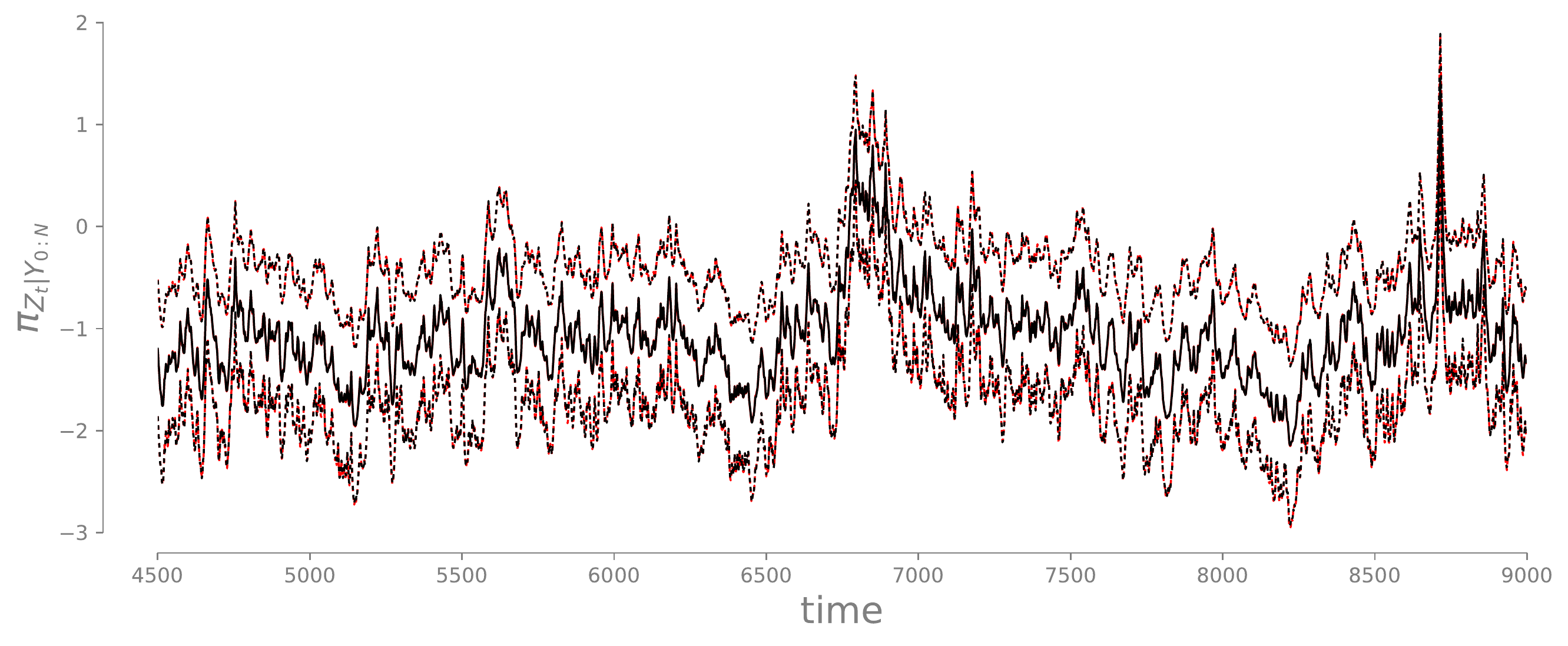}
\caption{Same as Figure \ref{fig:stoc-vol:long-smooth-vs-unbiased}, but
for a longer assimilation window, i.e., $\pi_{ \Zb_{k} \vert \thetab^*, \yb_{0:9000}}$.
The smoothing approximation remains excellent despite the widening inference horizon. }
\label{fig:stoc-vol:very-long-smooth-vs-unbiased} 
\end{center}
\end{sidewaysfigure}
\clearpage

%% file: sec_discuss.tex
This paper has focused on the problem of coupling a pair $(\genm_\eta,\genm_\pi)$ of
absolutely continuous measures on $\re^n$,
for the purpose of sampling or integration.
If $\genm_\eta$ is a tractable measure (e.g., an isotropic Gaussian) and
$\genm_\pi$ is an intractable measure of interest (e.g., a
posterior distribution), then a deterministic coupling
enables principled approximations of integrals 
via the identity $\int g\,{\rm d}\genm_\pi = \int g\circ T\,{\rm
d}\genm_\eta$. In other words, a deterministic coupling provides a
simple way to simulate $\genm_\pi$ by pushing forward samples from
$\genm_\eta$ through a transport map $T$.
This idea, modulo some variations, %
has been exploited in a variety of statistical and machine
learning applications---some old, some new---including
random number generation \cite{marsaglia2000ziggurat},
variational inference 
\cite{el2012bayesian,schillings2014scaling,rezende2015variational,liu2016stein,marzouk2016introduction},
 the 
computation of model evidence 
\cite{meng2002warp},
model learning and density estimation
\cite{tabak2013family,laparra2011iterative,anderes2012general,stavropoulou2015parametrization},
non-Gaussian proposals for MCMC or importance sampling
\cite{parno2014transport,Bardsley2014,morzfeld2015parameter,heng2015gibbs,oliver2015metropolized,han2017stein},
multiscale modeling \cite{parno2015multiscale}, and
filtering
\cite{daum2008particle,chorin2009implicit,reich2013ensemble,reich2013nonparametric},
to name a few. Indeed there are infinitely many ways to transport one
measure to another \cite{villani2008optimal} and as many ways to
compute one.  Yet these maps are not equally easy to
characterize.

This paper establishes an explicit link
between the conditional independence structure of 
$(\genm_\eta,\genm_\pi)$ and the
existence of low-dimensional couplings induced by transport
maps that are {\it sparse} and/or 
{\it decomposable}.
These results 
can enhance
a 
wide array of numerical approaches to
the transportation of measures, including 
\cite{reich2011dynamical,el2012bayesian,tabak2013family,rezende2015variational,liu2016stein,bigoni2016monotone}, 
and thus facilitate simulation with
respect to complex distributions in high dimensions. 
We briefly discuss our main results below.
\paragraph{Sparse transports.}

A sparse transport is a map whose components do not depend on all
input variables. Section \ref{sec:sparse} derives tight bounds on the
sparsity pattern of the Knothe--Rosenblatt (KR) rearrangement (a
triangular transport map) based solely on the Markov structure of
$\genm_\pi$, provided that $\genm_\eta$ is a %
product 
measure (Theorem \ref{thm:sparsityRosenblatt}). This analysis shows
that the inverse of the KR rearrangement
is the
natural generalization to the {\it non-Gaussian} case of the
Cholesky factor  
of the precision 
matrix 
of a Gaussian
MRF---in that
both the
inverse KR rearrangement (a potentially nonlinear map) and the Cholesky factor (a linear map)
have 
the same sparsity pattern given 
target
measures
with the 
same Markov structure.
Thus %
the KR rearrangement can be used to extend 
well-known modeling and sampling techniques for
high-dimensional Gaussian MRFs  \cite{rue2005gaussian} 
to non-Gaussian fields (Section \ref{sec_gmrf}).
These results are particularly useful when constructing a transport map
from samples via convex optimization \cite{parno2015transport} and suggest
novel approaches to model learning \cite{morrison2017beyond} and high-dimensional filtering
\cite[Chapter 6]{spantini2017inference}.
Section \ref{sec:sparse} 
shows that sparsity is usually a feature of inverse transports,
while direct transports tend to be dense, even for the
most trivial Markov structures.
In fact, the sparsity of direct transports stems from
\textit{marginal}
 (rather than conditional) independence---a property
frequently exploited in localization schemes for
high-dimensional covariance estimation \cite{gaspari1999construction,hamill2001distance}.

\paragraph{Decomposable transports.}  
A decomposable map is
a function that can be written as the %
composition of {\it finitely} many 
low-dimensional 
maps that
are triangular up to a permutation---i.e., 
$T=T_1 \circ \cdots \circ T_\ell$, where  each %
$T_i$ differs from the identity only along a small subset of its
components and is a generalized triangular function as
defined in Section \ref{sec:decomp}.
Theorem \ref{thm:decompTrans} shows 
that every target measure %
whose Markov network admits
a graph decomposition can be coupled 
with a %
product
(reference) measure %
 via a decomposable map.
Decomposable maps are important because
they are 
much easier to represent than 
arbitrary multivariate functions on $\re^n$. 
In general, these maps are non-triangular, even though %
each map in the composition is generalized triangular.
The notion of a decomposable map
is different from the composition-of-maps approaches
advocated in the literature for the approximation of
transport maps (e.g., consider normalizing flows \cite{rezende2015variational} or Stein variational algorithms 
\cite{anderes2012general,liu2016stein,detommaso2018stein}, but also \cite{tabak2013family,laparra2011iterative}).   
In these approaches, very simple maps
$(M_i)_{i \ge 1}$
are composed in growing number
to define a transport map of increasing complexity, 
$M=M_1\circ  \cdots \circ M_k$.
The %
number of layers in $M$ %
depends on the
desired accuracy of the transport and can be arbitrarily large.
On the other hand, a decomposable coupling is induced by a {special}
transport map that can be written {\it exactly} as the composition of
finitely many maps, $T=T_1 \circ \cdots \circ T_\ell$, where each $T_i$ has a specific sparsity
pattern that makes it low-dimensional.
This definition does not specify a
representation for
$T_i$.
In fact, each $T_i$ 
could itself be
approximated
by the composition of
simple maps using {\it any} of the aforementioned techniques.
The advantage of targeting a decomposable transport  is the
fact that the $(T_i)$ are {\it guaranteed} to be low-dimensional.
\paragraph{Approximate Markov properties.} 
Sparsity and decomposability of certain transport maps are  induced 
by the %
Markov properties of the target measure. 
A natural question is: %
what  
happens
when 
$\genm_\pi$ satisfies some Markov properties only {\it approximately}?
In particular, let $\genm_\pi$ be Markov with respect to $\Gcb$, 
and assume that
there exists a %
measure
$\agenm \in \borelmp(\re^n)$
which is Markov with respect to a graph $\agengr$ that is
{\it sparser} %
than $\Gcb$ and
such that 
$\Dkl(\,\agenm \, \vert\vert \, \genm_\pi \,)<\varepsilon$, for some
$\varepsilon>0$.
For small $\varepsilon$,
we would be tempted to use %
$\agengr$ to
characterize %
couplings of  $(\genm_\eta, \genm_\pi)$
that are possibly
sparser
or more decomposable
than those associated with  $\Gcb$.
Concretely, if we are interested in a 
triangular transport that 
pushes forward $\genm_\eta$ to $\genm_\pi$,
we could minimize 
$\Dkl(\,T\push\,\genm_\eta \, \vert\vert \, \genm_\pi \,)$
over the set of %
maps that have the
same sparsity pattern as the
KR rearrangement between $\genm_\eta$ and $\agenm$.
Bounds on this sparsity pattern are given by 
Theorem \ref{thm:sparsityRosenblatt} using only
graph operations on
$\agengr$;
no explicit knowledge of $\agenm$ is required.
Alternatively, if we are interested in
decomposable transports that 
push forward $\genm_\eta$ to $\genm_\pi$, %
we could minimize 
$\Dkl(\,T\push\,\genm_\eta \, \vert\vert \, \genm_\pi \,)$
over the set of %
maps that  factorize
as any of the decomposable transports
between $\genm_\eta$ and $\agenm$.
The shapes of these low-dimensional factorizations
are given by Theorem \ref{thm:decompTrans}
using, once again, only graph operations on
$\agengr$.

Now let $\hat{\spaceMap}$ denote the set of maps
whose structure is constrained by $\agengr$ in terms of
sparsity or decomposability.
It is easy to show that
\begin{equation}
	\min_{T \in \hat{\spaceMap}}	
	\Dkl(\,T\push\,\genm_\eta \, \vert\vert \, \genm_\pi \,) < \varepsilon,
\end{equation}
which 
means %
that the price %
of assuming 
that the coupling is either sparser or more
decomposable than it ought to be
is just a small
error in the approximation of $\genm_\pi$.

Of course, the pending 
question is  whether
$\genm_\pi$ can be well approximated  by a measure
that satisfies additional Markov properties.
There is some work on this topic, 
e.g., \cite{johnson2008recursive,jog2015model,cheng2015efficient}---especially in the
case of Gaussian measures---but a more thorough 
investigation of the problem remains an open and
important direction for future work.
Interestingly, the transport map framework also
allows one to \textit{adaptively}
discover information about low-dimensional
couplings. %
For instance, one might start with a very sparse transport map %
and then incrementally decrease the sparsity level of the map
until the resulting approximation of $\genm_\pi$ becomes
satisfactory. The same can be done for decomposable transports.
See \cite{bigoni2016monotone} for some details on this idea.

\paragraph{Filtering and smoothing.} 
Section \ref{sec:compDecompTrans} shows how not only the
representation, but also the \textit{computation}, of a decomposable
map, $T=T_1 \circ \cdots \circ T_\ell$, can be broken into a
sequence of $\ell$ simpler steps, each associated with a
low-dimensional optimization problem whose solution yields $T_i$. 
We give a concrete example of this idea for %
filtering, smoothing, and joint state--parameter inference in
nonlinear and non-Gaussian
state-space models (Section \ref{sec:dataAss}).
In this context, 
Theorems \ref{thm:decompSmooth} and \ref{thm:decompJoint} introduce
variational
approaches for characterizing the full posterior distribution of %
the sequential inference problem, 
essentially by performing only 
recursive lag--1 smoothing with transport maps.
The proposed approaches consist of a {\it single} forward pass on the
state-space model, and generalize the square-root
Rauch-Tung-Striebel  smoother to non-Gaussian models  (see Section~\ref{sec:smoothLinGauss}).
In practice, we should think of Theorems \ref{thm:decompSmooth} and
\ref{thm:decompJoint} as providing ``meta-algorithms'' within which
all kinds of approximations can be introduced, e.g.,
linearizations of the forward model, 
restriction to 
linear maps, and approximate flows \cite{daum2008particle,liu2016stein}, 
to name a few. %
These approximations are the workhorse of
modern approaches to 
large-scale filtering, e.g., data assimilation in geophysical applications
\cite{sarkka2013bayesian,evensen2007data}, and may play a key role
in further instantiations of the ``meta-algorithms'' 
proposed in Section \ref{sec:dataAss}. 
Of course, it would be desirable
to complement such variational approximations with a rigorous 
error analysis, analogous to the analysis available for
SMC methods (e.g., \cite{crisan2002survey,del2004feynman,smith2013sequential}). 
It is also important to note that one can always use functionals like
\eqref{eq:var_diag} 
to estimate the quality of a given approximate map, or use the map itself to build
sophisticated proposals for sampling techniques like MCMC \cite{parno2014transport}.

A recent approach that constructs an approximation of the KR rearrangement 
for sequential inference is the ``Gibbs flow'' of
\cite{heng2015gibbs}; here, the authors define a proposal for SMC (or
MCMC) methods using the solution map of a discretized ordinary
differential equation (ODE) whose drift term depends only on the full
conditionals of the target distribution.
Evaluating the solution map only requires the evaluation of
one-dimensional integrals, and the action of this map implicitly
defines a transport, without any explicit parameterization of the
transformation. 
Several other filtering approaches in the literature, e.g., \cite{daum2012particle,yang2013feedback}, rely on the solution of ODEs that are different from \cite{heng2015gibbs}, but also inspired by ideas from mass transportation.
Implicit sampling for particle filters~\cite{chorin2009implicit} also implicitly constructs a transport map, from a standard Gaussian to a particular approximation of the filtering distribution; the action of this transport is realized by solving an optimization/root-finding problem for each sample \cite{Morzfeld2012}. 
One of the first efforts to use \textit{optimal} transport in
filtering is \cite{reich2013nonparametric}, which constructs an
optimal transport {plan} between an empirical approximation of the
forecast distribution (given by simulating the prior dynamic) and a
corresponding empirical approximation of the filtering distribution,
obtained by reweighing the forecast ensemble according to the
likelihood.  Thus, \cite{reich2013nonparametric} solves a {\it
  discrete} Kantorovich optimal transport
problem %
instead of a continuous problem for a transport map (cf.\ Section
\ref{sec:filt}). A linear transformation of the forecast ensemble
is then derived from the optimal plan. In this approach, the explicit
construction of couplings is used only to update the forecast
distribution, instead of the previous filtering
marginal. %

\paragraph{Further extensions.}
We envision many additional ways to extend the present work.
For instance, it would be interesting to investigate the low-dimensional
structure of deterministic couplings between pair of measures
$(\genm_\eta,\genm_\pi)$ that 
are not absolutely continuous and that need not be
defined
on the same  space $\re^n$.
Such couplings are usually induced by ``random'' maps and
can be particularly effective for approximating
multi-modal distributions; see the warp bridge transformations in
\cite{meng2002warp,wang2016warp} for some examples.

Finally, we emphasize that this paper characterizes some
classes of low-dimensional maps, but certainly not all.
In particular, low dimensionality need not stem from the Markov
properties of the underlying measures.
In ongoing work we are exploring the notion of low-rank couplings:
these are induced by transport maps that are low-dimensional up to a
rotation of the space, i.e., maps whose action is nontrivial only
along a low-dimensional subspace.  This type of
structure appears quite naturally in certain high-dimensional
Bayesian inference problems (e.g., inverse problems
\cite{stuart2010inverse} and spatial statistics) where the data may be
informative only about a few linear combinations of the latent
parameters \cite{spantini2014optimal,cui2014likelihood,spantini2016goal}.  Low-rank
structure can be detected via certain average derivative functionals
\cite{samarov1993exploring,constantine2014active} but cannot be
deduced, in general, from the Markov structure of
$(\genm_\eta,\genm_\pi)$.

%% file: sec_genKR.tex
In this section we first review the classical
notion of KR rearrangement 
\cite{rosenblatt1952remarks,knothe1957contributions},
and
then give a formal
definition for a  {\it generalized} KR rearrangement, i.e.,
a transport map that is lower triangular up
to a permutation.
A disclaimer: these transports can also be defined
under  weaker conditions than those
considered here, at the expense, however, of
some useful regularity 
(e.g., see \cite{bogachev2005triangular}). %

The following definition introduces the one-dimensional version of the KR-rearrangement, and it
is key to extend the transport to higher dimensions.
\begin{definition}[Increasing rearrangement on $\re$]
\label{def:increasRearr}
Let $\genm_\eta,\genm_\pi \in \borelmp(\re)$, and let $F,G$ be their respective cumulative distribution functions, i.e., 
$F(t)=\genm_\eta( (-\infty, t) )$ and
$G(t)=\genm_\pi( (-\infty, t) )$.
Then 
the increasing rearrangement on $\re$ is given by $T=G^{-1} \circ F$.
\end{definition}
Under the hypothesis of Definition \ref{def:increasRearr}, it is easy to see that
both $F$ and $G$ are homeomorphisms, and that $T$ is a strictly increasing
map that pushes forward $\genm_\eta$ to $\genm_\pi$ \cite{santambrogio2015optimal}.
\begin{definition}[Knothe-Rosenblatt rearrangement]
\label{def:KRrearr}
Given 
$\Xb \sim \genm_\eta$, $\Zb \sim \genm_\pi$, with 
$\genm_\eta,\genm_\pi \in \borelmp(\re^n)$, and 
a pair $\eta,\pi$ of strictly positive densities for 
$\genm_\eta$ and $\genm_\pi$, respectively,
the corresponding KR rearrangement
is a triangular map $T:\re^n \ra \re^n$ defined,
recursively, as follows. %
For all $\xb_{1:k-1}\in \re^{k-1}$, the map 
$\xi \mapsto T^k(\xb_{1:k-1}, \xi )$---the restriction of 
the $k$th component of $T$
onto its first $k-1$ 
inputs---is
defined as the   increasing rearrangement on $\re$
that 
pushes forward $\xi \mapsto \eta_{X_k \vert \Xb_{1:k-1}}( \xi \vert \xb_{1:k-1})$ to
$\xi \mapsto \pi_{Z_k\vert \Zb_{1:k-1}}( \xi \vert T^1(x_1), \ldots, T^{k-1}(\xb_{1:k-1}) )$, where
$\eta_{X_k \vert \Xb_{1:k-1}}$ and $\pi_{Z_k\vert \Zb_{1:k-1}}$ 
are conditional densities defined as in \eqref{eq:cond_density}.
\end{definition}
Notice that for any measure $\genm$ in $\borelmp(\re^n)$ there always
exists a strictly positive {\it version} of its density.
By considering such positive densities in Definition \ref{def:KRrearr}, we can
define the KR rearrangement on the entire $\re^n$ \cite{bogachev2005triangular}.
In fact, we should really think of Definition \ref{def:KRrearr} as providing
a possible {\it version} of the KR rearrangement 
(recall that the increasing triangular
transport is unique up to sets of measure zero \cite{santambrogio2015optimal}).
Since in this case $\genm_\pi$ is equivalent to the
Lebesgue measure 
($\genm_\pi(\Aset)=\int_{\Aset} \pi(\xb)\,\lebm({\rm d}\xb) = 0 
\Rightarrow \lebm(\Aset)=0$ if $\pi>0$ a.e.), 
the component \eqref{eq:componentMap}
is also absolutely continuous on all compact intervals 
\cite[Lemma 2.4]{bogachev2005triangular}.
As a result, the rearrangement can be used to
define general change of variables 
as well as
pullbacks and pushforwards with respect to arbitrary densities, as shown
by the following lemma
adapted from \cite{bogachev2005triangular}. %

\begin{lem} \label{lem:GenchangeVarTriWeak}
Let $T$ be an increasing triangular bijection on $\re^n$ such that the functions
\begin{equation}  
  \xi \mapsto T^k( x_1 ,\ldots, x_{k-1}, \xi )
\end{equation}
are absolutely continuous on all compact intervals for a.e.
$(x_1,\ldots,x_{k-1})\in \re^{k-1}$. %
Then for any 
integrable function $\varphi$, it holds:
\begin{equation}
	\int \varphi(\yb)\,{\rm d}\yb =
	\int \varphi(T(\xb))\,\det \nabla T(\xb)\,{\rm d}\xb,	
\end{equation}
where $\det \nabla T \coloneqq \prod_{k=1}^n \partial_k T^k$.
In particular, if %
$\genm_{\rho}$ is a measure on $\re^n$ with density $\rho$,  then
we also have
$T \pull \genm_{\rho} \ll \lebm$ with density (a.e.): %
\begin{equation} \label{eq:GenpullbackDensTri}
  T\pull \rho(\xb) = \rho(T(\xb))\,\det \nabla T(\xb).
\end{equation}
\end{lem}
The lemma can also be applied to the inverse KR rearrangement $T^{-1}$ to show that
$T \push \genm_{\rho} \ll \lebm$, where
the form of the corresponding
pushforward density $T\push \rho$ is given by  replacing $T$
with $T^{-1}$ in \eqref{eq:GenpullbackDensTri}. 
We will use these results extensively in the proofs of Appendix \ref{sec:proofs}.
Notice, however, that Lemma \ref{lem:GenchangeVarTriWeak} does not hold for
a generic triangular function: the map must be somewhat regular, in the sense
specified by the lemma. See \cite{bogachev2005triangular} for an in depth discussion.

We now give a constructive definition for a generalized 
KR rearrangement.
\begin{definition}[Generalized 
Knothe-Rosenblatt rearrangement]
\label{def:genKRrearr}
Given 
$\Xb \sim \genm_\eta$, $\Zb \sim \genm_\pi$, with 
$\genm_\eta,\genm_\pi \in \borelmp(\re^n)$, 
a pair $\eta,\pi$ of strictly positive densities for 
$\genm_\eta$ and $\genm_\pi$, respectively, and
a permutation $\sigma$ of
$\mathbb{N}_n$, the corresponding $\sigma$-generalized KR rearrangement is 
a $\sigma$-triangular map\footnote{
See Definition \ref{def:genTri}.	
} $T:\re^n \ra \re^n$ %
defined at any $\xb \in \re^n$ using the following recursion in $k$.
The map 
$\xi \mapsto T^{\sigma(k)}(x_{\sigma(1)},\ldots,x_{\sigma(k-1)} , \xi )$
is
defined as the  increasing rearrangement on $\re$
that 
pushes forward $\xi \mapsto \eta_{X_{\sigma(k)} \vert \Xb_{\sigma(1:k-1)} }( \xi \vert 
\xb_{\sigma(1:k-1)})$ to
\begin{equation}
	\xi \mapsto \pi_{Z_{\sigma(k)}
\vert \Zb_{\sigma(1:k-1)}}( \xi \vert T^{\sigma(1)}(x_{\sigma(1)}), \ldots, 
T^{\sigma(k-1)}(\xb_{\sigma(1:k-1)}) ), 
\end{equation}
where
$\xb_{\sigma(1:k-1)} = x_{\sigma(1)},\ldots,x_{\sigma(k-1)}$.
\end{definition}
Existence %
of a generalized KR
rearrangement follows trivially from its definition. %
Moreover, the transport map satisfies all the regularity properties discussed
for the classic  KR rearrangement, including Lemmas 
\ref{lem:changeVarTriWeak} and \ref{lem:GenchangeVarTriWeak}.
Thus we will often cite these two results when dealing with
generalized KR rearrangements in our proofs.
The following lemma shows that the computation of a generalized KR rearrangement is also 
essentially no different than
the computation of a lower triangular transport (and thus all the
discussion of Section \ref{sec:compTransport} readily applies).

\begin{lem}
	Given 
	$\genm_\eta,\genm_\pi \in \borelmp(\re^n)$,
	let $T$ be a $\sigma$-generalized KR rearrangement that pushes
	forward $\genm_\eta$ to $\genm_\pi$ for some permutation $\sigma$.
	Then $T = Q_\sigma^\top \circ T_\ell \circ Q_\sigma$ a.e., where
	$Q_\sigma \in \re^{n \times n}$ is a matrix representing the permutation,
	i.e., $(Q^\sigma)_{ij}=(\eb_{\sigma(i)})_j$, and where
	$T_\ell$ is a (lower triangular) KR rearrangement that
	pushes forward 
	$(Q_\sigma)\push \genm_\eta$ to 
	$(Q_\sigma)\push \genm_\pi$.
\end{lem}
\begin{proof}
If $T_\ell$ pushes forward $(Q_\sigma)\push \genm_\eta$ to 
$(Q_\sigma)\push \genm_\pi$,
then $\genm_\eta \circ Q_\sigma^\top \circ T^{-1}_\ell = \genm_\pi \circ Q_\sigma^\top$, and so
$T = Q_\sigma^\top \circ T_\ell \circ Q_\sigma$ must 
push forward $\genm_\eta$ to $\genm_\pi$.
Moreover, notice that %
$T^{\sigma(k)}(\xb)=T_\ell^k(  \xb^\top \eb_{\sigma(1)} , \ldots,  \xb^\top \eb_{\sigma(k)})$, which
shows that $T$ is a monotone increasing $\sigma$-generalized triangular function
(see Definition \ref{def:genTri}).
The lemma then follows by $\genm_\eta$-uniqueness of a KR rearrangement.
\end{proof}

%% file: sec_proofs.tex
In this section we collect the proofs of the main results and claims of the paper,
together with useful additional lemmas to support the technical
derivations.
\medskip

\noindent
{\bf Proof of Lemma \ref{lem:pairwiseIndep}.}
The general solution of $\partial^2_{i,j} \log \pi = 0$ on $\re^n$ is given by
$\log \pi(\zb) = g(\zb_{1:i-1}, \zb_{i+1:n}) + h(\zb_{1:j-1}, \zb_{j+1:n})$ for
some  functions $g,h:\re^{n-1} \ra \re$. Hence 
$Z_i \orth Z_j \,\vert\, \Zb_{\Vc\setminus (i,j)}$ \cite{lauritzen1996graphical}.
Conversely, if $Z_i \orth Z_j \,\vert\, \Zb_{\Vc\setminus (i,j)}$, then
$\pi$---which is the density of 
$\genm_\pi$ with respect to a tensor product Lebesgue measure \cite{lauritzen1996graphical}---must factor as 
\begin{equation}
\pi = \pi_{Z_i \vert \Zb_{\Vc\setminus (i,j)} }
	   \pi_{Z_j \vert \Zb_{\Vc\setminus (i,j)} }
	   \pi_{\Zb_{\Vc\setminus (i,j)} },	
\end{equation}
so that
$\partial^2_{i,j} \log \pi = 0$ on $\re^n$.	   
\hfill $\square$
\smallskip

\noindent
{\bf Proof of Theorem \ref{thm:sparsityRosenblatt}.}
We begin with Part \ref{thm:sparsityRosenblatt:inverse} of the theorem.
Let $\eta,\pi$ be a pair of strictly positive densities for
$\genm_\eta$ and $\genm_\pi$, respectively 
(these positive densities exist since the measures are fully supported).
Now consider a {\it version} of the 
KR rearrangement, $S$, that pushes forward $\genm_\pi$ to $\genm_\eta$
as given by Definition \ref{def:KRrearr}
for the pair $\eta,\pi$ (Appendix \ref{sec:genKR}).
By definition, 
and for all
$\zb_{1:k-1}\in \re^{k-1}$, the map 
$\xi \mapsto S^k(\zb_{1:k-1}, \xi )$ is the %
monotone increasing 
rearrangement that pushes forward 
$\xi \mapsto \pi_{Z_k \vert \Zb_{1:k-1}}( \xi \vert \zb_{1:k-1})$ to the marginal 
$\eta_{X_k}$ (recall that $\genm_\eta$ is a tensor product measure).
Moreover, it follows easily
from \cite[Proposition 3.17]{lauritzen1996graphical}, that each
marginal $\pi_{\Zb_{1:k}}$---or better yet, the
corresponding measure---is globally Markov with 
respect to 
$\Gcb^k$, and that
$\pi_{\Zb_{1:k}}(\zb_{1:k})\,\pi_{\Cc}(\zb_{\Cc}) = 
\pi_{ \Zb_k , \Zb_\Cc }( \zb_k, \zb_{\Cc} ) \, \pi_{\Zb_{1:k-1}}(\zb_{1:k-1})$, where
$\Cc \coloneqq \neigh( k , \Gcb^k)$, possibly empty.
Thus, the conditional $\pi_{Z_k \vert \Zb_{1:k-1}}( \zb_k \vert \zb_{1:k-1})$
is constant along any input $\zb_j$  with $j \notin \neigh (k, \Gcb^k) $.
For any such $j$, $S^k$ must be constant along its $j$th input, so
that $(j,k)\in \widehat{\sparse}_S$.

Part \ref{thm:sparsityRosenblatt:direct} of the theorem follows similarly.
Consider the KR rearrangement, $T$, that pushes forward $\genm_\eta$ to 
$\genm_\pi$ as given by Definition \ref{def:KRrearr}.
For all
$\xb_{1:k-1}\in \re^{k-1}$, the map 
$\xi \mapsto T^k(\xb_{1:k-1}, \xi )$ is the 
monotone increasing 
rearrangement that pushes forward 
$\eta_{X_k}$ to
\begin{equation}
\xi \mapsto \pi_{Z_k \vert \Zb_{1:k-1}}( \xi \vert 
T^1(x_1),\ldots,T^{k-1}(\xb_{1:k-1})).	
\end{equation}
We already know that 
$\pi_{Z_k \vert \Zb_{1:k-1}}( \zb_k \vert \zb_{1:k-1})$ can
only depend (nontrivially) on $\zb_k$ and on 
$\zb_j$ for $j \in \neigh (k, \Gcb^k)$.
Hence, if none of the components $T^i$, with $i\in \neigh (k, \Gcb^k)$,
depends on the $j$th input, then 
$T^k$ is constant along its
$j$th
input as well,
so that $(j,k)\in \widehat{\sparse}_T$.

For Part \ref{thm:sparsityRosenblatt:inclusion}, let $(j,k)\in \widehat{\sparse}_T$. 
Then, by definition, 
$(j,i) \in \widehat{\sparse}_T$ for all $i \in \neigh (k, \Gcb^k)$,
which also implies that $j \notin \neigh (k, \Gcb^k)$ since
$j \neq i$ for all $(j,i)\in \widehat{\sparse}_T$. 
Hence $(j,k)\in \widehat{\sparse}_S$ and this shows the inclusion
$\widehat{\sparse}_T \subset \widehat{\sparse}_S$.

These arguments show that there exists at least a {\it version} of the
KR rearrangement that is exactly at least as sparse as predicted by the theorem.
\hfill $\square$
\smallskip

The following lemma specializes the results of Theorem 
\ref{thm:sparsityRosenblatt}[Part \ref{thm:sparsityRosenblatt:direct}]
to the case of I-maps $\Gcb$ with a disconnected component, and will be
useful in the proofs of Section \ref{sec:decomp}.

\begin{lem} \label{lem:productTarget}
Let $\Xb \sim \genm_\eta$, $\Zb \sim \genm_\pi$ with 
$\genm_\eta,\genm_\pi \in \borelmp(\re^n)$ and 
$\genm_\eta$ tensor product measure, and let
$\sigma$ be any permutation of $\mathbb{N}_n$.
Moreover, assume that
$\genm_\pi$ is globally Markov with respect to $\Gcb=(\Vc,\Ec)$, and
assume that there exists a nonempty set $\Ac \subset \Vc \simeq \mathbb{N}_n$ such that
$\Zb_{\Ac} \orth \Zb_{\Vc \setminus \Ac}$ and %
$\Zb_{\Ac} = \Xb_{\Ac}$ in distribution.
Then 
the $\sigma$-generalized KR rearrangement  $T$ given by
Definition \ref{def:genKRrearr}  (for a pair $\eta,\pi$ of nonvanishing densities for
$\genm_\eta$ and $\genm_\pi$, respectively)
is low-dimensional with respect to $\Ac$, i.e.,
\begin{enumerate} 
	\item  $T^k(\xb)=x_k$ for $k \in \Aset$
	\label{lem:productTarget_partI}
	\smallskip

	\item $\partial_j T^k = 0$ for $j\in \Ac$ and $k \in \Vc \setminus \Ac$.
	\label{lem:productTarget_partII}
\end{enumerate}
\end{lem}
\begin{proof}
It suffices to prove the lemma for a lower
triangular KR rearrangement; the result for an arbitrary $\sigma$ then follows
trivially.
If $\Ac = \Vc$, then $T$ is simply the identity map.
Thus we assume that
$\Vc \setminus \Ac$ is nonempty.

We begin with Part \ref{lem:productTarget_partI} of the lemma and use the results
of Theorem \ref{thm:sparsityRosenblatt}[Part \ref{thm:sparsityRosenblatt:direct}] 
to characterize the
sparsity of the rearrangement.
Let $k \in \Aset$ and notice that $\neigh(k,\Gcb^k)=\emptyset$, where 
$\Gcb^k$ is the marginal graph defined in Theorem \ref{thm:sparsityRosenblatt}.
Thus $(j,k)\in \widehat{\sparse}_T \subset \sparse_T$ for all $j=1,\ldots,k-1$,
so that $T^k(\xb)=x_k$ for all $k \in \Aset$.

Now let us focus on Part \ref{lem:productTarget_partII} and  prove that
$(j,k)\in \widehat{\sparse}_T$  for all $j\in \Ac$ and $k \in \Vc \setminus \Ac$.
We proceed by contradiction. 
Assume that there exists some pair $(j,k)\in \Ac \times (\Vc \setminus \Ac)$
such that $(j,k)\notin  \widehat{\sparse}_T$.
In particular, let $\Kcb$ be the set of $k \in \Vc \setminus \Ac$ for which
there exists at least a $j\in \Ac$ such that
$(j,k)\notin \widehat{\sparse}_T$.
Clearly $\Kcb$ is nonempty and finite.
Let $s$ be the minimum integer in $\Kcb$, and let 
$j\in \Aset$ be a corresponding index for which $(j,s)\notin \widehat{\sparse}_T$.
In this case, by Theorem \ref{thm:sparsityRosenblatt}[Part \ref{thm:sparsityRosenblatt:direct}], 
there must exist an 
$i\in \neigh(s , \Gcb^s)$ such that $(j,i) \notin \widehat{\sparse}_T$.
Now there are two cases: either $i \in \Aset$ (for which we reach a
contradiction by part \ref{lem:productTarget_partI} of the lemma) or
$i \in \Vc \setminus \Aset$.
In the latter case, we also reach a contradiction since $i<s$ and $s$ was defined
as the smallest index for which $(j,s) \notin \widehat{\sparse}_T$ for some $j\in \Aset$.
\end{proof}

\noindent
{\bf Proof of Theorem \ref{thm:decompTrans}.}
For notational convenience, we drop the subscript and superscript $i$ from
$\genm_i$, $\pi_i$, $\Zb^i$, and $\Gcb^i$.
Consider a factorization of $\pi$ of the form
\begin{equation} \label{eq:factPiDecomp}
	\pi(\zb)=\frac{1}{\mathfrak{c}} \, \psi_{\Aset \cup \Sset}(\zb_{\Aset \cup \Sset}) \, 
		\psi_{\Sset \cup \Bset }(\zb_{ \Sset \cup \Bset }),
\end{equation}
where $\psi_{\Aset \cup \Sset}$ 
is strictly positive and integrable, 
with $\mathfrak{c} = \int \psi_{ \Aset \cup \Sset} < \infty$.
A factorization like \eqref{eq:factPiDecomp} always exist since 
$\genm$ factorizes according to $\Gcb$---thus $\Gcb$ is 
an I-map for $\genm$---and since
$(\Aset, \Sset, \Bset)$
is a proper decomposition of $\Gcb$. %
For instance, one can set %
$\psi_{\Aset \cup \Sset} = \pi_{\Zb_{\Aset \cup \Sset}}$,
$\mathfrak{c}=1$, 
and $\psi_{\Sset \cup \Bset } = \pi_{\Zb_{\Bset} \vert \Zb_{\Sset}}$
since $\Zb_{\Aset} \orth \Zb_{\Bset} \vert \Zb_{\Sset}$ and since
$\pi$ is a nonvanishing density of $\genm$.
However, this is not the only possibility.
See Section \ref{sec:dataAss} for important
examples where it is not convenient to assume that %
$\psi_{\Aset \cup \Sset}$ corresponds to a marginal of $\pi$.
This proves Part \ref{thm:decompTrans_partFactorizationPi} 
of the theorem.

By \cite[Proposition 3.16]{lauritzen1996graphical}, 
we can rewrite $\psi_{\Sset \cup \Bset}$ as:
\begin{equation} \label{eq:factPsi}
\psi_{\Sset \cup \Bset }(\zb_{ \Sset \cup \Bset }) = \prod_{\Cc \in \Ccb_{\Sset \cup \Bset}} 
\psi_{\Cc}(\zb_{\Cc})
\end{equation}
for some nonvanishing
functions $(\psi_{\Cc})$, where $\Ccb_{\Sset \cup \Bset}$ denotes the set
of maximal cliques of the subgraph $\Gcb_{\Sset \cup \Bset}$.
Since $\Sset$ is a fully connected %
separator set (possibly empty) for $\Aset$ and $\Bset$, 
the maximal cliques of $\Gcb_{\Sset \cup \Bset}$ are precisely the maximal cliques of
$\Gcb$ that are a subset of $\Sset \cup \Bset$.
We are going to use  \eqref{eq:factPsi} shortly.
Define $\widetilde{\pi}:\re^n \ra \re$ as 
$\widetilde{\pi}(\zb) = \psi_{\Aset \cup \Sset}(\zb_{\Aset \cup \Sset}) 
\, \eta_{\Xb_{\Bset}}(\zb_{\Bset}) / \mathfrak{c}$, and notice that
$\widetilde{\pi}$ is a nonvanishing probability density.
Denote the corresponding measure by $\widetilde{\genm}\in \borelmp(\re^n)$. 
For an arbitrary permutation $\sigma$ of $\mathbb{N}_n$ that
satisfies
\eqref{eq:formPerm},
let $\lmap_i$ be the %
$\sigma$-generalized KR 
rearrangement 
that
pushes forward $\genm_\eta$ to $\widetilde{\genm}$ as given 
by Definition \ref{def:genKRrearr} in Appendix \ref{sec:genKR}.
By Lemma \ref{lem:productTarget},  %
$\lmap_i$ is low-dimensional with respect to
$\Bset$ (Part \ref{thm:decompTrans_partLeftMap}
of the theorem). 
To
see this, 
let $\widetilde{\Zb} \sim \widetilde{\genm}$, 
and 
notice %
that $\widetilde{\Zb}_{\Bset} \orth \widetilde{\Zb}_{\Aset \cup \Sset}$ %
and
$\widetilde{\Zb}_{\Bset} = \Xb_{\Bset}$ in distribution.
By Lemma \ref{lem:GenchangeVarTriWeak}, we can write
a density of the pullback measure $\lmap_i\pull \genm$ as:
\begin{eqnarray} \label{eq:factPullback}
	\lmap_i\pull \,\pi  & = & \pi \circ \lmap_i  \, |\det \nabla \lmap_i | \\
 				 & = &  \left( \lmap_i \pull \,\widetilde{\pi} \right) \,
 				 		\frac{ \prod_{\Cc \in \Ccb_{\Sset \cup \Bset}} \psi_{\Cc}
 				 		\circ \lmap_i^{\Cc} }{
				  		\eta_{\Xb_{\Bset}} 
				  		}  \nonumber \\
				 & = &  \eta_{\Xb_{\Aset \cup \Sset}} \,  
				 		\prod_{\Cc \in \Ccb_{\Sset \cup \Bset}} \psi_{\Cc}
				 		\circ \lmap_i^{\Cc},  \nonumber   
\end{eqnarray}
where we used the identity $\pi = \widetilde{\pi}\,\psi_{\Sset \cup \Bset} / \eta_{\Xb_{\Bset}}$ 
together with \eqref{eq:factPsi} and the fact that
$\lmap_i^k(\xb) = x_k$ for $k\in\Bset$ (Part \ref{thm:decompTrans_partLeftMap}), 
and where, for any
$\Cc= \{ c_1 , \ldots, c_{\ell} \}\in \Ccb_{\Sset \cup \Bset}$ with
$\psi_{\Cc}(\zb_{\Cc})=\psi_{\Cc}(z_{c_1}, \ldots, z_{c_{\ell}})$,
$ \lmap_i^{\Cc}$ is a map $\re^n \ra \re^\ell$ given by 
$\xb \mapsto (\lmap_i^{c_1}(\xb), \ldots, \lmap_i^{c_{\ell}}(\xb))$.

If $\Zb' \sim \lmap_i\pull \,\genm$, then 
\eqref{eq:factPullback} shows that 
$\Zb'_{\Aset} \orth \Zb'_{\Sset \cup \Bset}$ and that
$\Zb'_{\Aset} = \Xb_{\Aset}$ in distribution 
(Part \ref{thm:decompTrans_partRVimap} of the theorem).
Moreover, from the factorization in \eqref{eq:factPullback}, we can
easily construct a graph for which $\lmap_i\pull \,\genm$ factorizes:
it 
suffices to 
consider 
the scope of the factors 
$(\psi_{\Cc}\circ \lmap_i^{\Cc})$, i.e., 
the indices of the input variables that each $\psi_{\Cc}\circ \lmap_i^{\Cc}$ can depend on.
Recall that for a $\sigma$-triangular map, the $\sigma(k)$th component can only
depend
on the variables $x_{\sigma(1)},\ldots,x_{\sigma(k)}$.
For each $\Cc \in \Ccb_{\Sset \cup \Bset}$ there are two
possibilites:
Either  
$\Cc \cap \Sset = \emptyset$, in which case the scope of
$\psi_{\Cc}\circ \lmap_i^{\Cc}$ is simply $\Cc$ since
$\lmap_i^k(\xb)=x_k$ for $k \in \Bset$.
Or $\Cc \cap \Sset$ is nonempty, in which case 
let $j_{\Cc}$ be the maximum integer
$j$ such that $\sigma(j)\in \Cc \cap \Sset$, and notice that the scope of
$\psi_{\Cc}\circ \lmap_i^{\Cc}$ is simply
$\Cc \cup \{ \sigma(1),\ldots,\sigma(j_{\Cc})\}$.
Thus, we can modify $\Gcb$ to obtain an I-map for $\lmap_i\pull \, \genm$ as follows:
(1) Remove any edge that is incident to any node in $\Aset$ 
because of Part \ref{thm:decompTrans_partRVimap}.
(2) For every maximal clique $\Cc$ in $\Gcb$ that is a subset of 
$\Sset \cup \Bset$ and that has nonempty intersection with $\Sset$,  
turn $\Cc \cup \{ \sigma(1),\ldots,\sigma(j_{\Cc})\}$ into a clique.
This proves Part \ref{graphSparsification_partImap} of the theorem.

Now let $\mathfrak{R}_i$ be the set of maps $\re^n\ra \re^n$ that
are low-dimensional with respect to $\Aset$ and that push forward
$\genm_\eta$ to $\lmap_i\pull\,\genm$.
$\mathfrak{R}_i$ is nonempty.
To see this, let $\rmap$ be the 
$\sigma$-generalized KR rearrangement that 
pushes forward
$\genm_\eta$ to $\lmap_i\pull\,\genm$, for an
arbitrary permutation $\sigma$, as given by Definition \ref{def:genKRrearr} 
(for the pair of nonvanishing densities $\eta$ and $\lmap_i\pull \pi$).
By Part \ref{thm:decompTrans_partRVimap} and Lemma \ref{lem:productTarget}, 
$\rmap$ is low-dimensional with respect to $\Aset$. 
Thus $\rmap \in \mathfrak{R}_i$ (Part \ref{thm:decompTrans_partRightMap} of the theorem).

Let $\decset_i \coloneqq \lmap_i \circ \mathfrak{R}_i$ be the set of maps
that can be written as $\lmap_i \circ \rmap$ for some $\rmap \in \mathfrak{R}_i$.
By construction, each $T \in \decset_i$ pushes 
forward $\genm_\eta$ to $\genm$
(part \ref{thm:decompTrans_partT} of the theorem).
\hfill $\square$
\smallskip

In the following corollary every symbol should be interpreted as in Theorem \ref{thm:decompTrans}.

\begin{cor}
	\label{cor:lowdim_lmap}
	Given the hypothesis of Theorem \ref{thm:decompTrans}, 
	assume 	that there exists %
	$\Aind \subset \Aset$ such that 
	$\Zb^i_{\Aind}\orth \Zb^i_{\Vc \setminus \Aind}$ and 
	$\Zb^i_{\Aind} =  \Xb_{\Aind} $ in distribution.
	Then  $\lmap_i$ is low-dimensional with respect to $\Aind \cup \Bset$, 
	while each $T\in \decset_i$ 
	is low-dimensional with respect to $\Aind$.	
\end{cor}
\begin{proof}
	By Theorem \ref{thm:decompTrans}[Part \ref{thm:decompTrans_partLeftMap}], 
	$\lmap_i$ is low-dimensional with respect to $\Bset$, while
	Lemma \ref{lem:productTarget} shows that $\lmap_i$ is also  
	low-dimensional with respect to $\Aind$.
	Moreover, notice that if $\Aind$ is nonempty, then for all 
	$T = \lmap_i \circ \rmap$ in $\decset_i$, we have
	$T^k(\xb)=x_k$ for $k\in \Aind$ since $\lmap_i^k (\xb)=x_k$ and
	$\rmap^k(\xb)=x_k$ for $k\in \Aind$  
	(Theorem \ref{thm:decompTrans}[Parts \ref{thm:decompTrans_partRightMap}]).
	Additionally, $\partial_j \, T^k = 0$ for $j \in \Aind$ and 
	$k\in \Vc \setminus \Aind$. To see this, notice that 
	$T^k(\xb)=\lmap_i^k(\rmap(\xb))$ and that the following two facts hold:
	(1) The component
	$\lmap_i^k$, for $k\in \Vc \setminus \Aind$,
	does not depend
	on   input variables whose index is in $\Aind$ since
	$\lmap_i$ is low-dimensional with respect to $\Aind$; 
	(2)
	The $\ell$th component of $\rmap$ with $\ell \notin \Aind$ also does not depend on 
	$\xb_{\Aind}$
	since $\rmap$ is low-dimensional with respect to $\Aset$ 
	(Theorem \ref{thm:decompTrans}[Parts \ref{thm:decompTrans_partRightMap}]).
	Hence, $T$ must be a low-dimensional map with respect to $\Aind$. 
\end{proof}
\smallskip

\noindent
{\bf Proof of Lemma \ref{lem:recursiveDecomp}.}
Let $\genm_\eta , \genm_i, \pi_i, \Gcb^i, \decset_i,\lmap_i, \mathfrak{R}_i$, 
and $\Gcb^{i+1}$ be defined as in 
Theorem \ref{thm:decompTrans}
for a proper decomposition 
$(\Aset_i,\Sset_i,\Bset_i)$ of  $\Gcb^i$, 
a permutation $\sigma_i$ that satisfies \eqref{eq:formPerm}, and for any
factorization \eqref{eq:formFactoriz} of $\pi_i$.

We first want to prove that 
$\Sset_i \cup \Bset_i$ is fully connected in $\Gcb^{i+1}$ if and only if
the decomposition $(\Aset_{i+1}, \Sset_{i+1}, \Bset_{i+1})$  of
Part \ref{lem:recursiveDecomp_partPos} does not exist.
Let us start with one direction.
Assume that a decomposition like the one in Part 
\ref{lem:recursiveDecomp_partPos}
does not exist, despite the
possibility to add edges to $\Gcb^{i+1}$ in $\Vc \setminus \Aset_i$. 
We want to show that in this case 
$\Sset_i \cup \Bset_i$ must be a clique in $\Gcb^{i+1}$.
Since $\Bset_i$ is nonempty, 
there are two possibilities: either $|\Sset_i \cup \Bset_i|=1$ or
$|\Sset_i \cup \Bset_i|>1$.
If $|\Sset_i \cup \Bset_i|=1$, then $\Sset_i \cup \Bset_i$ consists of a single node
and thus it is a trivial clique.
If $|\Sset_i \cup \Bset_i|>1$, then $\Sset_i \cup \Bset_i$ contains at least
two nodes.
In this case, let us proceed by contradiction and assume that $\Sset_i \cup \Bset_i$
is not fully connected in $\Gcb^{i+1}=(\Vc, \Ec^{i+1})$, i.e.,
there exist a pair of nodes
$\alpha,\beta \in \Sset_i \cup \Bset_i$ such that 
$(\alpha,\beta)\notin \Ec^{i+1}$. 
Let $\Aset_{i+1}=\Aset_i \cup \{\alpha \}$, $\Bset_{i+1}= \{\beta\}$, and
$\Sset_{i+1} = ( \Vc \setminus \Aset_{i+1} ) \setminus \Bset_{i+1}$.
Notice that $(\Aset_{i+1},\Sset_{i+1},\Bset_{i+1})$ forms a partition
of $\Vc$, with nonempty $\Aset_{i+1},\Bset_{i+1}$  and with $\Aset_{i+1}$
strict superset of $\Aset_{i}$.
Moreover $\Sset_{i+1}$ must be a separator set for  
$\Aset_{i+1}$ and $\Bset_{i+1}$ since $(\alpha,\beta)\notin \Ec^{i+1}$ and
$\Aset_i$ is disconnected from $\Sset_i \cup \Bset_i$ 
in $\Gcb^{i+1}$ (Theorem \ref{thm:decompTrans}[Part \ref{graphSparsification_partImap}]).
Now there are two cases: If $\Sset_{i+1} = \emptyset$, then 
$(\Aset_{i+1},\Sset_{i+1},\Bset_{i+1})$ is a decomposition 
that satisfies Part \ref{lem:recursiveDecomp_partPos} of the lemma (contradiction).
If $\Sset_{i+1} \neq \emptyset$, then we can always add enough edges to
$\Gcb^{i+1}$ in $\Sset_i \cup \Bset_i \supset \Sset_{i+1}$ 
in order to make $\Sset_{i+1}$ fully connected.
Also in this case, the resulting decomposition 
$(\Aset_{i+1},\Sset_{i+1},\Bset_{i+1})$ 
satisfies Part \ref{lem:recursiveDecomp_partPos} of the lemma and thus leads to
a  contradiction.

Now the reverse direction.
Assume that $\Sset_i \cup \Bset_i$ is a clique in $\Gcb^{i+1}$.
If $|\Sset_i \cup \Bset_i|=1$, then the decomposition of
Part \ref{lem:recursiveDecomp_partPos} cannot exist since both
$\Aset_{i+1}\setminus \Aset_i$ and $\Bset_{i+1}$ should  be nonempty.
Hence, let $|\Sset_i \cup \Bset_i|>1$ and  
proceed by contradiction.
That is, let $(\Aset_{i+1},\Sset_{i+1},\Bset_{i+1})$  be
a proper decomposition that
satisfies Part \ref{lem:recursiveDecomp_partPos} of the lemma.
Notice that this decomposition must have been achieved without 
adding any edge to $\Gcb^{i+1}$ in
$\Sset_i \cup \Bset_i$ since this set is already fully connected.
By hypothesis, there must exist $\alpha,\beta$ such that
$\alpha \in \Aset_{i+1}\setminus \Aset_i$ and $ \beta \in \Bset_{i+1}$.
However, both $\alpha$ and $\beta$ are also in $\Sset_i \cup \Bset_i$, and so
they must be connected by an edge in $\Gcb^{i+1}$. Hence,
$\Sset_{i+1}$ is not a separator set for $\Aset_{i+1}$
and $\Bset_{i+1}$ (contradiction).

The latter result proves directly Part \ref{lem:recursiveDecomp_partNeg} of the lemma.
Moreover, it shows that if
$\Sset_i \cup \Bset_i$ is not a clique in $\Gcb^{i+1}$, then
there exists a proper decomposition $(\Aset_{i+1},\Sset_{i+1},\Bset_{i+1})$
of $\Gcb^{i+1}$, 
where $\Aset_{i+1}$ is a strict superset of $\Aset_i$,
obtained, possibly, by adding edges to $\Gcb^{i+1}$ in order
to turn $\Sset_{i+1}$ into a clique. 
Note that even if we add edges to $\Gcb^{i+1}$, 
$\lmap_i\pull \,\genm_i$ still factorizes according to the resulting
graph, which is then an I-map for  $\lmap_i\pull \,\genm_i$.
Moreover we can really only add edges in $\Vc \setminus \Aset_i$ since
$\Aset_i$ must be a strict subset of $\Aset_{i+1}$, and thus
$\Aset_i$ remains disconnected from $\Sset_i \cup \Bset_i$ in 
$\Gcb^{i+1}$.
Let $\decset_{i+1},\lmap_{i+1},\mathfrak{R}_{i+1}$ be defined as in 
Theorem \ref{thm:decompTrans} 
for the pair of measures $\genm_\eta , \genm_{i+1} = \lmap_i\pull \,\genm_i$, the
decomposition $(\Aset_{i+1},\Sset_{i+1},\Bset_{i+1})$ of $\Gcb^{i+1}$, 
a permutation  $\sigma_{i+1}$ that satisfies \eqref{eq:formPerm}, and for any
factorization \eqref{eq:formFactoriz} %
(note that $\lmap_i\pull \,\genm_i \in \borelmp(\re^n)$ 
by Theorem \ref{thm:decompTrans}[Part \ref{thm:decompTrans_partRightMap}]).
Fix $T\in \decset_{i+1}$.
By Theorem \ref{thm:decompTrans}[Part \ref{thm:decompTrans_partT}], 
$T$ pushes
forward $\genm_\eta$ to $\genm_{i+1} = \lmap_i\pull \,\genm_i$.
Moreover, if $\Zb^{i+1} \sim \lmap_i\pull \,\genm_i$, then by 
Theorem \ref{thm:decompTrans}[Part \ref{thm:decompTrans_partRVimap}]
we have
$\Zb^{i+1}_{\Aset_i} \orth \Zb^{i+1}_{\Sset_i \cup \Bset_i}$ and 
$\Zb^{i+1}_{\Aset_i} = \Xb_{\Aset_i}$ in distribution.
Then by Corollary \ref{cor:lowdim_lmap}
it must also be that $T$
is  low-dimensional with respect to $\Aset_i$.
Thus $T \in \mathfrak{R}_i$, and this proves the inclusion 
$\mathfrak{R}_i \supset \decset_{i+1}$.

Now fix
any $T \in \lmap_i \circ \lmap_{i+1} \circ \mathfrak{R}_{i+1} 
= \lmap_i \circ \decset_{i+1}$. 
It must be that $T = \lmap_i \circ g$ for some $g \in \decset_{i+1} \subset \mathfrak{R}_i$,
so that $T \in \lmap_i \circ \mathfrak{R}_i$, which shows the inclusion
$\lmap_i \circ \mathfrak{R}_i \supset 
\lmap_i \circ \lmap_{i+1} \circ \mathfrak{R}_{i+1}$
(Part \ref{lem:recursiveDecomp_partInclusion} of the lemma).
By Corollary \ref{cor:lowdim_lmap},
we have that
$\lmap_{i+1}$ is low-dimensional with respect to $\Aset_i \cup \Bset_{i+1}$, and so
its effective dimension is  bounded above by
$|\Vc \setminus (\Aset_i \cup \Bset_{i+1})|=|(\Aset_{i+1} \setminus \Aset_{i}) \cup \Sset_{i+1}|$
(Part \ref{lem:recursiveDecomp_lmap}).
Finally, by Theorem \ref{thm:decompTrans}[Part \ref{thm:decompTrans_partRightMap}],
each $R \in \mathfrak{R}_{i+1}$ is low-dimensional with respect to 
$\Aset_{i+1}$, and so its effective dimension is bounded by 
$|\Vc \setminus \Ac_{i+1}|$
(Part \ref{lem:recursiveDecomp_rmap}).
\hfill $\square$
\smallskip

\noindent
{\bf Proof of Theorem \ref{thm:decompSmooth}.} 
For the sake of clarity, we divide the proof in two parts:
First, we show that the maps $(\submap_i)_{i \ge 0}$ are well-defined.
Then, we prove the remaining claims of the theorem.

The maps $(\submap_i)_{i \ge 0}$ are well-defined as long as, for instance, 
we show that  $\pi^i$ is a   probability density %
for all $i\ge 0$, and as long
as there exist permutations $(\sigma_i)$ that guarantee the block
upper triangular structure of \eqref{eq:thm:decompSmooth:upperTriMap}.
As for the permutations, it suffices to consider 
$\sigma = \sigma_1 = \sigma_2 = \cdots$ with
$\sigma(\mathbb{N}_{2n}) = \{ 2n, 2n-1, \ldots, 1 \}$, i.e.,  upper
triangular maps. 
(If $n>1$, then there is some freedom in the choice of $\sigma$.)
As for the targets $(\pi^i)$, we now show that  
$\pi^i$ is a nonvanishing density
and that 
the marginal
$\int \pi^i(\zb_i,\zb_{i+1}) \, {\rm d}\zb_{i} =  
\pi_{\Zb_{i+1} \vert \yb_{0:i+1}}$, for
all $ i\ge 0$, using an induction argument over $i$.
For the base case ($i=0$), just notice that
\begin{equation} \label{thm:decompSmooth:c0}
\mathfrak{c}_0 = \int \widetilde{\pi}^0(\zb_0 , \zb_1) \, {\rm d}\zb_{0:1} = 
\pi_{\Yb_0,\Yb_1}(\yb_0,\yb_1) < \infty,
\end{equation}
so that
$\pi^0 = \widetilde{\pi}^0 / \mathfrak{c}_0 > 0$
is a valid density. %
Moreover, we have the desired marginal, i.e., %
\begin{equation}
	\int \pi^0(\zb_0 , \zb_{1}) \, {\rm d}\zb_0=
	\int \pi_{\Zb_0,\Zb_1\vert\Yb_0,\Yb_1}(\zb_0 , \zb_{1}\vert \yb_0, \yb_1) \, {\rm d}\zb_0 = 
	\pi_{\Zb_1\vert\Yb_0,\Yb_1}(\zb_{1}\vert \yb_0, \yb_1).	
\end{equation}
Now assume that $\pi^i$ is a nonvanishing density %
and that the marginal
$\int \pi^i(\zb_i,\zb_{i+1}) \, {\rm d}\zb_{i} =  
\pi_{\Zb_{i+1} \vert \yb_{0:i+1}}$ for some $i>0$. 
The map $\submap_i$ is then well-defined.
In particular, by definition of KR rearrangement,
the submap $\submap_i^1$ pushes forward 
$\eta_{\Xb_{i+1}}$ to the marginal 
$\int \pi^i(\zb_i,\zb_{i+1}) \, {\rm d}\zb_{i}$.
Moreover, by Lemma
\ref{lem:GenchangeVarTriWeak}, we have: %
\begin{eqnarray} \label{thm:decompSmooth:ck}
\mathfrak{c}_{i+1} & = & 
		\int \eta_{\Xb_{i+1}}(\zb_{i+1}) \,
		\widetilde{\pi}^{i+1}( \submap_{i}^1(\zb_{i+1}) , 
		\zb_{i+2})
		\, {\rm d}\zb_{i+1:i+2} \\ \nonumber
		& = & 
		\int \pi_{\Zb_{i+2},\Yb_{i+2}\vert \Yb_{0:i+1}}(\zb_{i+2}, \yb_{i+2} \vert \yb_{0:i+1})
		\, {\rm d}\zb_{i+2} \\ \nonumber
		& = &
\pi_{\Yb_{i+2}\vert \Yb_{0:i+1}}(\yb_{i+2} \vert \yb_{0:i+1}) < \infty, \nonumber
\end{eqnarray}
where we used the change of variables 
$\xb_{i+1} = \submap_{i}^1(\zb_{i+1})$ and the fact that
$(\submap_i^1)\push \, \eta_{\Xb_{i+1}} = \pi_{\Zb_{i+1}\vert\yb_{0:i+1}}$ 
(induction hypothesis).
Thus 
$\pi^{i+1}$ is a nonvanishing density
and
by \eqref{thm:decompSmooth:ck}  we can easily verify that $\pi^{i+1}$ has the
desired marginal, i.e., 
$\int \pi^{i+1}(\zb_{i+1},\zb_{i+2}) \, {\rm d}\zb_{i+1} =  
\pi_{\Zb_{i+2} \vert \yb_{0:i+2}}$. 
This argument completes the induction step  and shows that not only the maps
$(\submap_i)_{i \ge 0}$ are well-defined---together with the maps $(T_i)_{i \ge 0}$ in 
\eqref{eq:thm:decompSmooth:upperTriMapEmb}---but also that
$(\submap_i^1)\push \, \eta_{\Xb_{i+1}} = \pi_{\Zb_{i+1}\vert\yb_{0:i+1}}$
for all $i \ge 0$ (Part \ref{thm:decompSmooth:partFilt} of the theorem).

Now we move to Part \ref{thm:decompSmooth:partFull} of the theorem and
use another induction 
argument over $k\ge 0$. 
For the base case ($k=0$), notice that
$\mathfrak{T}_0 = T_0 = \submap_0$, and that, by definition,
$\submap_0$ pushes forward $\eta_{\Xb_{0},\Xb_{1}}$ to 
$\pi^0 = \pi_{\Zb_0,\Zb_1\vert \yb_0, \yb_1}$.

Assume that $\mathfrak{T}_k$ pushes forward $\eta_{\Xb_{0:k+1}}$ to 
$\pi_{\Zb_{0:k+1}\vert \yb_{0:k+1}}$ for
some $k>0$ ($\mathfrak{T}_k$ is well-defined for all
$k$ since the maps $(T_i)_{i \ge 0}$ in 
\eqref{eq:thm:decompSmooth:upperTriMapEmb}  are also well-defined),
and notice that
\begin{equation}
	\pi_{\Zb_{0:k+2}\vert \yb_{0:k+2}} = \pi_{\Zb_{0:k+1}\vert \yb_{0:k+1}} \,
	\frac{  \pi_{\yb_{k+2} \vert \Zb_{k+2} } \,
			\pi_{\Zb_{k+2}\vert\Zb_{k+1} } }
			 {\pi_{\yb_{k+2}\vert \yb_{0:k+1}}} = 
	\pi_{\Zb_{0:k+1}\vert \yb_{0:k+1}} \,
	\frac{ \widetilde{\pi}^{k+1}  }{\mathfrak{c}_{k+1}},
\end{equation}
where we used \eqref{thm:decompSmooth:ck} and 
the definition of the collection $(\widetilde{\pi}^i)$.
Let $\mathfrak{T}_{k+1} = T_0 \circ \cdots \circ T_{k+1}$ be defined
as in Part \ref{thm:decompSmooth:partFull} of the theorem, and observe that
$\mathfrak{T}_{k+1} =
A_{k+1} \circ T_{k+1}$ with
\begin{equation}  
A_{k+1}( \xb_{0:k+2} ) = 
	\left[\begin{array}{l}
		\mathfrak{T}_k(\xb_{0:k+1})\\[\spacinglines] 
		\xb_{k+2}\\
		 \end{array}\right], \quad
		T_{k+1}( \xb_{0:k+2} ) = \left[\begin{array}{l}
		\xb_0 \\ 
		\vdots \\
		\xb_{k} \\[\spacinglines]
		\submap_{k+1}^0( \xb_{k+1} , \xb_{k+2}) \\[\spacinglines]
		\submap_{k+1}^1( \xb_{k+2}) \\
		\end{array}\right].
\end{equation}
Thus the following hold:
\begin{eqnarray} \label{thm:decompSmooth:pullbacks}
 \mathfrak{T}_{k+1}\pull\,\pi_{\Zb_{0:k+2}\vert \Yb_{0:k+2}}  
 & = & 
 T_{k+1}\pull \left( \left(\mathfrak{T}_{k}\pull\,\pi_{\Zb_{0:k+1}\vert \yb_{0:k+1}} \right) \,
 \frac{\pi^{k+1}}{\eta_{\Xb_{k+1}}}  \right) \\
 & = & 
 T_{k+1}\pull \left( \eta_{\Xb_{0:k}} \, \pi^{k+1} \right) \nonumber \\
 & = &
 \eta_{\Xb_{0:k}} \,  \submap_{k+1}\pull \,\pi^{k+1}  =  \eta_{\Xb_{0:k+2}} , \nonumber
\end{eqnarray}
where we used the fact that
by Lemma \ref{lem:GenchangeVarTriWeak} 
(applied
iteratively)
it must be that 
$(A_{k+1} \circ T_{k+1})\pull \rho = T_{k+1} \pull \, A_{k+1} \pull \rho$
for all densities $\rho$.
(Notice that $A_{k+1}$ 
is the composition of functions which are trivial embeddings into
the identity map of KR rearrangements that couple a pair
of measures in $\borelmp(\re^n \times \re^n)$, %
and thus each map in the composition satisfies
the hypothesis of Lemma \ref{lem:GenchangeVarTriWeak}.)
In particular, 
$(\mathfrak{T}_{k+1})\push\, \eta_{\Xb_{0:k+2}} = \pi_{\Zb_{0:k+2}\vert \yb_{0:k+2}}$ 
(Part \ref{thm:decompSmooth:partFull} of the theorem).

Now notice that each $\mathfrak{T}_{k}$ can also be written as
\begin{equation}  
\mathfrak{T}_{k}( \xb_{0:k+1} ) = 
	\left[\begin{array}{l}
		B_k(\xb_{0:k+1})\\[\spacinglines] 
		\overline{\submap}_k(\xb_k ,\xb_{k+1})\\
		 \end{array}\right]
\end{equation}  
for a multivariate function $B_k$---whose particular form is not relevant to this argument---and for a map, $\overline{\submap}_k$, %
defined in \eqref{eq:thm:decompSmooth:upperTriMapSmooth} as a function
on $\re^n \times \re^n$. 
Since $(\mathfrak{T}_{k})\push\,\eta_{\Xb_{0:k+1}} = 
\pi_{\Zb_{0:k+1}\vert \yb_{0:k+1}}$, the map  
$\overline{\submap}_k$ must also push forward
$\eta_{\Xb_{k},\Xb_{k+1}}$ to the
lag-$1$ smoothing marginal $\pi_{\Zb_k,\Zb_{k+1}\vert \yb_{0:k+1}}$.
This proves Part \ref{thm:decompSmooth:partSmot} of the theorem.

For Part \ref{thm:decompSmooth:partEvidence}, just notice that
\begin{equation} \label{thm:decompSmooth:data}
	\pi_{\Yb_{0:k+1}}(\yb_{0:k+1}) = 
	\pi_{\Yb_0,\Yb_1}(\yb_0,\yb_1)\,
	\prod_{i=1}^k \pi_{\Yb_{i+1}\vert \Yb_{0:i}}(\yb_{i+1}\vert \yb_{0:i})	=
	\prod_{i=0}^k \mathfrak{c}_i,
\end{equation}
where we used both \eqref{thm:decompSmooth:c0} and \eqref{thm:decompSmooth:ck}.
\hfill $\square$
\smallskip

\noindent
{\bf Proof of Lemma \ref{lem:kalmanRec}.} 
First a remark about notation:
we denote by $\Gauss( \xb ; \mub, \Sigmab )$ the
density (as a function of $\xb$) of a Gaussian  with mean $\mub$ and covariance $\Sigmab$.

Now let $k>0$ and
notice that $\pi_{\Zb_{k+1} \vert \Zb_k}(\zb_{k+1}\vert \zb_k) = 
\Gauss( \zb_{k+1}; \Fb_k\,\zb_k, \Qb_k )$,
$\pi_{\Yb_{k+1} \vert \Zb_{k+1}}(\yb_{k+1} \vert \zb_{k+1}) = 
\Gauss( \yb_{k+1}; \Hb_{k+1}\,\zb_{k+1} , \Rb_{k+1} )$ and
$\eta_{\Xb_k}(\zb_k) = \Gauss( \zb_k; 0, {\bf I} )$.
By definition of the target $\pi^k$ in Theorem \ref{thm:decompSmooth}, we have:
\begin{eqnarray*}
	\pi^k(\zb_k, \zb_{k+1}) & = &	\eta_{\Xb_k}(\zb_k) \,
									\pi_{\Yb_{k+1} \vert \Zb_{k+1}}(\yb_{k+1} \vert \zb_{k+1}) \,
									\pi_{\Zb_{k+1} \vert \Zb_k}(\zb_{k+1}\vert 
									\submap_{k-1}^1(\zb_k) )  \\
							& = &   \Gauss( \zb_k; 0, {\bf I} ) \,
									\Gauss( \yb_{k+1}; \Hb_{k+1}\,\zb_{k+1} , \Rb_{k+1} )\,
									 \\
							&  & 		
									\Gauss( \zb_{k+1};
									\Fb_k\,( \Cb_{k-1} \, \zb_k + \cb_{k-1}), \Qb_k ) \\
							& \propto & 	
									\exp( -\frac{1}{2} \zb^\top \Jb\,\zb + \zb^\top \hb ),	
\end{eqnarray*}
where $\zb = (\zb_k , \zb_{k+1})\in \re^{2n}$, and where
$\Jb \in \re^{2n \times 2n} , \hb \in \re^{2n}$ are defined as
\begin{equation}
 \Jb = 	\left[\begin{array}{cc}
	\Jb_{11} & \Jb_{12} \\[3pt] 
	\Jb_{12}^\top & \Jb_{22}
	\end{array}\right], \quad
 \hb = \left[\begin{array}{c}
	\hb_{1}  \\ 
	\hb_{2}
	\end{array}\right],
\end{equation}
with:
\begin{equation}
	\begin{cases} 
		\Jb_{11} = {\bf I} + \Cb_{k-1}^\top\,\Fb_k^\top \, \Qb_k^{-1}\,\Fb_k\,\Cb_{k-1}    \\ 
		\Jb_{12} = - \Cb_{k-1}^\top\,\Fb_k^\top\,\Qb_k^{-1} \\
		\Jb_{22} = \Qb_k^{-1} + \Hb_{k+1}^\top\,\Rb_{k+1}^{-1}\,\Hb_{k+1} \\
		\hb_1 =  \Jb_{12}\,\Fb_k\,\cb_{k-1} \\
		\hb_2 =  \Qb_k^{-1}\,\Fb_k\,\cb_{k-1} +
				 \Hb_{k+1}^\top \, \Rb_{k+1}^{-1}\,\yb_{k+1}.  %
	\end{cases}		
\end{equation}
In particular, we can rewrite $\pi^k$ in {\it information form} \cite{koller2009probabilistic} 
as
$\pi^k(\zb) = \Gauss^{-1}( \zb ; \hb , \Jb )$.
Moreover we know by Theorem \ref{thm:decompSmooth}[Part \ref{thm:decompSmooth:partFilt}], that
the submap $\submap_k^1(\zb_{k+1}) = \Cb_k\,\zb_{k+1} + \cb_{k}$ pushes
forward $\eta_{\Xb_{k+1}}$ to the filtering marginal $\pi_{\Zb_{k+1}\vert\yb_{0:k+1}}$.
Hence $(\cb_k , \Cb_k)$ should be, respectively,  
the mean and a square root of the covariance of 
$\pi_{\Zb_{k+1}\vert\yb_{0:k+1}}$---thus the output of any
square-root Kalman filter at time $k+1$.
Now we just need to determine the submap 
$\submap^0_k(\zb_k, \zb_{k+1}) = \Ab_k \,\zb_k + \Bb_k \, \zb_{k+1} + \ab_{k}$.
Given that $\submap_k$ is a block upper triangular function, the map
$\zb_k \mapsto \submap^0_k(\zb_k, \zb_{k+1})$ should push forward
$\eta_{\Xb_k}$ to $\zb_k \mapsto \pi^k_{\Zb_k \vert \Zb_{k+1}}( 
\zb_k \vert \submap_k^1(\zb_{k+1}) )$.
Notice that $\pi^k_{\Zb_k \vert \Zb_{k+1} }(\zb_k \vert \zb_{k+1}) = 
\Gauss^{-1}( \zb_k ; \hb_1 - \Jb_{12}\,\zb_{k+1}, \Jb_{11} ) = 
\Gauss(\zb_k; \Jb_{11}^{-1}(\hb_1 - \Jb_{12}\,\zb_{k+1}) , \Jb_{11}^{-1} )$.
Hence 
$\pi^k_{\Zb_k \vert \Zb_{k+1} }( \zb_k \vert \submap_k^1(\zb_{k+1}) ) = 
\Gauss(\zb_k; 
\Jb_{11}^{-1} \Jb_{12}(\Fb_k\,\cb_{k-1} - \Cb_k\,\zb_{k+1}-\cb_k) , \Jb_{11}^{-1} )$, and so:
\begin{equation}
\submap^0_k(\zb_k, \zb_{k+1}) = 
\Jb_{11}^{-1} \Jb_{12}(\Fb_k\,\cb_{k-1} - \Cb_k\,\zb_{k+1}-\cb_k) + \Jb_{11}^{-1/2} \zb_k.
\end{equation}
Simple algebra then leads to \eqref{eq:recursionKalm}. 
\hfill $\square$
\smallskip

\noindent
{\bf Proof of Theorem \ref{thm:decompJoint}.}
We use a very similar argument to Theorem \ref{thm:decompSmooth}.
We first show that the maps $(\submap_i)_{i \ge 0}$ are well-defined.
These maps are well-defined as long as, for instance,  we show that
$\pi^i$ is a probability density 
for all $i\ge 0$, and as long
as there exist permutations $(\sigma_i)$ that guarantee the generalized
block
triangular structure of \eqref{eq:thm:decompJoint:genTriMap}.
As for the permutations, it suffices to consider 
$\sigma = \sigma_1 = \sigma_2 = \cdots$ with
$\sigma(\mathbb{N}_{\dhyp+2n}) = \{1,\ldots,\dhyp, \dhyp+2n, \dhyp+2n-1, \ldots, \dhyp+1 \}$.
As for the targets $(\pi^i)$, we now
use a (complete) induction argument over $i$
to show that,
for
all $ i\ge 0$,   
$\pi^i$ is a nonvanishing density and
$\int \pi^i(\zb_{\vhyps},\zb_i,\zb_{i+1}) \, {\rm d}\zb_{i} =  
A_i\pull\,\pi_{\vhyp, \Zb_{i+1} \vert \yb_{0:i+1}}(\zb_{\vhyps},\zb_{i+1})$ 
for a map $A_i$ defined
on $\re^{\dhyp} \times \re^n$ as
		\begin{equation}   
			A_i( \xb_{\vhyps} , \xb_{i+1} ) = 
			\left[\begin{array}{l}
			\mathfrak{T}^{\vhyp}_{i-1}(  \xb_{\vhyps} )\\[\spacinglines] 
			\xb_{i+1}
			\end{array}\right],%
		\end{equation}
with $\mathfrak{T}^{\vhyp}_{i-1}(\xb_{\vhyps}) = \xb_{\vhyps}$ if $i=0$.

For the base case ($i=0$), just notice that
$\mathfrak{c}_0 =  \pi_{\Yb_0,\Yb_1}(\yb_0,\yb_1) < \infty$,
so that
$\pi^0 = \widetilde{\pi}^0 / \mathfrak{c}_0 > 0$ is a valid density.
Moreover, we have the desired marginal, i.e., %
\begin{equation}
	\int \pi^0(\zb_{\vhyps},\zb_0 , \zb_{1}) \, {\rm d}\zb_0=
	\pi_{\vhyp,\Zb_1\vert\Yb_0,\Yb_1}(\zb_{\vhyps},\zb_{1}\vert \yb_0, \yb_1) =
	A_0\pull\,\pi_{\vhyp,\Zb_1\vert\yb_0,\yb_1}(\zb_{\vhyps},\zb_{1}),	
\end{equation}
since $A_0$ is the identity map on $\re^{\dhyp}\times \re^n$.
Now assume that %
$\pi^j$ is a nonvanishing density
for all 
$j \le i$ (complete induction) with $i>0$, and that the marginal
$\int \pi^i(\zb_{\vhyps},\zb_i,\zb_{i+1}) \, {\rm d}\zb_{i} =  
A_i\pull\,\pi_{\vhyp, \Zb_{i+1} \vert \yb_{0:i+1}}(\zb_{\vhyps},\zb_{i+1})$.
Under this hypothesis,
the maps $(\submap_j)_{j \le i}$ are well-defined, and so are
$A_i,A_{i+1}$ since $\mathfrak{T}^{\vhyp}_{i} = 
\submap_0^{\vhyp} \circ \cdots \circ  \submap_{i}^{\vhyp}$.
Before checking the integrability of $\pi^{i+1}$, 
notice that by definition of $\submap_i$ (a KR
rearrangement), the map
$B_i$, given by %
		\begin{equation}   
			B_i( \xb_{\vhyps} , \xb_{i+1} ) = 
			\left[\begin{array}{l}
			\submap^{\vhyp}_{i}(  \xb_{\vhyps} )\\[\spacinglines] 
			\submap^1_i( \xb_{\vhyps} , \xb_{i+1}  )
			\end{array}\right],%
		\end{equation}
pushes forward $\eta_{\Xb_{\vhyp},\Xb_{i+1}}$ to the marginal
$\int \pi^i(\zb_{\vhyps},\zb_i,\zb_{i+1}) \, {\rm d}\zb_{i}$, which equals
$A_i\pull\,\pi_{\vhyp, \Zb_{i+1} \vert \yb_{0:i+1}}$
(inductive hypothesis), i.e., %
$(B_i)\push \,\eta_{\Xb_{\vhyp},\Xb_{i+1}} = 
A_i\pull\,\pi_{\vhyp, \Zb_{i+1} \vert \yb_{0:i+1}}$.	
In particular, it must also be that
$(A_i \circ B_i)\push\,\eta_{\Xb_{\vhyp},\Xb_{i+1} }
=
\pi_{ \vhyp, \Zb_{i+1} \vert \yb_{0:i+1}}$, where
$A_i \circ B_i$ corresponds precisely to the
map $\widetilde{\submap}_i$ defined in
\eqref{eq:thm:decompJoint:TriMapFilt}, so that
$(\widetilde{\submap}_i) \push\,\eta_{\Xb_{\vhyp},\Xb_{i+1} }
=
\pi_{ \vhyp, \Zb_{i+1} \vert \yb_{0:i+1}}$.

Now we can prove that $\mathfrak{c}_{i+1} < \infty$ using the following identities:
\begin{eqnarray} \label{thm:decompJoint:ck}
\mathfrak{c}_{i+1} & = & 
	\int 
	\eta_{\Xb_{\vhyp},\Xb_{i+1}}(\zb_{\vhyps},\zb_{i+1})  \\
	&  & 
	\widetilde{\pi}^{i+1}(\mathfrak{T}^{\vhyp}_{i}(\zb_{\vhyps}),
	\submap_{i}^1(\zb_{\vhyps},\zb_{i+1}), \zb_{i+2})
	\, {\rm d}\zb_{\vhyps} \, {\rm d}\zb_{i+1:i+2} \nonumber \\ 
	&  = & 
	\int 
	(\widetilde{\submap}_i)\push\,\eta_{\Xb_{\vhyp},\Xb_{i+1}}(\xb_{\vhyps},\xb_{i+1}) \,
	\pi_{\Zb_{i+2} \vert \Zb_{i+1}, \vhyp}(\zb_{i+2}\vert \xb_{i+1}, \xb_{\vhyps}) 
		\nonumber \\
	&  &
	\pi_{\Yb_{i+2} \vert \Zb_{i+2} , \vhyp }(\yb_{i+2}\vert \zb_{i+2}, \xb_{\vhyps}) 
	\, {\rm d}\xb_{\vhyps} \, {\rm d}\xb_{i+1} \,{\rm d}\zb_{i+2}
		\nonumber \\
	&  = &
	\int
	\pi_{ \vhyp, \Zb_{i+1} \vert \yb_{0:i+1}}(\xb_\theta, \xb_{i+1}) \\ \nonumber
	&   & %
	\pi_{\Zb_{i+2},\Yb_{i+2} \vert \Zb_{i+1}, \vhyp}(\zb_{i+2}, \yb_{i+2} \vert \xb_{i+1}, \xb_\theta)
	\, {\rm d}\xb_{\vhyps} \, {\rm d}\xb_{i+1} \,{\rm d}\zb_{i+2}
		\nonumber \\ 
	&  = & 
	\pi_{\Yb_{i+2}\vert \Yb_{0:i+1}}(\yb_{i+2} \vert \yb_{0:i+1}) < \infty,  \nonumber
\end{eqnarray}
where we used the change of variables:
		\begin{equation}    \label{thm:decompJoint:changeVar}
			\left[\begin{array}{l}
				\xb_{\vhyps}  \\[\spacinglines] 
			    \xb_{i+1}
			\end{array}\right]
			= 
			\left[\begin{array}{l}
				\mathfrak{T}^{\vhyp}_{i}(  \zb_{\vhyps} )\\[\spacinglines] 
				\submap^1_i( \zb_{\vhyps} , \zb_{i+1}  )
			\end{array}\right] =
			\widetilde{\submap}_i( \zb_{\vhyps} , \zb_{i+1}  ),
		\end{equation}
and the fact that 
$(\widetilde{\submap}_i) \push\,\eta_{\Xb_{\vhyp},\Xb_{i+1}} = 
\pi_{ \vhyp, \Zb_{i+1} \vert \yb_{0:i+1}}$
(induction hypothesis).
(The change of variables in \eqref{thm:decompJoint:changeVar} is
valid 
for the following reason:
the map $\widetilde{\submap}_i$ can be factorized as the composition of
$i+1$ (generalized) triangular functions, all that fit the hypothesis of
Lemma \ref{lem:GenchangeVarTriWeak}, so that \eqref{thm:decompJoint:changeVar}
should really be interpreted as a sequence of $i+1$ change of
variables---each associated with one
map in the composition and justified by Lemma \ref{lem:GenchangeVarTriWeak}.)
Therefore %
$\pi^{i+1}$ is a nonvanishing density.
Following the same derivations as in \eqref{thm:decompJoint:ck},
it is not hard to show that $\pi^{i+1}$ has also the
desired marginal, i.e., 
\begin{equation}
\int \pi^{i+1}(\zb_{\vhyps},\zb_{i+1} , \zb_{i+2}) \, {\rm d}\zb_{i+1} 
	=
	A_{i+1}\pull\,\pi_{ {\vhyp} ,\Zb_{i+2} \vert \yb_{0:i+2}}
	(\zb_{\vhyps}, \zb_{i+2} ).%
\end{equation}
This argument completes the induction step  and shows that not only the maps
$(\submap_i)_{i \ge 0}$ are well-defined---together with the maps $(T_i)_{i \ge 0}$ in 
\eqref{eq:thm:decompJoint:genTriMapEmb}---but also that
$(\widetilde{\submap}_i^1)\push \, \eta_{\Xb_{\vhyp},\Xb_{i+1}} = 
\pi_{\vhyp,\Zb_{i+1}\vert\yb_{0:i+1}}$
for all $i \ge 0$ (Part \ref{thm:decompJoint:partFilt} of the theorem). 

Now we prove Part \ref{thm:decompJoint:partFull} of the theorem using another induction 
argument on $k\ge 0$. For the base case ($k=0$), notice that
$\mathfrak{T}_0 = T_0 = \submap_0$, and that, by definition,
$\submap_0$ pushes forward $\eta_{\Xb_{\vhyp},\Xb_0,\Xb_{1}}$ to 
$\pi^0=\pi_{ \vhyp,\Zb_0,\Zb_1\vert \yb_0, \yb_1}$.

Assume that $\mathfrak{T}_k$ pushes forward $\eta_{\Xb_{\vhyp},\Xb_{0:k+1}}$ to 
$\pi_{\vhyp,\Zb_{0:k+1}\vert \yb_{0:k+1}}$ for
some $k>0$ ($\mathfrak{T}_k$ is well-defined for all
$k$ since the maps $(T_i)_{i \ge 0}$ in 
\eqref{eq:thm:decompJoint:genTriMapEmb}  are also well-defined),
and notice that
\begin{equation}
	\pi_{{\vhyp},\Zb_{0:k+2}\vert \yb_{0:k+2}} = 
	\pi_{{\vhyp},\Zb_{0:k+1}\vert \yb_{0:k+1}} \,
	\frac{  \pi_{\yb_{k+2} \vert \Zb_{k+2} ,  {\vhyp}} \,
			\pi_{\Zb_{k+2}\vert\Zb_{k+1} , {\vhyp} } }
			 {\pi_{\yb_{k+2}\vert \yb_{0:k+1}}}=
	\pi_{{\vhyp},\Zb_{0:k+1}\vert \yb_{0:k+1}} \, 
	\frac{ \widetilde{\pi}^{k+1}  }{\mathfrak{c}_{k+1}}, \nonumber
\end{equation}
where we used \eqref{thm:decompJoint:ck} and 
the definition of the collection $(\widetilde{\pi}^i)$.
Let $\mathfrak{T}_{k+1} = T_0 \circ \cdots \circ T_{k+1}$ be defined
as in Part \ref{thm:decompJoint:partFull} of the theorem, and observe that
$\mathfrak{T}_{k+1} =
C_{k+1} \circ T_{k+1}$ with
\begin{equation}  
C_{k+1}(\xb_{\vhyps}, \xb_{0:k+2} ) = 
	\left[\begin{array}{l}
		\mathfrak{T}_k(\xb_{\vhyps},\xb_{0:k+1})\\[\spacinglines] 
		\xb_{k+2}\\
		 \end{array}\right], \quad
		T_{k+1}(\xb_{\vhyps},\xb_{0:k+2} ) = \left[\begin{array}{l}
		\submap_{k+1}^{\vhyp}( \xb_{\vhyps}) \\[\spacinglines]
		\xb_0 \\ 
		\vdots \\
		\xb_{k} \\[\spacinglines]
		\submap_{k+1}^0( \xb_{\vhyps}, \xb_{k+1} , \xb_{k+2}) \\[\spacinglines]
		\submap_{k+1}^1( \xb_{\vhyps}, \xb_{k+2}) \\
		\end{array}\right]. \nonumber
\end{equation}
Thus the following hold:
\begin{eqnarray} %
 \mathfrak{T}_{k+1}\pull\,\pi_{\vhyp,\Zb_{0:k+2}\vert \yb_{0:k+2}}  
 & = & 
 T_{k+1}\pull \left( \left( \mathfrak{T}_{k}\pull\,\pi_{\vhyp,\Zb_{0:k+1}\vert \yb_{0:k+1}} \right) \,
 \frac{\pi^{k+1}}{\eta_{\Xb_{\vhyp},\Xb_{k+1}}}  \right) \\
 & = & 
 T_{k+1}\pull \left( \eta_{\Xb_{0:k}} \, \pi^{k+1} \right) \nonumber \\
 & = &
 \eta_{\Xb_{0:k}} \,  \submap_{k+1}\pull \,\pi^{k+1}  =  
 \eta_{\Xb_{\vhyp},\Xb_{0:k+2}} , \nonumber
\end{eqnarray}
where we used the fact that
by Lemma \ref{lem:GenchangeVarTriWeak} (applied
iteratively) it must be that 
$(C_{k+1} \circ T_{k+1})\pull \rho = T_{k+1} \pull \, C_{k+1} \pull \rho$
for all densities $\rho$. 
(Notice that $C_{k+1}$ 
is the composition of functions which are trivial embeddings into
the identity map of KR rearrangements that couple a pair
of measures in $\borelmp(\re^\dhyp \times \re^n \times \re^n)$, %
and thus each map in the composition satisfies
the hypothesis of Lemma \ref{lem:GenchangeVarTriWeak}.)
Thus 
$(\mathfrak{T}_{k+1})\push\,\eta_{\Xb_{\vhyp},\Xb_{0:k+2}} = 
\pi_{{\vhyp},\Zb_{0:k+2}\vert \yb_{0:k+2}}$, and this concludes
the induction argument and the proof of Part \ref{thm:decompJoint:partFull} of the theorem.

The proof of Part \ref{thm:decompJoint:partEvidence} follows from
$\mathfrak{c}_0 =  \pi_{\Yb_0,\Yb_1}(\yb_0,\yb_1)$,
\eqref{thm:decompJoint:ck}, and \eqref{thm:decompSmooth:data}. \hfill $\square$
\smallskip

%% file: sec_algs.tex
\hrevone{
Here we digest the smoothing and joint state-parameter inference methodologies discussed in Section
  \ref{sec:dataAss} into a handful of algorithms, described with pseudocode.
Algorithms \ref{alg:computeMap} and \ref{alg:regresMap} below are
building blocks: they describe, respectively, how to approximate a
transport map given an (unnormalized) target
density, and how to project a given transport map onto a
set of monotone transformations.
Algorithm \ref{alg:assimilate} shows how to build a
recursive approximation of  $\pi_{\vhyp,\Zb_{0:k+1} \vert 
\yb_{0:k+1} }$---i.e., the full Bayesian
solution to the problem of sequential inference in state-space models with static parameters---using a decomposable transport map.
See details in Section \ref{sec:joint}.
For simplicity, we always use a standard normal reference process $\eta_{\Xb}$, 
although more general choices are possible. 
Algorithm \ref{alg:SampleSmoothing} shows how to sample from 
the resulting approximation of the joint distribution 
$\pi_{\vhyp,\Zb_{0:k+1} \vert 
\yb_{0:k+1} }$,
whereas
Algorithm \ref{alg:SampleFiltering}  
focuses on a particular ``filtering'' marginal, i.e.,
$\pi_{ \vhyp ,\Zb_{k+1} \vert \yb_{0:k+1}}$.
The problem of sequential inference on state-space models {\it without} static parameters (see Section \ref{sec:filt}) can be tackled via
a simplified version of Algorithm \ref{alg:assimilate}, wherein 
the formal  dependence on $\vhyp$ is dropped.
The actual implementation of these algorithms is
available online at \url{http://transportmaps.mit.edu}.
}
\begin{algorithm}[H]
  \caption{
    {\bf (Computation of a monotone map)} \\
    Given an
    unnormalized target density $\bar{\pi}$
    and a parametric triangular monotone map $T[{\bf c}]$ 
    of the form \eqref{eq:monotone}, defined by an arbitrary 
    set of coefficients
    ${\bf c}\in\mathbb{R}^N$, find
    the optimal coefficients ${\bf c}^\star$ 
    according to \eqref{OptimDirectSAA}.
  }
  \label{alg:computeMap}
  \begin{algorithmic}[1]
    \Procedure{ComputeMap}{$\bar{\pi}$, 
      $T[{\bf c}]$, $m$}
    \State Generate samples 
    $({\bm x}_i)_{i=1}^m \overset{\text{i.i.d.}} \sim \mathcal{N}(0,{\bf I})$
    \State Solve (e.g., via a quasi-Newton or Newton method),
    $$
    {\bf c}^\star = \argmin_{{\bf c}\in\mathbb{R}^N} 
    -\frac{1}{m}\sum_{i=1}^m 
    \left( 
      \,\log \bar{\pi}(T[{\bf c}]({\bm x}_i)) + 
      \sum_k \log \partial_k T[{\bf c}]^k({\bm x}_i) 
    \right)
    $$
    \State \Return $T[{\bf c}^\star]$
    \EndProcedure
  \end{algorithmic}
\end{algorithm}

\begin{algorithm}[H]
  \caption{
    {\bf (Regression of a monotone map)} \\
    Given a map $M$ and a parametric triangular monotone map 
    $T[{\bf c}]$ of the form \eqref{eq:monotone} ,
    defined by an arbitrary set of coefficients
    ${\bf c}\in\mathbb{R}^N$, find
    the coefficients ${\bf c}^\star$ 
    minimizing the discrete $L^2$ norm between the two maps.
  }
  \label{alg:regresMap}
  \begin{algorithmic}[1]
    \Procedure{RegressionMap}{$M$, $T[{\bf c}]$, $m$}
    \State Generate samples 
    $({\bm x}_i)_{i=1}^m \overset{\text{i.i.d.}} \sim \mathcal{N}(0,{\bf I})$
    \State Solve
    \begin{equation*}
      \begin{aligned}
        {\bf c}^\star = \argmin_{{\bf c}\in\mathbb{R}^N} 
        \frac{1}{m} \sum_{i=1}^m 
        \left( M({\bm x}_i) - T[{\bf c}]({\bm x}_i) \right)^2
      \end{aligned}
    \end{equation*}
    \State \Return $T[{\bf c}^\star]$
    \EndProcedure
  \end{algorithmic}
\end{algorithm}

\begin{algorithm}
  \caption{
  {\bf (Joint parameter and state inference)} \\
    Given observations $({\bm y}_i)_{i=0}^{k+1}$, construct a
    transport map 
    approximation
    of the smoothing distribution 
    $\pi_{\vhyp,\Zb_{0},\ldots, \Zb_{k+1} \vert \yb_{0},\ldots, \yb_{k+1} }$
    in terms of a list of maps $(\frak{M}_j)_{j=0}^{k}$.
  }
  \label{alg:assimilate}
  \begin{algorithmic}[1]
    \begin{varwidth}[t]{\linewidth}
      \Procedure{Assimilate}{$({\bm y}_i)_{i=0}^{k+1}$, $m$}
        \For{$i\gets 0$ to $k$}
        \Comment see Thm. \ref{thm:decompJoint} 
      \If{$i=0$} %
      \State Define $\widetilde{\mathfrak{T}}^{\vhyp}_{i-1}$ to be the identity map
      \State Define $\pi^i$ as in 
      \eqref{eq:target_stepmap_0}  
      \Else
      \State $\widetilde{\mathfrak{T}}^{\vhyp}_{i-1}[{\bf c}^\star]\gets$
      \Call{RegressionMap}{
        $\widetilde{\mathfrak{T}}^{\vhyp}_{i-2} \circ \submap_{i-1}^{\vhyp}$, 
        $\widetilde{\mathfrak{T}}^{\vhyp}_{i-1}[{\bf c}]$, $m$
      }
      \State Define $\pi^i$ as in \eqref{eq:target_stepmap} 
      \EndIf
      \State $\frak{M}_{i}[{\bf c}^\star] \gets$ 
      \Call{ComputeMap}{$\pi^i$, $\frak{M}_i[{\bf c}]$, $m$}
      \State Append $\frak{M}_i$ to the list $(\frak{M}_j)_{j=0}^{i-1}$
      \EndFor
      \State \Return $(\frak{M}_j)_{j=0}^{k}\,, \widetilde{\mathfrak{T}}^{\vhyp}_{k-1}$
      \EndProcedure
    \end{varwidth}
  \end{algorithmic}
\end{algorithm}

\begin{algorithm}[H]
  \caption{
  {\bf (Sample the smoothing distribution)} \\
    Generate a sample from the smoothing distribution 
    $\pi_{\vhyp,\Zb_{0},\ldots, \Zb_{k+1} \vert \yb_{0},\ldots, \yb_{k+1} }$%
    using the maps computed in Algorithm \ref{alg:assimilate}.
  }
  \label{alg:SampleSmoothing}
  \begin{algorithmic}
    \Procedure{SampleSmoothing}{
      $(\frak{M}_j)_{j=0}^{k}$}%
    \State Generate 
    ${\bm x} \sim \mathcal{N}(0,{\bf I})$, with ${\bf I}$ the identity in $d_\theta + k \cdot d_{\bf z}$ dimensions
    \For{$j\gets k$ to $0$} 
    \Comment see Thm. \ref{thm:decompJoint} Part. \ref{thm:decompJoint:partFull}
    \State ${\bm x}_{\theta} \gets \frak{M}_j^{\bm\Theta}({\bm x}_{\theta})$
    \State ${\bm x}_{j} \gets \frak{M}_j^{0}({\bm x}_{\theta}, {\bm x}_{j}, {\bm x}_{j+1})$
    \State ${\bm x}_{j+1} \gets \frak{M}_j^{1}({\bm x}_{\theta}, {\bm x}_{j+1})$

    \EndFor
    \State \Return ${\bm x}$
    \EndProcedure
  \end{algorithmic}
\end{algorithm}

\begin{algorithm}[H]
  \caption{
  {\bf (Sample the filtering distribution)} \\
    Generate a sample from the %
    marginal distribution %
    $\pi_{ \vhyp ,\Zb_{k+1} \vert \yb_{0}, \ldots, \yb_{k+1}}$ %
    using the maps computed in Algorithm \ref{alg:assimilate}.
  }
  \label{alg:SampleFiltering}
  \begin{algorithmic}
    \Procedure{SampleFiltering}{
      $\frak{M}_k$, 
      $\widetilde{\mathfrak{T}}^{\vhyp}_{k-1}$}
    \State Generate 
    ${\bm x} \sim \mathcal{N}(0,{\bf I})$, with ${\bf I}$ the identity in $d_\theta+d_{\bf z}$ dimensions
    \State Define
    $$ 
    \widetilde{\submap}_k( \xb_{\vhyps} , \xb_{k+1} ) := 
    \left[
      \begin{array}{l}
        \widetilde{\mathfrak{T}}^{\vhyp}_{k-1} ( \frak{M}_k^{\vhyp} (  \xb_{\vhyps} ))\\ 
        \submap_k^1(  \xb_{\vhyps}, \xb_{k+1})
      \end{array}
    \right]
    $$
    \State ${\bm y} \gets \widetilde{\submap}_k( \xb_{\vhyps} , \xb_{k+1} )$
    \Comment see Thm. \ref{thm:decompJoint} Part. \ref{thm:decompJoint:partFilt}
    \State \Return ${\bm y}$
    \EndProcedure
  \end{algorithmic}
\end{algorithm}

%% file: sec_additional_results.tex
\hrevone{We revisit the numerical example of Section~\ref{sec:numerics} and
re-run both the joint state\slash parameter inference problem and the
long-time smoothing problem with \textit{linear} rather than
nonlinear maps. The results are less accurate, but substantially
faster; see Table~\ref{tab:cost_vs_accuracy} and the discussion of
this comparison in Section~\ref{sec:numerics}.}

\begin{figure}[H]
  \centering
  \includegraphics[width=0.90\textwidth, bb=25bp 0bp 750bp 325bp, clip]{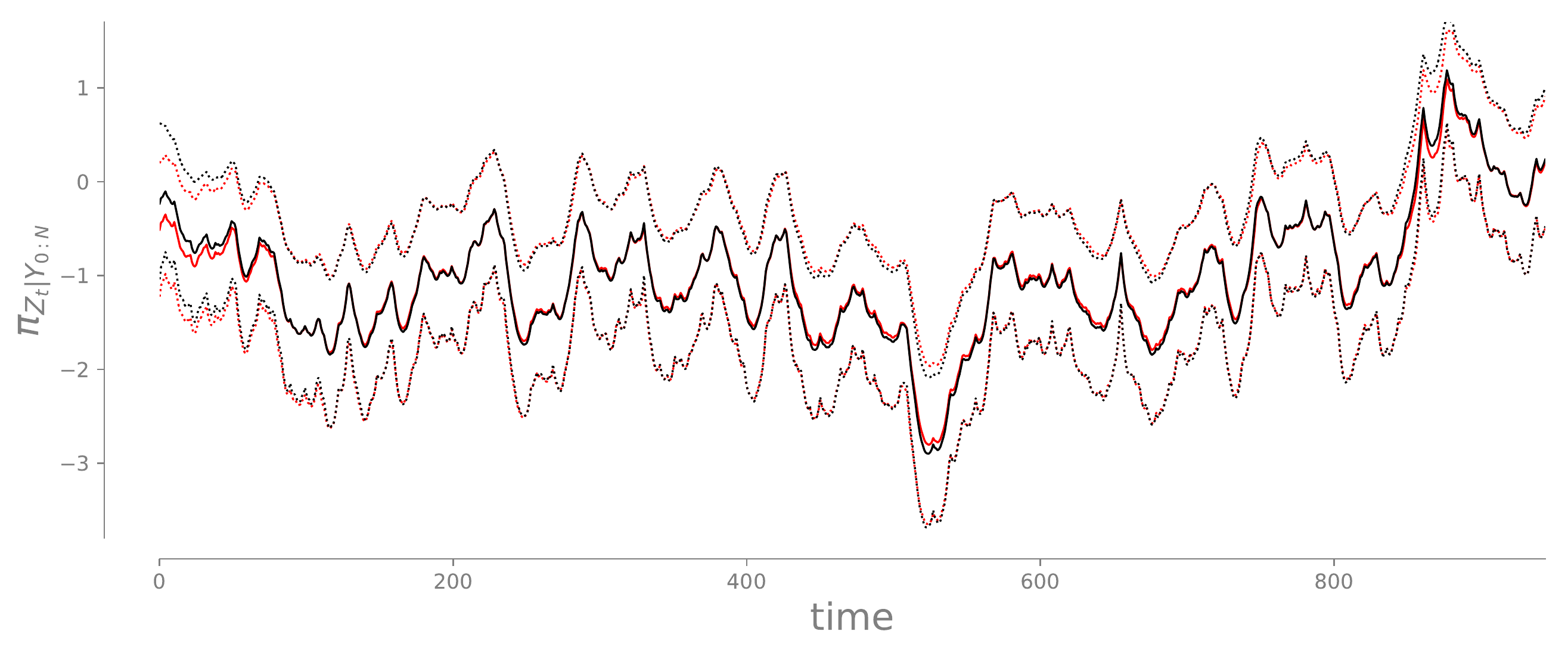}
  \caption{
    Same as Figure \ref{fig:stoc-vol:smoothing-vs-unbiased}, but using linear maps.
    Compared to a high-order map, there seems to be only a 
    minimal loss of accuracy, more
    prominent at earlier times.
  }
  \label{fig:stoc-vol:smoothing-vs-unbiased-linear}
\end{figure}

\begin{figure}[H]
  \centering
  \includegraphics[width=0.90\textwidth, bb=25bp 0bp 750bp 325bp, clip]{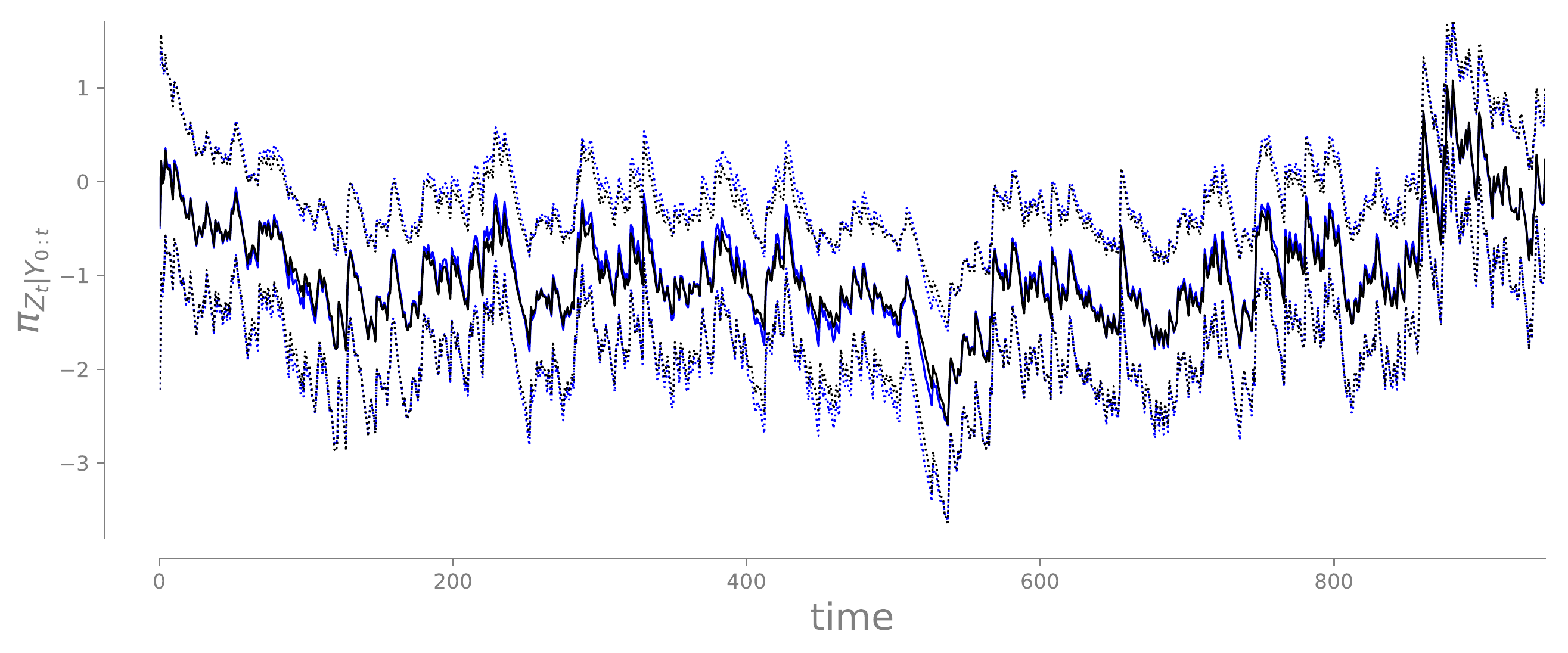}
  \caption{ Comparison of the $\{5,95\}$--percentiles (dashed lines) and the mean (solid line)
      of the numerical approximation of the
      filtering marginals using {\it linear} transport maps
      (blue lines) with those of a ``reference'' solution obtained
      via seventh-order maps (as shown in Figure \ref{fig:stoc-vol:filtering}).
      The two solutions look remarkably similar despite
      the enormous difference in computational cost (see Table \ref{tab:cost_vs_accuracy}).
  }
  \label{fig:stoc-vol:filtering-vs-unbiased-linear}
\end{figure}

\begin{figure}[H]
  \centering
  \includegraphics[width=1.0\textwidth, bb=65bp 20bp 775bp 280bp, clip]{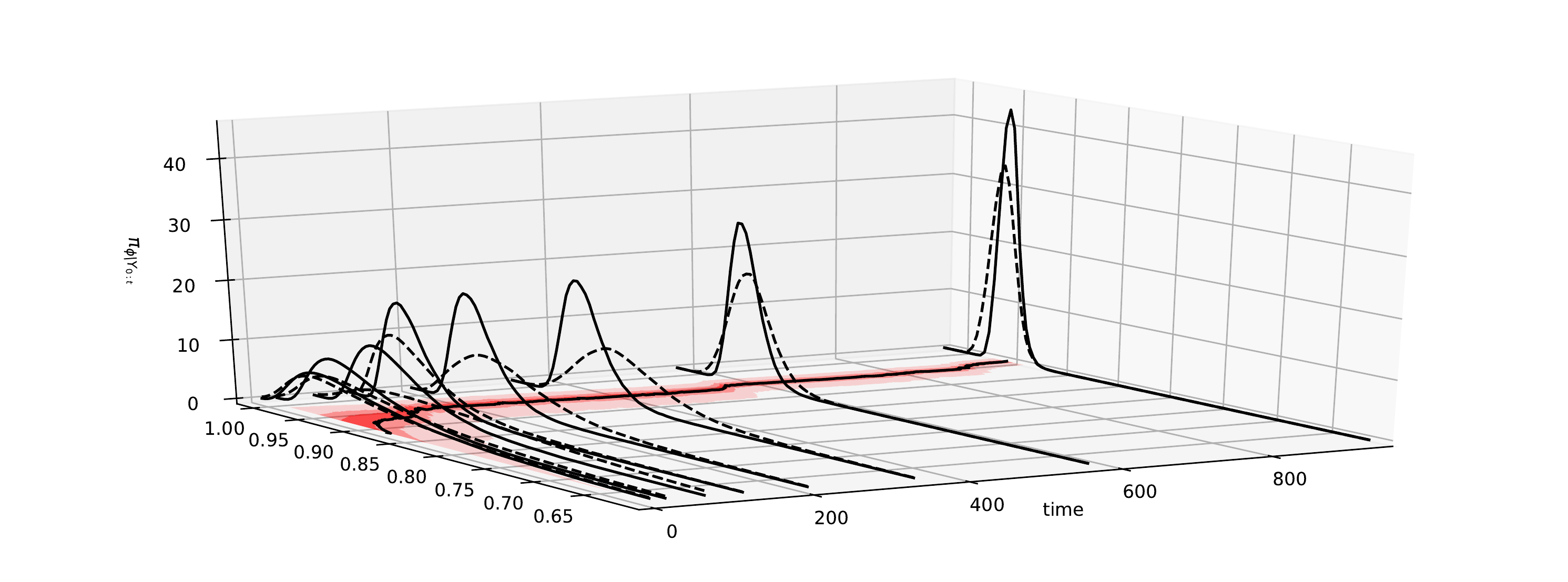}
  \caption{
    Same as Figure \ref{fig:stoc-vol:3d-phi-vs-unbiased}, but using
    linear maps. Here, the loss of accuracy is more dramatic than for
    the smoothing distribution of the state in
    Figure \ref{fig:stoc-vol:smoothing-vs-unbiased-linear}.
    Even though the approximate marginal captures the bulk of the
    true parameter marginals, for this specific
    problem of static parameter inference, a linear map is largely
    inadequate; hence the need for a higher-order  nonlinear transformation.
  }
  \label{fig:stoc-vol:3d-phi-vs-unbiased-linear}
\end{figure}

\begin{figure}[H]
  \centering
  \includegraphics[width=1.0\textwidth, bb=0bp 0bp 925bp 355bp, clip]{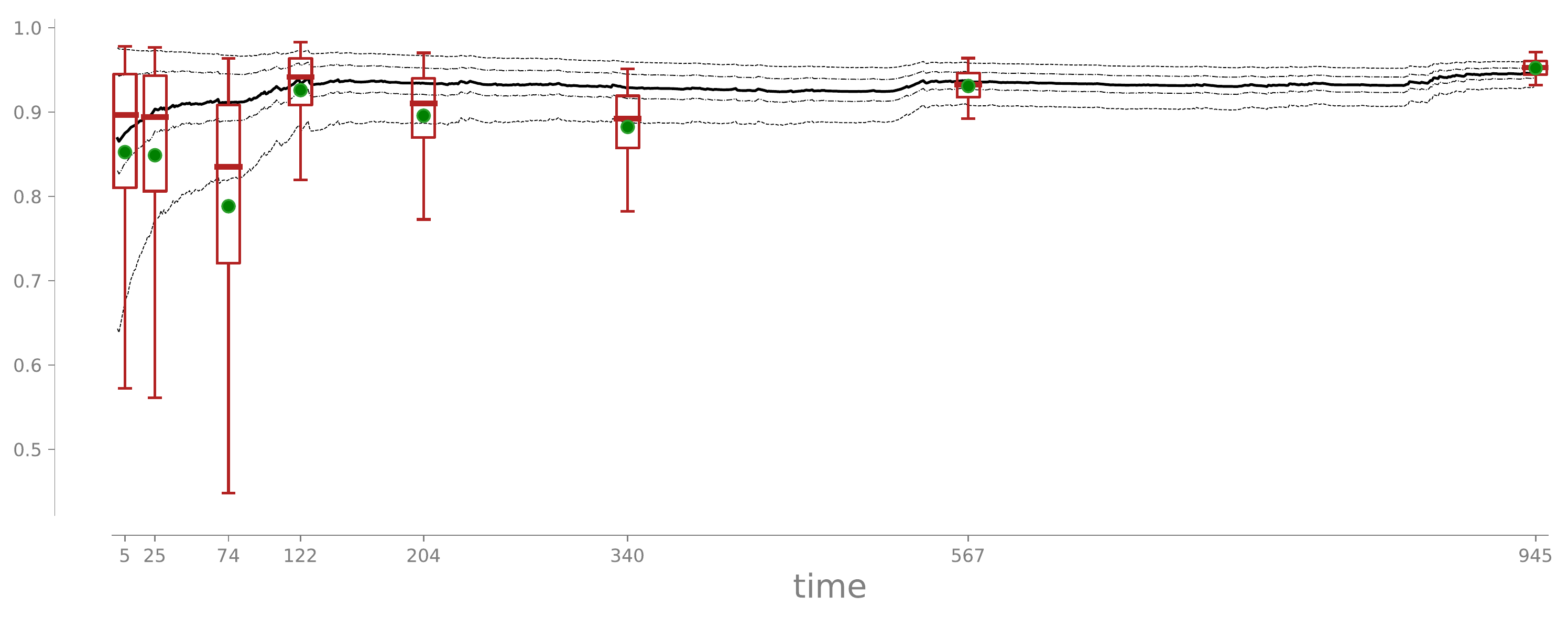}
  \caption{The horizontal plane of 
  Figure \ref{fig:stoc-vol:3d-phi-vs-unbiased-linear} ({\it black lines})
  overlaid with a selected number of box-and-whisker plots 
   associated with the marginals of a ``reference'' MCMC solution.
  The ends of the whiskers represent the $\{5,95\}$--percentiles, while the green dots correspond to the means of the reference
  distribution. 
  Linear maps are insufficient to correctly characterize the
  parameter marginals, especially the transition at time 74 (cf.\ Figures \ref{fig:stoc-vol:3d-phi-vs-unbiased} and 
  \ref{fig:stoc-vol:3d-phi-vs-unbiased-linear})
  \label{fig:stoc-vol:2d-phi-vs-unbiased-linear}
  }
\end{figure}

\begin{figure}[H]
  \centering
  \includegraphics[width=1.0\textwidth, bb=65bp 20bp 775bp 280bp, clip]{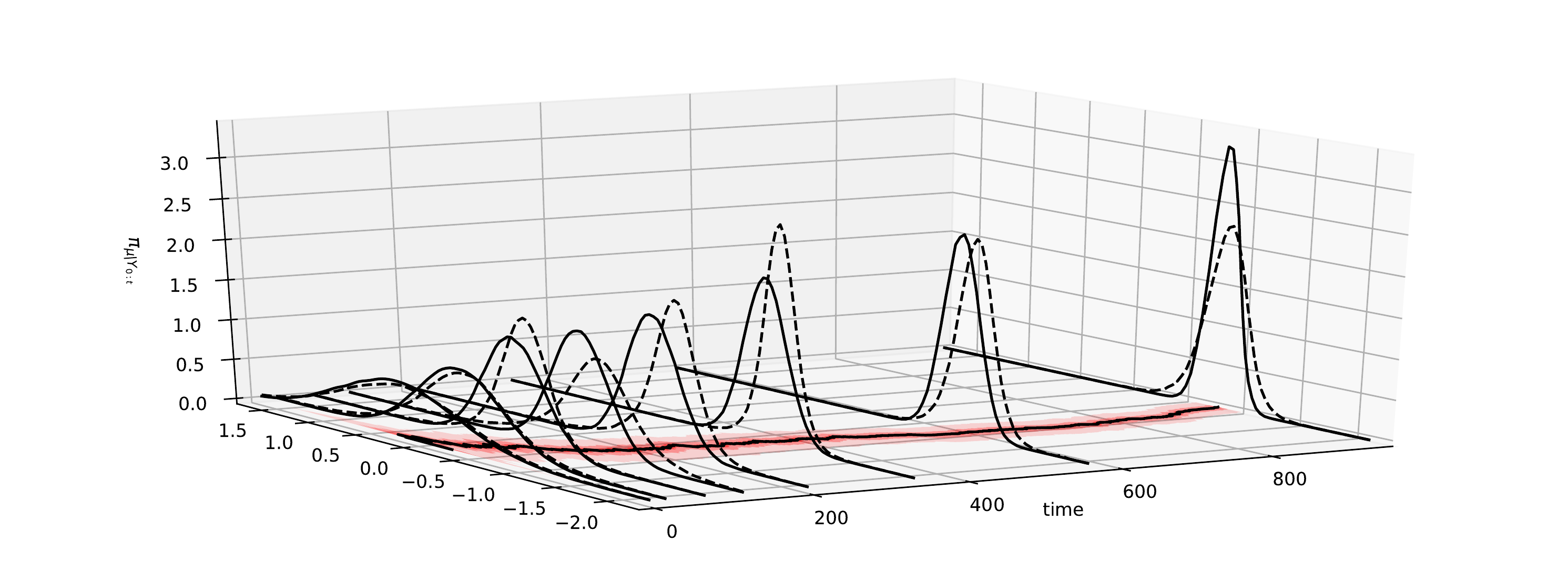}
  \caption{ Same as Figure \ref{fig:stoc-vol:3d-mu-vs-unbiased}, but
    using linear maps. Once again, the linear map provides plausible,
    but somewhat inaccurate, results for sequential parameter inference. A nonlinear transformation is
    better suited for this problem.  }
  \label{fig:stoc-vol:3d-mu-vs-unbiased-linear}
\end{figure}

\begin{figure}[H]
  \centering
  \includegraphics[width=1.0\textwidth, bb=0bp 0bp 925bp 355bp, clip]{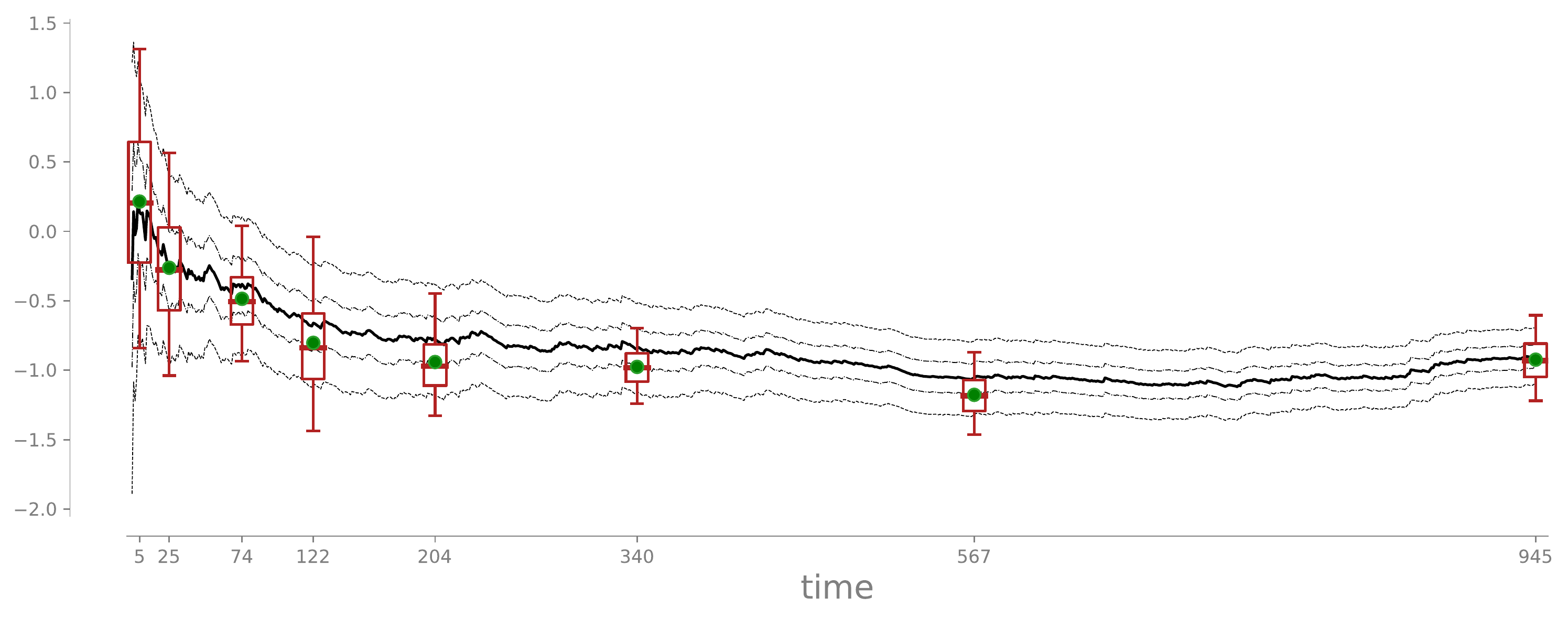}
  \caption{
    The horizontal plane of 
  Figure \ref{fig:stoc-vol:3d-mu-vs-unbiased-linear} ({\it black lines})
  overlaid with a selected number of box-and-whisker plots 
   associated with the marginals of a ``reference'' MCMC solution.
  See Figure \ref{fig:stoc-vol:2d-phi-vs-unbiased-linear} caption
  for more details. 
}
  \label{fig:stoc-vol:2d-mu-vs-unbiased-linear}
\end{figure}

\clearpage
\begin{sidewaysfigure}
  \begin{center}
    \includegraphics[width=.9\textwidth, bb=25bp 25bp 800bp 290bp, clip]{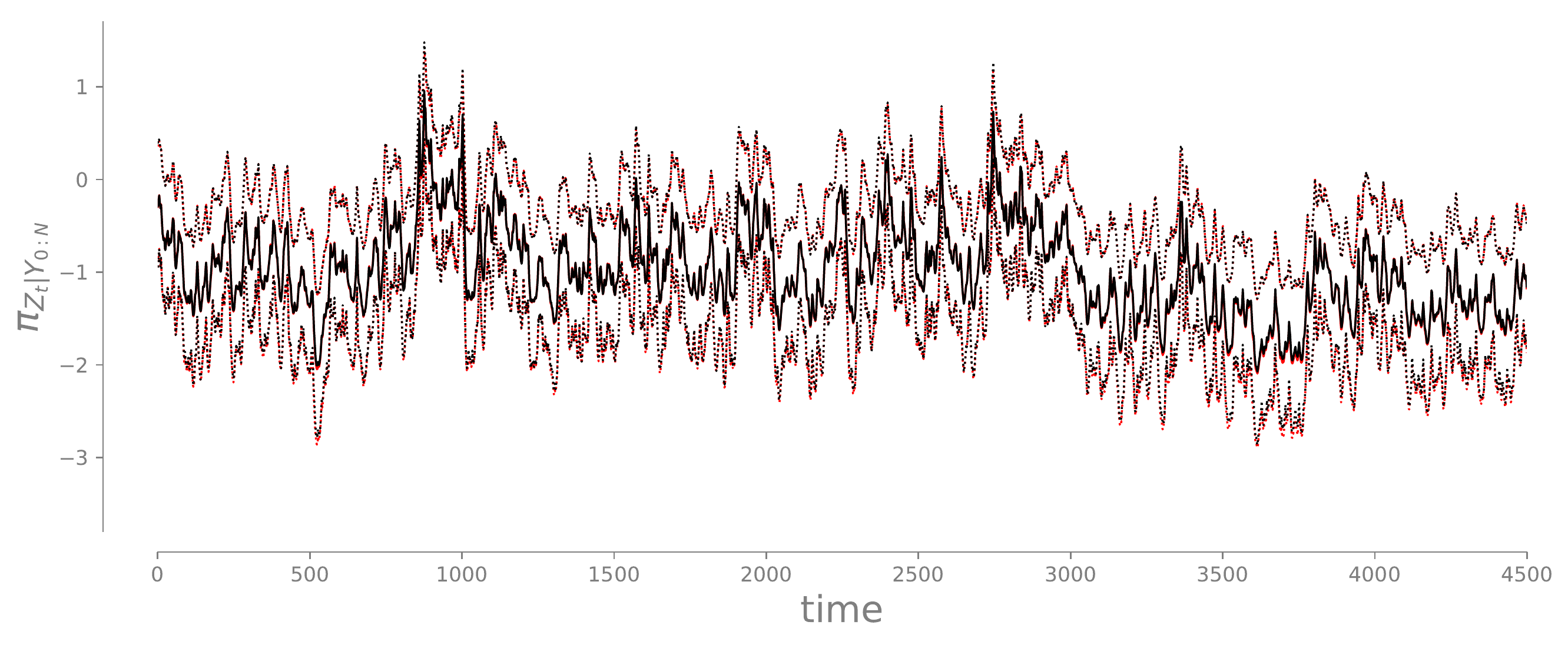}

    \vspace{10pt}

    \includegraphics[width=.9\textwidth, bb=25bp 0bp 800bp 290bp, clip]{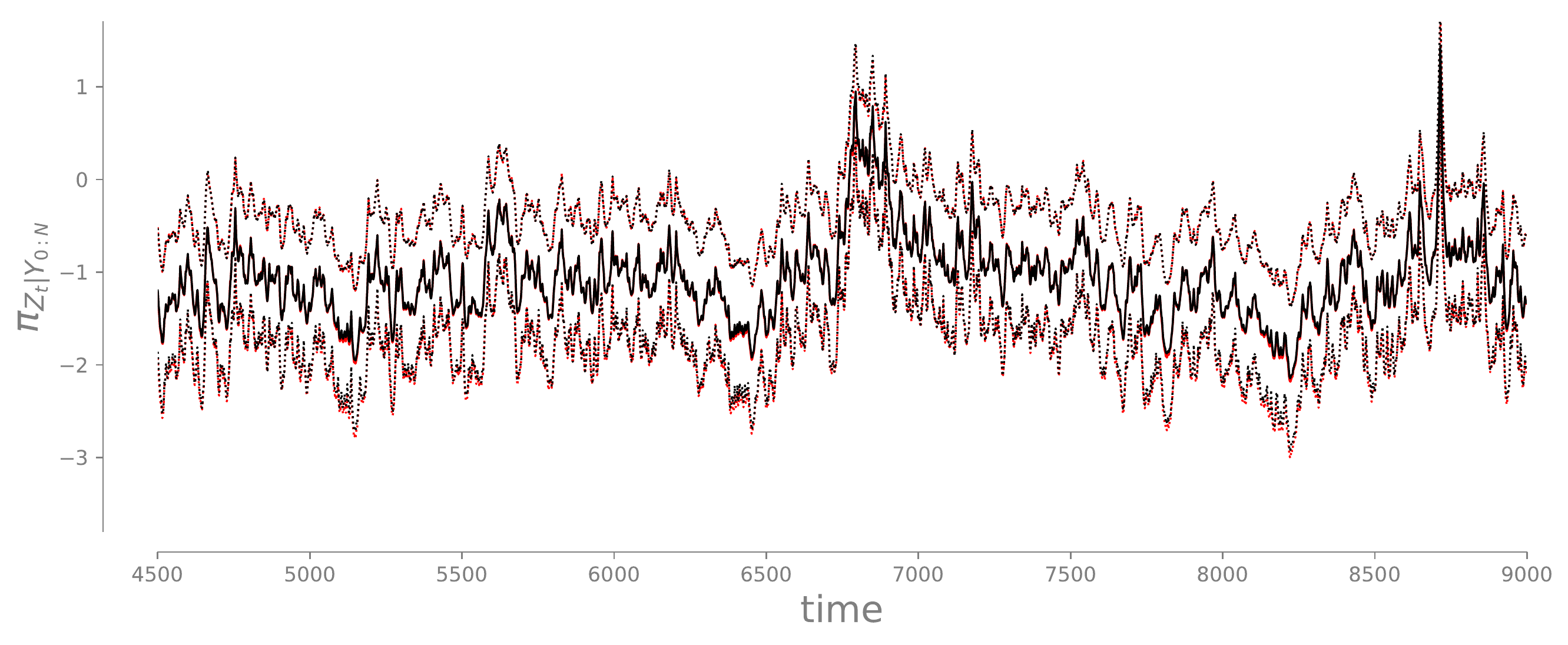}
    \caption{
      Same as Figure \ref{fig:stoc-vol:very-long-smooth-vs-unbiased}, but 
      using linear maps.
      Long-time smoothing with no static parameters via linear maps
      yields accurate characterizations of the marginal distributions
      across all times, at a fraction of the cost of a high-order
      nonlinear transformation (see Table \ref{tab:cost_vs_accuracy}).
    }
    \label{fig:stoc-vol:very-long-smooth-vs-unbiased-linear} 
  \end{center}
\end{sidewaysfigure}
\clearpage

%% file: mainMap.bbl
\begin{thebibliography}{124}
\providecommand{\natexlab}[1]{#1}
\providecommand{\url}[1]{\texttt{#1}}
\expandafter\ifx\csname urlstyle\endcsname\relax
  \providecommand{\doi}[1]{doi: #1}\else
  \providecommand{\doi}{doi: \begingroup \urlstyle{rm}\Url}\fi

\bibitem[Anderes and Coram(2012)]{anderes2012general}
E.~Anderes and M.~Coram.
\newblock A general spline representation for nonparametric and semiparametric
  density estimates using diffeomorphisms.
\newblock \emph{arXiv:1205.5314}, 2012.

\bibitem[Andrieu et~al.(2010)Andrieu, Doucet, and
  Holenstein]{andrieu2010particle}
C.~Andrieu, A.~Doucet, and R.~Holenstein.
\newblock Particle {M}arkov chain {M}onte {C}arlo methods.
\newblock \emph{Journal of the Royal Statistical Society: Series B},
  72\penalty0 (3):\penalty0 269--342, 2010.

\bibitem[Asmussen and Glynn(2007)]{asmussen2007stochastic}
S.~Asmussen and P.~W. Glynn.
\newblock \emph{Stochastic simulation: algorithms and analysis}, volume~57.
\newblock Springer Science \& Business Media, 2007.

\bibitem[Bardsley et~al.(2014)Bardsley, Solonen, Haario, and
  Laine]{Bardsley2014}
J.~M. Bardsley, A.~Solonen, H.~Haario, and M.~Laine.
\newblock Randomize-then-optimize: a method for sampling from posterior
  distributions in nonlinear inverse problems.
\newblock \emph{{SIAM Journal on Scientific Computing}}, 36\penalty0
  (4):\penalty0 A1895--A1910, 2014.
\newblock \doi{10.1137/140964023}.

\bibitem[Bertsekas(1995)]{bertsekas1995dynamic}
D.~P. Bertsekas.
\newblock \emph{Dynamic programming and optimal control}, volume~1.
\newblock Athena Scientific Belmont, MA, 1995.

\bibitem[Bierman(2006)]{bierman2006factorization}
G.~J. Bierman.
\newblock \emph{Factorization methods for discrete sequential estimation}.
\newblock Courier Corporation, 2006.

\bibitem[Bigoni et~al.(2017)Bigoni, Spantini, and Marzouk]{bigoni2016monotone}
D.~Bigoni, A.~Spantini, and Y.~Marzouk.
\newblock On the computation of monotone transports.
\newblock \emph{In preparation}, 2017.

\bibitem[Blake et~al.(2011)Blake, Kohli, and Rother]{blake2011markov}
A.~Blake, P.~Kohli, and C.~Rother.
\newblock \emph{Markov random fields for vision and image processing}.
\newblock MIT Press, 2011.

\bibitem[Blei et~al.(2016)Blei, Kucukelbir, and McAuliffe]{blei2016variational}
D.~Blei, A.~Kucukelbir, and J.~McAuliffe.
\newblock Variational inference: a review for statisticians.
\newblock \emph{arXiv:1601.00670}, 2016.

\bibitem[Bogachev et~al.(2005)Bogachev, Kolesnikov, and
  Medvedev]{bogachev2005triangular}
V.~I. Bogachev, A.~V. Kolesnikov, and K.~V. Medvedev.
\newblock Triangular transformations of measures.
\newblock \emph{Sbornik: Mathematics}, 196\penalty0 (3):\penalty0 309, 2005.

\bibitem[Bottou et~al.(2016)Bottou, Curtis, and
  Nocedal]{bottou2016optimization}
L.~Bottou, F.~Curtis, and J.~Nocedal.
\newblock Optimization methods for large-scale machine learning.
\newblock \emph{arXiv:1606.04838}, 2016.

\bibitem[Boyd(2001)]{boyd2001chebyshev}
J.~Boyd.
\newblock \emph{Chebyshev and {F}ourier spectral methods}.
\newblock Dover, 2 edition, 2001.

\bibitem[Carlier et~al.(2010)Carlier, Galichon, and
  Santambrogio]{carlier2010knothe}
G.~Carlier, A.~Galichon, and F.~Santambrogio.
\newblock From {K}nothe's transport to {B}renier's map and a continuation
  method for optimal transport.
\newblock \emph{SIAM Journal on Mathematical Analysis}, 41\penalty0
  (6):\penalty0 2554--2576, 2010.

\bibitem[Cheng et~al.(2015)Cheng, Cheng, Liu, Peng, and
  Teng]{cheng2015efficient}
D.~Cheng, Y.~Cheng, Y.~Liu, R.~Peng, and S.~Teng.
\newblock Efficient sampling for {G}aussian graphical models via spectral
  sparsification.
\newblock In \emph{Conference on Learning Theory}, pages 364--390, 2015.

\bibitem[Chopin et~al.(2013)Chopin, Jacob, and
  Papaspiliopoulos]{chopin2013smc2}
N.~Chopin, P.~Jacob, and O.~Papaspiliopoulos.
\newblock {SMC2}: an efficient algorithm for sequential analysis of state space
  models.
\newblock \emph{Journal of the Royal Statistical Society: Series B (Statistical
  Methodology)}, 75\penalty0 (3):\penalty0 397--426, 2013.

\bibitem[Chorin and Tu(2009)]{chorin2009implicit}
A.~J. Chorin and X.~Tu.
\newblock Implicit sampling for particle filters.
\newblock \emph{Proceedings of the National Academy of Sciences}, 106\penalty0
  (41):\penalty0 17249--17254, 2009.

\bibitem[Constantine et~al.(2014)Constantine, Dow, and
  Wang]{constantine2014active}
P.~G. Constantine, E.~Dow, and Q.~Wang.
\newblock Active subspace methods in theory and practice: Applications to
  kriging surfaces.
\newblock \emph{SIAM Journal on Scientific Computing}, 36\penalty0
  (4):\penalty0 A1500--A1524, 2014.

\bibitem[Crisan and Doucet(2002)]{crisan2002survey}
D.~Crisan and A.~Doucet.
\newblock A survey of convergence results on particle filtering methods for
  practitioners.
\newblock \emph{IEEE Transactions on signal processing}, 50\penalty0
  (3):\penalty0 736--746, 2002.

\bibitem[Crisan and Miguez(2013)]{crisan2013nested}
D.~Crisan and J.~Miguez.
\newblock Nested particle filters for online parameter estimation in
  discrete-time state-space {M}arkov models.
\newblock \emph{arXiv:1308.1883}, 2013.

\bibitem[Csill\'{e}ry et~al.(2010)Csill\'{e}ry, Blum, Gaggiotti, and
  Fran\c{c}ois]{Csillery2010}
K.~Csill\'{e}ry, M.~G.~B. Blum, O.~E. Gaggiotti, and O.~Fran\c{c}ois.
\newblock {Approximate Bayesian Computation (ABC) in practice}.
\newblock \emph{Trends in ecology \& evolution}, 25\penalty0 (7):\penalty0
  410--8, 2010.

\bibitem[Cui et~al.(2014)Cui, Martin, Marzouk, Solonen, and
  Spantini]{cui2014likelihood}
T.~Cui, J.~Martin, Y.~Marzouk, A.~Solonen, and A.~Spantini.
\newblock Likelihood-informed dimension reduction for nonlinear inverse
  problems.
\newblock \emph{Inverse Problems}, 30\penalty0 (11):\penalty0 114015, 2014.

\bibitem[Daum and Huang(2008)]{daum2008particle}
F.~Daum and J.~Huang.
\newblock Particle flow for nonlinear filters with log-homotopy.
\newblock In \emph{SPIE Defense and Security Symposium}, pages 696918--696918.
  International Society for Optics and Photonics, 2008.

\bibitem[Daum and Huang(2012)]{daum2012particle}
F.~Daum and J.~Huang.
\newblock Particle flow and {M}onge-{K}antorovich transport.
\newblock In \emph{Information Fusion (FUSION), 2012 15th International
  Conference on}, pages 135--142. IEEE, 2012.

\bibitem[Davis and Rabinowitz(2007)]{davis2007methods}
P.~J. Davis and P.~Rabinowitz.
\newblock \emph{Methods of numerical integration}.
\newblock Courier Corporation, 2007.

\bibitem[Del~Moral(2004)]{del2004feynman}
P.~Del~Moral.
\newblock Feynman-{K}ac formulae.
\newblock In \emph{Feynman-Kac Formulae: Genealogical and Interacting Particle
  Systems with Applications}, pages 47--93. Springer, 2004.

\bibitem[Del~Moral et~al.(2017)Del~Moral, Jasra, and Zhou]{del2017biased}
P.~Del~Moral, A.~Jasra, and Y.~Zhou.
\newblock Biased online parameter inference for state-space models.
\newblock \emph{Methodology and Computing in Applied Probability}, 19\penalty0
  (3):\penalty0 727--749, 2017.

\bibitem[Detommaso et~al.(2018)Detommaso, Cui, Marzouk, Scheichl, and
  Spantini]{detommaso2018stein}
G.~Detommaso, T.~Cui, Y.~Marzouk, R.~Scheichl, and A.~Spantini.
\newblock A {S}tein variational {N}ewton method.
\newblock \emph{arXiv:1806.03085}, 2018.

\bibitem[Dick et~al.(2013)Dick, Kuo, and Sloan]{dick2013high}
J.~Dick, F.~Y. Kuo, and I.~H. Sloan.
\newblock High-dimensional integration: the {Q}uasi-{M}onte {C}arlo way.
\newblock \emph{Acta Numerica}, 22:\penalty0 133--288, 2013.

\bibitem[Dick et~al.(2016)Dick, Gantner, Gia, and Schwab]{dick2016higher}
J.~Dick, R.~N. Gantner, Q.~T.~L. Gia, and C.~Schwab.
\newblock Higher order {Q}uasi-{M}onte {C}arlo integration for {B}ayesian
  estimation.
\newblock \emph{arXiv:1602.07363}, 2016.

\bibitem[Dinh et~al.(2016)Dinh, Sohl-Dickstein, and Bengio]{dinh2016density}
L.~Dinh, J.~Sohl-Dickstein, and S.~Bengio.
\newblock Density estimation using {R}eal {NVP}.
\newblock \emph{arXiv:1605.08803}, 2016.

\bibitem[Doucet and Johansen(2009)]{doucet2009tutorial}
A.~Doucet and A.~M. Johansen.
\newblock A tutorial on particle filtering and smoothing: Fifteen years later.
\newblock \emph{Handbook of nonlinear filtering}, 12\penalty0
  (656-704):\penalty0 3, 2009.

\bibitem[Douglas(1999)]{douglas1999applications}
R.~J. Douglas.
\newblock Applications of the {M}onge-{A}mpere equation and {M}onge transport
  problem to meteorology and oceanography.
\newblock In \emph{Monge Amp{\`e}re Equation: Applications to Geometry and
  Optimization}, volume 226, page~33. American Mathematical Soc., 1999.

\bibitem[Durbin and Koopman(2000)]{durbin2000time}
J.~Durbin and S.~J. Koopman.
\newblock Time series analysis of non-{G}aussian observations based on
  state-space models from both classical and {B}ayesian perspectives.
\newblock \emph{Journal of the Royal Statistical Society: Series B},
  62\penalty0 (1):\penalty0 3--56, 2000.

\bibitem[Erol et~al.(2017)Erol, Wu, Li, and Russell]{erol2017nearly}
Y.~B. Erol, Y.~Wu, L.~Li, and S.~J. Russell.
\newblock A nearly-black-box online algorithm for joint parameter and state
  estimation in temporal models.
\newblock In \emph{AAAI}, pages 1861--1869, 2017.

\bibitem[Evensen(2003)]{evensen2003ensemble}
G.~Evensen.
\newblock The ensemble {K}alman filter: Theoretical formulation and practical
  implementation.
\newblock \emph{Ocean dynamics}, 53\penalty0 (4):\penalty0 343--367, 2003.

\bibitem[Evensen(2007)]{evensen2007data}
G.~Evensen.
\newblock \emph{Data Assimilation}.
\newblock Springer, 2007.

\bibitem[Evensen and Van~Leeuwen(2000)]{evensen2000ensemble}
G.~Evensen and P.~J. Van~Leeuwen.
\newblock An ensemble {K}alman smoother for nonlinear dynamics.
\newblock \emph{Monthly Weather Review}, 128\penalty0 (6):\penalty0 1852--1867,
  2000.

\bibitem[Fremlin(2000)]{fremlin2000measure}
D.~H. Fremlin.
\newblock \emph{Measure Theory}, volume~4.
\newblock Torres Fremlin, 2000.

\bibitem[Gaspari and Cohn(1999)]{gaspari1999construction}
G.~Gaspari and S.~E. Cohn.
\newblock Construction of correlation functions in two and three dimensions.
\newblock \emph{Quarterly Journal of the Royal Meteorological Society},
  125\penalty0 (554):\penalty0 723--757, 1999.

\bibitem[George and Liu(1989)]{george1989evolution}
A.~George and J.~W.~H. Liu.
\newblock The evolution of the minimum degree ordering algorithm.
\newblock \emph{SIAM Review}, 31\penalty0 (1):\penalty0 1--19, 1989.

\bibitem[Glasserman(1991)]{glasserman1991gradient}
P.~Glasserman.
\newblock \emph{Gradient estimation via perturbation analysis}.
\newblock Springer Science \& Business Media, 1991.

\bibitem[Glynn(1990)]{glynn1990likelihood}
P.~W. Glynn.
\newblock Likelihood ratio gradient estimation for stochastic systems.
\newblock \emph{Communications of the ACM}, 33\penalty0 (10):\penalty0 75--84,
  1990.

\bibitem[Godsill et~al.(2004)Godsill, Doucet, and West]{godsill2012monte}
S.~J. Godsill, A.~Doucet, and M.~West.
\newblock Monte {C}arlo smoothing for nonlinear time series.
\newblock \emph{Journal of the American Statistical Association}, 99\penalty0
  (465):\penalty0 156--168, 2004.

\bibitem[Goodfellow et~al.(2014)Goodfellow, Pouget-Abadie, Mirza, Xu,
  Warde-Farley, Ozair, Courville, and Bengio]{goodfellow2014generative}
I.~Goodfellow, J.~Pouget-Abadie, M.~Mirza, B.~Xu, D.~Warde-Farley, S.~Ozair,
  A.~Courville, and Y.~Bengio.
\newblock Generative adversarial nets.
\newblock In \emph{Advances in Neural Information Processing Systems}, pages
  2672--2680, 2014.

\bibitem[Goodfellow et~al.(2016)Goodfellow, Bengio, Courville, and
  Bengio]{goodfellow2016deep}
I.~Goodfellow, Y.~Bengio, A.~Courville, and Y.~Bengio.
\newblock \emph{Deep learning}, volume~1.
\newblock MIT press Cambridge, 2016.

\bibitem[Hamill et~al.(2001)Hamill, Whitaker, and Snyder]{hamill2001distance}
T.~M. Hamill, J.~S. Whitaker, and C.~Snyder.
\newblock Distance-dependent filtering of background error covariance estimates
  in an ensemble {K}alman filter.
\newblock \emph{Monthly Weather Review}, 129\penalty0 (11):\penalty0
  2776--2790, 2001.

\bibitem[Hammersley and Clifford(1971)]{hammersley1971markov}
J.~M. Hammersley and P.~Clifford.
\newblock Markov fields on finite graphs and lattices.
\newblock \emph{Unpublished manuscript}, 1971.

\bibitem[Han and Liu(2017)]{han2017stein}
J.~Han and Q.~Liu.
\newblock Stein variational adaptive importance sampling.
\newblock \emph{arXiv:1704.05201}, 2017.

\bibitem[Heng et~al.(2015)Heng, Doucet, and Pokern]{heng2015gibbs}
J.~Heng, A.~Doucet, and Y.~Pokern.
\newblock Gibbs flow for approximate transport with applications to {B}ayesian
  computation.
\newblock \emph{arXiv:1509.08787}, 2015.

\bibitem[Ho and Cao(1983)]{ho1983perturbation}
Y.~C. Ho and X.~Cao.
\newblock Perturbation analysis and optimization of queueing networks.
\newblock \emph{Journal of Optimization Theory and Applications}, 40\penalty0
  (4):\penalty0 559--582, 1983.

\bibitem[Hyv{\"a}rinen(2005)]{hyvarinen2005estimation}
A.~Hyv{\"a}rinen.
\newblock Estimation of non-normalized statistical models by score matching.
\newblock \emph{Journal of Machine Learning Research}, 6\penalty0
  (Apr):\penalty0 695--709, 2005.

\bibitem[Jacob(2015)]{jacob2015sequential}
P.~E. Jacob.
\newblock Sequential {B}ayesian inference for implicit hidden {M}arkov models
  and current limitations.
\newblock \emph{ESAIM: Proceedings and Surveys}, 51:\penalty0 24--48, 2015.

\bibitem[Jog and Loh(2015)]{jog2015model}
V.~Jog and P.~Loh.
\newblock On model misspecification and {KL} separation for {G}aussian
  graphical models.
\newblock In \emph{IEEE International Symposium on Information Theory}, pages
  1174--1178, 2015.

\bibitem[Johnson and Willsky(2008)]{johnson2008recursive}
J.~K. Johnson and A.~S. Willsky.
\newblock A recursive model-reduction method for approximate inference in
  {G}aussian {M}arkov random fields.
\newblock \emph{IEEE Transactions on Image Processing}, 17\penalty0
  (1):\penalty0 70--83, 2008.

\bibitem[Kantas et~al.(2015)Kantas, Doucet, Singh, Maciejowski, and
  Chopin]{kantas2015particle}
N.~Kantas, A.~Doucet, S.~S. Singh, J.~Maciejowski, and N.~Chopin.
\newblock On particle methods for parameter estimation in state-space models.
\newblock \emph{Statistical Science}, 30\penalty0 (3):\penalty0 328--351, 2015.

\bibitem[Kantorovich(1965)]{kantorovich1965best}
L.~V. Kantorovich.
\newblock \emph{The best use of economic resources.}
\newblock Oxford \& London: Pergamon Press., 1965.

\bibitem[Kim et~al.(1998)Kim, Shephard, and Chib]{kim1998stochastic}
S.~Kim, N.~Shephard, and S.~Chib.
\newblock Stochastic volatility: likelihood inference and comparison with
  {ARCH} models.
\newblock \emph{The Review of Economic Studies}, 65\penalty0 (3):\penalty0
  361--393, 1998.

\bibitem[Kingma and Welling(2013)]{kingma2013auto}
D.~P. Kingma and M.~Welling.
\newblock Auto-encoding variational {B}ayes.
\newblock \emph{arXiv:1312.6114}, 2013.

\bibitem[Kitagawa(1987)]{kitagawa1987non}
G.~Kitagawa.
\newblock Non-{G}aussian state-space modeling of nonstationary time series.
\newblock \emph{Journal of the American Statistical Association}, 82\penalty0
  (400):\penalty0 1032--1041, 1987.

\bibitem[Kitagawa(1998)]{kitagawa1998self}
G.~Kitagawa.
\newblock A self-organizing state-space model.
\newblock \emph{Journal of the American Statistical Association}, pages
  1203--1215, 1998.

\bibitem[Knothe(1957)]{knothe1957contributions}
H.~Knothe.
\newblock Contributions to the theory of convex bodies.
\newblock \emph{The Michigan Mathematical Journal}, 4\penalty0 (1):\penalty0
  39--52, 1957.

\bibitem[Koller and Friedman(2009)]{koller2009probabilistic}
D.~Koller and N.~Friedman.
\newblock \emph{Probabilistic graphical models: principles and techniques}.
\newblock MIT press, 2009.

\bibitem[Kushner and Yin(2003)]{kushner2003stochastic}
H.~Kushner and G.~G. Yin.
\newblock \emph{Stochastic approximation and recursive algorithms and
  applications}, volume~35.
\newblock Springer Science \& Business Media, 2003.

\bibitem[Laparra et~al.(2011)Laparra, Camps-Valls, and
  Malo]{laparra2011iterative}
V.~Laparra, G.~Camps-Valls, and J.~Malo.
\newblock Iterative {G}aussianization: from {ICA} to random rotations.
\newblock \emph{IEEE transactions on neural networks}, 22\penalty0
  (4):\penalty0 537--549, 2011.

\bibitem[Lauritzen(1996)]{lauritzen1996graphical}
S.~L. Lauritzen.
\newblock \emph{Graphical Models}.
\newblock Oxford University Press, 1996.

\bibitem[Law et~al.(2015)Law, Stuart, and Zygalakis]{law2015data}
K.~Law, A.~Stuart, and K.~Zygalakis.
\newblock \emph{Data assimilation: a mathematical introduction}, volume~62.
\newblock Springer, 2015.

\bibitem[Lin et~al.(2015)Lin, Drton, and Shojaie]{lin2015high}
L.~Lin, M.~Drton, and A.~Shojaie.
\newblock High-dimensional inference of graphical models using regularized
  score matching.
\newblock \emph{arXiv:1507.00433}, 2015.

\bibitem[Liu and West(2001)]{liu2001combined}
J.~Liu and M.~West.
\newblock Combined parameter and state estimation in simulation-based
  filtering.
\newblock In \emph{Sequential Monte Carlo methods in practice}, pages 197--223.
  Springer, 2001.

\bibitem[Liu and Wang(2016)]{liu2016stein}
Q.~Liu and D.~Wang.
\newblock Stein variational gradient descent: a general purpose {B}ayesian
  inference algorithm.
\newblock In \emph{Advances in Neural Information Processing Systems}, pages
  2370--2378, 2016.

\bibitem[M. et~al.(2012)M., T., A., and C.]{Morzfeld2012}
Matthias M., Xuemin T., Ethan A., and Alexandre~J. C.
\newblock A random map implementation of implicit filters.
\newblock \emph{Journal of Computational Physics}, 231\penalty0 (4):\penalty0
  2049--2066, 2012.

\bibitem[Marin et~al.(2012)Marin, Pudlo, Robert, and Ryder]{Marin2011}
J.~M. Marin, P.~Pudlo, C.~P. Robert, and R.~J. Ryder.
\newblock {Approximate Bayesian computational methods}.
\newblock \emph{Statistics and Computing}, 22\penalty0 (6):\penalty0
  1167--1180, 2012.

\bibitem[Marsaglia and Tsang(2000)]{marsaglia2000ziggurat}
G.~Marsaglia and W.~W. Tsang.
\newblock The ziggurat method for generating random variables.
\newblock \emph{Journal of Statistical Software}, 5\penalty0 (8):\penalty0
  1--7, 2000.

\bibitem[Marzouk et~al.(2016)Marzouk, Moselhy, Parno, and
  Spantini]{marzouk2016introduction}
Y.~Marzouk, T.~Moselhy, M.~Parno, and A~Spantini.
\newblock Sampling via measure transport: An introduction.
\newblock In \emph{Handbook of Uncertainty Quantification, R. Ghanem, D.
  Higdon, and H. Owhadi, editors}. Springer, 2016.

\bibitem[Meinshausen and B{\"u}hlmann(2006)]{meinshausen2006high}
N.~Meinshausen and P.~B{\"u}hlmann.
\newblock High-dimensional graphs and variable selection with the {L}asso.
\newblock \emph{The Annals of Statistics}, pages 1436--1462, 2006.

\bibitem[Mendoza et~al.(2018)Mendoza, Allegra, and
  Coleman]{mendoza2018bayesian}
M.~Mendoza, A.~Allegra, and T.~P. Coleman.
\newblock Bayesian {L}asso posterior sampling via parallelized measure
  transport.
\newblock \emph{arXiv:1801.02106}, 2018.

\bibitem[Meng and Schilling(2002)]{meng2002warp}
X.~Meng and S.~Schilling.
\newblock Warp bridge sampling.
\newblock \emph{Journal of Computational and Graphical Statistics}, 11\penalty0
  (3):\penalty0 552--586, 2002.

\bibitem[Mohamed and Lakshminarayanan(2016)]{mohamed2016learning}
S.~Mohamed and B.~Lakshminarayanan.
\newblock Learning in implicit generative models.
\newblock \emph{arXiv:1610.03483}, 2016.

\bibitem[Morrison et~al.(2017)Morrison, Baptista, and
  Marzouk]{morrison2017beyond}
R.~Morrison, R.~Baptista, and Y.~Marzouk.
\newblock Beyond normality: Learning sparse probabilistic graphical models in
  the non-{G}aussian setting.
\newblock \emph{Advances in Neural Information Processing Systems}, 2017.
\newblock arXiv:1711.00950.

\bibitem[Morzfeld et~al.(2015)Morzfeld, Tu, Wilkening, and
  Chorin]{morzfeld2015parameter}
M.~Morzfeld, X.~Tu, J.~Wilkening, and A.~Chorin.
\newblock Parameter estimation by implicit sampling.
\newblock \emph{Communications in Applied Mathematics and Computational
  Science}, 10\penalty0 (2):\penalty0 205--225, 2015.

\bibitem[Moselhy and Marzouk(2012)]{el2012bayesian}
T.~Moselhy and Y.~Marzouk.
\newblock Bayesian inference with optimal maps.
\newblock \emph{Journal of Computational Physics}, 231\penalty0 (23):\penalty0
  7815--7850, 2012.

\bibitem[Oksendal(2013)]{oksendal2013stochastic}
B.~Oksendal.
\newblock \emph{Stochastic differential equations: an introduction with
  applications}.
\newblock Springer Science \& Business Media, 2013.

\bibitem[Oliver(2015)]{oliver2015metropolized}
D.~S. Oliver.
\newblock Metropolized randomized maximum likelihood for sampling from
  multimodal distributions.
\newblock \emph{arXiv:1507.08563}, 2015.

\bibitem[Parno(2015)]{parno2015transport}
M.~Parno.
\newblock \emph{Transport maps for accelerated Bayesian computation}.
\newblock PhD thesis, Massachusetts Institute of Technology, 2015.

\bibitem[Parno and Marzouk(2018)]{parno2014transport}
M.~Parno and Y.~Marzouk.
\newblock Transport map accelerated {M}arkov chain {M}onte {C}arlo.
\newblock \emph{SIAM/ASA Journal on Uncertainty Quantification}, 6\penalty0
  (2):\penalty0 645--682, 2018.

\bibitem[Parno et~al.(2016)Parno, Moselhy, and Marzouk]{parno2015multiscale}
M.~Parno, T.~Moselhy, and Y.~Marzouk.
\newblock A multiscale strategy for {B}ayesian inference using transport maps.
\newblock \emph{SIAM/ASA Journal on Uncertainty Quantification}, 4\penalty0
  (1):\penalty0 1160--1190, 2016.

\bibitem[Polson et~al.(2008)Polson, Stroud, and
  M{\"u}ller]{polson2008practical}
N.~G. Polson, J.~R. Stroud, and P.~M{\"u}ller.
\newblock Practical filtering with sequential parameter learning.
\newblock \emph{Journal of the Royal Statistical Society: Series B},
  70\penalty0 (2):\penalty0 413--428, 2008.

\bibitem[Raanes(2016)]{raanes2016ensemble}
P.~N. Raanes.
\newblock On the ensemble {R}auch-{T}ung-{S}triebel smoother and its
  equivalence to the ensemble {K}alman smoother.
\newblock \emph{Quarterly Journal of the Royal Meteorological Society},
  142\penalty0 (696):\penalty0 1259--1264, 2016.

\bibitem[Ramsay(1998)]{ramsay1998estimating}
J.~O. Ramsay.
\newblock Estimating smooth monotone functions.
\newblock \emph{Journal of the Royal Statistical Society: Series B}, pages
  365--375, 1998.

\bibitem[Ranganath et~al.(2014)Ranganath, Gerrish, and
  Blei]{ranganath2013black}
R.~Ranganath, S.~Gerrish, and D.~M. Blei.
\newblock Black box variational inference.
\newblock In \emph{Artificial Intelligence and Statistics}, pages 814--822,
  2014.

\bibitem[Rauch et~al.(1965)Rauch, Striebel, and Tung]{rauch1965maximum}
H.~E. Rauch, C.~T. Striebel, and F.~Tung.
\newblock Maximum likelihood estimates of linear dynamic systems.
\newblock \emph{AIAA Journal}, 3\penalty0 (8):\penalty0 1445--1450, 1965.

\bibitem[Reich(2011)]{reich2011dynamical}
S.~Reich.
\newblock A dynamical systems framework for intermittent data assimilation.
\newblock \emph{BIT Numerical Mathematics}, 51\penalty0 (1):\penalty0 235--249,
  2011.

\bibitem[Reich(2013)]{reich2013nonparametric}
S.~Reich.
\newblock A nonparametric ensemble transform method for {B}ayesian inference.
\newblock \emph{SIAM Journal on Scientific Computing}, 35\penalty0
  (4):\penalty0 A2013--A2024, 2013.

\bibitem[Reich and Cotter(2013)]{reich2013ensemble}
S.~Reich and C.~Cotter.
\newblock Ensemble filter techniques for intermittent data assimilation.
\newblock \emph{Large Scale Inverse Problems. Computational Methods and
  Applications in the Earth Sciences}, 13:\penalty0 91--134, 2013.

\bibitem[Reich and Cotter(2015)]{reich2015probabilistic}
S.~Reich and C.~Cotter.
\newblock \emph{Probabilistic Forecasting and Bayesian Data Assimilation}.
\newblock Cambridge University Press, 2015.

\bibitem[Rezende and Mohamed(2015)]{rezende2015variational}
D.~J. Rezende and S.~Mohamed.
\newblock Variational inference with normalizing flows.
\newblock \emph{arXiv:1505.05770}, 2015.

\bibitem[Robert and Casella(2013)]{robert2013monte}
C.~Robert and G.~Casella.
\newblock \emph{Monte Carlo statistical methods}.
\newblock Springer Science \& Business Media, 2013.

\bibitem[Rosenblatt(1952)]{rosenblatt1952remarks}
M.~Rosenblatt.
\newblock Remarks on a multivariate transformation.
\newblock \emph{The Annals of Mathematical Statistics}, pages 470--472, 1952.

\bibitem[Rudin(1987)]{rudin1987real}
W.~Rudin.
\newblock \emph{Real and complex analysis}.
\newblock Tata McGraw-Hill Education, 1987.

\bibitem[Rue and Held(2005)]{rue2005gaussian}
H.~Rue and L.~Held.
\newblock \emph{Gaussian Markov random fields: theory and applications}.
\newblock CRC Press, 2005.

\bibitem[Rue et~al.(2009)Rue, Martino, and Chopin]{rue2009approximate}
H.~Rue, S.~Martino, and N.~Chopin.
\newblock Approximate {B}ayesian inference for latent {G}aussian models by
  using integrated nested {L}aplace approximations.
\newblock \emph{Journal of the Royal Statistical Society: Series B},
  71\penalty0 (2):\penalty0 319--392, 2009.

\bibitem[Saad(2003)]{saad2003iterative}
Y.~Saad.
\newblock \emph{Iterative Methods for Sparse Linear Systems}.
\newblock SIAM, 2003.

\bibitem[Samarov(1993)]{samarov1993exploring}
A.~M. Samarov.
\newblock Exploring regression structure using nonparametric functional
  estimation.
\newblock \emph{Journal of the American Statistical Association}, 88\penalty0
  (423):\penalty0 836--847, 1993.

\bibitem[Santambrogio(2015)]{santambrogio2015optimal}
F.~Santambrogio.
\newblock \emph{Optimal Transport for Applied Mathematicians}, volume~87.
\newblock Springer, 2015.

\bibitem[S{\"a}rkk{\"a}(2013)]{sarkka2013bayesian}
S.~S{\"a}rkk{\"a}.
\newblock \emph{Bayesian Filtering and Smoothing}, volume~3.
\newblock Cambridge University Press, 2013.

\bibitem[Schillings and Schwab(2016)]{schillings2014scaling}
C.~Schillings and C.~Schwab.
\newblock Scaling limits in computational {B}ayesian inversion.
\newblock \emph{ESAIM: Mathematical Modelling and Numerical Analysis},
  50\penalty0 (6):\penalty0 1825--1856, 2016.

\bibitem[Shapiro(2013)]{Shapiro2013}
A.~Shapiro.
\newblock \emph{Sample Average Approximation}, pages 1350--1355.
\newblock Springer US, Boston, 2013.

\bibitem[Smith et~al.(2013)Smith, Doucet, de~Freitas, and
  Gordon]{smith2013sequential}
A.~Smith, A.~Doucet, N.~de~Freitas, and N.~Gordon.
\newblock \emph{Sequential Monte Carlo methods in practice}.
\newblock Springer Science \& Business Media, 2013.

\bibitem[Spall(2005)]{spall2005introduction}
J.~C. Spall.
\newblock \emph{Introduction to stochastic search and optimization: estimation,
  simulation, and control}, volume~65.
\newblock John Wiley \& Sons, 2005.

\bibitem[Spantini(2017)]{spantini2017inference}
A.~Spantini.
\newblock \emph{On the low-dimensional structure of Bayesian inference}.
\newblock PhD thesis, Massachusetts Institute of Technology, 2017.

\bibitem[Spantini et~al.(2015)Spantini, Solonen, Cui, Martin, Tenorio, and
  Marzouk]{spantini2014optimal}
A.~Spantini, A.~Solonen, T.~Cui, J.~Martin, L.~Tenorio, and Y.~Marzouk.
\newblock Optimal low-rank approximations of {B}ayesian linear inverse
  problems.
\newblock \emph{SIAM Journal on Scientific Computing}, 37\penalty0
  (6):\penalty0 A2451--A2487, 2015.

\bibitem[Spantini et~al.(2017)Spantini, Cui, Willcox, Tenorio, and
  Marzouk]{spantini2016goal}
A.~Spantini, T.~Cui, K.~Willcox, L.~Tenorio, and Y.~M. Marzouk.
\newblock Goal-oriented optimal approximations of {B}ayesian linear inverse
  problems.
\newblock \emph{SIAM Journal on Scientific Computing}, 39\penalty0
  (5):\penalty0 S167--S196, 2017.

\bibitem[Spivak(1965)]{spivak1965calculus}
M.~Spivak.
\newblock \emph{Calculus on manifolds}, volume~1.
\newblock WA Benjamin New York, 1965.

\bibitem[Stavropoulou and M{\"u}ller(2015)]{stavropoulou2015parametrization}
F.~Stavropoulou and J.~M{\"u}ller.
\newblock Parametrization of random vectors in polynomial chaos expansions via
  optimal transportation.
\newblock \emph{SIAM Journal on Scientific Computing}, 37\penalty0
  (6):\penalty0 A2535--A2557, 2015.

\bibitem[Stuart(2010)]{stuart2010inverse}
A.~M. Stuart.
\newblock Inverse problems: a {B}ayesian perspective.
\newblock \emph{Acta Numerica}, 19:\penalty0 451--559, 2010.

\bibitem[Tabak and Turner(2013)]{tabak2013family}
E.~G. Tabak and C.~V. Turner.
\newblock A family of nonparametric density estimation algorithms.
\newblock \emph{Communications on Pure and Applied Mathematics}, 66\penalty0
  (2):\penalty0 145--164, 2013.

\bibitem[Tao(2011)]{tao2011introduction}
T.~Tao.
\newblock \emph{An introduction to measure theory}, volume 126.
\newblock American Mathematical Soc., 2011.

\bibitem[Villani(2008)]{villani2008optimal}
C.~Villani.
\newblock \emph{Optimal transport: old and new}, volume 338.
\newblock Springer Science \& Business Media, 2008.

\bibitem[Wang and Meng(2016)]{wang2016warp}
L.~Wang and X.~Meng.
\newblock Warp bridge sampling: the next generation.
\newblock \emph{arXiv:1609.07690}, 2016.

\bibitem[Wilkinson(2011)]{wilkinson2011stochastic}
D.~J. Wilkinson.
\newblock \emph{Stochastic modelling for systems biology}.
\newblock CRC press, 2011.

\bibitem[Wright and Nocedal(1999)]{wright1999numerical}
S.~J. Wright and J.~Nocedal.
\newblock \emph{Numerical Optimization}, volume~2.
\newblock Springer New York, 1999.

\bibitem[Xiu(2010)]{xiu2010numerical}
D.~Xiu.
\newblock \emph{Numerical methods for stochastic computations: a spectral
  method approach}.
\newblock Princeton University Press, 2010.

\bibitem[Yang et~al.(2013)Yang, Mehta, and Meyn]{yang2013feedback}
T.~Yang, P.~G. Mehta, and S.~P. Meyn.
\newblock Feedback particle filter.
\newblock \emph{IEEE transactions on Automatic control}, 58\penalty0
  (10):\penalty0 2465--2480, 2013.

\bibitem[Yannakakis(1981)]{yannakakis1981computing}
M.~Yannakakis.
\newblock Computing the minimum fill-in is {NP}-complete.
\newblock \emph{SIAM Journal on Algebraic Discrete Methods}, 2\penalty0
  (1):\penalty0 77--79, 1981.

\bibitem[Yuan and Lin(2007)]{yuan2007model}
M.~Yuan and Y.~Lin.
\newblock Model selection and estimation in the {G}aussian graphical model.
\newblock \emph{Biometrika}, 94\penalty0 (1):\penalty0 19--35, 2007.

\end{thebibliography}
